\documentclass[12pt]{book} 
\usepackage[marginpar=2cm]{geometry}
\usepackage[utf8]{inputenc} 
\usepackage[T1]{fontenc}

\usepackage[english,french]{babel}

\usepackage{titletoc}
\usepackage{titlesec}
\assignpagestyle{\chapter}{empty}

\usepackage{csquotes,xpatch} 
\usepackage[backend=biber,
            bibencoding=utf8,
            texencoding=ascii,
            style=alphabetic,
            backref=true,
            maxbibnames=10]{biblatex}
\DefineBibliographyStrings{english}{
    backrefpage = {page},
    backrefpages = {pages},
}
\usepackage{tikz}
\usepackage{pgfplots}
\pgfplotsset{compat = 1.18}
\usetikzlibrary{external} 
\tikzsetexternalprefix{tikz/}
\tikzexternaldisable

\usepackage{mystyle}
\tcbset{shield externalize} 

\usepackage[pdfusetitle,hyperfootnotes=false]{hyperref} 
\hypersetup{colorlinks=true,
            linkcolor=cyan,
            citecolor=magenta,      
            urlcolor=cyan,
            pdfpagemode=UseOutlines}
\begin{document}

\selectlanguage{english}

\definecolor{TitleColor}{HTML}{B92948}

\newgeometry{vmargin=2cm,hmargin=1cm}

\begin{titlepage}
    \tikzexternaldisable
    \tikz[remember picture, overlay] \node[opacity = 1, inner sep = 0pt] at ([xshift = 5cm, yshift = -5cm] current page.center) {\includegraphics[width = 0.5\paperwidth, height = 0.5\paperheight]{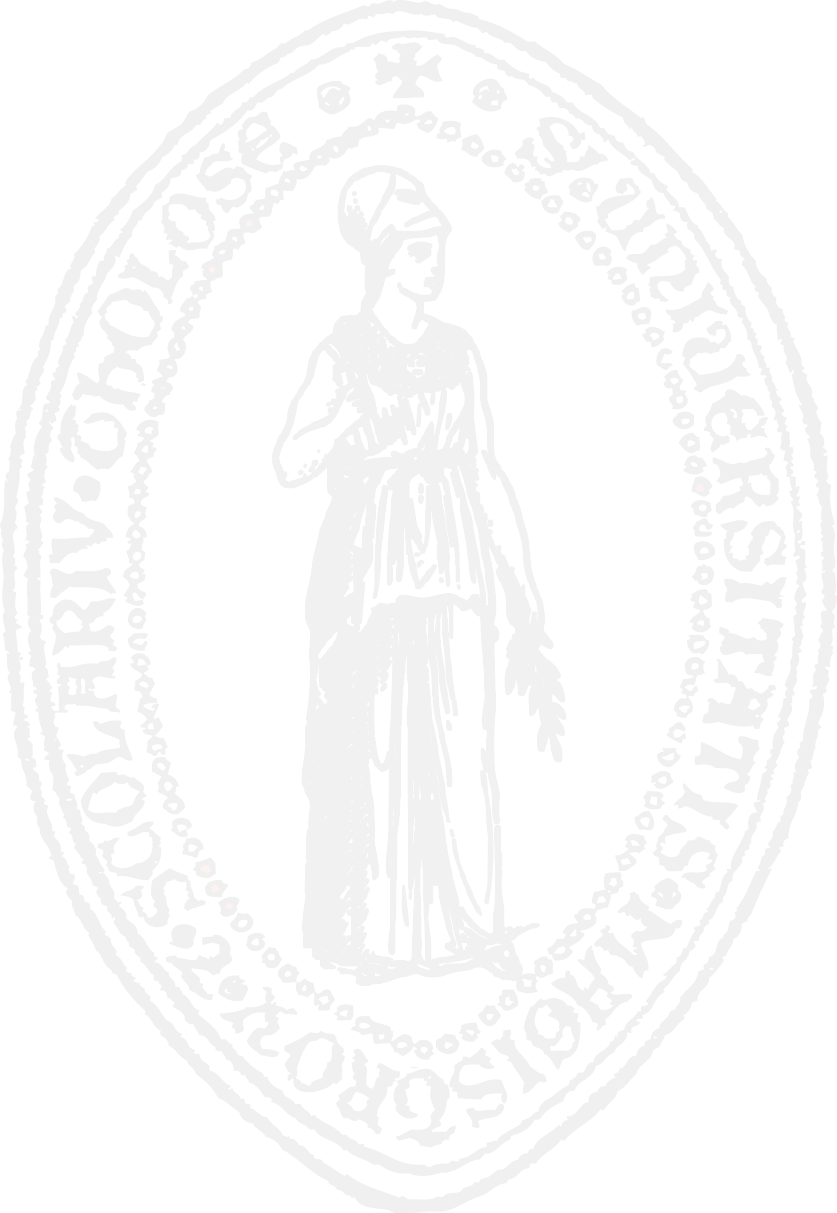}};
    \tikz[remember picture, overlay] \node[opacity = 1, inner sep = 0pt] at (2cm, -1cm) {\includegraphics[width = 4cm]{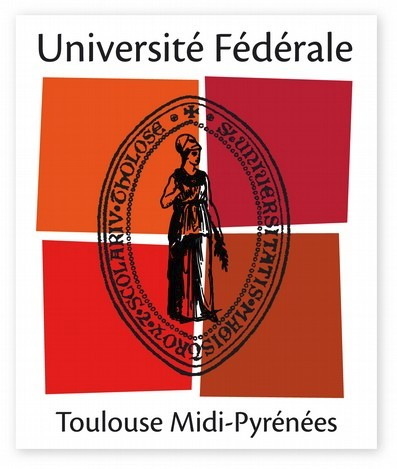}};
    \tikzexternalenable

    \centering
    
    {\fontfamily{phv}\fontseries{b}\fontsize{60}{0}\selectfont\color{TitleColor} THÈSE} \\
    
    \vspace{50pt}
    
    {\fontfamily{phv}\fontseries{b}\fontsize{19}{0}\selectfont\color{black} En vue de l'obtention du} \\
    \vspace{10pt}
    {\fontfamily{phv}\fontseries{b}\fontsize{20}{0}\selectfont\color{black} DOCTORAT DE L'UNIVERSITÉ DE TOULOUSE} \\
    \vspace{15pt}
    {\fontfamily{phv}\fontseries{b}\fontsize{17}{0}\selectfont\color{black} Délivré par l'Université Toulouse 3 -- Paul Sabatier} \\

    \vspace{50pt}

    {\color{TitleColor}\rule{250pt}{2pt}} \\
    \vspace{10pt}
    {\fontfamily{phv}\fontseries{b}\fontsize{17}{0}\selectfont\color{black} Présentée et soutenue par} \\
    \vspace{10pt}
    {\fontfamily{phv}\fontseries{b}\fontsize{20}{0}\selectfont\color{black} DENIS ROCHETTE} \\
    \vspace{20pt}
    {\fontfamily{phv}\fontsize{15}{0}\selectfont\color{black} Le 11 juillet 2023} \\
    \vspace{20pt}
    {\fontfamily{phv}\fontseries{b}\fontsize{18}{0}\selectfont\color{black} {ASYMMETRIC CLONING IN QUANTUM INFORMATION THEORY}} \\
    \vspace{5pt}
    {\color{TitleColor}\rule{250pt}{2pt}} \\
    
    \vspace{10pt}
    
    {\fontfamily{phv}\fontsize{14}{0}\selectfont\color{black} École doctorale~: {\fontseries{b}\selectfont Mathématiques, Informatique, Télécommunications de Toulouse}} \\
    \vspace{10pt}
    {\fontfamily{phv}\fontsize{14}{0}\selectfont\color{black} Spécialité~: {\fontseries{b}\selectfont Mathématiques}} \\
    \vspace{10pt}
    {\fontfamily{phv}\fontsize{14}{0}\selectfont\color{black} Unité de recherche~: {\fontseries{b}\selectfont Institut de Mathématiques de Toulouse} (UMR 5219)} \\
    \vspace{15pt}
    {\fontfamily{phv}\fontsize{14}{0}\selectfont\color{black} Thèse dirigée par} \\[5pt]
    {\fontfamily{phv}\fontseries{b}\fontsize{14}{16}\selectfont\color{black} M. Clément Pellegrini \\ M. Ion Nechita} \\
    \vspace{15pt}
    {\fontfamily{phv}\fontsize{14}{0}\selectfont\color{black} Jury} \\[5pt]
    {\fontfamily{phv}\fontsize{14}{16}\selectfont\color{black}
    \begin{tabular}{rl}
        {\fontseries{b}\selectfont M. Guillaume Aubrun} & Rapporteur \\  
        {\fontseries{b}\selectfont M. Omar Fawzi} & Rapporteur \\ 
        {\fontseries{b}\selectfont M. Michael Wolf} & Rapporteur \\
        {\fontseries{b}\selectfont M. Antonio Ácin} & Examinateur \\
        {\fontseries{b}\selectfont M. Fabrice Gamboa} & Examinateur \\
        {\fontseries{b}\selectfont Mme. Maria Jivulescu} & Examinatrice \\
    \end{tabular}
    }
\end{titlepage}

\restoregeometry

\newpage\null\thispagestyle{empty}\newpage

\titleclass{\chapter}{top}
\titleformat{\chapter}
    [display]
    {\vspace{-100pt}\centering\Huge\bfseries}
    {\appendixname\ \thechapter}
    {0pt}
    {\huge}
    
\chapter*{Acknowledgements}

My first words of thanks go to my supervisors \textbf{Clément Pellegrini} and \textbf{Ion Nechita} for giving me the opportunity to do a PhD under their supervision and with their advice. Throughout these three years, their kindness and support have been indispensable to the success of this thesis. Thank you for giving me the opportunity to attend many conferences, workshops and summer schools. These invaluable experiences have greatly enriched my research and academic network, for which I am deeply grateful.

I would like to thank \textbf{Guillaume Aubrun}, \textbf{Omar Fawzi} and \textbf{Michael Wolf} who accepted to review this manuscript and did so with great care. I would also like to thank \textbf{Antonio Ácin}, \textbf{Fabrice Gamboa} and \textbf{Maria Jivulescu} for accepting to be members of my thesis defence jury.

I would like to express my deep gratitude to the permanent members of the Institut de Mathématiques de Toulouse and the Probability team, with whom I have had the opportunity to discuss mathematics and other subjects. My special thanks go to \textbf{Tristan Benoist} and \textbf{Reda Chhaibi}, whose advice and encouragement have been crucial to my progress.

I would like to thank all the students who are part of our Quantum team in Toulouse or who have joined us in the past for a thesis, an internship or just a visit. Thanks to \textbf{Arnaud}, \textbf{Pierre}, \textbf{Khurshed}, \textbf{Jan-Luka}, \textbf{Linda}, \textbf{Satvik}, \textbf{Qing-Hua}, \textbf{Sang-Jun}, and all the new students who have joined us this year.
Thanks to \textbf{Anna}, our first postdoc!
And of course a special thanks to \textbf{Faedi}, with whom I spent many hours over a blackboard or a drink.

I would like to thank all the PhD students and postdocs at the Institut de Mathématiques de Toulouse whom I have had the pleasure to meet, and with whom I've spent some great times, in the lab, over a picnic in the park, or a drink in the Biergarten. Thank you for organising all the student activities in the lab, especially the student seminar. And of course, thank you for helping me to improve my Tarot skills. There are so many of you, so to mention just a few, thank you to
\textbf{Nicolas},
\textbf{Sophia},
\textbf{Paola},
\textbf{Lucas},
\textbf{Perla},
\textbf{Alberto},
\textbf{Fanny},
\textbf{Michèle},
\textbf{Clément},
\textbf{Étienne},
\textbf{Viviana},
\textbf{Anthony},
{\fontsize{11}{14}\selectfont \textbf{Virgile},}
{\fontsize{10}{14}\selectfont \textbf{Joachim},}
{\fontsize{9}{14}\selectfont \textbf{Corentin},}
{\fontsize{8}{14}\selectfont \textbf{Mahmoud},}
{\fontsize{7}{14}\selectfont \textbf{Alain},}
{\fontsize{6}{14}\selectfont \textbf{Fu-Hsuan},}
{\fontsize{5}{14}\selectfont \textbf{Javier},}
\scalebox{.75}{{\fontsize{5}{14}\selectfont \textbf{Axel},}}
\scalebox{.5}{{\fontsize{5}{14}\selectfont \textbf{Louis}, \rule[1pt]{10pt}{0.2pt}}}

I would also like to thank all the PhD students and postdocs at the Theoretical Physics Laboratory for welcoming me among the physicists over the last three years. Thank you for all the pool, restaurant and cinema evenings\ldots Thanks for all the discussions in the best coffee room and the seminars. Special thanks to \textbf{Bhupen}, \textbf{Sasank}, \textbf{Jeanne}, \textbf{Naïmo}, \textbf{Nathan}, \textbf{Noam}, \textbf{Quentin}, \textbf{Théo}, \textbf{Thomas\footnote[2]{It wasn't a footnote, just the two Thomas!}}.

I would like to express my sincere gratitude to \textbf{Malika} for the invaluable support and dedication provided throughout the three years of my thesis. Your assistance, enthusiasm, and kindness have been truly exceptional and irreplaceable.

I would like to thank the Amsterdam QuSoft quantum team for inviting me for a week. In particular, I would like to thank, \textbf{Dmitry} and \textbf{Māris} for your welcome and for our joint work.

I would like to thank Quantum Munich at TUM for their kind invitation. I would especially like to thank \textbf{Cambyse}.

Finally, I'd like to thank all the people I've met during the three years of my thesis. Thank you to \textbf{Cécilia}, who was part of our Quantum team in Toulouse for a while.
Thanks to \textbf{Andreas}, \textbf{Mehdi}, \textbf{Leevi}, \textbf{Teiko}, \textbf{Stephan}, \textbf{Yvan}, \textbf{Adrian}, \textbf{Matthias} and many others.
I would especially like to thank \textbf{Anne} for giving me the opportunity to continue doing quantum science during a postdoc.

My next thanks go to my friends \textbf{Pierre} and \textbf{Herménégilde}. Congratulations to you \textbf{Julien} and \textbf{Heloïse} on your wedding, congratulations to you \textbf{Joseph} and \textbf{Shiva-chan}, and welcome to you \textbf{Léo}. Thank you \textbf{Julien}, \textbf{Jérémy}, \textbf{Tristan}, \textbf{Rémi}, \textbf{Alexandre}, \textbf{Steven}, \textbf{Johanna}, \textbf{Youri}, \textbf{Viven}, \textbf{Mickaël}, \textbf{Maëlys} and all the others for the time we spent together.

I would like to thank my family for their support and help during the three years of my thesis in Toulouse. I know the long distance has not always been easy.

Finally, I'd like to thank you, \textbf{Dr. Claire Lacouture}. Thank you for your presence, your love, and your support that made the writing of this manuscript so much easier.

\setcounter{tocdepth}{1}
\setcounter{secnumdepth}{2}
\tableofcontents
\setcounter{page}{5} 

\titleclass{\chapter}{page}
\titleformat{\chapter}
    [display]
    {\flushright\Huge\bfseries}
    {
    \vspace{-30pt}
    \begin{tikzpicture}
        \node [text = darkgray, font = \huge, rotate=90] (name) {\MakeUppercase{\bfseries\sffamily\chaptername}};
        \node [scale = 6.25, minimum size = 22pt, inner sep = 0pt, text = white, fill = lightgray, font = \Huge] (number) [xshift = 14pt] at (name) {\thechapter};
    \end{tikzpicture}
    }
    {20pt}
    {}
    [\clearpage]

\chapter{Introduction}

\subsection{Motivation and context}

Quantum information theory, an interdisciplinary field that combines principles of quantum mechanics and information theory, bridges theoretical physics, computer science, and mathematics. The primary goal of the field is to understand quantum properties within physical systems, which can then be used to manipulate and transmit information. As a rapidly evolving field of research, quantum information theory has the potential to catalyze significant advances in cryptography, computing, and communications.

A fundamental concept within quantum information theory is quantum entanglement, which refers to the correlations that exist between two or more quantum systems. These non-classical correlations allow tasks to be performed that are impossible in classical systems. Quantum entanglement has emerged as a key resource for quantum information processing, enabling operations such as teleportation, superdense coding, and quantum error correction.

Quantum channels, another key concept in quantum information theory, describe the transmission of quantum information. The development of efficient and reliable techniques for quantum information transmission is essential for the realization of quantum communication and quantum computing. Quantum cloning channels, a specific category of quantum channels, refer to the notion of quantum cloning, which involves the creation of multiple identical copies of an unknown quantum state. Although perfect quantum cloning is impossible due to the no-cloning theorem, resulting from the linearity of quantum mechanics, the creation of approximate copies remains feasible.

The present work aims to provide an exhaustive investigation of quantum cloning problems, along with related quantum entanglement problems. The analysis of these topics is based on the application of the core concepts of representation theory, in particular those associated with the symmetric group. The use of these concepts allows the unification of different topics and a more extensive comprehension of the matters at hand.

To achieve this goal, the initial exploration involves the fundamental notion of Schur-Weyl duality, which provides a critical link between the symmetric group and the unitary group. This duality allows efficient representation and manipulation of quantum systems, making it a valuable tool for further research in quantum information theory. Additionally, various extensions of the Schur-Weyl duality, involving other groups and algebras, are studied.

A primary application of Schur-Weyl duality that receives special attention is the quantum cloning, which involves the creation of multiple copies of an unknown quantum state. Both the $1 \to 2$ case and the more general $1 \to N$ case, where $N$ copies of an unknown state are created, are studied, providing new insights into the constraints imposed by the no-cloning theorem.

The investigation then extends to a more general quantum entanglement problem, exploring its relation to Schur-Weyl duality and developing novel techniques for analyzing and solving the problem.

\subsection{Preliminaries}

Let $\mathcal{H} \coloneqq \mathbb{C}^d$ be a finite $d$-dimensional complex Hilbert space, and let $\mathcal{M}_d$ denote the space of $d \times d$ complex matrices acting on $\mathcal{H}$. Given a matrix $M \in \mathcal{M}_d$, its conjugate transpose ${\bar{M}}^\T$ is denoted $M^*$. The Frobenius inner product on $\mathcal{M}_d$ is defined by,
\begin{equation*}
    \scalar{A}{B} \coloneqq \Tr \1[ A^* B \1].
\end{equation*}

Using the Dirac notation, vectors are denoted as kets, represented by $\ket{\psi}$, while their duals are called bras, denoted by $\bra{\psi}$. The inner product on $\mathcal{H}$ becomes simply,
\begin{equation*}
    \braket{\psi}{\phi} \in \mathbb{C},
\end{equation*}
and the outer product,
\begin{equation*}
    \ketbra{\psi}{\phi} \in \mathcal{M}_d.
\end{equation*}

On a tensor product space $\mathcal{H}_1 \otimes \cdots \otimes \mathcal{H}_n$,
the notations $M_{(i)}$ and $v_{(i)}$ for a matrix $M \in \mathcal{M}_{d_i}$ and a vector $v \in \mathcal{H}_i$ are used to denote the position of the matrix and the vector on the tensor space $\mathcal{H}_i$. Given a matrix $M \in \mathcal{M}_{d_1} \otimes \cdots \otimes \mathcal{M}_{d_n}$, the partial transpose $M^\Tpartial$ denotes the transpose operation of the first tensor space $\mathcal{H}_1$, and the partial trace $\Tr_i [M]$ denotes the trace operation on the tensor space $\mathcal{H}_i$.

Additionally, the notation $[n]$ is used to denote alternatively the set $\{1, \ldots, n\}$ or the set $\{0, 1, \ldots, n\}$.  Both usages are unambiguous in their context.

\subsection{Summary of results}

In quantum information theory, the process of copying of a state, written
\begin{equation*}
    \psi \mapsto \psi \otimes \psi,
\end{equation*}
and called quantum cloning, can be performed perfectly if and only if $\psi$ is an element of a known orthonormal basis. Otherwise, perfect cloning becomes impossible, and the resulting copies turn out to be imperfect. This phenomenon is known as the no-cloning theorem, and is part of the family of no-go theorems in theoretical physics that describe an intrinsic impossibility of quantum mechanics.

This impossibility gives rise to an optimization problem called the $1 \to N$ quantum cloning problem, which is defined as follows: identify a specific quantum channel, called the quantum cloning channel, denoted by $\Phi: \mathcal{M}_d \to {\1( \mathcal{M}_d \1)}^{\otimes N}$, which maps an input pure quantum state on $\mathcal{H}$ to an output mixed quantum state on $\mathcal{H}^{\otimes N}$, such that the output marginals of $\Phi$ are as close as possible to the input. In the most general case, each marginal of $\Phi$ can be different, resulting in asymmetric copies.

It can be showed, after a suitable symmetrisation procedure, that the marginals $\Phi_i$ of such a quantum cloning channel are of the form:
\begin{align*}
    \Phi_i(\rho) &= p_i \cdot \rho + (1 - p_i) \frac{I_d}{d}, & &\forall \rho~\text{pure},
\end{align*}
for some probabilities $p_i \in [0, 1]$. The quantum cloning problem can be reformulated as the problem of identifying the set of achievable probabilities $p_i$, and their associated quantum cloning channel $\Phi$.

The first part of this thesis focuses on the $1 \to 2$ quantum cloning problem. In this special case, a simple description of both the quantum cloning channels $\Phi$ and the achievable probabilities $p_i$ can be given. The quantum cloning channels are parameterised, independently of the dimension of $\mathcal{H}$, by only $6$ coefficients, and this number of parameters decreases to only $4$ in the optimal quantum cloning channel with respect to the optimisation problem. The achievable probabilities $(p_1, p_2)$ are described by the union of a family of ellipses  indexed by $\lambda \in [0,d]$, given by
\begin{equation*}
    \qquad \frac{x^2}{a^2_\lambda} + \frac{(y - c_\lambda)^2}{b^2_\lambda} \leq 1,
\end{equation*}
with $a_\lambda \coloneqq \frac{\lambda}{\sqrt{d^2 - 1}}, b_\lambda \coloneqq \frac{\lambda}{d^2 - 1}$, $c_\lambda \coloneqq \frac{\lambda d - 2}{d^2 - 1}$, and $x = p_1 - p_2$, $y = p_1 + p_2$.

The second part of this thesis focuses on the general $1 \to N$ quantum cloning problem. The quantum cloning channels are this time parameterised, independently of the dimension of $\mathcal{H}$, by the $N$ coefficients of a largest eigenvector of a matrix $S_x$, defined for all $x \in \mathbb{R}^N$ by,
\begin{equation*}
    S_x \coloneqq \sum^N_{k=1} \sum^{d-1}_{i,j=0} |x_k| \cdot \ketbra{ii}{jj}_{(0,k)} \otimes I^{\otimes (N-1)}.
\end{equation*}
From this matrix is derived the $\mathcal{Q}$-norm, a norm on $\mathbb{R}^N$ defined on $x \in \mathbb{R}^N$ by,
\begin{equation*}
    \norm{x}_{\mathcal{Q}} \coloneqq \frac{d \lambda_{\text{max}}(S_x) - \norm{x}_1}{d^2 - 1},
\end{equation*}
The achievable probabilities $p_i$ are exactly the non-negative part of the unit ball of the dual $\mathcal{Q}$-norm.

Exploiting the close relationship between quantum channels and quantum states, the quantum cloning problem can be seen as the quantum entanglement problem on a graph: identifying a quantum state $\rho_G$ on a star graph $G$, such that the reduced quantum state $\rho_e$ on each edge $e \in G$ of the graph is as close as possible to the maximally entangled state $\omega$, i.e.
\begin{align*}
    \rho_e &= p_e \cdot \omega + (1 - p_e) \frac{I^{\otimes 2}_d}{d^2}, & \forall e~\text{edge}.
\end{align*}
This leads to the more general optimization problem of finding such a quantum state given any graph, formally defined by the semi-definite programming,
\begin{align*}
    \max_{\rho,p_e} \quad & \sum_{e~\text{edge}} p_e \\
    \text{s.t.} \quad & \rho_e = p_e \cdot \omega + (1 - p_e) \frac{I^{\otimes 2}_d}{d^2}, & \forall e~\text{edge} \\
    & \Tr \1[ \rho \1] = 1, \quad \rho \geq 0.
\end{align*}

The third part of the thesis focuses on the exact solution for this problem in the case of the complete graph on $N$ vertices, with equal reduced quantum state, that is,
\begin{align*}
    p(N,d) \coloneqq \max_{\rho,p} \quad & p \\
    \text{s.t.} \quad & \rho_e = p \cdot \omega + (1 - p) \frac{I_d}{d}, & \forall e~\text{edge} \\
    & \Tr \1[ \rho \1] = 1, \quad \rho \geq 0.
\end{align*}
The close formula for $p(N,d)$ depend on both the values of $d$ and $N$, and their parity,
\begin{equation*}
    p(N,d) =
    \begin{cases}
        \frac{1}{N + N \bmod 2 - 1} &\text{ if $d > N$ or either $d$ or $N$ is even} \\
        \min \1\{ \frac{2 d + 1}{2 d N + 1}, \frac{1}{N - 1} \1\} &\text{ if $N \geq d$ and both $d$ and $N$ are odd}.
    \end{cases}
\end{equation*}

\subsection{Outline of the thesis}

The thesis is structured as follows. \hyperref[appA]{Appendix~\ref*{appA}} provides an overview of the fundamental principles of representation theory that are used throughout the thesis, it can be read first, or skipped. \hyperref[chap2]{Chapter~\ref*{chap2}} is devoted to the Schur-Weyl duality between the symmetric group and the unitary group, but also various extensions of the Schur-Weyl duality involving other groups and algebras. \hyperref[chap3]{Chapter~\ref*{chap3}} provides a comprehensive presentation of the mathematical foundations of quantum mechanics in the context of quantum information theory. \hyperref[chap4]{Chapter~\ref*{chap4}} studies the above quantum cloning problems in both the $1 \to 2$ and $1 \to N$ cases. \hyperref[chap4]{Chapter~\ref*{chap4}} looks at the above quantum entanglement problem on the complete graph.

\titleformat{\chapter}
    [display]
    {\flushright\Huge\bfseries}
    {
    \vspace{-30pt}
    \begin{tikzpicture}
        \node [text = darkgray, font = \huge, rotate=90] (name) {\MakeUppercase{\bfseries\sffamily\chaptername}};
        \node [scale = 6.25, minimum size = 22pt, inner sep = 0pt, text = white, fill = lightgray, font = \Huge] (number) [xshift = 14pt] at (name) {\thechapter};
    \end{tikzpicture}
    }
    {20pt}
    {}
    [
    \vfill
    \flushleft\normalsize\bfseries
    {
        \startcontents[chapters]
        \printcontents[chapters]{}{1}[2]{\rule[0.3em]{50pt}{2pt} \: {\Large Chapter contents} \: \leaders\hbox{\rule[0.3em]{2pt}{2pt}\kern-1pt}\hfill\null}
    }
    \clearpage
    ]

\chapter{Schur-Weyl duality} \label{chap2}

The primary goal of this chapter consists in establishing the foundational result of Issai Schur and Hermann Weyl, known as the Schur-Weyl duality, which relates the symmetric group $\mathfrak{S}_n$ and the complex general linear group $\mathrm{GL}_d$. Moreover, this chapter explores other adaptations of this theorem for distinct groups and algebras.
\begin{theorem*}[Schur-Weyl duality \cite{schur1927rationalen,weyl1946classical}]
    The space of $n$-fold tensors over $\mathbb{C}^d$ decomposes under the action of the direct product group $\mathrm{GL}_d \times \mathfrak{S}_n$ as follows:
    \begin{equation*}
        {\1( \mathbb{C}^d \1)}^{\otimes n} \simeq \bigoplus_{\substack{\lambda \vdash n \\ \lambda^\prime_1 \leq d}} V^d_\lambda \otimes V_\lambda.
    \end{equation*}
    Where $V^d_\lambda$ is an irreducible representation space of $\mathrm{GL}_d$ and $V_\lambda$ is an irreducible representation space of $\mathfrak{S}_n$.
\end{theorem*}

For a comprehensive exploration of the representation theory concerning the symmetric group $\mathfrak{S}_n$ and the complex general linear group $\mathrm{GL}_d$, refer to \hyperref[appA]{Appendix~\ref*{appA}}. However, the current section aims to provide a self-contained exposition.

Subsequently, two notations for permutations and partitions of the set $\{1, \ldots, n\}$ are employed. The conventional cyclic notation $\emph{(1\:2\:3)(4\:5)}$ denotes the permutation:
\begin{equation*}
    \begin{pmatrix}
    1 & 2 & 3 & 4 & 5 \\
    2 & 3 & 1 & 5 & 4
    \end{pmatrix},
\end{equation*}
and the notation $\emph{1\,2\,3\:|\:4\,5}$ represents the partition:
\begin{equation*}
    \1\{ \{ 1, 2, 3 \}, \{ 4, 5 \} \1\}.
\end{equation*}

\section{Diagrammatic algebras} \label{chap2:sec:diagrammaticAlgebra}

The term diagram algebras has no specific definition by axiomatic properties or other rigorous means. In the present thesis, a diagram algebra refers to a finite unital associative algebra over the complex field, where the basis consists of homotopy classes of diagrams. The multiplication operation within this algebraic structure finds its definition through the process of concatenation. For a survey on diagram algebras, see \cite{koenig2008panorama,halverson2020set}.

In the context of finite dimensional algebras over the algebraically closed field $\mathbb{C}$, the notion of \emph{semisimplicity} is employed interchangeably with that of the direct sum of full matrix algebras, closed under matrix multiplication.

\subsection{Symmetric group algebra \texorpdfstring{$\mathbb{S}_k$}{}} \label{chap2:sec:symmetricGroupAlgebra}

Define $\emph{\mathfrak{S}_k}$ as the \emph{symmetric group}, which is the group of order $k!$ containing all the \emph{permutations} of the set $\{ 1, \ldots, k \}$. Given a permutation $\sigma$ belonging to the symmetric group $\mathfrak{S}_k$, it is possible to represent this permutation as a \emph{diagram} via a graph consisting of $2k$ vertices. These vertices are divided equally between two columns.

Interpretation of the diagram proceeds from right to left. A connection between the $i$-th vertex in the right column and the $j$-th vertex in the left column is established if and only if the relation $\sigma(i) = j$ holds. For example,
\begin{equation*}
    \begin{tikzpicture}[site/.style = {circle,
                                       draw = white,
                                       outer sep = 2pt,
                                       fill = black!80!white,
                                       inner sep = 2.5pt}]
        \node (m1) {};

        \node[site, label={[font=\small,text=darkgray]left:{$\sigma(3)$}}] (l1) [xshift = -2em] at (m1.center) {};
        \node[site, label={[font=\small,text=darkgray]right:{$1$}}] (r1) [xshift = 2em] at (m1.center) {};

        \node[site, label={[font=\small,text=darkgray]left:{$\sigma(1)$}}] (l2) [yshift = -2em] at (l1.center) {};
        \node[site, label={[font=\small,text=darkgray]right:{$2$}}] (r2) [yshift = -2em] at (r1.center) {};
        
        \node[site, label={[font=\small,text=darkgray]left:{$\sigma(2)$}}] (l3) [yshift = -2em] at (l2.center) {};
        \node[site, label={[font=\small,text=darkgray]right:{$3$}}] (r3) [yshift = -2em] at (r2.center) {};

        \node[inner sep = 6pt]  (l4) [xshift = -3.25em, yshift = -2em] at (l3.center) {};
        \node[inner sep = 6pt]  (r4) [xshift = 2em, yshift = -2em] at (r3.center) {};
        
        \draw[-, line width = 4.5pt, draw = white] (l2) to (r1);
        \draw[-, line width = 2.5pt, draw = black] (l2) to (r1);
         
        \draw[-, line width = 4.5pt, draw = white] (l3) to (r2);
        \draw[-, line width = 2.5pt, draw = black] (l3) to (r2);
         
        \draw[-, line width = 4.5pt, draw = white] (l1) to (r3);
        \draw[-, line width = 2.5pt, draw = black] (l1) to (r3);

        \draw[<-, line width = 3pt, draw = gray] (l4) to (r4);
    \end{tikzpicture}
\end{equation*}
represents the permutation $(1\:2\:3)$ of the symmetric group $\mathfrak{S}_3$. 

The composition, denoted by $\emph{\sigma \circ \tau}$, of two permutations $\sigma$ and $\tau$ of the symmetric group $\mathfrak{S}_k$, is obtained by positioning the diagram of $\tau$ immediately to the right of the diagram of $\sigma$, and subsequently associating the leftmost column of $\tau$'s diagram with the rightmost column of $\sigma$'s diagram. For example, consider the following two permutations $\sigma \coloneqq (1 \: 2 \: 4)$ and $\tau \coloneqq (1 \: 2)(3)$, which are elements of the symmetric group $\mathfrak{S}_3$,
\begin{equation*}

\end{equation*}

Numerous generating sets for the symmetric group $\mathfrak{S}_k$ exist, with varying cardinalities. A particularly notable generating set is the collection of adjacent \emph{transpositions}. This set is characterized by containing the $k - 1$ permutations of the form $(i \:\: i\!+\!1)$, for all $1 \leq i < k$:
\begin{equation*}
    %
\end{equation*}

Let $\sigma$ be an element of the symmetric group $\mathfrak{S}k$. The \emph{cycle type} associated with $\sigma$, denoted by $\lambda$, is defined as the $l$-tuple containing the lengths of the $l$ disjoint \emph{cycles} composing $\sigma$, arranged in non-increasing order. As a consequence, the cycle type $\lambda$ corresponds to a \emph{partition} of the integer $k$ into $l$ parts, denoted by $\emph{\lambda \vdash k}$. This partition $\lambda$ obeys the following conditions:
\begin{equation*}
    \lambda_1 \geq \cdots \geq \lambda_l \quad \text{ and } \quad \sum^l_{i = 1} \lambda_i = k.
\end{equation*}
Given a partition $\lambda \vdash k$, let $\emph{\lambda^\prime}$ denote the \emph{conjugate} partition associated with $\lambda$, defined by: $\lambda^\prime_i$ is the number of parts in $\lambda$ that are greater than or equal to $i$. A partition $\lambda \vdash k$ may be represented as a \emph{Young diagram}, which is a collection of $k$ empty boxes arranged in left-justified rows such that the $i$-th row contains $\lambda_i$ boxes. The conjugate partition $\lambda^\prime$, is the partition corresponding to transposing the Young diagram representing $\lambda$. For example, consider the permutation $\sigma \in \mathfrak{S}_7$, expressed as a product of disjoint cycles, arranged in non-increasing order, by
\begin{equation*}
    \sigma \coloneqq (1 \: 2 \: 3 \: 4) (5 \: 6 \: 7).
\end{equation*}
The cycle type of $\sigma$ is the partition $\lambda \vdash 7$ given by $\lambda \coloneqq (4,3)$, where the $i$-th entry of $\lambda$ denotes the length of the $i$-th cycle in the disjoint cycle decomposition of $\sigma$. Moreover, the partition $\lambda$ can be represented using a Young diagram, which consists of $2$ rows with $4$ and $3$ boxes, respectively
\begin{equation*}
    \ydiagram{4, 3}
\end{equation*}
The conjugate partition $\lambda^\prime \coloneqq (2,2,2,1)$ is represented using a Young diagram, which consists of $4$ rows with $2$, $2$, $2$, and $1$ boxes, respectively:
\begin{equation*}
    \ydiagram{2, 2, 2, 1}
\end{equation*}
Note that $\lambda^\prime_1$ is the length of the first column of $\lambda$. Within the context of the symmetric group $\mathfrak{S}_k$, the concept of cycle type plays a crucial role in characterizing \emph{conjugacy classes}. Specifically, two permutations in $\mathfrak{S}_k$ are said to be \emph{conjugate} if and only if their respective cycle types are identical.

\begin{remark*}
    It is essential to note that the symmetric group $\mathfrak{S}_k$, defined as permutations of the set $\{1, \ldots, k\}$ or as the above diagrams conveys the same underlying mathematical structure. These two forms are equivalent and can be employed interchangeably.
\end{remark*}

The \emph{group algebra} of the symmetric group $\mathfrak{S}_k$, denoted by $\emph{\mathbb{S}_k}$, is the complex vector space spanned by the permutations of $\mathfrak{S}_k$, i.e.
\begin{equation*}
    \mathbb{S}_k \coloneqq \Span_{\mathbb{C}} \1\{ \sigma \in \mathfrak{S}_k \1\}.
\end{equation*}
The multiplication in the group algebra $\mathbb{S}_k$, is defined on the elements of the symmetric group $\mathfrak{S}_k$ by its group law, and is denoted $\emph{\sigma \cdot \tau}$, for some $\sigma$ and $\tau$ in $\mathfrak{S}_k$.

\begin{remark*}
    The symmetric group algebra $\mathbb{S}_k$ is a finite group algebra. As a consequence, the symmetric group algebra $\mathbb{S}_k$ is always semisimple \cite{serre1977linear,fulton2013representation}.
\end{remark*}

\subsection{Partition algebra \texorpdfstring{$\mathbb{P}_k(\delta)$}{}} \label{chap2:sec:partitionAlgebra}

The \emph{partition monoid}, denoted $\mathbb{P}_k$, is a diagrammatic monoid generated by the $3 \times (k-1)$ diagrams
\begin{equation*}
    \hskip \textwidth minus \textwidth

    \hskip \textwidth minus \textwidth
\end{equation*}
for all $1 \leq i < k$, as well as the $k$ disconnected diagrams
\begin{equation*}
    %
\end{equation*}
for all $1 \leq i \leq k$. These $4k - 3$ diagrams distributed in $4$ distinct collections do not constitute a minimal generating set, as it is possible to choose a single representative diagram from each collection and subsequently use the transpositions to generate the remaining diagrams in the collection.

An element of $\mathbb{P}_k$ is a \emph{partition} of the set $\{ 1, \ldots, 2k \}$, corresponding to the connected components of the associated diagram, where then enumeration of vertices located in the right column ranges from $1$ to $k$, while the enumeration of vertices situated in the left column ranges from $k+1$ to $2k$. As a monoid, $\mathbb{P}_k$ has an identity $1_{\mathbb{P}_k}$ given by the partition $1\,(k+1)\:|\:\cdots\:|\:k\,(2k)$:
\begin{equation*}
    \begin{tikzpicture}[site/.style = {circle,
                                       draw = white,
                                       outer sep = 2pt,
                                       fill = black!80!white,
                                       inner sep = 2.5pt}]
        \node[inner sep = 5pt] (m1) {};
        \node[inner sep = 5pt] (m2) [yshift = -1.75em] at (m1.center) {};
        \node[inner sep = 5pt] (m3) [yshift = -1.75em] at (m2.center) {};

        \node[site, label={[font=\small,text=darkgray]left:{$k+1$}}] (l1) [xshift = -2em] at (m1.center) {};
        \node[site, label={[font=\small,text=darkgray]right:{$1$}}] (r1) [xshift = 2em] at (m1.center) {};

        \node[label={left:{$1_{\mathbb{P}_k} \coloneqq \phantom{k+1}$}}] (l2) [yshift = -1.75em] at (l1.center) {};
        \node (r2) [yshift = -1.75em] at (r1.center) {};
        
        \node[site, label={[font=\small,text=darkgray]left:{$2k$}}] (l3) [yshift = -1.75em] at (l2.center) {};
        \node[site, label={[font=\small,text=darkgray]right:{$k$}}] (r3) [yshift = -1.75em] at (r2.center) {};

        \draw[-, line width = 4.5pt, draw = white] (l1) to (r1);
        \draw[-, line width = 2.5pt, draw = black] (l1) to (r1);

        \draw[-, line width = 4.5pt, draw = white] (l3) to (r3);
        \draw[-, line width = 2.5pt, draw = black] (l3) to (r3);

        \draw[-, dotted, line width = 2pt, draw = gray] (m3) to (m1);
    \end{tikzpicture}
\end{equation*}

Consider two partitions $p$ and $q$ of $\mathbb{P}_k$, the composition $\emph{p \circ q}$ is obtained by positioning the diagram of $q$ immediately to the right of the diagram of $p$, associating the leftmost column of $q$'s diagram with the rightmost column of $p$'s diagram, and finally removing any \emph{loops}, which are the components of the resulting diagram not connected to either the left or the right column. For example, given the two partitions $p \coloneqq 1\,3\:|\:2\,6\:|\:4\,5$ and $q \coloneqq 1\,2\:|\:3\,5\:|\:4\,6$ of $\mathbb{P}_3$,
\begin{equation*}

\end{equation*}
where the gray loop is removed.

\begin{remark*}
    Considering the generators of the partition monoid, the inclusion of diagrams $\mathfrak{S}_k \subseteq \mathbb{P}_k$ holds for every $k \in \mathbb{N}$. However, it is important to note that the partition monoid $\mathbb{P}_k$ does not constitute a group, e.g. the partition $1\,2\:|\:3\,6\:|\:4\,5$ of $\mathbb{P}_3$ has no inverse with respect to the composition in $\mathbb{P}_3$.
\end{remark*}

The order of the partition monoid $\mathbb{P}_k$ is the number of partition of the set $\{ 1, \ldots, 2k \}$, denoted as the even \emph{Bell number} $\mathrm{B}_{2k}$. In general, the Bell number $\mathrm{B}_k$ is given by a recursive formula, with initial condition $\mathrm{B}_0 \coloneqq 1$, and 
\begin{equation*}
    \mathrm{B}_{k + 1} = \sum^k_{i=0} \binom{k}{i} \mathrm{B}_i.
\end{equation*}
Starting at $k = 0$, the first values of the Bell numbers are \cite[sequence A000110]{oeis}:
\begin{equation*}
    1, 1, 2, 5, 15, 52, 203, 877, 4140, \ldots
\end{equation*}

The \emph{partition algebra}, denoted by $\emph{\mathbb{P}_k(\delta)}$ is defined for some $\delta \in \mathbb{C}$, as the complex vector space spanned by the diagrams of $\mathbb{P}_k$, i.e.
\begin{equation*}
    \mathbb{P}_k(\delta) \coloneqq \Span_{\mathbb{C}} \1\{ p \in \mathbb{P}_k \1\}.
\end{equation*}
The multiplication in $\mathbb{P}_k(\delta)$, given two elements $p$ and $q$ of $\mathbb{P}_k$, is denoted by $\emph{p\cdot q}$ and is defined by $p\cdot q \coloneqq \delta^l (p \circ q)$, where $l$ is the number of loops removed during the composition in $\mathbb{P}_k$. For example, let $p \coloneqq 1\,4\:|\:2\,8\:|\:3\,5\:|\:6\,7\:|\:9\,10$ and $q \coloneqq 1\,2\:|\:3\,5\:|\:4\,7\:|\:6\,9\:|\:8\,10$ be two partitions of $\mathbb{P}_5$,
\begin{equation*}

\end{equation*}
the composition $p \circ q = 1\,2\:|\:3\,5\:|\:4\,8\:|\:6\,7\:|\:9\,10$ becomes, in the partition algebra $\mathbb{P}_5(\delta)$,
\begin{equation*}
    %
\end{equation*}
where the number of gray loops removed is $2$.

\begin{remark*}
    The partition algebra $\mathbb{P}_k(\delta)$ is not semisimple for all $\delta \in \mathbb{C}$. Specifically, it is semisimple if and only if $\delta$ belongs to the set $\mathbb{C} \setminus \{ 0, \ldots, 2k - 2 \}$ \cite{martin1994algebras,halverson2005partition}.
\end{remark*}

\subsection{Others diagrammatic algebras}

In \hyperref[chap2:sec:symmetricGroupAlgebra]{Section~\ref*{chap2:sec:symmetricGroupAlgebra}}, the focus was on the symmetric group $\mathfrak{S}_k$, generated by by $k-1$ permutation diagrams. Then, in \hyperref[chap2:sec:partitionAlgebra]{Section~\ref*{chap2:sec:partitionAlgebra}}, attention was turned to the partition monoid $\mathbb{P}_k$, generated by by $4k-3$ diagrams distributed in $4$ distinct collections.

The purpose of the present Section is to highlight the relationships between the $4$ collections of diagrams generating the partition monoid $\mathbb{P}_k$, and the algebraic structures that emerge. Specifically, the choice of specific collections may yield distinct monoids.

\begin{remark*}
    A monoid is a unitary semigroup. To obtain the identity diagram $1_{\mathbb{P}_k}$ of the partition monoid $\mathbb{P}_k$, corresponding to the partition $1\,(k+1)\:|\:\cdots\:|\:k\,(2k)$, the collection of transpositions is required,  as none of the $3$ other collections is composed of invertible elements.
\end{remark*}

\subsubsection{Uniform block permutation algebra \texorpdfstring{$\mathbb{U}_k(\delta)$}{}}

The \emph{uniform block permutation monoid}, denoted by $\emph{\mathbb{U}_k}$, is the submonoid of $\mathbb{P}_k$ generated by the $2$ collections of diagrams
\begin{equation*}
    \hskip \textwidth minus \textwidth

    \hskip \textwidth minus \textwidth
\end{equation*}
for all $1 \leq i < k$.

The elements of the uniform block permutation monoid $\mathbb{U}_k$ are precisely those elements from $\mathbb{P}_k$ that satisfy the following condition: the number of vertices located on the left column equals the number of vertices located on the right column for each connected component of the diagram. This condition is a consequence of the fact that the $2$ collections generating $\mathbb{U}_k$ satisfy it and that it is preserved under multiplication. As a consequence, each connected component contains vertices situated in both the left and right columns. For example, the partition $p \coloneqq 1\,2\,4\,5\:|\:3\,6$ is in $\mathbb{U}_3$,
\begin{equation*}

\end{equation*}
while the partition $q \coloneqq 1\,2\:|\:3\,6\:|\:4\,5$ is not in $\mathbb{U}_3$,
\begin{equation*}
    %
\end{equation*}

The order of the uniform block permutation monoid $\mathbb{U}_k$ is given by a recursive formula \cite{sixdeniers2001extended}, with initial condition $\1| \mathbb{U}_0 \1| \coloneqq 1$, and 
\begin{equation*}
    \1| \mathbb{U}_{k+1} \1| = \sum^k_{i=1} \binom{k}{i} \binom{k + 1}{i} \1| \mathbb{U}_i \1|.
\end{equation*}
Starting at $k = 0$, the first values of $\1| \mathbb{U}_k \1|$ are \cite[sequence A023998]{oeis}:
\begin{equation*}
    1, 1, 3, 16, 131, 1496, \ldots
\end{equation*}

The \emph{uniform block permutation algebra}, denoted by $\emph{\mathbb{U}_k(\delta)}$ is defined for some $\delta \in \mathbb{C}$, as the subalgebra of the partition algebra $\mathbb{P}_k(\delta)$, spanned by the element of the uniform block permutation monoid $\mathbb{U}_k$, i.e
\begin{equation*}
    \mathbb{U}_k(\delta) \coloneqq \Span_{\mathbb{C}} \1\{ p \in \mathbb{U}_k \1\}.
\end{equation*}
The multiplication in $\mathbb{U}_k(\delta)$ does not yield any loops. Consequently, all the uniform block permutation algebras are isomorphic, for all $\delta \in \mathbb{C}$.

\begin{remark*}
    The uniform block permutation monoid $\mathbb{U}_k$ is a finite \emph{inverse monoid}, i.e. forall $x \in \mathbb{U}_k$ there exists a unique $x^* \in \mathbb{U}_k$ satifying $x \circ x^* \circ x = x$ and $x^* \circ x \circ x^* = x^*$. As a consequence, all the uniform block permutations algebras $\mathbb{U}_k(\delta)$ are semisimple for all $\delta \in \mathbb{C}$ \cite{orellana2021plethysm,steinberg2016representation}.
\end{remark*}

\subsubsection{Brauer algebra \texorpdfstring{$\mathbb{B}_k(\delta)$}{}}

The \emph{Brauer monoid}, denoted by $\emph{\mathbb{B}_k}$, is the submonoid of $\mathbb{P}_k$ generated by the $2$ collections of diagrams
\begin{equation*}
    \hskip \textwidth minus \textwidth

    \hskip \textwidth minus \textwidth
\end{equation*}
for all $1 \leq i < k$.

The elements of the Brauer monoid $\mathbb{B}_k$ are precisely all the pairings on the set $\{1, \ldots, 2k\}$. Specifically, each vertex of a diagram in $\mathbb{B}_k$ has precisely a degree of $1$. Given a diagram of $\mathbb{B}_k$, the \emph{vertical edges} are the edges that connect vertices within the same column, whereas the \emph{horizontal edges} are the edges that connect vertices of both the left and right columns. For example, the partition $p \coloneqq 1\,3\:|\:2\,6\:|\:4\,5$ is in $\mathbb{B}_3$,
\begin{equation*}
    \begin{tikzpicture}[site/.style = {circle,
                                       draw = white,
                                       outer sep = 2pt,
                                       fill = black!80!white,
                                       inner sep = 2.5pt}]
        \node (m1) {};
        \node (m2) [yshift = -1.5em] at (m1.center) {};
        \node (m3) [yshift = -1.5em] at (m2.center) {};

        \node[site] (l1) [xshift = -2em] at (m1.center) {};
        \node[site] (r1) [xshift = 2em] at (m1.center) {};

        \node[site, label={left:{$p \coloneqq$}}] (l2) [yshift = -1.5em] at (l1.center) {};
        \node[site] (r2) [yshift = -1.5em] at (r1.center) {};
        
        \node[site] (l3) [yshift = -1.5em] at (l2.center) {};
        \node[site] (r3) [yshift = -1.5em] at (r2.center) {};

        \draw[-, line width = 4.5pt, draw = white] (r1) .. controls (m1) and (m3) .. (r3);
        \draw[-, line width = 2.5pt, draw = black] (r1) .. controls (m1) and (m3) .. (r3);

        \draw[-, line width = 4.5pt, draw = white] (l1) .. controls (m1) and (m2) .. (l2);
        \draw[-, line width = 2.5pt, draw = black] (l1) .. controls (m1) and (m2) .. (l2);

        \draw[-, line width = 4.5pt, draw = white] (l3) to (r2);
        \draw[-, line width = 2.5pt, draw = black] (l3) to (r2);
    \end{tikzpicture}
\end{equation*}

The order of the Brauer monoid $\mathbb{B}_k$ is given by the \emph{odd factorial},
\begin{equation*}
    \1| \mathbb{B}_k \1| = (2k - 1)!!.
\end{equation*}
Starting at $k = 0$, the first values of $\1| \mathbb{B}_k \1|$ are \cite[sequence A001147]{oeis}:
\begin{equation*}
    1, 1, 3, 15, 105, 945, 10395, \ldots
\end{equation*}

The \emph{Brauer algebra}, denoted by $\emph{\mathbb{B}_k(\delta)}$ is defined for some $\delta \in \mathbb{C}$, as the subalgebra of the partition algebra $\mathbb{P}_k(\delta)$, spanned by the element of the Brauer monoid $\mathbb{B}_k$, i.e
\begin{equation*}
    \mathbb{B}_k(\delta) \coloneqq \Span_{\mathbb{C}} \1\{ p \in \mathbb{B}_k \1\}.
\end{equation*}

\begin{remark*}
    The Brauer algebra $\mathbb{B}_k(\delta)$ is not semisimple for all $\delta \in \mathbb{C}$. Specifically it is semisimple if and only if one of the following conditions holds \cite{wenzl1988structure,doran1999semisimplicity,rui2005criterion,rui2006criterion,andersen2017semisimplicity}:
    \begin{itemize}
        \item $\delta = 0$ and $k \in \{1, 3, 5\}$;
        \item $\delta \in \mathbb{Z} \setminus \{ 0 \}$ and $k \leq |\delta| + 1$;
        \item $\delta \not\in \mathbb{Z}$.
    \end{itemize}
\end{remark*}

\subsubsection{Walled Brauer algebra \texorpdfstring{$\mathbb{B}_ {k,l}(\delta)$}{}}

The \emph{walled Brauer monoid}, denoted by $\emph{\mathbb{B}_{k,l}}$, is the submonoid of $\mathbb{B}_{k+l}$ generated by the $2$ collections of diagrams
\begin{equation*}
    \hskip \textwidth minus \textwidth

    \hskip \textwidth minus \textwidth
\end{equation*}
for all $1 \leq i < k < j < k + l$, as well as the diagram
\begin{equation*}
        %
\end{equation*}

The \emph{wall} of the walled Brauer monoid $\mathbb{B}_{k,l}$ denotes the vertical separation between the uppermost $2k$ vertices and the lowermost $2l$ vertices. The diagram elements of $\mathbb{B}_{k,l}$ are precisely those elements from $\mathbb{B}_{k+l}$ that satisfy the following condition: every vertical edge must cross the wall, while no horizontal edge shall cross the wall. The condition arises from the fact that the $k + l - 1$ diagrams generating $\mathbb{B}_{k,l}$ satisfy the condition, in addition to preserving this property under multiplication. For example, the partition $p \coloneqq 1\,4\:|\:2\,7\:|\:3\,9\:|\:5\,10\:|\:6\,12\:|\:8\,11$ is in $\mathbb{B}_{3, 3}$,
\begin{equation*}

\end{equation*}
while the partition $q \coloneqq 1\,3\:|\:2\,7\:|\:4\,9\:|\:5\,10\:|\:6\,12\:|\:8\,11$ is in $\mathbb{B}_6$ but not in $\mathbb{B}_{3, 3}$,
\begin{equation*}
    %
\end{equation*}

\begin{remark*}
    Given the generators of the walled Brauer monoid $\mathbb{B}_{k,l}$, the inclusion of diagrams $\mathfrak{S}_k \times \mathfrak{S}_l \subseteq \mathbb{B}_{k,l}$ holds for every $k,l \in \mathbb{N}$. The diagrams of $\mathfrak{S}_k \times \mathfrak{S}_l$ consists of those from $\mathbb{B}_{k,l}$ with no edges crossing the wall.
\end{remark*}

The operation denoted as \emph{partial transposition} corresponds to the process that exchanges vertex $i$ and vertex $k+l+i$, both situated on the same row. However, it should be noted that the diagram obtained after performing such an operation may not necessarily belong to $\mathbb{B}_{k,l}$ anymore. For example, the transposition of the $1$-st row of the partition $1\,4\:|\:2\,7\:|\:3\,9\:|\:5\,10\:|\:6\,12\:|\:8\,11$ in $\mathbb{B}_{3, 3}$ becomes,
\begin{equation*}

\end{equation*}
which correspond to the partition $1\,2\:|\:3\,9\:|\:4\,7\:|\:5\,10\:|\:6\,12\:|\:8\,11$ in $\mathbb{B}_6$ but not in $\mathbb{B}_{3, 3}$.

The partial transposition involving the $l$ lowest vertices constitutes a one-to-one mapping from the walled Brauer monoid $\mathbb{B}_{k,l}$ to the symmetric group $\mathfrak{S}_{k+l}$. As a consequence,
\begin{equation*}
    \1| \mathbb{B}_{k,l} \1| = (k+l)!.
\end{equation*}

The \emph{walled Brauer algebra}, denoted by $\emph{\mathbb{B}_{k,l}(\delta)}$ is defined for some $\delta \in \mathbb{C}$, as the subalgebra of the Brauer algebra $\mathbb{B}_{k+l}(\delta)$, spanned by the element of the walled Brauer monoid $\mathbb{B}_{k,l}$, i.e
\begin{equation*}
    \mathbb{B}_{k,l}(\delta) \coloneqq \Span_{\mathbb{C}} \1\{ p \in \mathbb{B}_{k,l} \1\}.
\end{equation*}

\begin{remark*}
    The walled Brauer algebra $\mathbb{B}_{k,l}(\delta)$ is not semisimple for all $\delta \in \mathbb{C}$. Specifically it is semisimple if and only if one of the following conditions hold \cite{cox2008blocks,bulgakova2020some}:
    \begin{itemize}
        \item $k = 0$ or $l = 0$;
        \item $\delta = 0$ and $(k, l) \in \1\{ (1, 2), (2, 1), (1, 3), (3, 1) \1\}$;
        \item $\delta \in \mathbb{Z} \setminus \{ 0 \}$ and $k + l \leq |\delta| + 1$;
        \item $\delta \not\in \mathbb{Z}$.
    \end{itemize}
\end{remark*}

\section{Tensor representation} \label{chap2:sec:tensorRepresentation}

The diagrammatic algebras described in \hyperref[chap2:sec:diagrammaticAlgebra]{Section~\ref*{chap2:sec:diagrammaticAlgebra}} act on the $d^n$-dimensional tensor product complex vector space ${\1( \mathbb{C}^d \1)}^{\otimes n}$ by considering the mapping $\emph{\psi}$, called \emph{tensor representation}, from the diagrammatic monoids on $2n$ vertices to $\mathcal{M}_{d^n}$, and defined for each diagram $p$ by,
\begin{equation*}
    {\1( \psi(p) \1)}^{i_1, \ldots, i_n}_{i_{n+1}, \ldots, i_{2n}} \coloneqq
    \begin{cases}
    1 &\text{ if $i_k = i_l$, for all vertices $k$ and $l$ connected in $p$} \\
    0 &\text{ otherwise},
    \end{cases}
\end{equation*}
and linearly extended to the entire diagrammatic algebra. For example, let the partition $p \coloneqq 1\,3\:|\:2\,6\:|\:4\,5$ in the Brauer algebra $\mathbb{B}_3(d)$,
\begin{equation*}

\end{equation*}
with some vector $x_1, x_2, x_3 \in \mathbb{C}^d$ and $y_1, y_2, y_3 \in \mathbb{C}^d$, then
\begin{align*}
    \bra{y_1 \otimes y_2 \otimes y_3} \psi(p) \ket{x_1 \otimes x_2 \otimes x_3} &=
        %
 \\[1em]
    &= \scalar{x_1}{x_3} \cdot \scalar{x_2}{y_3} \cdot \scalar{y_1}{y_2}. 
\end{align*}

\begin{remark*}
    In the case where the diagrammatic algebras depends on a complex parameter $\delta \in \mathbb{C}$, the tensor representation acting on the tensor product complex vector space ${\1( \mathbb{C}^d \1)}^{\otimes n}$, requires that $\delta = d$.
\end{remark*}

The tensor representations on $\mathbb{C}^2 \otimes \mathbb{C}^2$, of all the $3$ diagrams spanning the Brauer algebra $\mathbb{B}_2(2)$ are
\begin{equation*}
    \psi \left(
    %
\end{equation*}
and the identity tensor
\begin{equation*}
    \psi \left(
    %
.
\end{equation*}

For the symmetric group $\mathfrak{S}_n$, spanning the symmetric group algebra $\mathbb{S}_n$, this action corresponds to the permutation of the tensor positions, i.e. for all $\sigma \in \mathfrak{S}_n$ and all $v_1, \ldots, v_n$ in $\mathbb{C}^d$, the tensor representation of $\sigma$ on ${\1( \mathbb{C}^d \1)}^{\otimes n}$ gives the action
\begin{equation*}
    \psi(\sigma) \cdot (v_1 \otimes \cdots \otimes v_n) = v_{\sigma^{\shortminus 1}(1)} \otimes \cdots \otimes v_{\sigma^{\shortminus 1}(n)}.
\end{equation*}

The partial transposition operation on a row, for a given diagram $p$, corresponds to the partial transposition of a tensor, for the matrix $\psi(p)$. For example, let the partition $1\,2\:|\:3\,4$ in the Brauer algebra $\mathbb{B}_2(2)$, then taking the partial transposition on the $1$-st row gives:
\begin{equation*}
    \psi \4(
^\Tpartial.
\end{equation*}

The operation referred to as \emph{closing} a diagram consists in connecting each vertex of the diagram, to the vertex located in the same row, yielding a collection of loops. The trace of the tensor representation on ${\1( \mathbb{C}^d \1)}^{\otimes n}$, for a diagram $p$, can be obtained by $\Tr \1[ \psi(p) \1] = d^l$, where $l$ is the number of loops after closing the diagram $p$. For example, let the partition $p \coloneqq 1\,2\,4\,5\:|\:3\,6$ in the uniform block permutation algebra $\mathbb{U}_3(d)$,
\begin{equation*}
    %
\end{equation*}
the closing of $p$ is
\begin{equation*}
    %
\end{equation*}
with $2$ loops, then the trace of $\psi(p)$ becomes $\Tr \1[ \psi(p) \1] = d^2$.

The tensor representation of a diagrammatic algebra is in general non-faithfull. For example, let $\mathbb{S}_3$ be the symmetric group algebra with the its tensor representation on $\mathbb{C}^2 \otimes \mathbb{C}^2 \otimes \mathbb{C}^2$ given by the map $\psi$, and define $\emph{\mathrm{sign}(\sigma)}$ to be a \emph{signature} of the permutation $\sigma \in \mathfrak{S}_3$, then
\begin{equation*}
    \sum_{\sigma \in \mathfrak{S}_3} \mathrm{sign}(\sigma) \cdot \psi(\sigma) = 0.
\end{equation*}

\begin{remark*}
    In certain cases, diagrammatic algebras may not exhibit semisimplicity. However, the algebra $\mathcal{A}$ defined by the complex span of $\psi(p)$, for all $p$ in some diagrammatic monoid, always exhibits semisimplicity, as a matrix algebra:
    \begin{equation*}
        \mathcal{A} \simeq \mathcal{M}^{\oplus n_1}_{d_1} \oplus \cdots \oplus \mathcal{M}^{\oplus n_k}_{d_k},
    \end{equation*}
    with some multiplicities $n_i$ and dimensions $d_i$. In this basis, an element $A \in \mathcal{A}$ is written,
    \begin{equation*}
        A \simeq I_{n_1} \otimes A_1 \oplus \cdots \oplus I_{n_k} \otimes A_k.
    \end{equation*}
\end{remark*}

\section{Schur-Weyl dualities} \label{chap2:sec:SchurWeylDualities}

\subsection{Commutant}

Given a matrix algebra $\mathcal{A} \subseteq \mathcal{M}_d$, the \emph{commutant} of $\mathcal{A}$, denoted $\emph{\mathcal{A}^\prime}$, is the set of matrices that commute with all elements of $\mathcal{A}$:
\begin{equation*}
    \mathcal{A}^\prime \coloneqq \set{M \in \mathcal{M}_d}{MA = AM, \text{ for all } A \in \mathcal{A}}.
\end{equation*}

\begin{theorem}[\cite{serre1977linear,fulton2013representation}] \label{chap1:thm:doubleCommutant}
    Let $\mathcal{A}$ be a matrix algebra, and $\mathcal{B}$ the commutant of $\mathcal{A}$. Suppose $\mathcal{A}$ decomposes as $\mathcal{A} \simeq \mathcal{M}^{\oplus n_1}_{d_1} \oplus \cdots \oplus \mathcal{M}^{\oplus n_k}_{d_k}$. Then for all $A \in \mathcal{A}$ and $B \in \mathcal{B}$, 
    \begin{align*}
        A &\simeq \bigoplus^k_{i=1} I_{n_i} \otimes A_i,\\
        B &\simeq \bigoplus^k_{i=1} B_i \otimes I_{d_i}.
    \end{align*}
    Furthermore both $\mathcal{A}$ and $\mathcal{B}$ are commutants of each other, i.e. $\mathcal{B} = \mathcal{A}^\prime$ and $\mathcal{A} = \mathcal{B}^\prime$.
\end{theorem}

\subsection{Schur-Weyl duality for \texorpdfstring{$\mathfrak{S}_n$}{the symmetric group}}

Let $\emph{\mathrm{GL}_d}$ be the \emph{complex general linear group} of degree $d$, which consists of the $d \times d$ invertible complex matrices acting on $\mathbb{C}^d$. This action extends diagonally to an action on the $d^n$-dimensional tensor product complex vector space ${\1( \mathbb{C}^d \1)}^{\otimes n}$, defined for $M \in \mathrm{GL}_d$ on tensor $v_1 \otimes \cdots \otimes v_n \in {\1( \mathbb{C}^d \1)}^{\otimes n}$ by,
\begin{equation*}
    M^{\otimes n} \cdot (v_1 \otimes \cdots \otimes v_n) = M \cdot v_1 \otimes \cdots \otimes M \cdot v_n,
\end{equation*}
and extended linearly.

Let $\mathcal{A}$ and $\mathcal{B}$ be the matrix algebras generated, respectively, by the actions of the symmetric group $\mathfrak{S}_n$ and the complex general linear group $\mathrm{GL}_d$, on the $d^n$-dimensional tensor product complex vector space ${\1( \mathbb{C}^d \1)}^{\otimes n}$, i.e.
\begin{align*}
    \mathcal{A} &\coloneqq \Span_{\mathbb{C}} \set{\psi(\sigma)}{\sigma \in \mathfrak{S}_n} \\
    \mathcal{B} &\coloneqq \Span_{\mathbb{C}} \set{M^{\otimes n}}{M \in \mathrm{GL}_d}.
\end{align*}

\begin{theorem*}[\cite{goodman2009symmetry}] \label{chap2:thm:commutantSchurWeylDuality}
    Both $\mathcal{A}$ and $\mathcal{B}$ are commutants of each other.
\end{theorem*}

The matrix algebra $\mathcal{A}$, generated by the tensor representation of the symmetric group $\mathfrak{S}_n$, can be decomposed as the direct sum,
\begin{equation*}
    \mathcal{A} \simeq \bigoplus_{\substack{\lambda \vdash n \\ \lambda^\prime_1 \leq d}} \mathcal{M}^{\oplus n_\lambda}_{d_\lambda},
\end{equation*}
indexed by the Young diagrams $\lambda$ with $n$ boxes and at most $d$ rows.\footnote{see \hyperref[appA]{Appendix~\ref*{appA}}} Then acording to \hyperref[chap1:thm:doubleCommutant]{Theorem~\ref*{chap1:thm:doubleCommutant}}, for all $A \in \mathcal{A}$ and $B \in \mathcal{B}$,
\begin{equation*}
    A \simeq \bigoplus_{\substack{\lambda \vdash n \\ \lambda^\prime_1 \leq d}} I_{n_\lambda} \otimes A_\lambda \qquad \text{ and } \qquad B \simeq \bigoplus_{\substack{\lambda \vdash n \\ \lambda^\prime_1 \leq d}} B_\lambda \otimes I_{d_\lambda},
\end{equation*}
where the $A_\lambda$ act on a space denoted $\emph{V_\lambda}$, and the $B_\lambda$ act on a space denoted $\emph{V^d_\lambda}$.

\begin{theorem}[Schur-Weyl duality \cite{schur1927rationalen,weyl1946classical}] \label{chap2:thm:SchurWeylDuality}
    The space of $n$-fold tensors over $\mathbb{C}^d$ decomposes under the action of the direct product group $\mathrm{GL}_d \times \mathfrak{S}_n$ as follows:
    \begin{equation*}
        {\1( \mathbb{C}^d \1)}^{\otimes n} \simeq \bigoplus_{\substack{\lambda \vdash n \\ \lambda^\prime_1 \leq d}} V^d_\lambda \otimes V_\lambda.
    \end{equation*}
\end{theorem}

The diagonal action of $\mathrm{GL}_d$, on the $d^n$-dimensional tensor product complex vector space ${\1( \mathbb{C}^d \1)}^{\otimes n}$, can be restricted to the subgroup of the \emph{unitary group} of degree $d$, denoted $\emph{\mathrm{U}_d}$, which consists of the $d \times d$ unitary matrices acting on $\mathbb{C}^d$.

\begin{theorem} \label{chap2:thm:generalAUnitaryMatrixAlgebra}
    Let $\mathcal{C}$ be the matrix algebra generated by the action of the unitary group $\mathrm{U}_d$, on the $d^n$-dimensional tensor product complex vector space ${\1( \mathbb{C}^d \1)}^{\otimes n}$, i.e. $\mathcal{C} \coloneqq \Span_{\mathbb{C}} \set{U^{\otimes n}}{U \in \mathrm{U}_d}$. Then
    \begin{equation*}
        \mathcal{B} \simeq \mathcal{C}.
    \end{equation*}
\end{theorem}
\begin{proof}
    The present proof is from an unpublished note by Guillaume Aubrun \cite{aubrun2018naive}. Due to the inclusion $\mathrm{U}_d \subseteq \mathrm{GL}_d$, it follows that $\mathcal{C} \subseteq \mathcal{B}$. To establish the result, it suffices to prove that for $M \in \mathrm{GL}_d$, the $n$-fold $M^{\otimes n}$ can be expressed as a limit of linear combinations of $n$-fold $U^{\otimes n}$, for some $U \in \mathrm{U}_d$.

    Without loss of generality, assume that $M$ can be multiplied by a real scalar to obtain a singular value decomposition given by,
    \begin{equation*}
        M = \sum^d_{i=1} s_i \ketbra{e_i}{f_i},
    \end{equation*}
    where $e_1, \ldots, e_d$ and $f_1, \ldots, f_d$ are two orthonormal bases of $\mathbb{C}^d$, and $s_i$ are non-negative real numbers satisfing $-1 < s_i < 1$. Notice that the matrix obtained by replacing each $s_i$ with a complex number $z_i$ satisfying $|z_i| = 1$ is unitary. Then
    \begin{align*}
        M^{\otimes n} &= {\3( \sum^d_{i=1} s_i \ketbra{e_i}{f_i} \3)}^{\otimes n} \\
        &= \sum^d_{i_1, \ldots, i_n = 1} \prod^n_{j=1} s_{i_j} \ketbra{e_{i_1} \otimes \cdots \otimes e_{i_n}}{f_{i_1} \otimes \cdots \otimes f_{i_n}}.
    \end{align*}

    Let $\gamma$ be a counterclockwise unit circle in the complex plane. It is well-known from Cauchy's formula, that for all $s \in \mathbb{R}$ such that $-1 < s < 1$, and for all $k \in \mathbb{N}$,
    \begin{equation*}
        s^k = \frac{1}{2 \pi i} \int_\gamma z^k \frac{\mathrm{d}z}{z-s}.
    \end{equation*}
    Then, using Fubini’s theorem, for all $s \in \mathbb{R}^d$ such that $-1 < s_i < 1$, and for all $i_1, \ldots, i_n \in \{1, \ldots d\}$,
    \begin{equation*}
        \prod^n_{j=1} s_{i_j} = \frac{1}{(2 \pi i)^d} \int_{\gamma^{\times d}} \prod^n_{j=1} z_{i_j} \frac{\mathrm{d}z_1}{z_1 - s_1} \cdots \frac{\mathrm{d}z_d}{z_d - s_d}.
    \end{equation*}
    Finally,
    \begin{align*}
        &\phantom{{}={}} M^{\otimes n} \\
        &= \sum^d_{i_1, \ldots, i_n = 1} \prod^n_{j=1} s_{i_j} \ketbra{e_{i_1} \otimes \cdots \otimes e_{i_n}}{f_{i_1} \otimes \cdots \otimes f_{i_n}} \\
        &= \sum^d_{i_1, \ldots, i_n = 1} \frac{1}{(2 \pi i)^d} \int_{\gamma^{\times d}} \prod^n_{j=1} z_{i_j} \frac{\mathrm{d}z_1}{z_1 - s_1} \cdots \frac{\mathrm{d}z_d}{z_d - s_d} \ketbra{e_{i_1} \otimes \cdots \otimes e_{i_n}}{f_{i_1} \otimes \cdots \otimes f_{i_n}} \\
        &= \frac{1}{(2 \pi i)^d} \int_{\gamma^{\times d}} U^{\otimes n}_{z_1, \ldots, z_d} \frac{\mathrm{d}z_1}{z_1 - s_1} \cdots \frac{\mathrm{d}z_d}{z_d - s_d},
    \end{align*}
    where $U^{\otimes n}_{z_1, \ldots, z_d}$ is the $n$-fold unitary matrix defined by
    \begin{equation*}
        U_{z_1, \ldots, z_d} \coloneqq \sum^d_{i=1} z_i \ketbra{e_i}{f_i}.
    \end{equation*}
\end{proof}
\begin{corollary*}
    The space of $n$-fold tensors over $\mathbb{C}^d$ decomposes under the action of the direct product group $\mathrm{U}_d \times \mathfrak{S}_n$ as follows:
    \begin{equation*}
        {\1( \mathbb{C}^d \1)}^{\otimes n} \simeq \bigoplus_{\substack{\lambda \vdash n \\ \lambda^\prime_1 \leq d}} V^d_\lambda \otimes V_\lambda.
    \end{equation*}
\end{corollary*}

\subsection{Other Schur-Weyl dualities}

\subsubsection{Schur-Weyl duality for \texorpdfstring{$\mathbb{P}_n$}{the partition monoid}}

The symmetric group $\mathfrak{S}_d$ acts on the $d$-dimensional complex vector space $\mathbb{C}^d$ by considering the mapping $\emph{\phi}$, called \emph{permutation matrix}, from $\mathfrak{S}_d$ to $\mathcal{M}_{d}$, and defined for each permutation $\sigma$ by,
\begin{equation*}
    {\1( \phi(p) \1)}_{i,j} \coloneqq
    \begin{cases}
    1 &\text{ if $i = \sigma(j)$} \\
    0 &\text{ otherwise},
    \end{cases}
\end{equation*}
To be explicit, given a basis $e_1, \ldots, e_d$ of $\mathbb{C}^d$ and vector $v = \sum^d_{i=1} v_i \ket{e_i}$ in $\mathbb{C}^d$, a permutation $\sigma \in \mathfrak{S}_d$ acts on $v$ by
\begin{equation*}
    \phi(\sigma) \cdot v = \sum^d_{i=1} v_{\sigma^{\shortminus1}(i)} \ket{e_i}.
\end{equation*}

This action extends diagonally to an action on the $d^n$-dimensional tensor product complex vector space ${\1( \mathbb{C}^d \1)}^{\otimes n}$, defined for $\sigma \in \mathfrak{S}_d$ on tensor $v_1 \otimes \cdots \otimes v_n \in {\1( \mathbb{C}^d \1)}^{\otimes n}$ by,
\begin{equation*}
    \phi(\sigma)^{\otimes n} \cdot (v_1 \otimes \cdots \otimes v_n) = \phi(\sigma) \cdot v_1 \otimes \cdots \otimes \phi(\sigma) \cdot v_n,
\end{equation*}
and extended linearly.

Let $\mathcal{A}$ and $\mathcal{B}$ be the matrix algebras generated, respectively, by the actions of the partition monoid $\mathbb{P}_n$ and the symmetric group $\mathfrak{S}_d$, on the $d^n$-dimensional tensor product complex vector space ${\1( \mathbb{C}^d \1)}^{\otimes n}$, i.e.
\begin{align*}
    \mathcal{A} &\coloneqq \Span_{\mathbb{C}} \set{\psi(p)}{p \in \mathbb{P}_n} \\
    \mathcal{B} &\coloneqq \Span_{\mathbb{C}} \set{\phi(\sigma)^{\otimes n}}{\sigma \in \mathfrak{S}_d}.
\end{align*}

\begin{theorem*}[\cite{martin1998potts,martin2000partition,halverson2005partition}]
    Both $\mathcal{A}$ and $\mathcal{B}$ are commutants of each other.
\end{theorem*}

\subsubsection{Schur-Weyl duality for \texorpdfstring{$\mathbb{U}_n$}{the uniform block permutation monoid}}

Let $\emph{\text{diag.}~\mathrm{U}_d}$ be the subgroup of $\mathrm{U}_d$ consisting of $d \times d$ \emph{diagonal unitary matrices} acting on $\mathbb{C}^d$. That is for all $U$ in $\text{diag.}~\mathrm{U}_d$, there is $\theta \in [0, 2\pi)^{d}$ such that,
\begin{equation*}
    U = \mathrm{diag} \1( e^{i\theta_1}, \ldots, e^{i\theta_d} \1).
\end{equation*}
As subgroups of $\mathrm{U}_d$, the action of $\text{diag.}~\mathrm{U}_d$ extends diagonally on the $d^n$-dimensional tensor product complex vector space ${\1( \mathbb{C}^d \1)}^{\otimes n}$.

The product of a diagonal unitary matrix $U \coloneqq \mathrm{diag} \1( e^{i\theta_1}, \ldots, e^{i\theta_d} \1)$ in $\text{diag.}~\mathrm{U}_d$, and a permutation matrix $\phi(\sigma)$ for some $\sigma \in \mathfrak{S}_d$, is a \emph{monomial matrix} in $[0, 2\pi)$, i.e. a permutation matrix whose nonzero components are $e^{i\theta}$ with $\theta \in [0, 2\pi)$:
\begin{equation*}
    {\1( U \cdot \phi(\sigma) \1)}_{i,j} =
    \begin{cases}
    \theta_i &\text{ if $i = \sigma(j)$} \\
    0 &\text{ otherwise},
    \end{cases}
\end{equation*}

Let $\mathcal{A}$ and $\mathcal{B}$ be the matrix algebras generated, respectively, by the actions of the uniform block permutation monoid $\mathbb{U}_n$ and the monomial matrices in $[0, 2\pi)$, on the $d^n$-dimensional tensor product complex vector space ${\1( \mathbb{C}^d \1)}^{\otimes n}$, i.e.
\begin{align*}
    \mathcal{A} &\coloneqq \Span_{\mathbb{C}} \set{\psi(p)}{p \in \mathbb{U}_n} \\
    \mathcal{B} &\coloneqq \Span_{\mathbb{C}} \set{U^{\otimes n} \cdot \phi(\sigma)^{\otimes n}}{U \in \text{diag.}~\mathrm{U}_d \text{ and } \sigma \in \mathfrak{S}_d}.
\end{align*}

\begin{theorem}[\cite{tanabe1997centralizer}] \label{chap2:thm:commutantSchurWeylDualityUn}
    Both $\mathcal{A}$ and $\mathcal{B}$ are commutants of each other.
\end{theorem}

\subsubsection{Schur-Weyl duality for \texorpdfstring{$\mathbb{B}_n$}{the Brauer monoid}}

Let $\emph{\mathrm{O}_d}$ denote the \emph{ortogonal group} of degree $d$, which consists of the $d \times d$ orthogonal matrices acting on $\mathbb{C}^d$. As subgroups of $\mathrm{U}_d$, the action of $\mathrm{O}_d$ extends diagonally on the $d^n$-dimensional tensor product complex vector space ${\1( \mathbb{C}^d \1)}^{\otimes n}$.

Let $\mathcal{A}$ and $\mathcal{B}$ be the matrix algebras generated, respectively, by the actions of the Brauer monoid $\mathbb{B}_n$ and the orthogonal group $\mathrm{O}_d$, on the $d^n$-dimensional tensor product complex vector space ${\1( \mathbb{C}^d \1)}^{\otimes n}$, i.e.
\begin{align*}
    \mathcal{A} &\coloneqq \Span_{\mathbb{C}} \set{\psi(p)}{p \in \mathbb{B}_n} \\
    \mathcal{B} &\coloneqq \Span_{\mathbb{C}} \set{O^{\otimes n}}{O \in \mathrm{O}_d}.
\end{align*}

\begin{theorem*}[\cite{brauer1937algebras}] \label{chap2:thm:commutantSchurWeylDualityBn}
    Both $\mathcal{A}$ and $\mathcal{B}$ are commutants of each other.
\end{theorem*}

\subsubsection{Schur-Weyl duality for \texorpdfstring{$\mathbb{B}_{m,n}$}{the walled Brauer monoid}}

Let the action of the complex general linear group $\mathrm{GL}_d$ on the $d^n$-dimensional mixed tensor product complex vector space ${\1( \mathbb{C}^d \1)}^{\otimes n} \otimes {\1( \mathbb{C}^d \1)}^{\otimes m}$, defined for $M \in \mathrm{GL}_d$ by,
\begin{equation*}
    M^{\otimes m} \otimes {\2( {\1( M^{\shortminus 1} \1)}^\T \2)}^{\otimes n}.
\end{equation*}
As a subgroup of $\mathrm{GL}_d$, this action is defined for $U \in \mathrm{U}_d$ by,
\begin{equation*}
    U^{\otimes m} \otimes {\bar{U}}^{\otimes n}.
\end{equation*}

Let $\mathcal{A}, \mathcal{B}$ and $\mathcal{C}$ be the matrix algebras generated, respectively, by the actions of the walled Brauer monoid $\mathbb{B}_{m,n}$, the complex general linear group $\mathrm{GL}_d$ and the unitary group $\mathrm{U}_d$, on the $d^n$-dimensional mixed tensor product complex vector space ${\1( \mathbb{C}^d \1)}^{\otimes n}$, i.e.
\begin{align*}
    \mathcal{A} &\coloneqq \Span_{\mathbb{C}} \set{\psi(p)}{p \in \mathbb{B}_n} \\
    \mathcal{B} &\coloneqq \Span_{\mathbb{C}} \set{M^{\otimes m} \otimes {\2( {\1( M^{\shortminus 1} \1)}^\T \2)}^{\otimes n}}{M \in \mathrm{GL}_d} \\
    \mathcal{C} &\coloneqq \Span_{\mathbb{C}} \set{U^{\otimes m} \otimes {\bar{U}}^{\otimes n}}{U \in \mathrm{U}_d}.
\end{align*}

\begin{theorem*}[\cite{benkart1994tensor}]
    Both $\mathcal{A}$ and $\mathcal{B}$ (or $\mathcal{C}$) are commutants of each other.
\end{theorem*}

\begin{remark*}
    It is important to note that the various Schur-Weyl dualities presented in \hyperref[chap2:sec:SchurWeylDualities]{Section~\ref*{chap2:sec:SchurWeylDualities}} are given only in terms of the matrix algebras generated by the map $\psi$, rather than the diagrammatic algebras. The map $\psi$ may not always exhibit faithfulness.
\end{remark*}

\chapter{Quantum Information Theory} \label{chap3}

This chapter provides a comprehensive overview of the mathematical foundations of quantum mechanics in the context of quantum information theory, focusing on the postulates of quantum mechanics and their inherent probabilistic nature.

In quantum information theory, the attention is mainly directed to quantum systems with a finite number of degrees of freedom.

First, the formalism of pure quantum states is introduced, which is particularly suitable for closed quantum systems. Then then formalism of mixed quantum states is then introduced, in particular to describe open quantum systems that interact with an environment that is not intended to be described.

References for the different postulates and the mathematical foundations of quantum mechanics can be found in the textbooks \cite{watrous2018theory,nielsen2002quantum,aubrun2017alice}.

\section{Postulates of quantum mechanics}

Let $\emph{\mathcal{H}} \coloneqq \mathbb{C}^d$ be a finite $d$-dimensional Hilbert complex vector space. The convex set of unit trace, positive semi-definite $d \times d$ matrices, acting on $\mathcal{H}$ is denoted,
\begin{equation*}
    \emph{\mathcal{D}_d} \coloneqq \set{\rho \in \mathcal{M}_d}{\Tr \rho = 1 \text{ and } \rho \geq 0}.
\end{equation*}
An element of $\mathcal{D}_d$ is called a \emph{density matrix}, to highlight the fact that its eigenvalues represent a probability distribution. The extremal points of $\mathcal{D}_d$ are the unit rank projections $\ketbra{\psi}{\psi}$, for some $\psi \in \mathcal{H}$ with $\norm{\psi} = 1$.

\begin{definition*}
    A \emph{quantum system} is represented by a finite $d$-dimensional Hilbert complex vector space $\mathcal{H}$.
\end{definition*}

Let $\mathcal{H} \coloneqq \mathbb{C}^d$ be a finite $d$-dimensional Hilbert complex vector space, the computational basis of the quantum system $\mathcal{H}$ is denoted:
\begin{equation*}
    \ket{0}, \ldots, \ket{d-1}.
\end{equation*}

\begin{definition*}
    A \emph{composite quantum system} is represented by a tensor product of Hilbert complex vector spaces $\mathcal{H}_1 \otimes \cdots \otimes \mathcal{H}_n$.
\end{definition*}

The computational basis of the $n$-fold composite quantum system $\mathcal{H}^{\otimes n}$, with the finite $d$-dimensional Hilbert complex vector space $\mathcal{H} \coloneqq \mathbb{C}^d$, is the set:
\begin{equation*}
    \set{\ket{i_1} \otimes \cdots \otimes \ket{i_n}}{i_1, \ldots, i_n \in \{0, \ldots, d-1\}}.
\end{equation*}

\subsection{Pure states}

\begin{definition*}
    A \emph{pure quantum state} on $\mathcal{H}$ is an extremal point of $\mathcal{D}_d$.
\end{definition*}

A pure quantum state on $\mathcal{H} \coloneqq \mathbb{C}^2$ is a unit rank projection $\ketbra{\psi}{\psi}$, for some vector $\ket{\psi} \coloneqq \alpha \cdot \ket{0} + \beta \cdot \ket{1}$ with $\alpha, \beta$ satisfying,
\begin{equation*}
    |\alpha|^2 + |\beta|^2 = 1.
\end{equation*}

\begin{remark*}
    In the following the description of a pure quantum state as a unit rank projection $\ketbra{\psi}{\psi}$ or as a unit norm vector $\ket{\psi}$ is used interchangeably. Note that $\ketbra{\psi}{\psi}$ is simply the orthogonal projection onto the complex line spanned by $\ket{\psi}$.
\end{remark*}

\begin{definition*}
    The \emph{evolution} of a pure quantum state $\rho$ on $\mathcal{H}$ is governed by a unitary matrix $U$ on $\mathcal{H}$, through the conjugation mapping
    \begin{equation*}
        \rho \longmapsto U \rho U^*.
    \end{equation*}
\end{definition*}

The evolution of pure quantum states is a transitive action: for all pure quantum states $\rho$ and $\sigma$ there exists a unitary matrix $U$ such that $\rho = U \sigma U^*$.

As a unit norm vector $\ket{\psi}$, the evolution a pure quantum state, through a unitary matrix $U$, is given by $\ket{\psi} \mapsto U \ket{\psi}$. 

\begin{definition*}
    The \emph{projective measure} of a pure quantum state $\rho$ on $\mathcal{H}$ is described by a set of orthogonal projections $\{P_1, \ldots, P_n\}$ on $\mathcal{H}$, which sum to the identity. The \emph{outcome} of the measure is $i \in \{1, \ldots, n\}$ with probability,
    \begin{equation*}
        \Tr \1[ P_i \rho P^*_i \1],
    \end{equation*}
    and the resulting pure quantum state after the measure becomes,
    \begin{equation*}
        \frac{P_i \rho P^*_i}{\Tr \1[ P_i \rho P^*_i \1]}.
    \end{equation*}
\end{definition*}

Let $\ket{\psi} \coloneqq \alpha \cdot \ket{0} + \beta \cdot \ket{1}$ be a quantum pure state on $\mathcal{H} \coloneqq \mathbb{C}^2$. A projective measure in the computational basis $\ket{0}, \ket{1}$ is described by the two orthogonal projections $\ketbra{0}{0}$ and $\ketbra{1}{1}$, and yields outcome $0$ with probability $|\alpha|^2$ and outcome $1$ with probability $|\beta|^2$.

\begin{remark*}
    In the context of a projective measure, due to the orthogonality property of the projections $\{P_1, \ldots, P_n\}$, and the cyclic property of the trace, the probability of the outcome $i$ given a pure quantum state $\rho \coloneqq \ketbra{\psi}{\psi}$ is just $\Tr \1[ P_i \rho \1] = \bra{\psi} P_i \ket{\psi}$, and the resulting pure quantum state after the measure becomes,
    \begin{equation*}
        \frac{P_i \rho P_i}{\bra{\psi} P_i \ket{\psi}}.
    \end{equation*}
\end{remark*}

Let $\rho \coloneqq \rho_A \otimes \rho_B$ be a pure quantum state on $\mathcal{H}_A \otimes \mathcal{H}_B$, a composite quantum system. The projective measure of $\rho$ described by the set of orthogonal projections $\{P_1, \ldots, P_n\}$ on $\mathcal{H}_A \otimes \mathcal{H}_B$, yields outcome $i \in \{1, \ldots, n\}$ with probability
\begin{align*}
    \Tr \1[ P_i \rho \1] &= \Tr_A \2[ \underbrace{P_i \Tr_B \1[ I_A \otimes \rho_B \1]}_{M_i} \rho_A \2] \\
    &= \Tr \1[ M_i \rho_A \1],
\end{align*}
for some positive semidefinite $M_i$ on $\mathcal{H}_A \otimes \mathcal{H}_B$, which sum to the identity. Such a measure on part of a composite system is called a \emph{generalized measure}.

\subsection{Mixed states}

Let $\rho \coloneqq \ketbra{\psi}{\psi}$ be a pure quantum state on $\mathcal{H}_A \otimes \mathcal{H}_B$, a composite quantum system with $\mathcal{H}_A \coloneqq \mathbb{C}^d$ and $\mathcal{H}_B$ an unknown quantum system. The projective measure of $\rho$ described by the set of orthogonal projections $\{P_1, \ldots, P_n\}$ on $\mathcal{H}_A$ only, yields outcome $i \in \{1, \ldots, n\}$ with probability
\begin{align*}
    \bra{\psi} (P_i \otimes I_B) \ket{\psi} &= \Tr_A \2[ P_i \underbrace{\Tr_B \1[ \rho \1]}_\sigma \2] \\
    &= \Tr \1[ P_i \sigma \1],
\end{align*}
for some $\sigma \in \mathcal{D}_d$, which is in general not a pure quantum state. 

\begin{definition*}
    A \emph{mixed quantum state} on $\mathcal{H}$ is an element of $\mathcal{D}_d$.
\end{definition*}

Since the set of mixed quantum state $\mathcal{D}_d$ is a convex set, and since the extremal point are the pure quantum states, a mixed quantum state is a convex combination of pure quantum states, of the form,
\begin{equation*}
    \sum^k_{i=1} p_i \cdot \ketbra{\psi_i}{\psi_i},
\end{equation*}
with some positive real numbers $p_i$ satisfying $\sum^k_{i=1} p_i = 1$, and $\ketbra{\psi_i}{\psi_i}$ some unit rank projections. From the spectral Theorem, a mixed quantum state on $\mathcal{D}_d$ is a convex sum of at most $d$ terms. The most central element of $\mathcal{D}_d$ is the mixed quantum state $\emph{\mathrm{I}} \coloneqq \frac{I_d}{d}$ called \emph{maximally mixed}.

As a fundamental consequence, maximazing a convex function or minimizing a concave function over the set $\mathcal{D}_d$ of mixed quantum states will lead to extremal values of the function on a pure quantum state.

In the case $\mathcal{H} \coloneqq \mathbb{C}^2$, the mixed quantum states $\mathcal{D}_2$ can be written in a spherical representation, called \emph{Bloch sphere},
\begin{equation*}
    \mathcal{D}_2 = \set{\frac{1}{2}(I_2 + r_x \cdot \sigma_x + r_y \cdot \sigma_y + r_z \cdot \sigma_z)}{r \coloneqq (r_x, r_y, r_z) \in \mathbb{R}^3 \text{ and } \norm{r} \leq 1},
\end{equation*}
where $\emph{\sigma_x}, \emph{\sigma_y}$ and $\emph{\sigma_z}$ are the \emph{Pauli matrices} defined by
\begin{align*}
    \sigma_x &=
    \begin{pmatrix}
        0 & 1 \\
        1 & 0
    \end{pmatrix}
    &
    \sigma_y &=
    \begin{pmatrix}
        0 & -i \\
        i & 0
    \end{pmatrix}
    &
    \sigma_z &=
    \begin{pmatrix}
        1 & 0 \\
        0 & -1
    \end{pmatrix}.
\end{align*}
The pure quantum states of $\mathcal{D}_2$ are the elements satisfying $\norm{r} = 1$.

\section{Quantum entanglement}

The \emph{quantum entanglement} is a fundamental concept in quantum information theory that refers to the non-classical correlations that exist between two or more quantum systems.

\subsection{Schmidt decomposition}

Recall the \emph{singular value decomposition} (\textsc{svd}) for a matrix $M \in \mathcal{M}_d$ acting on $\mathbb{C}^d$: there exists two orthonormal bases $e_1, \ldots, e_d$ and $f_1, \ldots, f_d$ for the vector space $\mathbb{C}^d$, and $d$ non-negative real numbers $s_1, \ldots, s_d$, such that,
\begin{equation*}
    M = \sum^d_{i=1} s_i \ketbra{e_i}{f_i}.
\end{equation*}

Using the isomorphism $\mathrm{Hom}(V, W) \simeq V \otimes W$ between two finite dimensional complex vector spaces $V$ and $W$, the singular value decomposition becomes the \emph{Schmidt decomposition} of vector on a bipartite tensor product $V \otimes W$.
\begin{theorem*}[Schmidt decomposition \cite{schmidt1907theorie,everett1957relative}]
    Let $\psi$ be a vector in a bipartite tensor product of $d$-dimensional Hilbert complex vector spaces $\mathcal{H}_1 \otimes \mathcal{H}_2$. Then there exists two orthonormal bases $e_1, \ldots, e_d$ and $f_1, \ldots, f_d$ for $\mathcal{H}_1$ and $\mathcal{H}_2$, respectively, and $d$ non-negative real numbers $s_1, \ldots, s_d$ called \emph{Schmidt coefficients}, such that,
    \begin{equation*}
        \psi = \sum^d_{i=1} s_i \cdot e_i \otimes f_i.
    \end{equation*}
    The number of nonzero Schmidt coefficients is called the \emph{Schmidt number}.
\end{theorem*}
Apart from the bipartite scenario, direct multipartite extension of the Schmidt decomposition does not exist in general \cite{peres1995higher}.

\subsection{Pure quantum state entanglement}

A bipartite pure quantum state $\ket{\psi}$ on $\mathcal{H}_A \otimes \mathcal{H}_B$ is said to be \emph{entangled} if its Schmidt number is strictly greater than $1$, otherwise it is said to be \emph{separable} and can be written,
\begin{equation*}
    \ket{\psi} = \ket{\psi}_A \otimes \ket{\psi}_B,
\end{equation*}
for some $\ket{\psi}_A \in \mathcal{H}_A$ and $\ket{\psi}_B \in \mathcal{H}_B$.

The Schmidt number of a bipartite pure quantum state gives a canonical quantitative \emph{measure of entanglement}. 

A bipartite pure quantum state $\emph{\omega \coloneqq \ketbra{\Omega}{\Omega}}$, on the $2$-fold composite quantum system $\mathcal{H} \otimes \mathcal{H}$, with $\mathcal{H} \simeq \mathbb{C}^d$, is called \emph{maximally entangled} if it has, on the computational basis, the form
\begin{equation*}
    \ket{\Omega} = \frac{1}{\sqrt{d}} \sum^{d-1}_{i = 0} \ket{ii}.
\end{equation*}
The partial transpose of an unormalized maximally entangled pure quantum state $d \cdot \omega$ is the \emph{flip operator},
\begin{equation*}
    d \cdot \omega^\Tpartial = \sum^{d-1}_{i,j = 0} \ketbra{ij}{ij},
\end{equation*}
on the computational basis. On product vector $x \otimes y$, the flip operator acts through $d \cdot  \omega^\Tpartial (x \otimes y) = y \otimes x$. However, since $y \otimes x - x \otimes y$ is a eigenvector for the negative eigenvalue $-1$, the flip operator $\omega^\Tpartial$ is not a quantum state. The normalized flip operator on $\mathcal{D}_{d^2}$ is denoted $\emph{\mathrm{F}} \coloneqq \frac{\omega^\Tpartial}{d}$.

\begin{theorem*}[PPT criterion \cite{peres1996separability,horodecki2001separability}]
    If $\rho$ is a separable bipartite pure quantum state on $\mathcal{H}_A \otimes \mathcal{H}_B$, then the partially transposed $\rho^\Tpartial$ is a quantum state. The converse is true if and only if $\dim \mathcal{H}_A \times \mathcal{H}_B \leq 6$.
\end{theorem*}

In general, in the multipartite scenario $\mathcal{H}_1 \otimes \cdots \otimes \mathcal{H}_k$, a pure quantum state $\ket{\psi}$ is separable if it can be written as a product vector
\begin{equation*}
    \ket{\psi} = \ket{\psi_1} \otimes \cdots \otimes \ket{\psi_k},
\end{equation*}
and entangled otherwise.

\begin{remark*}
    As a unit rank projection $\ketbra{\psi}{\psi}$, a separable pure quantum state on a multipartite $\mathcal{H}_1 \otimes \cdots \otimes \mathcal{H}_k$ can be written
    \begin{equation*}
        \ketbra{\psi}{\psi} = \ketbra{\psi_1}{\psi_1} \otimes \cdots \otimes \ketbra{\psi_k}{\psi_k}.
    \end{equation*}
\end{remark*}

\subsection{Mixed quantum state entanglement}

A mixed quantum state is said to be separable if it can be written as a convex combination of separable pure quantum states. Therefore the convex set of separable quantum states on $\mathcal{H}_1 \otimes \cdots \otimes \mathcal{H}_k$ is
\begin{equation*}
    \Conv \setAlt{\ketbra{\psi_1}{\psi_1} \otimes \cdots \otimes \ketbra{\psi_k}{\psi_k}}{\psi_i \in \mathcal{H}_i \text{ for all } i \in \{1, \ldots, k\}},
\end{equation*}
with extremal points the separable pure quantum states.

An \emph{isotropic state} is a convex combination of a maximally entangled and maximally mixed quantum states:
\begin{equation*}
    \rho = \lambda \cdot \omega + (1 - \lambda) \mathrm{I},
\end{equation*}
with $0 \leq \lambda \leq 1$ and $\omega, \mathrm{I} \in \mathcal{D}_{d^2}$. If $\frac{-1}{d^2 - 1} \leq \lambda \leq 0$, then $\rho$ is still in $\mathcal{D}_{d^2}$, and thus still a mixed quantum state.

A \emph{Werner state} is an affine combination of a normalized flip operator and maximally mixed quantum states:
\begin{equation*}
    \sigma = \lambda \cdot \mathrm{F} + (1 - \lambda) \mathrm{I},
\end{equation*}
with $\frac{1}{1-d} \leq \lambda \leq \frac{1}{1+d}$ and $\mathrm{F}, \mathrm{I} \in \mathcal{D}_{d^2}$.

Any isotropic state $\rho$ and Werner state $\sigma$ satisfy the following commutation relations:
\begin{equation*}
    [\rho, \bar{U} \otimes U] = 0 \qquad \text{ and } \qquad [\sigma, U \otimes U] = 0,
\end{equation*}
for all unitary matrices $U$.

\begin{theorem*}[\cite{horodecki1999reduction}]
    Let $\rho \coloneqq \lambda \cdot \omega + (1 - \lambda) \mathrm{I}$ be an isotropic state on $\mathcal{D}_{d^2}$. Then $\rho$ is separable if and only if $\lambda \leq \frac{1}{d+1}$.
\end{theorem*}
\begin{theorem*}[\cite{werner1989quantum}]
    Let $\sigma \coloneqq \lambda \cdot \mathrm{F} + (1 - \lambda) \mathrm{I}$ be a Werner state, on $\mathcal{D}_{d^2}$. Then $\sigma$ is separable if and only if $\lambda \geq \frac{1}{1-d^2}$.
\end{theorem*}

In general, deciding whether a given quantum state is separable is known to be NP-hard \cite{gurvits2003classical}.

\subsection{Monogamy of entanglement} \label{chap3:sec:monogamyOfEntanglement}

A bipartite quantum state $\rho$ on $\mathcal{H}_A \otimes \mathcal{H}_B$ is said to be $k$-extendible, with respect to $\mathcal{H}_B$ if there exists a quantum state $\sigma$ on $\mathcal{H}_A \otimes \mathcal{H}^{\otimes k}_B$ such that for all $i \in \{1, \ldots, k\}$:
\begin{equation*}
    \rho = \Tr_{\{B_1, \ldots, B_k\} \setminus \{B_i\}} \1[ \sigma \1].
\end{equation*}

\begin{theorem*}[Entanglement hierarchy \cite{doherty2004complete}]
    A bipartite quantum state $\rho$ on $\mathcal{H}_A \otimes \mathcal{H}_B$ is separable if and only if it is $k$-extendible for all $k \in \mathbb{N}$.
\end{theorem*}

Let $\rho_{A,B,C}$ be a tripartite quantum state on $\mathcal{H}_A \otimes \mathcal{H}_B \otimes \mathcal{H}_C$ with $\mathcal{H}_A \simeq \mathcal{H}_B \simeq \mathcal{H}_C$, and such that the \emph{reduced quantum state} on $\mathcal{H}_A \otimes \mathcal{H}_B$,
\begin{equation*}
    \rho_{A,B} \coloneqq \Tr_C [\rho_{A,B,C}],
\end{equation*}
is maximally entangled, i.e. $\rho_{A,B} = \omega$. From the spectral Theorem, the quantum state $\rho_{A,B,C}$ is a convex combination,
\begin{equation*}
    \rho_{A,B,C} = \sum^k_{i=1} p_i \cdot \ketbra{\psi_i}{\psi_i},
\end{equation*}
for some orthonormal pure quantum states $\ket{\psi_i}$ on $\mathcal{H}_A \otimes \mathcal{H}_B \otimes \mathcal{H}_C$. Then the reduced quantum state $\rho_{A,B}$ becomes
\begin{equation*}
        \rho_{A,B} = \sum^k_{i=1} p_i \cdot \Tr_C \2[ \ketbra{\psi_i}{\psi_i} \2].
\end{equation*}
But since $\rho_{A,B}$ is a pure quantum state, i.e. $\rho_{A,B} = \ketbra{\Omega}{\Omega}$ is an extremal point of the convex set of density matrices, every $\Tr_C \2[ \ketbra{\psi_i}{\psi_i} \2]$ in the sum must necessarily be equal to $\rho_{A,B}$. This implies that $\rho_{A,B,C} = \rho_{A,B} \otimes \rho_C$, for some reduced quantum state $\rho_C$ on $\mathcal{H}_C$. Thus none of the reduced quantum states $\rho_{A,C}$ and $\rho_{B,C}$ can be maximally entangled. This phenomenon is known as \emph{monogamy of entanglement}.

\section{Quantum channels}

Let $\Phi: \mathcal{M}_d \to \mathcal{M}_{d^\prime}$ be a map such that $\Phi \1( \mathcal{D}_d \1) \subseteq \mathcal{D}_{d^\prime}$, i.e. mapping $d$-dimensional quantum states to $d^\prime$-dimensional quantum states. Under linearity assumption this is equivalent to:
\begin{itemize}
    \item $\Phi$ \emph{positive}: $X \geq 0 \implies \Phi(X) \geq 0$.
    \item $\Phi$ \emph{trace preserving}: $\Tr X = \Tr \Phi(X)$.
\end{itemize}

The transpose of a matrix is a linear map that satisfies both positivity and trace preserving properties. But when partially applied to a composite quantum system, this can lead to non-quantum states, e.g. the partial transpose of a maximally entangled pure quantum state is the flip operator $F = \omega^\Tpartial$.

A linear map $\Phi: \mathcal{M}_d \to \mathcal{M}_{d^\prime}$ is called \emph{completely positive} if all the partial applications of $\Phi$ on any positive semidefinite matrix results in another positive semidefinite matrix, i.e. $\forall D \in \mathbb{N}, \mathcal{M}_d \otimes \mathcal{M}_D \ni X \geq 0 \implies (\Phi \otimes \id_D) (X) \geq 0$.

\begin{definition*}
    The most general transformations of quantum states, called \emph{quantum channels}, are the Completely Positive Trace Preserving (\textsc{cptp}) linear maps.
\end{definition*}

\subsection{Structure of quantum channels}

The \emph{Choi matrix} of a linear map $\Phi: \mathcal{M}_d \to \mathcal{M}_{d^\prime}$ is the matrix $\emph{C_\Phi}$ in $\mathcal{M}_{d \times d^\prime}$, defined by,
\begin{align*}
    C_\Phi &\coloneqq (\id_d \otimes \Phi) \3( \sum^d_{i,j = 1} \ketbra{ii}{jj} \3) \\
    &= (\id_d \otimes \Phi) (d \cdot \omega).
\end{align*}
It is possible to retrieve the original linear map $\Phi$ from its associated Choi matrix $C_\Phi$ using the formula,
\begin{equation*}
    \Phi(X) = \Tr_d \1[ C_\Phi (X^\T \otimes I_{d^{\prime}}) \1].
\end{equation*}

\begin{theorem}[\cite{watrous2018theory}] \label{chap3:thm:quantumChannelStructure}
    Let $\Phi: \mathcal{M}_d \to \mathcal{M}_{d^\prime}$ be a linear map. The following are equivalent:
    \begin{itemize}
        \item the map $\Phi$ is \textsc{cptp};
        \item the Choi matrix $C_\Phi$ is positive semidefinite and $\Tr_{d^\prime} \1[ C_\Phi \1] = I_d$;
        \item there exist $A_1, \ldots, A_k \in \mathcal{M}_{d \times d^\prime}$ such that,
        \begin{equation*}
            \Phi(X) = \sum^k_{i=1} A_i X A^*_i \qquad \text{ and } \qquad \sum^k_{i=1} A^*_i A_i = I_d;
        \end{equation*}
        \item there exist $D \in \mathbb{N}$ and an isometry $V: \mathbb{C}^d \to \mathbb{C}^{d^\prime} \otimes \mathbb{C}^D$ such that,
        \begin{equation*}
            \Phi(X) = \Tr_D \1[ V X V^* \1].
        \end{equation*}
    \end{itemize}
\end{theorem}

\subsection{Compatibility of quantum channels}

Let $\Phi: \mathcal{M}_d \to \mathcal{M}^{\otimes n}_{d^\prime}$ be a quantum channel from $1$ to $n$ quantum states, the $i$-th \emph{marginal} of $\Phi$, denoted $\emph{\Phi_i}$ is defined by,
\begin{equation*}
    \Phi_i(X) \coloneqq \Tr_{[n] \setminus \{i\}} \1[ \Phi(X) \1].
\end{equation*}
A marginal of a quantum channel is also a quantum channel.

Let $\Phi_i: \mathcal{M}_d \to \mathcal{M}_{d_i}$ be a family of $k$ quantum channels. The \emph{quantum channel compatibility problem} consists in determining whether there exists a global quantum channel $\Psi: \mathcal{M}_d \to \mathcal{M}_{d_1} \otimes \cdots \otimes \mathcal{M}_{d_k}$ \emph{compatible} with all the $\Phi_i$, that is,
\begin{equation*}
    \Psi_i = \Phi_i,
\end{equation*}
for all marginals $\Psi_i$.

\begin{remark*}
    Quantum channels can be incompatible with themselves.
\end{remark*}

\section{Quantum fidelity}

The \emph{quantum fidelity} is a measure of the closeness between two quantum states $\rho$ and $\sigma$, defined as the function \emph{$F$} on $\mathcal{D}_d \times \mathcal{D}_d$ by,
\begin{equation*}
    F(\rho, \sigma) = \Tr {\3[ \sqrt{\sqrt{\sigma} \rho \sqrt{\sigma}} \3]}^2.
\end{equation*}

\begin{proposition}[\cite{watrous2018theory}] \label{chap3:prop:quantumFidelityProperties}
    The quantum fidelity $F$ between two quantum states has the following properties:
    \begin{itemize}
        \item $F(\rho, \sigma) = F(\sigma, \rho)$;
        \item $F(\rho, \sigma) \in [0, 1]$;
        \item $F(\rho, \sigma) = 1 \Longleftrightarrow \rho = \sigma$;
        \item $F \2( \rho, \ketbra{\psi}{\psi} \2) = \bra{\psi} \rho \ket{\psi} = \Tr \2[ \rho \ketbra{\psi}{\psi} \2]$;
        \item $F \1( U \rho U^* , U \sigma U^* \1) = F(\rho , \sigma),$ for all unitary matrices $U$;
        \item $F \1( \Phi(\rho) , \Phi(\sigma) \1) \geq F(\rho , \sigma),$ for all quantum channels $\Phi$;
        \item $F$ is jointly concave, i.e. for all $\lambda \in [0, 1]$,
        \begin{equation*}
            F \1( \lambda \cdot \rho_1 + (1 - \lambda) \rho_2, \lambda \cdot \sigma_1 + (1 - \lambda) \sigma_2 \1) \geq \lambda \cdot F(\rho_1, \sigma_1) + (1 - \lambda) \cdot F(\rho_2, \sigma_2).
        \end{equation*}
    \end{itemize}
\end{proposition}

\section{Graphical calculus} \label{chap3:sec:graphicalCalculus}

The present Section introduces a graphical calculus for tensors, which is built upon the graphical notation developed by Penrose \cite{penrose1971applications}. Recently, analogous calculi have been formulated within the tensor network states framework and the framework of categorical quantum information theory, which are elaborated in \cite{wood2015tensor,bridgeman2017hand,coecke2017picturing}. The graphical calculus introduced here is consistent with the tensor representation of a diagram algebra, introduced in \hyperref[chap2:sec:tensorRepresentation]{Section~\ref*{chap2:sec:tensorRepresentation}}.

In this graphical notation, tensors are represented by \emph{boxes} and \emph{wires},
\begin{equation*}

\end{equation*}
More specifically, the wires are labeled by \emph{indices}, such that a box represents the value of the tensor at the given indices,
\begin{equation*}
    %
\end{equation*}
In particular, a vector is a box with only $1$ wire pointing to the left, and labeled by the index of the coordinate. A dual vector has its wire pointing to the right,
\begin{equation*}
    \hskip \textwidth minus \textwidth
    %
    \hskip \textwidth minus \textwidth
\end{equation*}
Note that these left and right directions are just a convention corresponding to the usual \emph{right-to-left} composition in linear algebra.

The tensor diagrams can be combined in two ways. The \emph{tensor product} combines two diagrams vertically
\begin{equation*}
    %
\end{equation*}
or horizontally, with the isomorphism $\mathrm{Hom}(V, W) \simeq V^* \otimes W$ between two finite dimensional complex vector spaces $V$ and $W$,
\begin{equation*}
    %
\end{equation*}
The operation of joining together two diagrams by one of their wire labeled with the same index corresponds to multiplying the values of the two tensors at the given indices,
\begin{equation*}
    %
\end{equation*}
The \emph{tensor contraction} combines two diagrams using the previous operation by taking the sum over all common indices
\begin{equation*}
    %
\end{equation*}
In particular, this leads to the scalar product,
\begin{equation*}
    %
\end{equation*}
the matrix product,
\begin{equation*}
    %
\end{equation*}
and the matrix trace,
\begin{equation*}
    %
\end{equation*}
Scalars multiply diagrams and are depicted next to them. The following three special tensors of $\mathbb{C}^d$, important in quantum information theory, have wire-only diagrams,
\begin{equation*}
    \hskip \textwidth minus \textwidth
    %
    \hskip \textwidth minus \textwidth
\end{equation*}
Finally, the \emph{matrix transpose} can be graphically depicted by swapping the \textit{input} and the \textit{output} wires of a box depicting a matrix. The partial transpose permutes just the wires of the corresponding spaces,
\begin{equation*}
    %
\end{equation*}
\chapter{Quantum Cloning} \label{chap4}

\noindent\textbf{The present Section discusses the results of the papers ``A geometrical description of the universal $1 \to 2$ asymmetric quantum cloning region'' \cite{nechita2021geometrical}, and ``The asymmetric quantum cloning region'' \cite{nechita2022asymmetric}, of which I am a co-author.}

\medskip

The problem of quantum cloning has received considerable attention over the last thirty years. This area of research began with the early investigations of universal quantum cloners \cite{buvzek1996quantum} and has since expanded to include various cloning scenarios: symmetric \cite{werner1998optimal, keyl1999optimal}, asymmetric \cite{cerf1998asymmetric, fiuravsek2005highly}, qubits \cite{gisin1997optimal}, qudits \cite{cerf2000asymmetric}, universal \cite{bruss1998optimal, iblisdir2006generalised, iblisdir2005multipartite}, equatorial \cite{bruss2000phase, d2003optimal, durt2004characterization}, economical \cite{durt2004economic, durt2005economical}, probabilistic \cite{duan1998probabilistic, duan1998probabilistic}, continuous quantum systems \cite{cerf2004optimal, cerf2005non}.

Of particular interest are two sets of papers dealing with the most general case of asymmetric universal $1 \to N$ quantum cloning problem. The first set of papers, by Kay and collaborators, deals with the optimisation of this problem \cite{kay2012optimal, kay2014optimal, kay2016optimal}. The second set of papers, by \'Cwikli\'nski, Horodecki, Mozrzymas, and Studzi\'nski, uses techniques from group representation theory \cite{cwiklinski2012region, studzinski2013commutant, mozrzymas2014structure, studzinski2014group, mozrzymas2018simplified}.

In particular, Hashagen's article on the asymmetric $1 \to 2$ quantum cloning problem \cite{hashagen2016universal},  should be read in conjunction with the \hyperref[chap4:sec:universal1to2QuantumCloningProblem]{Section~\ref*{chap4:sec:universal1to2QuantumCloningProblem}}. Indeed, the quantum cloning problem is studied there using similar techniques: twirling of the quantum channel, Choi matrix as an element the algebra of operator permutations, decomposition this algebra into matrix algebras of size $2 \times 2$ and $1 \times 1$.
The main contribution of \cite{nechita2021geometrical} is the complete spectral analysis of the Choi matrix, which allows a block diagonalization of the matrix and a characterisation of the figures of merit as ellipses.

\section{Quantum cloning problem}

The problem known as the \emph{quantum cloning problem} aims to identify a specific quantum channel, called the \emph{quantum cloning channel}, denoted by $\Phi: \mathcal{M}_d \to {\1( \mathcal{M}_d \1)}^{\otimes N}$, which maps an \textit{input} pure quantum state to an \textit{output} mixed quantum state on a $N$-fold composite quantum system, such that the \textit{output} marginals of $\Phi$ are as close as possible to the \textit{input}. To achieve this, a \emph{direction vector} denoted $\emph{a}$ satisfying $a \in [0,1]^N$ and $\sum^N_{i=1}a_i = 1$, together with a subset $\emph{\Gamma}$ of the pure quantum states, are introduced. The quantum cloning problem is defined as the optimisation problem given by,
\begin{equation*}
    \sup_{\Phi~\textsc{cptp}} \sum^N_{i = 1} a_i \cdot \operatornamewithlimits{\mathbb{E}}_{\rho \in \Gamma} \2[ F \1( \Phi_i(\rho), \rho \1) \2],
\end{equation*}
where the expected value is taken with respect to the uniform measure on $\Gamma$.
\begin{remark*}
    The quantum cloning problem belongs to a more general class of quantum marginal problems, which includes the quantum state marginal problem \cite{haapasalo2021quantum}, the quantum channel marginal problem \cite{hsieh2022quantum}, and the quantum measurement marginal problem \cite{loulidi2021compatibility}.
\end{remark*}

The \emph{no-cloning theorem} states that a perfect quantum cloning channel, i.e. a quantum channel with marginals the identity quantum channels, cannot exist in general. Indeed, a linearity argument shows that, if $\ket{0}, \ldots, \ket{d-1}$ is the computational basis of $\mathbb{C}^d$, and $U \in \mathrm{U}_{d^2}$ is a unitary matrix such that on this basis,
\begin{equation*}
    U \ket{i} \otimes \ket{v} = \ket{i} \otimes \ket{i},
\end{equation*}
Then $U$ is a perfect quantum cloning unitary for this basis, given an auxiliary pure quantum state $\ket{v}$. But if $\ket{\psi} \coloneqq \frac{\ket{i} + \ket{j}}{\sqrt{2}}$ for distinct $i,j \in \{0, \ldots, d-1 \}$, then,
\begin{align*}
    U \ket{\psi} \otimes \ket{v} &= \frac{\ket{i} \otimes \ket{i} + \ket{j} \otimes \ket{j}}{\sqrt{2}} \\
    &\neq \ket{\psi} \otimes \ket{\psi}.
\end{align*}

Many quantum cryptographic schemes rely on the existence of the no-cloning theorem \cite{bennett1984quantum,broadbent2020uncloneable}.
\begin{theorem}[No-cloning theorem \cite{wootters1982single,dieks1982communication}] \label{chap4:thm:noCloning}
    For any subset $\Gamma$ of the pure quantum states, there is no quantum cloning channel $\Psi: \mathcal{M}_d \to {\1( \mathcal{M}_d \1)}^{\otimes N}$ such that for all pure quantum states $\rho \in \Gamma$ and all marginals $\Psi_i$,
    \begin{equation*}
        \Psi_i(\rho) = \rho,
    \end{equation*}
    unless $\Gamma$ is a set of mutually orthogonal pure quantum states.
\end{theorem}
\begin{remark*}
    A perfect quantum cloning channel $\Psi: \mathcal{M}_d \to {\1( \mathcal{M}_d \1)}^{\otimes N}$ has marginals $\Phi_i$ of the form $\Phi_i = \id_d$ the identity quantum channels, and Choi matrices $C_{\Phi_i} = d \cdot \omega$ an unormalized maximally entangled quantum state. The no-cloning \hyperref[chap4:thm:noCloning]{Theorem~\ref*{chap4:thm:noCloning}} is a reformulation of the monogamy of the entanglement, and the fact that there is no quantum channel compatible with the identity quantum channels.
\end{remark*}

A generalisation of the quantum cloning problem can be considered given a quantum cloning channel $\Phi: {\1( \mathcal{M}_d \1)}^{\otimes M} \to {\1( \mathcal{M}_d \1)}^{\otimes N}$, which maps an \textit{input} pure quantum state as an identical product state on a $M$-fold composite quantum system, to an \textit{output} mixed quantum state on a $N$-fold composite quantum system,
\begin{equation*}
    \sup_{\Phi~\textsc{cptp}} \sum^N_{i = 1} a_i \cdot \operatornamewithlimits{\mathbb{E}}_{\rho \in \Gamma} \3[ F \2( \Phi_i \1( \rho^{\otimes M} \1), \rho \2) \3].
\end{equation*}
Even with $M$ identical copies as \textit{input}, the no-cloning theorem hold.
\begin{theorem}[$M \to N$ no-cloning theorem \cite{werner1998optimal}] \label{chap4:thm:MtoNnoCloning}
    For any subset $\Gamma$ of the pure quantum states, there is no quantum cloning channel $\Phi: {\1( \mathcal{M}_d \1)}^{\otimes M} \to {\1( \mathcal{M}_d \1)}^{\otimes N}$, with $M < N$, such that for all pure quantum states $\rho \in \Gamma$ and all marginals $\Psi_i$,
    \begin{equation*}
        \Phi_i \1( \rho^{\otimes M} \1) = \rho,
    \end{equation*}
    unless $\Gamma$ is a set of mutually orthogonal pure quantum states.
\end{theorem}

It is important to note that the quantum cloning problem, the $1 \to N$ no-cloning \hyperref[chap4:thm:noCloning]{Theorem~\ref*{chap4:thm:noCloning}} and $M \to N$ no-cloning \hyperref[chap4:thm:MtoNnoCloning]{Theorem~\ref*{chap4:thm:MtoNnoCloning}} are states on pure states only. When mixed states are considered, the problem becomes the \emph{quantum broadcasting problem}, and the equivalent no-broadcasting theorems do not hold in the same generality. In particular for mixed enough states, it is possible to broadcast.
\begin{theorem*}[Superbroacasting \cite{d2005superbroadcasting}]
    For $0 < M < N$ large enough, there exists a quantum channel $\Phi: {\1( \mathcal{M}_2 \1)}^{\otimes M} \to {\1( \mathcal{M}_2 \1)}^{\otimes N}$ such that for all $2$-dimensional mixed states $\rho \in \mathcal{D}_2$, where $\rho = \frac{1}{2}(I_2 + r_x \cdot \sigma_x + r_y \cdot \sigma_y + r_z \cdot \sigma_z)$ and with $\norm{r}$ small enough, the mixed state, 
    \begin{equation*}
        \sigma \coloneqq \Phi_i \1( \rho^{\otimes M} \1),
    \end{equation*}
    commutes with $\rho$, for all marginals $\Phi_i$. I.e. they are collinear in the spherical representation:
    \begin{equation*}
        \mathcal{D}_2 = \set{\frac{1}{2}(I_2 + r_x \cdot \sigma_x + r_y \cdot \sigma_y + r_z \cdot \sigma_z)}{r \coloneqq (r_x, r_y, r_z) \in \mathbb{R}^3 \text{ and } \norm{r} \leq 1}.
    \end{equation*}
    Moreover, if $\sigma = \frac{1}{2}(I_2 + r^\prime_x \cdot \sigma_x + r^\prime_y \cdot \sigma_y + r^\prime_z \cdot \sigma_z)$, then $\norm{r^\prime} \geq \norm{r}$.
\end{theorem*}

From \hyperref[chap3:prop:quantumFidelityProperties]{Proposition~\ref*{chap3:prop:quantumFidelityProperties}} and the joint concavity quantum fidelity, given two quantum cloning channels $\Phi$ and $\Psi$ from $\mathcal{M}_d$ to ${\1( \mathcal{M}_d \1)}^{\otimes N}$, then for all subset $\Gamma$ of the pure quantum states and for all $1 \leq i \leq N$,
\begin{equation*}
    \operatornamewithlimits{\mathbb{E}}_{\rho \in \Gamma} \4[ F \3( \frac{(\Phi + \Psi)_i (\rho)}{2}, \rho \3) \4] \geq \frac{\operatornamewithlimits{\mathbb{E}}_{\rho \in \Gamma} \2[ F \1( \Phi_i(\rho), \rho \1) \2] + \operatornamewithlimits{\mathbb{E}}_{\rho \in \Gamma} \2[ F \1( \Psi_i(\rho), \rho \1) \2]}{2}.
\end{equation*}
Hence, in order to address the quantum cloning optimization problem, the approach would be to identify the largest uniform sum of quantum cloning channels.

Let $\Gamma$ be a subset of the pure quantum states, and $G$ be a compact subgroup of the unitary group $\mathrm{U}_d$ acting on $\mathbb{C}^d$, such that for all $\rho \in \Gamma$ and for all $M \in G$,
\begin{equation*}
    M \rho M^* \in \Gamma.
\end{equation*}
The \emph{twirling} of a quantum channel $\Phi: \mathcal{M}_d \to {\1( \mathcal{M}_d \1)}^{\otimes N}$, with respect to $\Gamma$ and $G$, denoted $\emph{\widetilde{\Phi}}$, is defined for all $\rho \in \Gamma$ by
\begin{equation*}
    \widetilde{\Phi}(\rho) \coloneqq \int_G {\1( M^* \1)}^{\otimes N} \2( \Phi \1( M \rho M^* \1) \2) M^{\otimes N} ~\mathrm{d}M,
\end{equation*}
where the integral is taken with respect to the normalized Haar measure on the group $G$.
Then for all $1 \leq i \leq N$,
\begin{equation*}
    \operatornamewithlimits{\mathbb{E}}_{\rho \in \Gamma} \2[ F \1( \widetilde{\Phi}_i (\rho), \rho \1) \2] \geq \operatornamewithlimits{\mathbb{E}}_{\rho \in \Gamma} \2[ F \1( \Phi_i(\rho), \rho \1) \2].
\end{equation*}
The approach would be to consider the largest group $G$. In particular the partially transposed Choi matrix $C^\Tpartial_{\widetilde{\Phi}}$ is in the commutant of $G$.
\begin{proposition} \label{chap4:prop:twirledChoiMatrixCommutationRelation}
    Let $\Gamma$ be a subset of the pure quantum states, let $G$ be a compact subgroup of $\mathrm{U}_d$, and $\Phi: \mathcal{M}_d \to {\1( \mathcal{M}_d \1)}^{\otimes N}$ a quantum cloning channel. If $\widetilde{\Phi}$ is the twirling of $\Phi$ with respect to $\Gamma$ and $G$, then for all $M \in G$,
    \begin{equation*}
        \1[ C^\Tpartial_{\widetilde{\Phi}}, M^{\otimes (N+1)} \1] = 0.
    \end{equation*}
\end{proposition}
\begin{proof}
    Since $G$ is a subgroup of $\mathrm{U}_d$, for any $M$ in $G$, the following two equalities hold,
    \begin{equation*}
        \1( \bar{M} \otimes M \1) \ket{\Omega} = \ket{\Omega} \qquad \text{ and } \qquad \bra{\Omega} \1( M^\T \otimes M^* \1) = \bra{\Omega}.
    \end{equation*}
    Then for any $M \in G$,
    \begin{align*}
        C^\Tpartial_{\widetilde{\Phi}} &= {\3( \2( \id_d \otimes \widetilde{\Phi} \2) \1( d \cdot \ketbra{\Omega}{\Omega} \1) \3)}^\Tpartial \\
        &= {\3( \2( \id_d \otimes \widetilde{\Phi} \2) \2( d \1( \bar{M} \otimes M \1) \ketbra{\Omega}{\Omega} \1( M^\T \otimes M^* \1) \2) \3)}^\Tpartial.
    \end{align*}
    From the definition of the twirling $\widetilde{\Phi}$, then for all $M \in G$ and for all $\rho \in \Gamma$
    \begin{equation*}
        \widetilde{\Phi} \1( M \rho M^* \1) = M^{\otimes N} \1( \widetilde{\Phi} (\rho ) \1) {\1( M^* \1)}^{\otimes N}.
    \end{equation*}
    Then the following commutation relation on the partially transposed Choi matrix $C^\Tpartial_{\widetilde{\Phi}}$ holds for any $M \in G$,
    \begin{align*}
        C^\Tpartial_{\widetilde{\Phi}} &= {\4( \1( \bar{M} \otimes M^{\otimes N} \1) \3( \2( \id_d \otimes \widetilde{\Phi} \2) \1( d \cdot \ketbra{\Omega}{\Omega} \1) \3) \2( M^\T \otimes {\1( M^* \1)}^{\otimes N} \2) \4)}^\Tpartial \\
        &= \1( M \otimes M^{\otimes N} \1) {\3( \2( \id_d \otimes \widetilde{\Phi} \2) \1( d \cdot \ketbra{\Omega}{\Omega} \1) \3)}^\Tpartial \2( M^* \otimes {\1( M^* \1)}^{\otimes N} \2) \\
        &= \1( M \otimes M^{\otimes N} \1) C^\Tpartial_{\widetilde{\Phi}} \2( M^* \otimes {\1( M^* \1)}^{\otimes N} \2).
    \end{align*}
    That is $\1[ C^\Tpartial_{\widetilde{\Phi}}, M^{\otimes (N+1)} \1] = 0$.
\end{proof}
\begin{corollary}
    Let $\Gamma$ be a subset of the pure quantum states, let $G$ be a compact subgroup of $\mathrm{U}_d$, and $\Phi: \mathcal{M}_d \to {\1( \mathcal{M}_d \1)}^{\otimes N}$ a quantum cloning channel. If $\widetilde{\Phi}$ is the twirling of $\Phi$ with respect to $\Gamma$ and $G$, then for all $M \in G$ and all $i \in \{1, \ldots, N \}$,
    \begin{equation*}
        \1[ C^\Tpartial_{\widetilde{\Phi}_i}, M^{\otimes 2} \1] = 0.
    \end{equation*}
\end{corollary}

\begin{remark*}
    The importance of twirling a quantum channel is twofold: first, it allows to improve the performance of the quantum channel with respect to the optimization problem, and second, it allows to simplify the optimization problem through the induced symmetries. The study of optimization problems under symmetries has been the subject of substantial work \cite{fawzi2022hierarchy,grinko2022linear}.
\end{remark*}

\subsection{Universal quantum cloning problem} \label{chap4:sec:universalQuantumCloningProblem}

When the subset $\Gamma$ of the pure quantum states is the full set of pure quantum states, it can choosen the subgroup $G$ of $\mathrm{U}_d$ to be the full unitary group. The quantum cloning problem becomes the \emph{universal quantum cloning problem} Then from \hyperref[chap4:prop:twirledChoiMatrixCommutationRelation]{Proposition~\ref*{chap4:prop:twirledChoiMatrixCommutationRelation}}, the partially transposed Choi matrix $C^\Tpartial_{\widetilde{\Phi}}$ of a twirled quantum cloning channel $\widetilde{\Phi}$ commutes with all the unitary matrices $U^{\otimes (N+1)}$, i.e. it is in the commutant of the algebra
\begin{equation*}
    \Span_{\mathbb{C}} \set{U^{\otimes (N+1)}}{U \in \mathrm{U}_d}.
\end{equation*}

From \hyperref[chap2:thm:commutantSchurWeylDuality]{Theorem~\ref*{chap2:thm:commutantSchurWeylDuality}} and \hyperref[chap2:thm:generalAUnitaryMatrixAlgebra]{Theorem~\ref*{chap2:thm:generalAUnitaryMatrixAlgebra}}, the partially transposed Choi matrix $C^\Tpartial_{\widetilde{\Phi}}$ is in the algebra
\begin{equation*}
    \Span_{\mathbb{C}} \set{\psi(\sigma)}{\sigma \in \mathfrak{S}_{N+1}}.
\end{equation*}
That is, it exists a family of complex numbers $c_\sigma$ indexed by the permutations of $\mathfrak{S}_{N+1}$ such that
\begin{equation*}
    C_{\widetilde{\Phi}} = \sum_{\sigma \in \mathfrak{S}_{N+1}} c_\sigma \cdot \psi(\sigma)^\Tpartial.
\end{equation*}

Therefore, the Choi matrix $C_{\widetilde{\Phi}}$ is a sum of partially transposed tensor representation of the symmetric group $\mathfrak{S}_{N+1}$. Hence, a sum of partially transposed tensor representation of the symmetric group $\mathfrak{S}_{N+1}$ is a Choi matrix of a quantum channel if both the positivity and the trace conditions of \hyperref[chap3:thm:quantumChannelStructure]{Theorem~\ref*{chap3:thm:quantumChannelStructure}} hold. These conditions depend on the $(N+1)!$ coefficients, and in particular it does not depend on the dimension of the quantum system. Recall from the Stirling's formula that the asymptotic growth of the factorial function is,
\begin{equation*}
    n! \sim \sqrt{2 \pi n} {\2( \frac{n}{e} \2)}^n.
\end{equation*}

From \hyperref[chap4:prop:twirledChoiMatrixCommutationRelation]{Proposition~\ref*{chap4:prop:twirledChoiMatrixCommutationRelation}}, the partially transposed Choi matrix $C_{\widetilde{\Phi}_i}$ of each marginal $\widetilde{\Phi}_i$ has the form
\begin{align*}
    C_{\widetilde{\Phi}_i} &= \alpha_i \cdot \psi {\1( (1)(2) \1)}^\Tpartial + \beta_i \cdot \psi {\1( (1\:2) \1)}^\Tpartial \\[0.5em]
    &= \alpha_i \cdot \psi {\left(

    \right),
\end{align*}
for some complex numbers $\alpha_i$ and $\beta_i$. For all $\rho \in \Gamma$, the marginals $\widetilde{\Phi}_i$ on $\rho$ are
\begin{align*}
    \widetilde{\Phi}_i(\rho) &= \Tr_{\{0\}} \1[ C_{\widetilde{\Phi}_i} \1( \rho^\T \otimes I_d \1) \1] \\
    &= \alpha_i \cdot \Tr_{\{0\}} \2[ \psi {\1( (1)(2) \1)}^\Tpartial \1( \rho \otimes I_d \1) \2] + \beta_i \cdot \Tr_{\{0\}} \2[ \psi {\1( (1\:2) \1)}^\Tpartial \1( \rho \otimes I_d \1) \2] \\
    &= \alpha_i \cdot I_d + \beta_i \cdot \rho,
\end{align*}
and the fidelity $F \1( \widetilde{\Phi}_i(\rho), \rho \1)$ becomes,
\begin{align*}
    F \1( \widetilde{\Phi}_i(\rho), \rho \1) &= \Tr \2[ \1( \widetilde{\Phi}_i(\rho) \1) \rho \2] \\
    &= \Tr \3[ C_{\widetilde{\Phi}_i} \2( \1( \rho^\T \cdot \rho  \1) \otimes I_d \2) \3] \\
    &= \Tr \2[ C_{\widetilde{\Phi}_i} \1( \rho \otimes I_d \1) \2] \\
    &= \alpha_i \cdot \Tr \2[ \psi {\1( (1)(2) \1)}^\Tpartial \1( \rho \otimes I_d \1) \2] + \beta_i \cdot \Tr \2[ \psi {\1( (1\:2) \1)}^\Tpartial \1( \rho \otimes I_d \1) \2] \\
    &= d \cdot \alpha_i + \beta_i.
\end{align*}

\subsection{Equatorial quantum cloning problem} \label{chap4:sec:equatorialQuantumCloningProblem}

Recall that in the case $\mathcal{H} \coloneqq \mathbb{C}^2$, the pure quantum states can be written,
\begin{equation*}
    \frac{1}{2}(I_2 + r_x \cdot \sigma_x + r_y \cdot \sigma_y + r_z\cdot \sigma_z),
\end{equation*}
with $r \coloneqq (r_x, r_y, r_z) \in \mathbb{R}^3$ and $\norm{r} = 1$, and thus are isomorphic to the unit sphere in $\mathbb{R}^3$. The \emph{equatorial pure quantum states} are the pure quantum states located on the $x$ -- $y$ equator of this sphere, i.e. $r_z = 0$. The equatorial pure quantum states are of the form,
\begin{equation*}
    \frac{e^{i \theta_0} \ket{0} + e^{i \theta_1} \ket{1}}{\sqrt{2}},
\end{equation*}
for some $\theta \in [0,2\pi)^2$.

\begin{remark*}
    The $4$ states used in the BB84 protocol \cite{bennett1984quantum} are all in the $x$ -- $z$ equator, and the two equators $x$ -- $y$ and $x$ -- $z$ are connected by a change of basis.
\end{remark*}

When the subset $\Gamma$ of the pure quantum states is the set of states of the form,
\begin{equation*}
    \frac{1}{\sqrt{d}} \sum^{d-1}_{k=0} e^{i \theta_k} \ket{k},
\end{equation*}
for some $\theta \in [0,2\pi)^d$, the subgroup $G$ of $\mathrm{U}_d$ can be chosen to be the group of monomial matrices in $[0,2\pi)$. The quantum cloning problem becomes the \emph{equatorial quantum cloning problem}. Then from \hyperref[chap4:prop:twirledChoiMatrixCommutationRelation]{Proposition~\ref*{chap4:prop:twirledChoiMatrixCommutationRelation}}, the partially transposed Choi matrix $C^\Tpartial_{\widetilde{\Phi}}$ of a twirled quantum cloning channel $\widetilde{\Phi}$ commutes with all the diagonal unitary matrices $U^{\otimes (N+1)}$ and all the tensor representation of permutations $\psi(\sigma)^{\otimes (N+1)}$, i.e. it is in the commutant of the algebra
\begin{equation*}
    \Span_{\mathbb{C}} \set{U^{\otimes n} \cdot \phi(\sigma)^{\otimes n}}{U \in \text{diag.}~\mathrm{U}_d \text{ and } \sigma \in \mathfrak{S}_d}.
\end{equation*}

From \hyperref[chap2:thm:commutantSchurWeylDuality]{Theorem~\ref*{chap2:thm:commutantSchurWeylDualityUn}}, the partially transposed Choi matrix $C^\Tpartial_{\widetilde{\Phi}}$ is in the algebra
\begin{equation*}
    \Span_{\mathbb{C}} \set{\psi(p)}{p \in \mathbb{U}_n}.
\end{equation*}
That is, it exists a family of complex numbers $c_p$ indexed by the uniform block permutations of $\mathbb{U}_{N+1}$ such that
\begin{equation*}
    C_{\widetilde{\Phi}} = \sum_{p \in \mathbb{U}_{N+1}} c_p \cdot \psi(p)^\Tpartial.
\end{equation*}

From \hyperref[chap4:prop:twirledChoiMatrixCommutationRelation]{Proposition~\ref*{chap4:prop:twirledChoiMatrixCommutationRelation}}, the partially transposed Choi matrix $C_{\widetilde{\Phi}_i}$ of each marginal $\widetilde{\Phi}_i$ has the form
\begin{align*}
    C_{\widetilde{\Phi}_i} &= \alpha_i \cdot \psi {\1( 1\,3\:|\:2\,4 \1)}^\Tpartial + \beta_i \cdot \psi {\1( 1\,4\:|\:2\,3 \1)}^\Tpartial + \gamma_i \cdot \psi {\1( 1\,2\,3\,4 \1)}^\Tpartial \\[0.5em]
    &= \alpha_i \cdot \psi {\left(

    \right),
\end{align*}
for some complex numbers $\alpha_i, \beta_i$ and $\gamma_i$. For all $\rho \in \Gamma$, the marginals $\widetilde{\Phi}_i$ on $\rho$ are
\begin{align*}
    \widetilde{\Phi}_i(\rho) &= \Tr_{\{0\}} \1[ C_{\widetilde{\Phi}_i} \1( \rho^\T \otimes I_d \1) \1] \\
    &= \Tr_{\{0\}} \3[ \2( \alpha_i \cdot \psi {\1( 1\,3\:|\:2\,4 \1)}^\Tpartial + \beta_i \cdot \psi {\1( 1\,4\:|\:2\,3 \1)}^\Tpartial + \gamma_i \cdot \psi {\1( 1\,2\,3\,4 \1)}^\Tpartial \2) \1( \rho \otimes I_d \1) \3] \\
    &= \alpha_i \cdot I_d + \beta_i \cdot \rho + \gamma_i \cdot \mathrm{diag} (\rho_{11}, \ldots, \rho_{dd}) \\
    &= \alpha_i \cdot I_d + \beta_i \cdot \rho + \gamma_i \cdot I_d,
\end{align*}
and the fidelity $F \1( \widetilde{\Phi}_i(\rho), \rho \1)$ becomes,
\begin{align*}
    F \1( \widetilde{\Phi}_i(\rho), \rho \1) &= \Tr \2[ \1( \widetilde{\Phi}_i(\rho) \1) \rho \2] \\
    &= \Tr \3[ C_{\widetilde{\Phi}_i} \2( \1( \rho^\T \cdot \rho  \1) \otimes I_d \2) \3] \\
    &= \Tr \2[ C_{\widetilde{\Phi}_i} \1( \rho \otimes I_d \1) \2] \\
    &= \Tr \3[ \2( \alpha_i \cdot \psi {\1( 1\,3\:|\:2\,4 \1)}^\Tpartial + \beta_i \cdot \psi {\1( 1\,4\:|\:2\,3 \1)}^\Tpartial + \gamma_i \cdot \psi {\1( 1\,2\,3\,4 \1)}^\Tpartial \2) \1( \rho \otimes I_d \1) \3] \\
    &= d \cdot \alpha_i + \beta_i + \gamma_i.
\end{align*}

\subsection{Average vs. worst fidelity} \label{chap4:sec:averageVsWorstFidelity}

In \hyperref[chap4:sec:universalQuantumCloningProblem]{Section~\ref*{chap4:sec:universalQuantumCloningProblem}} and \hyperref[chap4:sec:equatorialQuantumCloningProblem]{Section~\ref*{chap4:sec:equatorialQuantumCloningProblem}}, the marginals $\widetilde{\Phi}_i$ and $\widetilde{\Psi}_i$ of a twirled quantum cloning channel $\widetilde{\Phi}$ for the universal quantum cloning problem and a twirled quantum cloning channel $\widetilde{\Psi}$ for the equatorial quantum cloning problem, were shown to be, on all $\rho \in \Gamma$, some linear combinations of the trace, identity maps. Using the trace preserving property of quantum channels, the marginals becomes on all $\rho \in \Gamma$, some affine combinations
\begin{equation*}
    \widetilde{\Phi}_i(\rho) = p \cdot \rho + (1 - p) \frac{I_d}{d} \qquad \text{ and } \qquad \widetilde{\Psi}_i(\rho) = q \cdot \rho + (1 - q) \frac{I_d}{d},
\end{equation*}
for some $p,q \in \mathbb{R}$, and with fidelities $F \1( \widetilde{\Phi}_i(\rho), \rho \1)$ and $F \1( \widetilde{\Psi}_i(\rho), \rho \1)$, 
\begin{equation*}
    F \1( \widetilde{\Phi}_i(\rho), \rho \1) = p  + \frac{1 - p}{d} \qquad \text{ and } \qquad F \1( \widetilde{\Psi}_i(\rho), \rho \1) = q + \frac{1 - q}{d},
\end{equation*}
In particular, none of the fidelities depend on the quantum state. Hence both optimization problems,
\begin{align*}
    &\textit{(average)} & \sup_{\Phi~\textsc{cptp}} \sum^N_{i = 1} a_i \cdot \operatornamewithlimits{\mathbb{E}}_{\rho \in \Gamma} \2[ F \1( \Phi_i(\rho), \rho \1) \2] \\
    &\textit{(worst)} & \sup_{\Phi~\textsc{cptp}} \sum^N_{i = 1} a_i \cdot \inf_{\rho \in \Gamma} \2[ F \1( \Phi_i(\rho), \rho \1) \2],
\end{align*}
are equal on universal and equatorial $\Gamma$'s. However, the expectation value will be more convenient to manipulate, especially for to determine an upper bound.

\section{Upper bound} \label{chap4:sec:upperBound}

This section is dedicated to the determination of an upper bound for both the universal and equatorial quantum cloning problems, given a direction vector $a \in [0, 1]^N$. For all quantum cloning channels $\Phi: \mathcal{M}_d \to {\1( \mathcal{M}_d \1)}^{\otimes N}$, and for all subset $\Gamma$ of the pure quantum states,
\begin{align*}
    \sum^N_{i = 1} a_i \cdot \operatornamewithlimits{\mathbb{E}}_{\rho \in \Gamma} \2[ F \1( \Phi_i(\rho), \rho \1) \2] &= \sum^N_{i = 1} a_i \cdot \operatornamewithlimits{\mathbb{E}}_{\rho \in \Gamma} \3[ \Tr \2[ \1( \Phi_i(\rho) \1) \rho \2] \3] \\
    &= \sum^N_{i = 1} a_i \cdot \operatornamewithlimits{\mathbb{E}}_{\rho \in \Gamma} \3[ \Tr \2[ \1( \Phi(\rho) \1) \1( \rho_{(i)} \otimes I^{\otimes (N-1)}_d \1) \2] \3] \\
    &= \sum^N_{i = 1} a_i \cdot \operatornamewithlimits{\mathbb{E}}_{\rho \in \Gamma} \3[ \Tr \2[ C_{\Phi} \1( \rho^\T_{(0)} \otimes \rho_{(i)} \otimes I^{\otimes (N-1)}_d \1) \2] \3] \\
    &= \sum^N_{i = 1} a_i \cdot \Tr \3[ C_{\Phi} \2( \operatornamewithlimits{\mathbb{E}}_{\rho \in \Gamma} \1[ \rho^\T_{(0)} \otimes \rho_{(i)} \1] \otimes I^{\otimes (N-1)}_d \2) \3].
\end{align*}

\begin{theorem} \label{chap4:thm:universalQuantumCloningProblemUpperBound}
    For any direction vector $a \in [0, 1]^N$ the universal quantum cloning problem is upper bounded by
    \begin{equation*}
        \sup_{\Phi~\textsc{cptp}} \sum^N_{i = 1} a_i \cdot \operatornamewithlimits{\mathbb{E}}_{\rho \in \Gamma} \2[ F \1( \Phi_i(\rho), \rho \1) \2] \leq \frac{\lambda_{\text{max}}(R_a)}{d + 1},
    \end{equation*}
    where $\lambda_{\text{max}}(R_a)$ is the largest eigenvalue of the matrix
    \begin{equation*}
        R_a \coloneqq \sum^N_{i=1} a_i \cdot \1( d^2 \cdot \mathrm{I}_{(0,i)} + d \cdot \omega_{(0,i)} \1) \otimes I^{\otimes (N-1)}_d,
    \end{equation*}
    with quantum states $\mathrm{I}_{(0,i)}$ and $ \omega_{(0,i)}$, respectively maximally mixed and maximally entangled, between between the $0$-th and $i$-th quantum systems.
\end{theorem}
\begin{proof}
    Let $\emph{\vee^N \mathbb{C}^d}$ be the \emph{symmetric subspace} of ${\1( \mathbb{C}^d \1)}^{\otimes N}$ defined by
    \begin{equation*}
        \vee^N \mathbb{C}^d \coloneqq \Span_{\mathbb{C}} \set{v^{\otimes N}}{v \in \mathbb{C}^d}.
    \end{equation*}
    It is well known \cite{harrow2013church} that $\vee^N \mathbb{C}^d$ is an irreducible representation vector space for the representation $U \mapsto U^{\otimes N}$ of the unitary group $\mathrm{U}_d$. Since the pure quantum states $\Gamma$ are generated by the unitary matrices, i.e. $\Gamma \simeq \set{U \rho U^*}{U \in \mathrm{U}_d}$ for any pure quantum state $\rho$, then
    \begin{align*}
        \operatornamewithlimits{\mathbb{E}}_{\rho \in \Gamma} \1[ \rho^\T_{(0)} \otimes \rho_{(i)} \1] &= {\2( \operatornamewithlimits{\mathbb{E}}_{\rho \in \Gamma} \1[ \rho_{(0)} \otimes \rho_{(i)} \1] \1] \2)}^\Tpartial \\
        &= {\3( \int_\Gamma \rho \otimes \rho ~\mathrm{d}\rho \3)}^\Tpartial \\
        &= {\3( \int_{\mathrm{U}_d} U \ketbra{0}{0} U^* \otimes U \ketbra{0}{0} U^* ~\mathrm{d}U \3)}^\Tpartial.
    \end{align*}
    Note that the integral, before taking the partial transpose, commutes with all the unitary matrices $U \otimes U$ and lives in $\mathrm{End} \1( \vee^2 \mathbb{C}^d \1)$, by Schur's Lemma it must be a multiple of the identity in $\vee^2 \mathbb{C}^d$. The identity of the bipartite symmetric subspace is
    \begin{equation*}
        \frac{\psi\1( (1)(2) \1) + \psi\1( (1\:2) \1)}{2}.
    \end{equation*}
    The unit trace condition, together with the partial transpose, give
    \begin{align*}
        \operatornamewithlimits{\mathbb{E}}_{\rho \in \Gamma} \1[ \rho \otimes \rho \1] &= \frac{\psi {\1( (1)(2) \1)}^\Tpartial + \psi {\1( (1\:2) \1)}^\Tpartial}{d (d + 1)} \\
        &= \frac{d^2 \cdot \mathrm{I}_{(0,1)} + d \cdot \omega_{(0,1)}}{d (d + 1)}.
    \end{align*}
    Finally the universal quantum cloning problem becomes
    \begin{align*}
        &\sup_{\Phi~\textsc{cptp}} \sum^N_{i = 1} a_i \cdot \operatornamewithlimits{\mathbb{E}}_{\rho \in \Gamma} \2[ F \1( \Phi_i(\rho), \rho \1) \2] \\
        &= \sup_{\Phi~\textsc{cptp}} \sum^N_{i = 1} a_i \cdot \Tr \3[ C_{\Phi} \2( \operatornamewithlimits{\mathbb{E}}_{\rho \in \Gamma} \1[ \rho^\T_{(0)} \otimes \rho_{(i)} \1] \otimes I^{\otimes (N-1)}_d \2) \3] \\
        &= \sup_{\Phi~\textsc{cptp}} \frac{\Tr \1[ C_\Phi R_a \1)}{d (d + 1)}.
    \end{align*}
    Then from the inequality $\Tr [CR] \leq \Tr [C] \cdot \lambda_\text{max}(R)$ that holds for any positive semidefinite matrix $C$ and symmetric matrix $R$, and the equality $\Tr \1[ C_\Phi \1] = d$ for any Choi matrix $C_\Phi$ of a quantum channel $\Phi: \mathcal{M}_d \to \mathcal{M}_{d^\prime}$,
    \begin{align*}
        \sup_{\Phi~\textsc{cptp}} \frac{\Tr \1[ C_\Phi R_a \1]}{d (d + 1)} &\leq \sup_{\Phi~\textsc{cptp}} \frac{\Tr \1[ C_\Phi \1]}{d (d + 1)} \lambda_\text{max}(R_a) \\
        &= \frac{\lambda_\text{max}(R_a)}{d + 1}.
    \end{align*}
\end{proof}

\begin{remark*}
    The upper bound in \hyperref[chap4:thm:universalQuantumCloningProblemUpperBound]{Theorem~\ref*{chap4:thm:universalQuantumCloningProblemUpperBound}} is a special case of the result of Jaromír Fiurášek on the extremal equation for optimal completely-positive maps \cite{fiuravsek2001extremal}.
\end{remark*}

The spectrum of the matrix $R_a$ has been considered in a recent series of papers for the port-based teleportation protocol \cite{studzinski2017port,mozrzymas2018optimal,leditzky2022optimality}. In particular, in \cite[Lemma 3.6]{christandl2021asymptotic}, all the eigenvalues of the operator $R_\alpha$, up to shift factor, are given in the special case of $a = \frac{1}{N}(1, \ldots, 1)$.

\begin{theorem} \label{chap4:thm:equatorialQuantumCloningProblemUpperBound}
    For any direction vector $a \in [0, 1]^N$ the equatorial quantum cloning problem is upper bounded by
    \begin{equation*}
        \sup_{\Phi~\textsc{cptp}} \sum^N_{i = 1} a_i \cdot \operatornamewithlimits{\mathbb{E}}_{\rho \in \Gamma} \2[ F \1( \Phi_i(\rho), \rho \1) \2] \leq \frac{\lambda_{\text{max}}(R_a)}{d},
    \end{equation*}
    where $\lambda_{\text{max}}(R_a)$ is the largest eigenvalue of the matrix
    \begin{equation*}
        R_a \coloneqq \sum^N_{i=1} a_i \cdot \1( d^2 \cdot \mathrm{I}_{(0,i)} + d \cdot \omega_{(0,i)} - d \cdot \mathrm{X}_{(0,i)} \1) \otimes I^{\otimes (N-1)},
    \end{equation*}
    with quantum states $\mathrm{I}_{(0,i)}$ and $ \omega_{(0,i)}$, respectively maximally mixed and maximally entangled, between between the $0$-th and $i$-th quantum systems, and quantum state $\mathrm{X}_{(0,i)}$ defined by
    \begin{equation*}
        \mathrm{X} \coloneqq \frac{1}{d} \sum^{d-1}_{i=0} \ketbra{ii}{ii},
    \end{equation*}
    between between the $0$-th and $i$-th quantum systems.
\end{theorem}
\begin{proof}
    Let $\ket{+}$ be the pure quantum state defined by $\ket{+} \coloneqq \frac{1}{\sqrt{d}} \sum^{d-1}_{i=0} \ket{i}$, then for any pure quantum state $\ket{\psi} \in \Gamma$, i.e., of the form $\frac{1}{\sqrt{d}} \sum^{d-1}_{k=0} e^{i \theta_k} \ket{k}$, there exists a diagonal unitary matrix $U$ in $\text{diag.}~\mathrm{U}_d$ such that $\ket{\psi} = U \ket{+}$. Then
    \begin{align*}
        \operatornamewithlimits{\mathbb{E}}_{\rho \in \Gamma} \1[ \rho^\T_{(0)} \otimes \rho_{(i)} \1] &= \int_\Gamma \rho^\T \otimes \rho ~\mathrm{d}\rho \\
        &= \int_{\substack{\text{diag.} \\ \mathrm{U}_d}} \bar{U} \ketbra{+}{+} U^\T \otimes U \ketbra{+}{+} U^* ~\mathrm{d}U.
    \end{align*}
    Using integration over random diagonal unitary matrices \cite{nechita2021graphical},
    \begin{equation*}
        \int_{\substack{\text{diag.} \\ \mathrm{U}_d}} \bar{U} \ketbra{+}{+} U^\T \otimes U \ketbra{+}{+} U^* ~\mathrm{d}U = \frac{d^2 \cdot \mathrm{I}_{(0,1)} + d \cdot \omega_{(0,1)} - d \cdot \mathrm{X}_{(0,1)}}{d^2}.
    \end{equation*}
    Which leads to the upper bound,
    \begin{equation*}
        \sup_{\Phi~\textsc{cptp}} \sum^N_{i = 1} a_i \cdot \operatornamewithlimits{\mathbb{E}}_{\rho \in \Gamma} \2[ F \1( \Phi_i(\rho), \rho \1) \2] \leq \frac{\lambda_{\text{max}}(R_a)}{d}.
    \end{equation*}
\end{proof}

\begin{remark*}
    For both the upper bounds in \hyperref[chap4:thm:universalQuantumCloningProblemUpperBound]{Theorem~\ref*{chap4:thm:universalQuantumCloningProblemUpperBound}} and \hyperref[chap4:thm:equatorialQuantumCloningProblemUpperBound]{Theorem~\ref*{chap4:thm:equatorialQuantumCloningProblemUpperBound}}, the identity terms in the matrices $R_a$ only lead to a shift in the largest eigenvalues $\lambda_{\text{max}}(R_a)$.
\end{remark*}

\section{Universal \texorpdfstring{$1 \longrightarrow 2$}{1 to 2} quantum cloning problem} \label{chap4:sec:universal1to2QuantumCloningProblem}

This section is devoted to the universal $1 \to 2$ quantum cloning problem, in the general case of quantum systems of arbitrary dimension \cite{nechita2021geometrical}. The problem is studied from the following perspective: given some target pair of fidelities $\emph{(f_1, f_2)}$, does there exist a quantum cloning channel $\Phi$ from $\mathcal{M}_d$ to ${\1( \mathcal{M}_d \1)}^{\otimes 2}$, such that $F \1( \Phi_i(\rho), \rho \1) = f_i$ for all pure quantum state $\rho$. As seen in \hyperref[chap4:sec:averageVsWorstFidelity]{Section~\ref*{chap4:sec:averageVsWorstFidelity}}, by twirling the quantum cloning channel $\Phi$, the marginals becomes for all pure quantum states $\rho$,
\begin{equation*}
    \widetilde{\Phi}_1(\rho) = p_1 \cdot \rho + (1 - p_1) \frac{I_d}{d} \qquad \text{ and } \qquad \widetilde{\Phi}_2(\rho) = p_2 \cdot \rho + (1 - p_2) \frac{I_d}{d},
\end{equation*}
for some $p_1,p_2 \in \mathbb{R}$. Therefore, the following transformations rules for $\widetilde{\Phi}$ hold: 
\begin{align*}
    f_1 &= p_1  + \frac{(1 - p_1)}{d} & f_2 &= p_2 + \frac{(1 - p_2)}{d} \\
    p_1 &= \frac{d f_1 - 1}{d - 1} & p_2 &= \frac{d f_2 - 1}{d - 1}
\end{align*}

The main result is expressed as the following: the \emph{achievable fidelity region}, defined as
\begin{equation*}
    \set{(f_1, f_2)}{\exists \Phi: \mathcal{M}_d \xrightarrow{\textsc{cptp}} {\1( \mathcal{M}_d \1)}^{\otimes 2} \text{ such that } \operatornamewithlimits{\mathbb{E}}_{\rho \in \mathrm{U}_d} \2[ F \1( \Phi_i(\rho), \rho \1) \2] = f_i},
\end{equation*}
is a \emph{union of ellipses}, with the optimal one coming from a restricted class of quantum cloning channel (see \hyperref[chap4:fig:achievableFidelityRegion]{Figure~\ref*{chap4:fig:achievableFidelityRegion}}).
\begin{theorem*}[\cite{nechita2021geometrical}]
    The achievable fidelity region for the universal $1 \to 2$ quantum cloning problem is the union of a family of ellipses indexed by $\lambda \in (0,d]$, given by
    \begin{equation*}
        \qquad \frac{x^2}{a^2_\lambda} + \frac{(y - c_\lambda)^2}{b^2_\lambda} \leq 1,
    \end{equation*}
    with $a_\lambda \coloneqq \frac{\lambda}{\sqrt{d^2 - 1}}, b_\lambda \coloneqq \frac{\lambda}{d^2 - 1}$ and $c_\lambda \coloneqq \frac{\lambda d - 2}{d^2 - 1}$. The parameters $x$ and $y$ can be expressed as,
    \begin{equation*}
        \begin{cases}
            x &= \frac{d(f_1-f_2) }{d-1} \\
            y &= \frac{d(f_1+f_2) - 2}{d-1}
        \end{cases}
        \qquad \text{ or } \qquad
        \begin{cases}
		      x &=  p_1 - p_2 \\
            y &=  p_1 + p_2.
        \end{cases}
    \end{equation*}
    The optimal quantum cloning channels correspond to $\lambda = d$. 
\end{theorem*}

\begin{figure}[h]
    \centering

    \caption
        [The achievable fidelity region in orange of the universal $1 \to 2$ quantum cloning problem is the union of a continuous family of ellipses in blue.]
        {The achievable fidelity region \tikz{\draw[fill = orange!50!white, opacity = 0.5] (0,0) rectangle (1.5ex,1.5ex);} of the universal $1 \to 2$ quantum cloning problem is the union of a continuous family of ellipses \tikz{\draw[fill = blue!30!white, opacity = 0.5] (0,0) rectangle (1.5ex,1.5ex);}.} \label{chap4:fig:achievableFidelityRegion}
\end{figure}

In \hyperref[chap4:sec:universalQuantumCloningProblem]{Section~\ref*{chap4:sec:universalQuantumCloningProblem}}, the Choi matrix $C_{\widetilde{\Phi}}$ of a twirled quantum cloning channel $\widetilde{\Phi}$ from $\mathcal{M}_d$ to ${\1( \mathcal{M}_d \1)}^{\otimes 2}$, was shown to be a linear combination of the $6$ partially transposed tensor representations of the symmetric group $\mathfrak{S}_3$:
\begin{equation*}
    C_{\widetilde{\Phi}} = c_1 \cdot \psi \left(

    \right),
\end{equation*}
for complex numbers $c_1, \ldots, c_6 \in \mathbb{C}$. Such linear combination is the Choi matrix of a quantum channel if both the positivity and the trace conditions of \hyperref[chap3:thm:quantumChannelStructure]{Theorem~\ref*{chap3:thm:quantumChannelStructure}} are satified.

By taking the corresponding partial traces, the two marginals of $\widetilde{\Phi}$ becomes on all pure quantum states $\rho$,
\begin{align*}
    \widetilde{\Phi}_1(\rho) = \underbrace{(d c_2 + c_5 + c_6)}_{p_1} \rho + \underbrace{(d^2 c_1 + d c_3 + d c_4)}_{1-p_1} \frac{I_d}{d} \\
    \widetilde{\Phi}_2(\rho) = \underbrace{(d c_3 + c_5 + c_6)}_{p_2} \rho + \underbrace{(d^2 c_1 + d c_2 + d c_4)}_{1-p_2} \frac{I_d}{d}.
\end{align*}
The partially transposed tensor representations $
\psi \1(

\1)
$ corresponding to the two cycles $(1 \: 2 \: 3)$ and $(3 \: 2 \: 1)$ of $\mathfrak{S}_3$ contribute, with coefficients $c_5$ and $c_6$, to both the $p_1$ and $p_2$ of $\widetilde{\Phi}_1$ and $\widetilde{\Phi}_2$. Hewever, due to the no-cloning \hyperref[chap4:thm:noCloning]{Theorem~\ref*{chap4:thm:noCloning}}, a Choi matrix composed only of the two cycles $(1 \: 2 \: 3)$ and $(3 \: 2 \: 1)$ of $\mathfrak{S}_3$ cannot be the Choi matrix of a quantum channel, since no linear combination,
\begin{equation*}
    \alpha \cdot \psi \left(
    %
    \right)
\end{equation*}
is positive semidefinite. Hence, contributions from the other $4$ partially transposed tensor representations are needed to ensure the positivity of the Choi matrix $C_{\widetilde{\Phi}}$. The trace condition implies that
\begin{equation*}
    d^2 c_1 + d (c_2 + c_3 + c_4) + c_5 + c_6 = 1.
\end{equation*}

In the following, the tensor representation function $\psi$ is dropped, and instead the partially transposed permutations are depicted using the graphical calculus from \hyperref[chap3:sec:graphicalCalculus]{Section~\ref*{chap3:sec:graphicalCalculus}}. The focus now shifts to the problem of characterizing the positivity of the Choi matrix $C_{\widetilde{T}}$ using the coefficients $c_i$. To achieve this, let the vectors $u_i$ and $v_i$ of be defined for all $i \in \{0, \ldots, d-1\}$ by,
\begin{equation*}
    u_i \coloneqq \sqrt{d} \cdot \ket{\Omega}_{(0,1)} \otimes \ket{i} \qquad \text{ and } \qquad v_i \coloneqq \sqrt{d} \cdot \ket{\Omega}_{(0,2)} \otimes \ket{i},
\end{equation*}
depicted as
\begin{equation*}
    u_i = 
.
\end{equation*}
Then, the action of the partially transposed permutations on these vectors is given by,
\begin{equation} \label{chap4:eq:uiAction}
    \begin{aligned}
        %
.
    \end{aligned}
\end{equation}

The vectors $u_i + v_i$ and $u_i - v_i$ are eigenvectors of $
%
$. Since they are all $d$-rank matrices, their complete (nonzero) spectrums is known:
\begin{align*}
    \Spec \2(
    %
    \2) &=
    \begin{cases}
        +(d + 1) &\times d \\
        -(d - 1) &\times d.
    \end{cases}
\end{align*}
The other two partially transposed permutations $
%
$, are unitary matrices and thus have full rank.

In the next two Sections the positivity of the Choid matrix $C_{\widetilde{\Phi}}$ is characterize in terms of the coefficients $c_i$, first by restricting to the first $4$, and then considering the general case. 

\subsection{Restricted quantum cloning channels}

In this Section, the universal $1 \to 2$ quantum cloning problem will be solved when the Choi matrix $C_{\widetilde{\Phi}}$ is a linear combination of only $4$ partially transposed permutations of $\mathfrak{S}_3$, that is,
\begin{equation*}
    C_{\widetilde{\Phi}} = c_1 \cdot
.
\end{equation*}
Since $C_{\widetilde{T}}$ must be positive semidefinite and in particular Hermitian, the following must hold: $c_1, c_2 \in \mathbb{R}$ and $c_3 = \bar{c}_4$. That is;
\begin{equation*}
    C_{\widetilde{\Phi}} = \alpha \cdot
    %
,
\end{equation*}
such that the trace condition becomes $d (\alpha + \beta) + 2 \Re(\gamma) = 1$. Using \hyperref[chap4:eq:uiAction]{Eq.~(\ref*{chap4:eq:uiAction})} and \hyperref[chap4:eq:viAction]{Eq.~(\ref*{chap4:eq:viAction})},
The $4$ partially transposed permutations can be block diagonalized in the basis of the $2 d$ vectors $u_i$ and $v_i$, i.e.,
\begin{equation*}
    {\2(
    %
\end{align*}
such that each partially transposed permutation operators is a direct sum of~$d$ such blocks and~$d^3-2d$ zero blocks, e.g.,
\begin{equation*}
    %
    }^{\oplus d}.
\end{equation*}
Then the block diagonal decomposition of $C_{\widetilde{\Phi}}$ becomes
\begin{equation*}
    {\1( C_{\widetilde{\Phi}} \1)}_i =
    %
.
\end{equation*}

It is well known that a $2 \times 2$ hermitian matrix $M$ is positive semidefinite if and only if $\det (M) \geq 0$ and $\Tr [M] \geq 0$. Then $C_{\widetilde{\Phi}}$ is positive semi-definite if and only if each of its $2 \times 2$ blocks, in its block diagonal decomposition, are is positive semidefinite. That is
\begin{equation*}
    \Tr \2[ {\1( C_{\widetilde{\Phi}} \1)}_i \2] = d (\alpha + \beta) + 2  \Re(\gamma) \geq 0
\end{equation*}
and
\begin{equation*}
     \det \2( {\1( C_{\widetilde{\Phi}} \1)}_i \2) = \alpha \beta - |\gamma|^2  \geq 0.
\end{equation*}

The first condition is always true since $d (\alpha + \beta) + 2 \Re(\gamma) = 1$. Finally, the Choi matrix $C_{\widetilde{\Phi}}$ is the Choi matrix of a quantum channel, and thus a quantum cloning channel, when both
\begin{equation*}
    d (\alpha + \beta) + 2  \Re(\gamma) = 1 \qquad \text{ and } \qquad \alpha \beta \geq |\gamma|^2.
\end{equation*}

The two marginals of $\widetilde{\Phi}$ becomes on all pure quantum states $\rho$,
\begin{align*}
    \widetilde{\Phi}_1(\rho) = \underbrace{\1( d \alpha + 2  \Re(\gamma) \1)}_{p_1} \rho + \underbrace{d \beta}_{1-p_1}  \frac{I_d}{d} \\
    \widetilde{\Phi}_2(\rho) = \underbrace{\1( d \beta + 2  \Re(\gamma) \1)}_{p_2} \rho + \underbrace{d \alpha}_{1-p_2} \frac{I_d}{d}.
\end{align*}

\begin{theorem*}
    The achievable fidelity region for the restricted universal $1 \to 2$ quantum cloning problem is the ellipse given by:
    \begin{equation*}
        \qquad \frac{x^2}{a^2} + \frac{(y - c)^2}{b^2} \leq 1,
    \end{equation*}
    with $a \coloneqq \frac{d}{\sqrt{d^2 - 1}}, b \coloneqq \frac{d}{d^2 - 1}$ and $c \coloneqq \frac{d^2 - 2}{d^2 - 1}$. The parameters $x$ and $y$ can be expressed as,
    \begin{equation*}
        \begin{cases}
            x &= \frac{d(f_1-f_2) }{d-1} \\
            y &= \frac{d(f_1+f_2) - 2}{d-1}
        \end{cases}
        \qquad \text{ or } \qquad
        \begin{cases}
		      x &=  p_1 - p_2 \\
            y &=  p_1 + p_2.
        \end{cases}
    \end{equation*}
    Equivalently, the achievable fidelity region for the restricted universal $1 \to 2$ quantum cloning problem is the set,
    \begin{equation*}
        \set{(p_1, p_2) \in \left[\frac{-1}{d^2-1}, 1\right]}{\frac{(1 - p_1)(1 - p_2)}{d^2} \geq {\3( \frac{p_1 + p_2 - 1}{2} \3)}^2}.
    \end{equation*}     
\end{theorem*}
\begin{proof}
	A pair $(p_1, p_2) \in  \1[ \frac{-1}{d^2-1}, 1 \1]$ is in the achievable fidelity region for the restricted universal $1 \to 2$ quantum cloning problem if and only if there exists coefficients $\alpha, \beta \in \mathbb{R}$ and $\gamma \in \mathbb{C}$ satisfying both,
    \begin{equation*}
        d (\alpha + \beta) + 2  \Re(\gamma) = 1 \qquad \text{ and } \qquad \alpha \beta \geq |\gamma|^2,
    \end{equation*}
    with,
    \begin{align*}
        p_1  &= d \alpha + 2 \Re(\gamma)  & p_2 &= d \beta + 2 \Re(\gamma) \\
        &= 1 - d \beta & &= 1 - d \alpha.
    \end{align*}
 
	Such a complex number $\gamma$ exists if and only if,
    \begin{equation*}
        {\3( \frac{1-d(\alpha+\beta)}{2} \3)}^2 \leq \alpha \beta.
    \end{equation*}
    Rewriting the this inequality in terms of $p_1$ and $p_2$ yields,
    \begin{equation*}
        \frac{(1 - p_1)(1 - p_2)}{d^2} \geq {\3( \frac{p_1 + p_2 - 1}{2} \3)}^2.
    \end{equation*}
\end{proof}

\begin{figure*}[h]
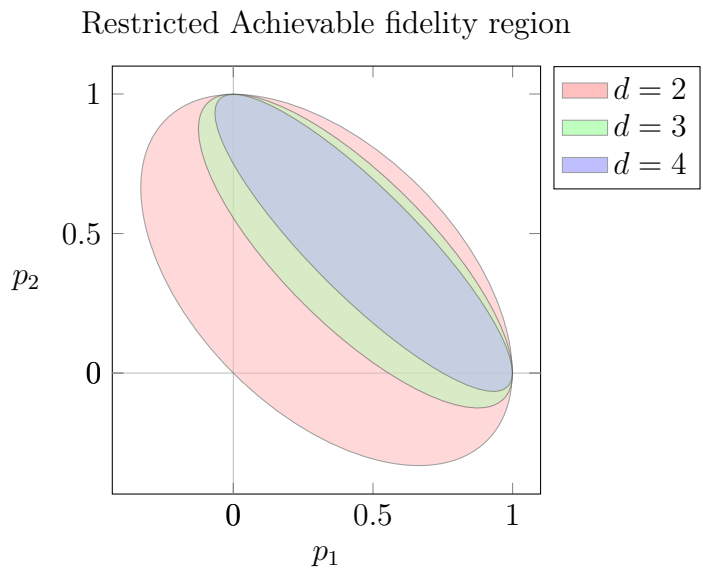

    \centering

    \caption
        {The achievable fidelity region for the restricted universal $1 \to 2$ quantum cloning problem is an ellipse.}
\end{figure*}

\subsection{General quantum cloning channels}

In this Section, the general case of the universal $1 \to 2$ quantum cloning problem will be solved when the Choi matrix $C_{\widetilde{\Phi}}$ is a linear combination of all the $6$ partially transposed permutations of $\mathfrak{S}_3$, that is,
\begin{equation*}
    C_{\widetilde{\Phi}} = \alpha \cdot
    %
.
\end{equation*}
The hermitian condition on $C_{\widetilde{\Phi}}$ imposes that $\alpha, \beta, \varepsilon_1, \varepsilon_2 \in \mathbb{R}$, and the trace condition reads $d (\alpha + \beta) + 2  \Re(\gamma) + d^2 \varepsilon_1 + d \varepsilon_2 = 1$. As in the previous Section, the action of $C_{\widetilde{\Phi}}$ decomposes on with respect to the the $2d$ vectors $u_i$ and $v_i$.

Let $\mathbf{V}$ be the complex span of the $2d$ vectors $u_i$ and $v_i$:
\begin{equation*}
    \mathbf{V} \coloneqq \Span_{\mathbb{C}} \2\{
    %
    \2\},
\end{equation*}
then $\mathbf{V} \subset \mathbb{C}^d \otimes \vee^2 \mathbb{C}^d$, and is invariant by the two partially transposed permutations $
%
$ that have both full rank. On $\mathbf{V}^\perp$, the spectrum of $\varepsilon_1 \cdot
%
    \2) =
    \begin{cases}
        \varepsilon_1 + \varepsilon_2 &\times d \frac{d (d + 1)}{2} - 2 d \\
        \varepsilon_1 - \varepsilon_2 &\times d \frac{d (d - 1)}{2}.
    \end{cases}
\end{equation*}
Then the complete block diagonal decomposition of $C_{\widetilde{\Phi}}$ is made of $2 \times 2$ and $1 \times 1$ blocks. The $1 \times 1$ blocks $(\varepsilon_1 + \varepsilon_2)$ and $(\varepsilon_1 - \varepsilon_2)$ are positive when $\varepsilon_1 \geq |\varepsilon_2|$. On $\boldsymbol{V}$ the block diagonalization of $C_{\widetilde{\Phi}}$ becomes,
\begin{equation*} 
    {\1( C_{\widetilde{\Phi}} \1)}_i =
.
\end{equation*}
After the change of basis
\begin{align*}
    %
    \2),
\end{align*}
the block diagonal decomposition becomes the Hermitian,
\begin{equation*}
    {\1( C_{\widetilde{\Phi}} \1)}_i =
    %
.
\end{equation*}

The two marginals of $\widetilde{\Phi}$ becomes on all pure quantum states $\rho$,
\begin{align*}
    \widetilde{\Phi}_1(\rho) = \underbrace{\1( d \alpha + 2  \Re(\gamma) \1)}_{p_1} \rho + \underbrace{)d \beta + d^2 \varepsilon_1 + d \varepsilon_2)}_{1-p_1}  \frac{I_d}{d} \\
    \widetilde{\Phi}_2(\rho) = \underbrace{\1( d \beta + 2  \Re(\gamma) \1)}_{p_2} \rho + \underbrace{(d \alpha + d^2 \varepsilon_1 + d \varepsilon_2)}_{1-p_2} \frac{I_d}{d}.
\end{align*}

\begin{theorem*}
    The achievable fidelity region for the universal $1 \to 2$ quantum cloning problem is the union of a family of ellipses indexed by $\lambda \in (0,d]$, given by
    \begin{equation*}
        \qquad \frac{x^2}{a^2_\lambda} + \frac{(y - c_\lambda)^2}{b^2_\lambda} \leq 1,
    \end{equation*}
    with $a_\lambda \coloneqq \frac{\lambda}{\sqrt{d^2 - 1}}, b_\lambda \coloneqq \frac{\lambda}{d^2 - 1}$ and $c_\lambda \coloneqq \frac{\lambda d - 2}{d^2 - 1}$. The parameters $x$ and $y$ can be expressed as,
    \begin{equation*}
        \begin{cases}
            x &= \frac{d(f_1-f_2) }{d-1} \\
            y &= \frac{d(f_1+f_2) - 2}{d-1}
        \end{cases}
        \qquad \text{ or } \qquad
        \begin{cases}
		      x &=  p_1 - p_2 \\
            y &=  p_1 + p_2.
        \end{cases}
    \end{equation*}
\end{theorem*}
\begin{proof}
    Setting $x \coloneqq p_1 - p_2$ and $y \coloneqq p_1 + p_2$, together with the relation $d (\alpha + \beta) + 2  \Re(\gamma) + d^2 \varepsilon_1 + d \varepsilon_2 = 1$ yields,
    \begin{align*}
        {\1( C_{\widetilde{\Phi}} \1)}_i &=
        \begin{pmatrix}
            \frac{(d^2 - 1) y + d (d - 1) \1( (d^2 - 2) \varepsilon_1 + d \varepsilon_2 - 1 \1) + 2}{2 d} & \frac{\sqrt{d^2 - 1}}{2 d} x \\[1em]
            \frac{\sqrt{d^2 - 1}}{2 d} x & - \frac{(d^2 - 1) y + d (d + 1) \1( (d^2 - 2) \varepsilon_1 + d \varepsilon_2 - 1 \1) + 2}{2 d}
        \end{pmatrix}.
    \end{align*}
    Let $\lambda \coloneqq - d \1( (d^2 - 2) \varepsilon_1 + d \varepsilon_2 - 1 \1)$, then the block diagonal decomposition of $C_{\widetilde{\Phi}}$ reduces to 
    \begin{equation*}
        {\1( C_{\widetilde{T}} \1)}_i = \frac{1}{2 d}
        \begin{pmatrix}
            (d^2 - 1) y - (d - 1) \lambda + 2 & \sqrt{d^2 - 1} x \\[1em]
            \sqrt{d^2 - 1} x & - \1( (d^2 - 1) y - (d + 1) \lambda + 2 \1)
        \end{pmatrix}.
    \end{equation*}
    In this way, the positivity of each of the $2 \times 2$ blocks becomes
    \begin{equation*}
        0 \leq \lambda \leq d \qquad \text{ and } \qquad \frac{x^2}{a^2_\lambda} + \frac{(y - c_\lambda)^2}{b^2_\lambda} \leq 1,
    \end{equation*}
    with $a_\lambda \coloneqq \frac{\lambda}{\sqrt{d^2 - 1}}, b_\lambda \coloneqq \frac{\lambda}{d^2 - 1}$ and $c_\lambda \coloneqq \frac{\lambda d - 2}{d^2 - 1}$.
\end{proof}

\section{Universal \texorpdfstring{$1 \longrightarrow N$}{1 to N} quantum cloning problem}

This section covers the arbitrary universal $1 \to N$ quantum cloning problem, in the general case of quantum systems of any dimension \cite{nechita2022asymmetric}.

\subsection{Partially transposed permutations}

This section is devoted to the study of the partially transposed tensor representations $\psi(\sigma)^\Tpartial$ of the symmetric group $\mathfrak{S}_{N+1}$, that appear in the Choi matrix of the twirled quantum cloning channels. In this section, the symmetric group $\mathfrak{S}_{N+1}$ is the group of permutations of the set $\emph{\{ 0, 1, \ldots, N \}}$, starting from $0$.

\begin{remark*}
    In the Choi matrix of a twirled quantum cloning channel, the \textit{input} tensor corresponds to the set $\{ 0 \}$ of its partially transposed tensor representations, and the \textit{ouput} tensors correspond to the set $\{ 1, \ldots, N \}$.
\end{remark*}

Let $\sigma \in \mathfrak{S}_{N+1}$ such that $0$ is a fixed point of $\sigma$, i.e. $\sigma(0)$. Assume that $\psi(\sigma)^\Tpartial$ appears in the Choi matrix of a twirled quantum cloning channel, then $\sigma$ does not contribute to the performance of the quantum cloning channel. Indeed on all pure quantul states $\rho$,
\begin{equation*}
    \Tr_{\{0\}} \2[ \psi(\sigma)^\Tpartial \1( \rho^\T \otimes I^{\otimes N}_d\1) \2] = \Tr \1[ \rho^\T \1] \cdot \psi(\hat{\sigma}) = \psi(\hat{\sigma}),
\end{equation*}
where $\hat{\sigma}$ is the permutation of the symmetric group $\mathfrak{S}_N$ on $\{1, \ldots, N \}$, obtained from $\sigma$ by dropping $\{0\}$. Each marginals in thus a scalar multiple of the identity.

Let $\emph{\Sigma_{a,b}}$ be the subet of $\mathfrak{S}_{N+1}$ be defined for all $1 \leq a,b \leq N$ by,
\begin{equation*}
    \Sigma_{a,b} \coloneqq \set{\sigma \in \mathfrak{S}_{N+1}}{\sigma(0) = a \text{ and } \sigma^{\shortminus 1}(0) = b}.
\end{equation*}
This gives a partition of
\begin{equation*}
    \set{\sigma \in \mathfrak{S}_{N + 1}}{\sigma(0) \neq 0} = \bigcup_{1 \leq a,b \leq N} \Sigma_{a,b},
\end{equation*}
where each set $\Sigma_{a,b}$ contains $(N - 1)!$ permutations.

Note that for all $1 \leq a,b \leq N$ and for all $\sigma \in \Sigma_{a,b}$, there exists a unique $\hat{\sigma} \in \mathfrak{S}_{N - 1}$ such that the partially transposed tensor representation $\psi(\sigma)^\Tpartial$ decomposes into, 
\begin{equation} \label{chap4:eq:sigmaHat}
    \psi(\sigma)^\Tpartial = \psi \1( (1 \: a) \1) \cdot \1( d \cdot \omega_{(0,1)} \otimes \psi(\hat{\sigma}) \1) \cdot \psi \1( (1 \: b) \1),
\end{equation}
where $\1( d \cdot \omega_{(0,1)} \otimes \psi(\hat{\sigma}) \1)$ is the partially transposes tensor representation of a permutation in $\Sigma_{1,1}$.

\begin{lemma} \label{chap4:lem:partialTracePermutation}
    Let distinct $1 \leq a,b,c \leq N$, then
    \begin{align*}
        \Tr_{[N+1] \setminus \{0\}} \3[ \sum_{\sigma \in \Sigma_{a,b}} \psi(\sigma)^\Tpartial \3] &= \frac{1}{d} \sum_{\sigma \in \mathfrak{S}_{N-1}} \Tr \1[ \psi(\sigma) \1] \cdot I_d \\
        \Tr_{[N+1] \setminus \{0\}} \3[ \sum_{\sigma \in \Sigma_{c,c}} \psi(\sigma)^\Tpartial \3] &= \sum_{\sigma \in \mathfrak{S}_{N-1}} \Tr \1[ \psi(\sigma) \1] \cdot I_d.
    \end{align*}
\end{lemma}
\begin{proof}
    For the second equation, using the decomposition of \hyperref[chap4:eq:sigmaHat]{Eq.~(\ref*{chap4:eq:sigmaHat})},
    \begin{align*}
        &\Tr_{[N+1] \setminus \{0\}} \3[ \sum_{\sigma \in \Sigma_{c,c}} \psi(\sigma)^\Tpartial \3] \\
        &= \Tr_{[N+1] \setminus \{0\}} \3[ \sum_{\hat{\sigma} \in \mathfrak{S}_{N-1}} \psi \1( (1 \: c) \1) \cdot \1( d \cdot \omega_{(0,1)} \otimes \psi(\hat{\sigma}) \1) \cdot \psi \1( (1 \: c) \1) \3] \\ 
        &= \Tr_{[N+1] \setminus \{0\}} \3[ \sum_{\hat{\sigma} \in \mathfrak{S}_{N-1}} \sum^{d-1}_{i,j = 0} \psi \1( (1 \: c) \1) \cdot \1( \ketbra{ii}{jj} \otimes \psi(\hat{\sigma}) \1) \cdot \psi \1( (1 \: c) \1) \3] \\ 
        &= \Tr \3[ \sum_{\hat{\sigma} \in \mathfrak{S}_{N-1}} \sum^{d-1}_{i,j = 0} \psi \1( (0 \: (c-1)) \1) \cdot \1( \ketbra{i}{j} \otimes \psi(\hat{\sigma}) \1) \cdot \psi \1( (0 \: (c-1)) \1) \3] \cdot \ketbra{i}{j} \\ 
        &= \sum_{\sigma \in \mathfrak{S}_{N-1}} \Tr \1[ \psi(\sigma) \1] \cdot I_d,
    \end{align*}
    where $\psi \1( (0 \: (c-1)) \1)$ is the tensor representations of the permutation $(0 \: (c-1)$ on $\{0, \ldots, (N-1)\}$. For the first equation, since the partial transpose is a linear operator,
    \begin{equation*}
        \Tr_{[N+1] \setminus \{0\}} \3[ \sum_{\sigma \in \Sigma_{a,b}} \psi(\sigma)^\Tpartial \3] = {\4( \Tr_{[N+1] \setminus \{0\}} \3[ \sum_{\sigma \in \Sigma_{a,b}} \psi(\sigma) \3] \4)}^\Tpartial.
    \end{equation*}
    For any $\sigma \in \mathfrak{S}_{N+1}$, the partial trace of the tensor representation $\psi(\sigma)$ is a multiple of the identity, i.e., 
    \begin{equation*}
        \Tr_{[N+1] \setminus \{0\}} \1[ \psi(\sigma) \1] = c \cdot I_d,
    \end{equation*}
    with $c = \frac{1}{d} \Tr \1[ \psi(\sigma) \1]$. Then
    \begin{equation*}
        {\4( \Tr_{[N+1] \setminus \{0\}} \3[ \sum_{\sigma \in \Sigma_{a,b}} \psi(\sigma) \3] \4)}^\Tpartial = \frac{1}{d} \Tr \3[ \sum_{\sigma \in \Sigma_{a,b}} \psi(\sigma) \3] \cdot I_d.
    \end{equation*}
    For a permutation $\sigma \in \mathfrak{S}_{N+1}$, let $\# \sigma$ denotes the number of disjoint cycles of $\sigma$. Then
    \begin{align*}
        \Tr \3[ \sum_{\sigma \in \Sigma_{a,b}} \psi(\sigma) \3] &= \sum_{\sigma \in \Sigma_{a,b}} d^{\# \sigma} \\
        &= \sum_{\sigma \in \Sigma_{a,a}} d^{\# [\sigma \circ (a \: b)]}.
    \end{align*}
    Let distinct $a,b \in \{0, \ldots, N\}$ and a permutation $\sigma \in \mathfrak{S}_{N+1}$ with its decomposition into disjoint cycles $\sigma = c_1 \circ \cdots \circ c_k$. If there exists $i \in \{1, \ldots, k\}$ such that both $a,b \in c_i$, then the permutation $c_i \circ (a \: b)$ can be decomposed into two disjoint cycles. Otherwise, if there exist distinct $i,j \in \{1, \ldots, k\}$ such that $a \in c_i$ and $b \in c_j$, then the permutation $c_i \circ c_j \circ (a \: b)$ can be decomposed into only one disjoint cycle. Finally
    \begin{equation*}
        \# \1[ \sigma \circ (a \: b) \1] =
        \begin{cases}
            \# \sigma + 1 &\text{if } \exists i \in \{1, \ldots, k\} \text{ s.t. } a,b \in c_i \\
            \# \sigma - 1 &\text{otherwise}
        \end{cases}.
    \end{equation*}
    But when $\sigma \in \Sigma_{a,a}$ and since $b \neq a$, in the decomposition of $\sigma$ into disjoint cycles, $a$ is in the cycle $(0 \: a)$, and $\# \1[ \sigma \circ (a \: b) \1] = \# \sigma - 1$. That is
    \begin{equation*}
        \Tr \3[ \sum_{\sigma \in \Sigma_{a,b}} \psi(\sigma) \3] = \frac{1}{d} \Tr \3[ \sum_{\sigma \in \Sigma_{a,a}} \psi(\sigma) \3].
    \end{equation*}
    Using the first equation,
    \begin{equation*}
        \Tr_{[N+1] \setminus \{0\}} \3[ \sum_{\sigma \in \Sigma_{a,b}} \psi(\sigma)^\Tpartial \3] = \frac{1}{d} \sum_{\sigma \in \mathfrak{S}_{N-1}} \Tr \1[ \psi(\sigma) \1] \cdot I_d.
    \end{equation*}
\end{proof}

The next lemma establishes the relationship between a Choi matrix of a partially transposed tensor representation $\psi(\sigma)^\Tpartial$, and its corresponding quantum channel.
\begin{lemma} \label{chap4:lem:ChoiMatrixQuantumChannelPermutation}
    Let some $1 \leq a, b \leq N$ and $\sigma \in \Sigma_{a,b}$, then there exist $\mu, \nu \in \mathfrak{S}_N$ satisfying $\mu(0) = a - 1$ and $\nu^{\shortminus 1}(0) = b - 1$, such that $\psi(\sigma)^\Tpartial$ is the Choi matrix of the linear map,
    \begin{equation*}
        X \mapsto \psi(\mu) \cdot \1( X \otimes I^{\otimes (N - 1)}_d \1) \cdot \psi(\nu).
    \end{equation*}
\end{lemma}
\begin{proof}
    Using the decomposition of \hyperref[chap4:eq:sigmaHat]{Eq.~(\ref*{chap4:eq:sigmaHat})},
    \begin{equation*}
        \psi(\sigma)^\Tpartial = \psi \1( (1 \: a) \1) \cdot \1( d \cdot \omega_{(0,1)} \otimes \psi(\hat{\sigma}) \1) \cdot \psi \1( (1 \: b) \1),
    \end{equation*}
    for some unique permutation $\hat{\sigma} \in \mathfrak{S}_{N-1}$. Then,
    \begin{align*}
        &\Tr_{\{0\}} \2[ \psi(\sigma)^\Tpartial \1( X^\T \otimes I^{\otimes N}_d \1) \2] \\
        &= \Tr_{\{0\}} \2[ \psi \1( (1 \: a) \1) \cdot \1( d \cdot \omega_{(0,1)} \otimes \psi(\hat{\sigma}) \1) \cdot \psi \1( (1 \: b) \1) \1( X^\T \otimes I^{\otimes N}_d \1) \2] \\
        &= \Tr_{\{0\}} \3[ \sum^{d-1}_{i,j = 0} \psi \1( (1 \: a) \1) \cdot \1( \ketbra{ii}{jj} \otimes \psi(\hat{\sigma}) \1) \cdot \psi \1( (1 \: b) \1) \1( X^\T \otimes I^{\otimes N}_d \1) \3] \\
        &= \sum^{d-1}_{i,j = 0} \bra{i} X \ket{j} \cdot \psi \1( (0 \: (a-1)) \1) \cdot \1( \ketbra{i}{j} \otimes \psi(\hat{\sigma}) \1) \cdot \psi \1( (0 \: (b-1)) \1) \\
        &= \psi \1( (0 \: (a-1)) \1) \cdot \1( X \otimes \psi(\hat{\sigma}) \1) \cdot \psi \1( (0 \: (b-1)) \1).
    \end{align*}
    Thus, by setting $\psi(\mu) \coloneqq \psi \1( (0 \: (a-1)) \1)$ and $\psi(\nu) \coloneqq \1( I_d \otimes \psi(\hat{\sigma}) \1) \cdot \psi \1( (0 \: (b-1)) \1)$, the result holds.
\end{proof}

Recall that the upper bound of the universal quantum cloning problem for a direction vector $a$, from \hyperref[chap4:sec:upperBound]{Section~\ref*{chap4:sec:upperBound}}, is given as the largest eigenvalue of the the matrix $\emph{R_a}$ defined by,
\begin{equation*}
    R_a \coloneqq \sum^N_{i=1} a_i \cdot \1( d^2 \cdot \mathrm{I}_{(0,i)} + d \cdot \omega_{(0,i)} \1) \otimes I^{\otimes (N-1)}_d.
\end{equation*}
The eigenvectors of $R_a$ are the same as those of the matrix $\emph{S_a}$ defined by,
\begin{equation*}
    S_a \coloneqq \sum^N_{i=1} a_i \cdot \1( d \cdot \omega_{(0,i)} \1) \otimes I^{\otimes (N-1)}.
\end{equation*}
\begin{lemma} \cite{nechita2022asymmetric} \label{chap4:lem:largestEigenvector}
    The normalized largest eigenvectors of $S_a$ are of the form
    \begin{equation*}
        \chi = \sum^N_{i=1} b_i \cdot \1( \sqrt{d} \cdot \ket{\Omega}_{(0,i)} \1) \otimes \ket{v},
    \end{equation*}
    for some vector $\ket{v}$ in the symmetric subspace $\vee^{(N - 1)} \mathbb{C}^d$, and some positive real numbers $b_i$ satisfy the equation,
    \begin{equation*}
        (d - 1) \sum^N_{i = 1} b^2_i + \3( \sum^N_{i = 1} b_i \3)^2 = 1.
    \end{equation*}
    The largest eigenvalue becomes $\lambda_{\text{max}} = \sum^N_{i = 1} a_i {\2( (d - 1) b_i + \sum^N_{j = 1} b_j \2)}^2$.
\end{lemma}
\begin{remark*}
    Note that the positive real numbers $b_i$ depend on the direction vector $a$.
\end{remark*}

\subsection{Optimal symmetric quantum cloning channels}

The \emph{symmetric} universal $1 \to N$ quantum cloning problem is a special case of the quantum cloning problem, where all the marginals $\Phi_i$ of the quantum cloning channel are asked to be equal on all pure quantum states $\rho$, i.e.,
\begin{equation*}
    \Phi_i(\rho) = p \cdot \rho + (1 - p) \frac{I_d}{d},
\end{equation*}
where $p$ does not depend on the choice of the marginal $\Phi_i$.

\begin{theorem*}[\cite{keyl1999optimal}] \label{chap4:thm:optimalSymmetricQuantumCloningChannel}
    The optimal quantum cloning channel $\Phi_{\text{opt}}$ from $\mathcal{M}_d$ to ${\1( \mathcal{M}_d \1)}^{\otimes N}$, for the symmetric universal $1 \to N$ quantum cloning problem, is defined on all pure quantum states $\rho$ by
    \begin{equation*}
        \Phi_{\text{opt}}(\rho) \coloneqq \frac{d}{\Tr \1[ P^+_{\mathfrak{S}_N} \1]} P^+_{\mathfrak{S}_N} \1( \rho \otimes I^{\otimes (N - 1)}_d \1) P^+_{\mathfrak{S}_N},
    \end{equation*}
    where $P^+_{\mathfrak{S}_N}$ is the orthogonal projector onto the symmetric subspace $\vee^N \mathbb{C}^d$, defined by
    \begin{equation*}
        P^+_{\mathfrak{S}_N} \coloneqq \frac{1}{N!} \sum_{\sigma \in \mathfrak{S}_N} \psi(\sigma).
    \end{equation*}
    Then each marginal ${\1( \Phi_{\text{opt}} \1)}_i$ is equal, on all pure quantum states $\rho$, to,
    \begin{equation*}
        {\1( \Phi_{\text{opt}} \1)}_i(\rho) = \underbrace{\frac{d + N}{N (d + 1)}}_{p_{\text{opt}}} \rho + \underbrace{\3( 1 - \frac{d + N}{N (d + 1)} \3)}_{1 - p_{\text{opt}}} \frac{I_d}{d}.
    \end{equation*}
\end{theorem*}

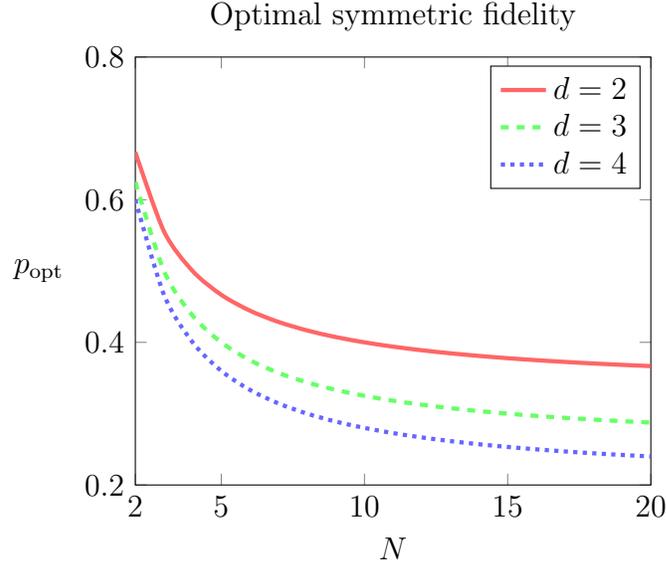
\begin{figure*}[h]
    \centering
    \begin{tikzpicture}
        \begin{axis}[title = {Optimal symmetric fidelity},
                     xlabel = $N$,
                     ylabel style = {anchor = north},
                     ylabel = $p_{\text{opt}}$,
                     ylabel style = {rotate = -90, anchor = east},
                     xmin = 2,
                     xmax = 20,
                     ymin = 0.2,
                     ymax = 0.8,
                     xtick = {2, 5, 10, 15, 20},
                     ]
            \addplot[domain = 2:20, 
                    samples = 19, 
                    color = red!60!white,
                    smooth,
                    ultra thick]
                    {(2 + x) / (x * (2 + 1))};
            \addlegendentry{$d = 2$}
            \addplot[domain = 2:20, 
                    samples = 19, 
                    color = green!60!white,
                    dashed,
                    smooth,
                    ultra thick]
                    {(3 + x) / (x * (3 + 1))};
            \addlegendentry{$d = 3$}
            \addplot[domain = 2:20, 
                    samples = 19, 
                    color = blue!60!white,
                    dotted,
                    smooth,
                    ultra thick]
                    {(4 + x) / (x * (4 + 1))};
            \addlegendentry{$d = 4$}
        \end{axis}
    \end{tikzpicture}
    \caption{The optimal fidelity for the symmetric universal $1 \to N$ quantum cloning problem.}
\end{figure*}

\subsection{Optimal asymmetric quantum cloning channels}

The \emph{asymmetric} universal $1 \to N$ quantum cloning problem is the general case of the quantum cloning problem, where the marginals $\Phi_i$ of the quantum cloning channel can be arbitrary.

In \hyperref[chap4:sec:universal1to2QuantumCloningProblem]{Section~\ref*{chap4:sec:universal1to2QuantumCloningProblem}}, the main technical difficulty was the positivity condition of the Choi matrix of the quantum cloning channels. Let $\Phi: \mathcal{M}_d \to {\1( \mathcal{M}_d \1)}^{\otimes N}$ be is a linear map defined by,
\begin{equation*}
    \Phi(X) \coloneqq A \1( X \otimes I^{\otimes (N - 1)}_d \1) A^*,
\end{equation*}
for some matrix $A \in \mathcal{M}_{d^N}$. Let $\Psi: \mathcal{M}_{d^N} \to \mathcal{M}_{d^N}$ be a linear map defined by $\Psi(X) \coloneqq A X A^*$, then,
\begin{equation*}
    \Psi \1( x \otimes I^{\otimes (N - 1)}_d \1) = \Phi(X).
\end{equation*}
Through this conjugacy  form $\Psi$ and in particular $\Phi$ are both completely positive.

\begin{theorem} \label{chap4:thm:optimalAsymmetricQuantumCloningChannel}
    Given a direction vector $a$, let $\chi = \sum^N_{i=1} b_i \cdot \1( \sqrt{d} \cdot \ket{\Omega}_{(0,i)} \1) \otimes \ket{v}$ be normalized largest eigenvectors of $S_a$. The optimal quantum cloning channel $\Phi^a_{\text{opt}}$ from $\mathcal{M}_d$ to ${\1( \mathcal{M}_d \1)}^{\otimes N}$, for the asymmetric universal $1 \to N$ quantum cloning problem in the direction $a$, is defined on all pure quantum states $\rho$ by
    \begin{equation*}
        \Phi^a_{\text{opt}}(\rho) \coloneqq \frac{d N (N + d - 1)}{\Tr \1[ P^+_{\mathfrak{S}_N} \1]} P^a_{\mathfrak{S}_N} \1( \rho \otimes I^{\otimes (N - 1)}_d \1) {\1( P^a_{\mathfrak{S}_N} \1)}^*,
    \end{equation*}
    where,
    \begin{equation*}
        P^a_{\mathfrak{S}_N} \coloneqq \frac{1}{N!} \sum_{\sigma \in \mathfrak{S}_N} b_{\sigma(0) + 1} \cdot \psi(\sigma).
    \end{equation*}
\end{theorem}
\begin{proof}
    The complete positivity is given by the conjugacy form of the quantum channel $\Phi^a_{\text{opt}}$.

    For the trace preserving property, the Choi matrix $C_{\Phi^a_{\text{opt}}}$ of $\Phi^a_{\text{opt}}$ will be first determined. For all pure quantum states $\rho$,
    \begin{align*}
        \Phi^a_{\text{opt}}(\rho) &= \frac{d N (N + d - 1)}{\Tr \1[ P^+_{\mathfrak{S}_N} \1]} P^a_{\mathfrak{S}_N} \1( \rho \otimes I^{\otimes (N - 1)}_d \1) {\1( P^a_{\mathfrak{S}_N} \1)}^* \\
        &= \frac{d N (N + d - 1)}{(N!)^2 \cdot \Tr \1[ P^+_{\mathfrak{S}_N} \1]} \sum_{\sigma,\tau \in \mathfrak{S}_N} \1( b_{\sigma(0)+1} \cdot b_{\tau^{\shortminus 1}(0)+1} \1) \cdot \psi(\sigma) \1( \rho \otimes I^{\otimes (N - 1)}_d \1) \psi(\tau).
    \end{align*}
    As a consequence of \hyperref[chap4:lem:ChoiMatrixQuantumChannelPermutation]{Lemma~\ref*{chap4:lem:ChoiMatrixQuantumChannelPermutation}}, for any $1 \leq a,b \leq N$, the matrix 
    \begin{equation*}
        (N-1)! \sum_{\sigma \in \Sigma_{a,b}} \psi(\sigma)^\Tpartial,
    \end{equation*}
    is the Choi matrix of the linear map
    \begin{equation*}
        X \mapsto \sum_{\substack{\mu,\nu \in \Sigma_{a,b} \\ \mu(0) = a - 1 \\ \nu^{\shortminus 1}(0) = b - 1}} \psi(\mu) \cdot \1( X \otimes I^{\otimes (N - 1)}_d \1) \cdot \psi(\nu).
    \end{equation*}
    This implies that the Choi matrix $C_{\Phi^a_{\text{opt}}}$ of $\Phi^a_{\text{opt}}$ is,
    \begin{equation*}
        C_{\Phi^a_{\text{opt}}} = \frac{d (N + d - 1)}{N! \cdot \Tr \1[ P^+_{\mathfrak{S}_N} \1]} \sum_{\substack{1 \leq a,b \leq N \\ \sigma \in \Sigma_{a,b}}} b_a b_b \cdot \psi(\sigma)^\Tpartial.
    \end{equation*}
    Finaly, using \hyperref[chap4:lem:partialTracePermutation]{Lemma~\ref*{chap4:lem:partialTracePermutation}}, the equation $(d - 1) \sum^N_{i = 1} b^2_i + \1( \sum^N_{i = 1} b_i \1)^2 = 1$ satisfied by the positive real numbers $b_i$, and the relation $\Tr \1[ P^+_{\mathfrak{S}_{N - 1}} \1] = \frac{N}{N + d - 1} \Tr \1[ P^+_{\mathfrak{S}_{N}} \1]$, the trace condition on the Choi matrix $C_{\Phi^a_{\text{opt}}}$ holds, i.e.,
    \begin{equation*}
        \Tr_{[N+1] \setminus \{0\}} \1[ C_{\Phi^a_{\text{opt}}} \1] = I_d.
    \end{equation*}
    Hence $\Phi^a_{\text{opt}}$ is trace preserving, and thus a quantum channel.

    It remains to prove that the optimal quantum cloning channel $\Phi^a_{\text{opt}}$ saturates the upper bound,
    \begin{equation*}
        \sum^N_{i = 1} a_i \cdot \operatornamewithlimits{\mathbb{E}}_{\rho \in \Gamma} \2[ F \1( {(\Phi^a_{\text{opt}})}_i(\rho), \rho \1) \2] \leq \frac{\lambda_{\text{max}}(R_a)}{d + 1}.
    \end{equation*}
    For all marginals ${(\Phi^a_{\text{opt}})}_i$, the average fidelity $\operatornamewithlimits{\mathbb{E}}_{\rho \in \Gamma} \2[ F \1( {(\Phi^a_{\text{opt}})}_i(\rho), \rho \1) \2]$ becomes,
    \begin{align*}
        \operatornamewithlimits{\mathbb{E}}_{\rho \in \Gamma} \2[ F \1( {(\Phi^a_{\text{opt}})}_i(\rho), \rho \1) \2] &= \operatornamewithlimits{\mathbb{E}}_{\rho \in \Gamma} \3[ \Tr \2[ C_{\Phi^a_{\text{opt}}} \1( \rho^\T_{(0)} \otimes \rho_{(i)} \otimes I^{\otimes (N-1)}_d \1) \2] \3] \\
        &= \frac{1}{d(d+1)} \Tr \3[ C_{\Phi^a_{\text{opt}}} \2( \1( d^2 \cdot \mathrm{I}_{(0,i)} + d \cdot \omega_{(0,i)} \1) \otimes I^{\otimes (N-1)}_d \2) \3] \\
        &= \frac{\Tr\2[ C_{\Phi^a_{\text{opt}}} \2] + \Tr \2[ C_{\Phi^a_{\text{opt}}} \1( d \cdot \omega_{(0,i)} \otimes I^{\otimes (N-1)}_d \1) \2]}{d(d+1)} \\
        &= \frac{d + \Tr \2[ C_{\Phi^a_{\text{opt}}} \1( d \cdot \omega_{(0,i)} \otimes I^{\otimes (N-1)}_d \1) \2]}{d(d+1)}
    \end{align*}
    Using the form of the Choi matrix $C_{\Phi^a_{\text{opt}}}$,
    \begin{align*}
        C_{\Phi^a_{\text{opt}}} \1( d \cdot \omega_{(0,i)} \otimes I^{\otimes (N-1)}_d \1) &= \frac{d (N + d - 1)}{N! \cdot \Tr \1[ P^+_{\mathfrak{S}_N} \1]} \sum_{\substack{1 \leq a,b \leq N \\ \sigma \in \Sigma_{a,b}}} b_a b_b \cdot \2( \psi(\sigma)^\Tpartial \cdot \psi \1( (0 \: i) \1)^\Tpartial \2) \\
        &= \frac{d (N + d - 1)}{N! \cdot \Tr \1[ P^+_{\mathfrak{S}_N} \1]} \sum_{\substack{1 \leq a \leq N \\ \sigma \in \Sigma_{a,i}}} b_a \3( (d - 1) b_i + \sum_{1 \leq b \leq N} b_b \3) \psi(\sigma)^\Tpartial.
    \end{align*}
    Such that taking the trace yields, using \hyperref[chap4:lem:partialTracePermutation]{Lemma~\ref*{chap4:lem:partialTracePermutation}},
    \begin{equation*}
        \Tr \2[ C_{\Phi^a_{\text{opt}}} \1( d \cdot \omega_{(0,i)} \otimes I^{\otimes (N-1)}_d \1) \2] = d {\3( (d - 1) b_i + \sum^N_{j = 1} b_j \3)}^2,
    \end{equation*}
    and finally the average fidelity becomes,
    \begin{equation*}
        \operatornamewithlimits{\mathbb{E}}_{\rho \in \Gamma} \2[ F \1( {(\Phi^a_{\text{opt}})}_i(\rho), \rho \1) \2] = \frac{1 + {\1( (d - 1) b_i + \sum^N_{j = 1} b_j \1)}^2}{d+1}.
    \end{equation*}
    But since
    \begin{equation*}
        \lambda_{\text{max}}(S_a) = \bra{\chi} S_a \ket{\chi} = \sum^N_{i = 1} a_i {\3( (d - 1) b_i + \sum^N_{j = 1} b_j \3)}^2,
    \end{equation*}
    then the equality $\sum^N_{i = 1} a_i \cdot \operatornamewithlimits{\mathbb{E}}_{\rho \in \Gamma} \2[ F \1( {(\Phi^a_{\text{opt}})}_i(\rho), \rho \1) \2] = \frac{\lambda_{\text{max}}(R_a)}{d + 1}$ holds.
\end{proof}

\subsection{Optimal quantum cloning channel necessary condition}

Let $a$ be a direction vector and $\Phi^a_{\text{opt}}$ be the optimal quantum cloning channel for the universal $1 \to N$ quantum cloning problem in the direction $a$. For all marginals ${(\Phi^a_{\text{opt}})}_i$, the average fidelity $f_i \coloneqq \operatornamewithlimits{\mathbb{E}}_{\rho \in \Gamma} \2[ F \1( {(\Phi^a_{\text{opt}})}_i(\rho), \rho \1) \2]$ is
\begin{equation*}
    f_i = \frac{1 + {\1( (d - 1) b_i + \sum^N_{j = 1} b_j \1)}^2}{d+1}.
\end{equation*}
Summing over all $i$ yields,
\begin{equation*}
    \sum^N_{i = 1} \sqrt{(d + 1) f_i  - 1} = (N + d - 1) \sum^N_{i = 1} b_i.
\end{equation*}
Using the equation $(d - 1) \sum^N_{i = 1} b^2_i + \1( \sum^N_{i = 1} b_i \1)^2 = 1$ satisfied by the positive real numbers $b_i$,
\begin{align*}
    \sum^N_{i = 1} \1( (d + 1) f_i - 1 \1) &= \sum^N_{i = 1} {\3( (d - 1) b_i + \sum^N_{j = 1} b_j \3)}^2 \\
    &= (d - 1) + (N + d - 1) {\3( \sum^N_{i = 1} b_i \3)}^2 \\
    &= (d - 1) + \frac{{\2( \sum^N_{i = 1} \sqrt{(d + 1) f_i  - 1} \2)}^2}{N + d - 1}.
\end{align*}
Finally from the relation $f_i = p_i + \frac{1 - p_i}{d}$, the necessary condition for the $p_i$'s of the optimal quantum cloning channels becomes,
\begin{equation} \label{chap4:eq:necessaryConditionOptimalQuantumCloningChannel}
    N + (d^2 - 1) \sum^N_{i = 1} p_i = d (d - 1) + \frac{{\2( \sum^N_{i=1} \sqrt{(d^2 - 1) p_i + 1} \2)}^2}{N + d - 1}.
\end{equation}

\begin{figure*}[h]
    \centering

    \caption
        {The necessary condition \hyperref[chap4:eq:necessaryConditionOptimalQuantumCloningChannel]{Eq.~(\ref*{chap4:eq:necessaryConditionOptimalQuantumCloningChannel})} of the optimal quantum cloning channels for the universal $1 \to 2$ quantum cloning problem.}
\end{figure*}

\subsection{The \texorpdfstring{$\mathcal{Q}$}{Q}-norm}

Let $x \in \mathbb{R}^N$, and define its $\emph{\mathcal{Q}}$-norm by
\begin{equation*}
    \norm{x}_{\mathcal{Q}} \coloneqq \frac{d \lambda_{\text{max}}(S_x) - \norm{x}_1}{d^2 - 1},
\end{equation*}
where $\norm{x}_1 = \sum_{i=1}^N |x_i|$ is the $\ell_1$ norm of the vector $x$ and the matrix $S_x \in \mathcal M_{d^{N+1}}$ is given by, 
\begin{equation*}
    S_x \coloneqq \sum^N_{i=1} |x_i| \cdot \1( d \cdot \omega_{(0,i)} \1) \otimes I^{\otimes (N-1)}_d.
\end{equation*}
Note that the $\mathcal{Q}$-norm depends implicitly on a parameter $d$, and for non-negative $x \in \mathbb{R}^N$, the matrix $S_x$ coincides with that of \hyperref[chap4:lem:largestEigenvector]{Lemma~\ref*{chap4:lem:largestEigenvector}}.

\begin{remark*}
   Note that on its own $ \lambda_{\text{max}}(S_x)$ defines a norm, but in general, except for trivial examples, subtracting two norms don't provide a new norm.
\end{remark*}

\begin{proposition}
	For all $N \geq 1$, the quantity $\norm{\cdot}_{\mathcal{Q}}$ is a norm on $\mathbb{R}^N$.
\end{proposition}
\begin{proof}
    The absolute homogeneity of $\norm{\cdot}_{\mathcal{Q}}$ comes from the fact that both $\lambda_{\text{max}}(S_x)$ and $\norm{x}_1$ are also absolutely homogeneous.

    Let $x \in \mathbb{R}^N$ such that $\norm{x}_{\mathcal{Q}} = 0$, then $d \cdot \lambda_{\text{max}}(S_x) = \norm{x}_1 = \Tr [S_x]$. In particular all eigenvalues of $S_x$ are equal to $\lambda_{\text{max}}(S_x)$, and hence
    \begin{equation*}
        S_x = \lambda_{\text{max}}(S_x) \cdot I^{\otimes (N+1)}_d,
    \end{equation*}
    with $\lambda_{\text{max}}(S_x) \geq 0$. Let distinct $i,j = \{0, \ldots, d-1\}$ and $\ket{\psi}$ defined by,
    \begin{equation*}
        \ket{\psi} = \ket{i \smash{\underbrace{j \cdots j}_{N \text{ times}}}}.
    \end{equation*}
    Then
    \begin{align*}
        \bra{\psi\mathclap{\phantom{\sum}}} \sum^N_{i=1} |x_i| \cdot \1( d \cdot \omega_{(0,i)} \1) \otimes I^{\otimes (N-1)}_d \ket{\mathclap{\phantom{\sum}}\psi} &= 0 \\
        \bra{\psi} \lambda_{\text{max}}(S_x) \cdot I^{\otimes (N+1)}_d \ket{\psi} &= \lambda_{\text{max}}(S_x).
    \end{align*}
    Finally $\lambda_{\text{max}}(S_x) = 0$ implies $x = 0$, and $\norm{\cdot}_{\mathcal{Q}}$ is positive definite.

    For the subadditivity of $\norm{\cdot}_{\mathcal{Q}}$, it is sufficient to prove that for any $x,y \in \mathbb{R}^N$ such that $\norm{x}_{\mathcal{Q}} \leq 1$ and $\norm{y}_{\mathcal{Q}} \leq 1$, then $\norm{\frac{x + y}{2}}_{\mathcal{Q}} \leq 1$. It is possible to assert, without loss of generality, that by multiplying the vectors by a scalar, the vector $x + y$ is in $\mathbb{R}^N_+$. Let $x^{\prime} \in \mathbb{R}^N_+$ be defined by
    \begin{equation*}
        x^{\prime} \coloneqq
        \begin{cases}
            \1[ x,\frac{x+y}{2} \1] \cap \partial \mathbb{R}^N_+ &\text{ if } \1[ x,\frac{x+y}{2} \1] \cap \partial \mathbb{R}^N_+ \neq \emptyset \\
            x &\text{ otherwise}
        \end{cases}
    \end{equation*}
    and similarly for $y^{\prime} \in \mathbb{R}^N_+$, where $\partial \mathbb{R}^N_+$ denote the boundary of $\mathbb{R}^N_+$. Then $\frac{x + y}{2} \in [x^{\prime} , y^{\prime}]$, and there exists $\lambda \in[0,1]$ such that $\frac{x+y}{2} = \lambda \cdot x^{\prime} + (1 - \lambda) \cdot y^{\prime}$. Then using,
    \begin{align*}
        &\lambda_{\text{max}} \3( \sum^{N}_{i = 1} \1( x^{\prime}_i + y^{\prime}_i \1) \cdot \1( d \cdot \omega_{(0,i)} \1) \otimes I^{\otimes (N-1)}_d \3) \\
        &\leq \lambda_{\text{max}} \3( \sum^{N}_{i = 1} x^{\prime}_i \cdot \1( d \cdot \omega_{(0,i)} \1) \otimes I^{\otimes (N-1)}_d \3) + \lambda_{\text{max}} \3( \sum^{N}_{i = 1} y^{\prime}_i \cdot \1( d \cdot \omega_{(0,i)} \1) \otimes I^{\otimes (N-1)}_d \3),
    \end{align*}
    and the absolute homogeneity of $\norm{\cdot}_{\mathcal{Q}}$,
    \begin{equation*}
        \norm[\Big]{\frac{x+y}{2}}_{\mathcal{Q}} \leq \lambda \cdot \norm{x^{\prime}}_{\mathcal{Q}} + (1 - \lambda) \cdot \norm{y^\prime}_{\mathcal{Q}} \leq \max \1(  \norm{x^{\prime}}_{\mathcal{Q}},  \norm{y^{\prime}}_{\mathcal{Q}} \1).
    \end{equation*}
    Now there exists $t_x \in [0,1]^N$ and $t_y \in [0,1]^N$ such that element-wise product equations $x^{\prime} = t_x \cdot x$ and $y^{\prime} = t_y \cdot y$ hold. It remains to prove that, for all $t \in [0,1]^N$ and $x \in \mathbb{R}^N_+$, the inequality $\norm{t \cdot x}_{\mathcal{Q}} \leq \norm{x}_{\mathcal{Q}}$ hold. With this both $\norm{x^{\prime}}_{\mathcal{Q}} \leq 1$ and $\norm{y^{\prime}}_{\mathcal{Q}} \leq 1$, such that $\norm{\frac{x+y}{2}}_{\mathcal{Q}} \leq 1$.

    Let $x,y \in \mathbb{R}^N_+$ such that $x_i \leq y_i$ for all $i \in \{1, \ldots, N\}$, then $\norm{x}_{\mathcal{Q}} \leq \norm{y}_{\mathcal{Q}}$. Indeed, let $\chi = \sum^N_{i=1} b_i \cdot \1( \sqrt{d} \cdot \ket{\Omega}_{(0,i)} \1) \otimes \ket{v}$ be normalized largest eigenvectors of $S_x$, from \hyperref[chap4:lem:largestEigenvector]{Lemma~\ref*{chap4:lem:largestEigenvector}}, then using that $\ket{v} \in \vee^{(N-1)} \mathbb{C}^d$ and than all the $b_i$ are positive real numbers,
    \begin{align*}
        &\bra{\chi} \1( d \cdot \omega_{(0,i)} \1) \otimes I^{\otimes (N-1)}_d \ket{\chi} \\
        &= \sum^N_{k,j = 1} (d^2 b_k b_l) \cdot \1( \bra{\Omega}_{(0,i)} \otimes \bra{v} \1) \1( \omega_{(0,i)} \otimes I^{\otimes (N-1)}_d \1) \1( \ket{\Omega}_{(0,i)} \otimes \ket{v} \1) \\
        &= {\3( (d - 1) b_i + \sum^N_{j=1} b_k \3)}^2 \\
        &\geq {\3( \sum^N_{j=1} b_k \3)}^2.
    \end{align*}
    Now, using the equation $(d - 1) \sum^N_{i = 1} b^2_i + \1( \sum^N_{i = 1} b_i \1)^2 = 1$ satisfied by the positive real numbers $b_i$,
    \begin{align*}
        1 &= \braket{\chi}{\chi} \\
        &= (d - 1) \sum^N_{i = 1} b^2_i + \3( \sum^N_{i = 1} b_i \3)^2 \\
        &\leq d \3( \sum^N_{i = 1} b_i \3)^2.
    \end{align*}
    this implies that $\bra{\chi} \1( d \cdot \omega_{(0,i)} \1) \otimes I^{\otimes (N-1)}_d \ket{\chi} \geq \dfrac{1}{d}$. This way since $x_i \leq y_i$ for all $i \in \{1, \ldots, N\}$, then $\lambda_{\text{max}}(S_y) \geq \bra{\chi} S_y \ket{\chi}$ and,
    \begin{align*}
        \lambda_{\text{max}}(S_y) - \lambda_{\text{max}}(S_x) & \bra{\chi} S_y - S_x \ket{\chi} \\
        &=\sum_{i=1}^N (y_i - x_i) \cdot \bra{\chi} \1( d \cdot \omega_{(0,i)} \1) \otimes I^{\otimes (N-1)}_d \ket{\chi} \\
        &\geq \sum_{i=1}^N (y_i - x_i) \frac{1}{d}.
    \end{align*}
    Then $ \norm{y}_{\mathcal{Q}} = \lambda_{\text{max}}(S_y) - \dfrac{1}{d} \norm{y}_1 \geq \lambda_{\text{max}}(S_x) - \dfrac{1}{d} \norm{x}_1 = \norm{x}_{\mathcal{Q}}$, which conclude the proof since $t \cdot y = x$ for some $t \in [0,1]^N$.
\end{proof}

With the help of the $\mathcal{Q}$-norm, the upper bound from \hyperref[chap4:thm:universalQuantumCloningProblemUpperBound]{Theorem~\ref*{chap4:thm:universalQuantumCloningProblemUpperBound}} can be reformulated as: for any direction vector $a \in [0, 1]^N$ the universal quantum cloning problem is upper bounded by
\begin{equation*}
    \sup_{\Phi~\textsc{cptp}} \sum^N_{i = 1} a_i \cdot \operatornamewithlimits{\mathbb{E}}_{\rho \in \Gamma} \2[ F \1( \Phi_i(\rho), \rho \1) \2] \leq \frac{1}{d} \norm{a}_1 + \1( 1 - \frac{1}{d} \1) \norm{a}_{\mathcal{Q}}.
\end{equation*}

Let the \emph{dual} $\mathcal{Q}$-norm, be defined on $y \in \mathbb{R}^N$ by,
\begin{equation*}
    \norm{y}^*_{\mathcal{Q}} \coloneqq \sup_{\substack{x \in \mathbb{R}^N \\ x \neq 0}} \frac{\scalar{y}{x}}{\norm{x}_\mathcal{Q}}.
\end{equation*}

\begin{figure}[h]
    \centering

    \caption
        [The $1$-dimensional unit spheres of the $\mathcal{Q}$-norm in red and the dual $\mathcal{Q}$-norm in green.]
        {The $1$-dimensional unit spheres of the $\mathcal{Q}$-norm \tikz{\draw[fill = red!60!white] (0,0) rectangle (1.5ex,1.5ex);} and the dual $\mathcal{Q}$-norm \tikz{\draw[fill = green!60!white] (0,0) rectangle (1.5ex,1.5ex);}.}
\end{figure}

\subsection{Achievable fidelity region}

In this Section, the characterization of the achievable fidelity region for the universal $1 \to N$ quantum cloning problem will be given geometrically, as the non-negative part of the unit ball of the dual $\mathcal{Q}$-norm.

Let $\emph{\mathcal{R}_{N,d}}$ be the achievable fidelity region for the universal $1 \to N$ quantum cloning problem, that is,
\begin{equation*}
    \mathcal{R}_{N,d} \coloneqq \set{p \in [0,1]^N}{\exists \Phi: \mathcal{M}_d \xrightarrow{\textsc{cptp}} {\1( \mathcal{M}_d \1)}^{\otimes N} \text{ s.t. } \Phi_i(\rho) = p_i \cdot \rho + (1 - p_i) \frac{I_d}{d}}.
\end{equation*}
\begin{theorem} \label{chap4:thm:achievableFidelityRegion}
    $\mathcal{R}_{N,d}$ is the non-negative part of the unit ball of the dual $\mathcal{Q}$-norm, i.e. $p \in [0,1]^N$ is in $\mathcal{R}_{N,d}$ if and only if $\norm{p}^*_{\mathcal{Q}} \leq 1$.
\end{theorem}
\begin{proof}
    The following more general result will be first proved. Let $X$ be a finite dimensional real vector space and let $X^*$ denote its dual space. Let $\norm{\cdot}_1$ be a norm on $X$, let $\norm{\cdot}_2$ be a norm on $X^*$, and define the dual norms as follows,
    \begin{align*}
        \norm{y}^*_1 &\coloneqq \sup_{\substack{x \in X \\ x \neq 0}} \frac{\scalar{y}{x}}{\norm{x}_1}, & &\forall y \in X^*, \\
        \norm{y}^*_2 &\coloneqq \sup_{\substack{x \in X^* \\ x \neq 0}} \frac{\scalar{y}{x}}{\norm{x}_2}, & &\forall y \in X.
    \end{align*}
    Provided that both conditions,
    \begin{align}
        \forall x \in X, \forall y \in X^*, &\quad \scalar{x}{y} \leq \norm{x}_1 \norm{y}_2 \label{chap4:eq:pairNormConditions1} \\
        \forall y \in X^*, \exists x \in X, &\quad \scalar{x}{y} = \norm{x}_1 \norm{y}_2, \label{chap4:eq:pairNormConditions2}
    \end{align}
    are satisfied, then $\forall x \in X, \norm{x}_1 = \norm{x}^*_2$. In particular the unit balls of $\norm{\cdot}_1$ and $\norm{\cdot}^*_2$ are equal. Indeed, from \hyperref[chap4:eq:pairNormConditions1]{Eq.~(\ref*{chap4:eq:pairNormConditions1})}, then $\forall x \in X$,
    \begin{equation*}
        \norm{x}_1 \geq \sup_{y \in X^*} \frac{\scalar{x}{y}}{\norm{y}_2} = \norm{x}^*_2,
    \end{equation*}
    and from \hyperref[chap4:eq:pairNormConditions2]{Eq.~(\ref*{chap4:eq:pairNormConditions2})}, then $\forall y \in X^*, \exists x \in X$,
    \begin{equation*}
        \norm{x}_1 = \frac{\scalar{x}{y}}{\norm{y}_2} \leq \sup_{y \in X^*} \frac{\scalar{x}{y}}{\norm{y}_2} \leq \norm{x}^*_2,
    \end{equation*}
    which conclude the proof of this first result.

    Now to prove the theorem, define the norm $\norm{\cdot}_{\mathcal{R}}$ on $x \in \mathbb{R}^N$ by,
    \begin{equation*}
        \norm{x}_{\mathcal{R}} \coloneqq \max \Big\{ t > 0 \: \big{|} \: t \cdot |x| \in \mathcal{R}_{N,d} \Big\},
    \end{equation*}
    where $t \cdot |x|$ denote the vector of $\mathbb{R}^N$ with coefficients $t \cdot |x_i|$. The norms $\norm{\cdot}_{\mathcal{R}}$ and $\norm{\cdot}_{\mathcal{Q}}$ are as in the settings above. If both conditions  \hyperref[chap4:eq:pairNormConditions1]{Eq.~(\ref*{chap4:eq:pairNormConditions1})} and \hyperref[chap4:eq:pairNormConditions2]{Eq.~(\ref*{chap4:eq:pairNormConditions2})}, i.e.
    \begin{align*}
        \forall p \in \mathbb{R}^N, \forall a \in \mathbb{R}^N, &\quad \scalar{p}{a} \leq \norm{p}_{\mathcal{R}} \norm{a}_{\mathcal{Q}} \\
        \forall a \in \mathbb{R}^N, \exists p \in \mathbb{R}^N, &\quad \scalar{p}{a} = \norm{p}_{\mathcal{R}} \norm{a}_{\mathcal{Q}},
    \end{align*}
    are satisfied, then the theorem holds.

    Let $p \in \mathcal{R}_{N,d}$ and a quantum cloning channel $\Phi_p : \mathcal{M}_d \to {\1( \mathcal{M}_d \1)}^{\otimes N}$ such that for all marginals ${(\Phi_p)}_i$ and all pure quantum states $\rho$,
    \begin{equation*}
        {(\Phi_p)}_i (\rho) = p_i \cdot \rho + (1 - p_i) \frac{I_d}{d}.
    \end{equation*}
    Then for all direction vectors $a \in [0,1]^N$,
    \begin{align*}
        \frac{1}{d} \norm{a}_1 + \1( 1 - \frac{1}{d} \1) \norm{a}_{\mathcal{Q}} &\geq \sup_{\Phi~\textsc{cptp}} \sum^N_{i = 1} a_i \cdot \operatornamewithlimits{\mathbb{E}}_{\rho \in \Gamma} \2[ F \1( \Phi_i(\rho), \rho \1) \2] \\
        &\geq \sum^N_{i = 1} a_i \cdot \operatornamewithlimits{\mathbb{E}}_{\rho \in \Gamma} \2[ F \1( {(\Phi_p)}_i (\rho), \rho \1) \2] \\
        &= \frac{1}{d} \norm{a}_1  + \1( 1 - \frac{1}{d} \1) \scalar{p}{a}.
    \end{align*}
    This gives $\scalar{p}{a} \leq \norm{a}_{\mathcal{Q}}$ for all direction vectors $a \in [0,1]^N$. For arbitrary $a \in \mathbb{R}^N$, note that
    \begin{equation*}
        \scalar{p}{a} \leq \scalar{p}{|a|} \leq \norm{\,|a|\,}_{\mathcal Q} = \norm{a}_{\mathcal{Q}},
    \end{equation*}
    showing that the condition \hyperref[chap4:eq:pairNormConditions1]{Eq.~(\ref*{chap4:eq:pairNormConditions1})} is satisfied, since $\norm{p}_{\mathcal{R}} = 1$.

    Let some direction vector $a \in [0,1]^N$, then from \hyperref[chap4:thm:optimalAsymmetricQuantumCloningChannel]{Therem~\ref*{chap4:thm:optimalAsymmetricQuantumCloningChannel}}, there a quantum cloning channel $\Phi^a_{\text{opt}}$ such that,
    \begin{equation*}
        \sum^N_{i = 1} a_i \cdot \operatornamewithlimits{\mathbb{E}}_{\rho \in \Gamma} \2[ F \1( {(\Phi^a_{\text{opt}})}_i(\rho), \rho \1) \2] = \frac{1}{d} \norm{a}_1 + \1( 1 - \frac{1}{d} \1) \norm{a}_{\mathcal{Q}},
    \end{equation*}
    with some $p \in \mathcal{R}_{N,d}$, such that,
    \begin{equation*}
        \sum^N_{i = 1} a_i \cdot \operatornamewithlimits{\mathbb{E}}_{\rho \in \Gamma} \2[ F \1( {(\Phi^a_{\text{opt}})}_i(\rho), \rho \1) \2] = \frac{1}{d} \norm{a}_1 + \1( 1 - \frac{1}{d} \1) \scalar{p}{a}.
    \end{equation*}
    This shows that condition \hyperref[chap4:eq:pairNormConditions2]{Eq.~(\ref*{chap4:eq:pairNormConditions2})} is also satisfied.
\end{proof}

From \hyperref[chap4:thm:achievableFidelityRegion]{Therem~\ref*{chap4:thm:achievableFidelityRegion}}, the achievable fidelity region $\mathcal{R}_{N,d}$ is a convex set delimited by a family of hyperplanes (see \hyperref[chap4:fig:hyperplaneFamilly]{Figure~\ref*{chap4:fig:hyperplaneFamilly}}):
\begin{equation*}
    \scalar{p}{a} = \norm{a}_{\mathcal{Q}}.
\end{equation*}

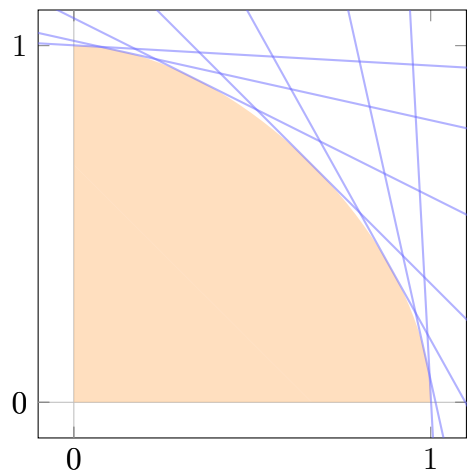
\begin{figure}[h]
    \centering
    \begin{tikzpicture}
        \begin{axis}[title = {$\mathcal{R}_{2,2}$ is delimited by a family of hyperplanes},
                     xmin =  -0.1,
                     xmax = {1 + 0.1},
                     ymin = -0.1,
                     ymax = {1 + 0.1},
                     axis equal image,
                     xtick = {0, 1},
                     ytick = {0, 1},
                     extra x ticks = {0},
                     extra x tick style = {grid = major},
                     extra y ticks = {0},
                     extra y tick style = {grid = major},
                     ]            

            \fill[orange!50!white, opacity = 0.5] (0, 0) -- (0, {2 / 3}) -- ({2 / 3}, 0) -- cycle;
            
            \begin{scope}
                \clip (0, {2 / 3}) -- (0, 1) -- (1, 1) -- (1, 0) -- ({2 / 3}, 0) -- cycle;
                \fill[rotate around = {-45:(0,0)}, orange!50!white, opacity = 0.5] (0, {(2 / 3) / sqrt(2)}) ellipse [x radius = {(2 / sqrt(3)) / sqrt(2)}, y radius = {(2 / 3) / sqrt(2)}];
            \end{scope}

            \draw[variable=\x, blue!60!white, opacity = 0.5, thick] plot ({\x}, {(0.998 - 0.0566 * \x) / 0.997});
            \draw[variable=\x, blue!60!white, opacity = 0.5, thick] plot ({\x}, {(0.974 - 0.214 * \x) / 0.961});
            \draw[variable=\x, blue!60!white, opacity = 0.5, thick] plot ({\x}, {(0.923 - 0.429 * \x) / 0.857});
            \draw[variable=\x, blue!60!white, opacity = 0.5, thick] plot ({\x}, {(0.889 - 0.667 * \x) / 0.667});
            \draw[variable=\x, blue!60!white, opacity = 0.5, thick] plot ({\x}, {(0.915 - 0.834 * \x) / 0.464});
            \draw[variable=\x, blue!60!white, opacity = 0.5, thick] plot ({\x}, {(0.974 - 0.961 * \x) / 0.214});
            \draw[variable=\x, blue!60!white, opacity = 0.5, thick] plot ({\x}, {(0.998 - 0.997 * \x) / 0.0536});
        \end{axis}
    \end{tikzpicture}
    \caption
        [$\mathcal{R}_{2,2}$ in orange is a convex set delimited by a family of hyperplanes in blue $\scalar{p}{a} = \norm{a}_{\mathcal{Q}}$.]
        {$\mathcal{R}_{2,2}$ \tikz{\draw[fill = orange!50!white, opacity = 0.5] (0,0) rectangle (1.5ex,1.5ex);} is a convex set delimited by a family of hyperplanes \tikz{\draw[fill = blue!60!white, opacity = 0.5] (0,0) rectangle (1.5ex,1.5ex);} $\scalar{p}{a} = \norm{a}_{\mathcal{Q}}$.} \label{chap4:fig:hyperplaneFamilly}
\end{figure}

Note that within the formulation of the quantum cloning problem, i.e. 
\begin{equation*}
    \sup_{\Phi~\textsc{cptp}} \sum^N_{i = 1} a_i \cdot \operatornamewithlimits{\mathbb{E}}_{\rho \in \Gamma} \2[ F \1( \Phi_i(\rho), \rho \1) \2],
\end{equation*}
the optimal quantum cloning channels are twirled quantum channels, and thus their marginals are of the form $\rho \mapsto p_i \cdot \rho + (1 - p_i) \dfrac{I_d}{d}$, but the $p_i$'s are not asked to be collinear with the $a_i$'s. Instead they have to maximize,
\begin{equation*}
    \sum^N_{i = 1} a_i \cdot \2( p_i + \frac{1 - p_i}{d} \2).
\end{equation*}
The optimal quantum cloning channels from \hyperref[chap4:thm:optimalAsymmetricQuantumCloningChannel]{Therem~\ref*{chap4:thm:optimalAsymmetricQuantumCloningChannel}}, can indeed give $p_i$'s in a different direction than the $a_i$'s, especially if the direction of the $a_i$ does not intersect an extreme point of $\mathcal{R}_{N,d}$. As a consequence, the optimal quantum cloning channels from \hyperref[chap4:thm:optimalAsymmetricQuantumCloningChannel]{Therem~\ref*{chap4:thm:optimalAsymmetricQuantumCloningChannel}} do not fill the boundary of $\mathcal{R}_{N,d}$, since some points in this boundary are not optimal with respect to the optimisation problem.

For example, from \hyperref[chap4:thm:optimalSymmetricQuantumCloningChannel]{Therem~\ref*{chap4:thm:optimalSymmetricQuantumCloningChannel}},
\begin{equation*}
    \max \set{p \in [0, 1]}{(p, p, 0) \in \mathcal{R}_{3,2}} = \frac{2}{3},
\end{equation*}
that is the $p_{\text{opt}}$ for the symmetric universal $1 \to 2$ quantum cloning problem is $\nicefrac{2}{3}$. However $\1( \nicefrac{2}{3}, \nicefrac{2}{3}, 0 \1) \in \mathcal{R}_{3,2}$ is not optimal for the asymmetric universal $1 \to 3$ quantum cloning problem, given the direction vector $a \coloneqq \1( \nicefrac{1}{2}, \nicefrac{1}{2}, 0 \1)$. Indeed, the the optimal quantum cloning channel in this direction, from \hyperref[chap4:thm:optimalAsymmetricQuantumCloningChannel]{Therem~\ref*{chap4:thm:optimalAsymmetricQuantumCloningChannel}}, gives $\1( \nicefrac{2}{3}, \nicefrac{2}{3}, \nicefrac{1}{9} \1) \in \mathcal{R}_{3,2}$ (see \hyperref[chap4:fig:flatRegion]{Figure~\ref*{chap4:fig:flatRegion}}).

\begin{figure}[h]
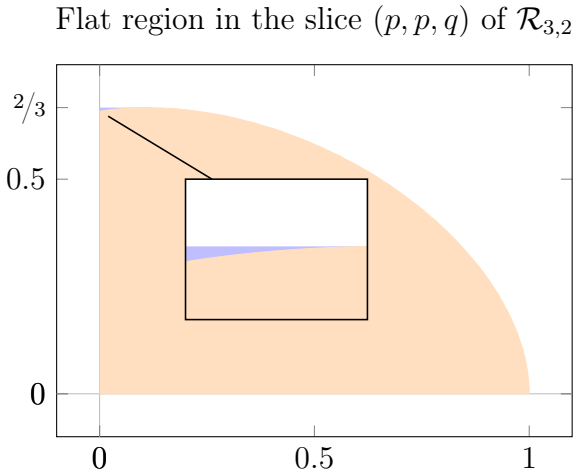

    \centering

    \caption
        [$\mathcal{R}_{3,2}$ is a convex set delimited by a family of hyperplanes, but some points may intersect more than $1$ hyperplanes. Here is a view of the slice $(p, p, q)$ in orage of $\mathcal{R}_{3,2}$, with a flat region in blue.]
        {$\mathcal{R}_{3,2}$ is a convex set delimited by a family of hyperplanes, but some points may intersect more than $1$ hyperplanes. Here is a view of the slice $(p, p, q)$ \tikz{\draw[fill = orange!25!white] (0,0) rectangle (1.5ex,1.5ex);} of $\mathcal{R}_{3,2}$, with a flat region \tikz{\draw[fill = blue!25!white] (0,0) rectangle (1.5ex,1.5ex);}.} \label{chap4:fig:flatRegion}
\end{figure}

In conclusion, as the previous example shows, given a $p \in [0,1]^N$, in order to decide whether $p$ is in the achievable fidelity region $\mathcal{R}_{N,d}$, it does not suffice to check the necessary condition \hyperref[chap4:eq:necessaryConditionOptimalQuantumCloningChannel]{Eq.~(\ref*{chap4:eq:necessaryConditionOptimalQuantumCloningChannel})}.

\chapter{Extendibility} \label{chap5}

\noindent\textbf{The present Section discusses the results of the papers ``Monogamy of highly symmetric states'' \cite{allerstorfer2023monogamy}, of which I am a co-author. At the time of writing this manuscript, that paper was in preparation. Only the section on the extendibility of isotropic states is presented below.}

\medskip

In recent years, the problem of entanglement in many-body systems has received particular attention in quantum physics and quantum information theory. This complex but fundamental quantum phenomenon has been the subject of much work because of its essential role in the description and understanding of quantum many-body systems \cite{wolf2003entanglement,buonsante2007ground,amico2008entanglement,tabia2022entanglement}.

\section{Extendibility of quantum states on a graph}

Let $\emph{G}$ be a graph with $N$ \emph{vertices}. A quantum state on the graph $G$ is a quantum state $\rho$ on a $N$-fold composite quantum system, each associated to a vertex of $G$, i.e.
\begin{equation*}
    \rho \in {\1( \mathbb{C}^d \1)}^{\otimes N}.
\end{equation*}
Let $e \coloneqq (u,v)$ be an \emph{edge} of the graph $G$, the reduced quantum state $\rho_e$ is the reduced quantum state on quantum systems $u$ and $v$, i.e.
\begin{equation*}
    \rho_e \coloneqq \Tr_{G \setminus \{u, v\}} \1[ \rho \1].
\end{equation*}

The \emph{complete graph} on $N$ vertices, denoted $\emph{K_N}$ is the graph $G$ where for all distinct vertices $u$ and $v$ in $G$, the pair $(u, v)$ is and edge of $G$, i.e. every pair of distinct vertices is connected by a unique edge (see \hyperref[chap5:fig:K5]{Figure~\ref*{chap5:fig:K5}}). The complete graph $K_N$ has $\frac{N (N - 1)}{2}$ edges.

The \emph{star graph} on $N$ vertices, denoted $\emph{S_N}$ is the graph $G$ with a distinct \emph{central} vertex $v \in G$ such that the pair $(u, v)$ is an edge of $G$ for all other vertex $u \in G$, i.e. it is the $1$-depth tree of order $N$ (see \hyperref[chap5:fig:S6]{Figure~\ref*{chap5:fig:S6}}). The star graph $S_N$ has $n - 1$ edges.

\begin{figure}[h]
    \centering
    \begin{subfigure}[b]{0.45\textwidth}
        \centering
        \begin{tikzpicture}[site/.style = {circle,
                                            draw = white,
                                            line width = 2pt,
                                            fill = black!80!white,
                                            inner sep = 5pt}]
        
            \coordinate (1) at (90:4em) {};
            \coordinate (2) at (162:4em) {};
            \coordinate (3) at (234:4em) {};
            \coordinate (4) at (306:4em) {};
            \coordinate (5) at (18:4em) {};

            \foreach \i in {1,...,5}
            \foreach \j in {\i,...,5} {
                \draw[-, line width = 4.5pt, draw = white] (\j.center) -- (\i.center);
                \draw[-, line width = 2.5pt, draw = black] (\j.center) -- (\i.center);
            }

            \foreach \i in {1,...,5}
                \node[site] (v\i) at (\i.center) {};
        \end{tikzpicture}
        \caption{$K_5$}
        \label{chap5:fig:K5}
     \end{subfigure}
     \hfill
     \begin{subfigure}[b]{0.45\textwidth}
         \centering
         \begin{tikzpicture}[site/.style = {circle,
                                            draw = white,
                                            line width = 2pt,
                                            fill = black!80!white,
                                            inner sep = 5pt}]

            \coordinate (0);
            \coordinate (1) at (90:4em);
            \coordinate (2) at (162:4em);
            \coordinate (3) at (234:4em);
            \coordinate (4) at (306:4em);
            \coordinate (5) at (18:4em);

            \foreach \i in {1,...,5} {
                \draw[-, line width = 4.5pt, draw = white] (0.center) -- (\i.center);
                \draw[-, line width = 2.5pt, draw = black] (0.center) -- (\i.center);
            }

            \foreach \i in {0,...,5}
                \node[site] (v\i) at (\i.center) {};
        \end{tikzpicture}
         \caption{$S_6$}
         \label{chap5:fig:S6}
     \end{subfigure}
     \caption{Complete graph on $5$ vertices $K_5$, and star graph on $6$ vertices $S_6$.}
\end{figure}
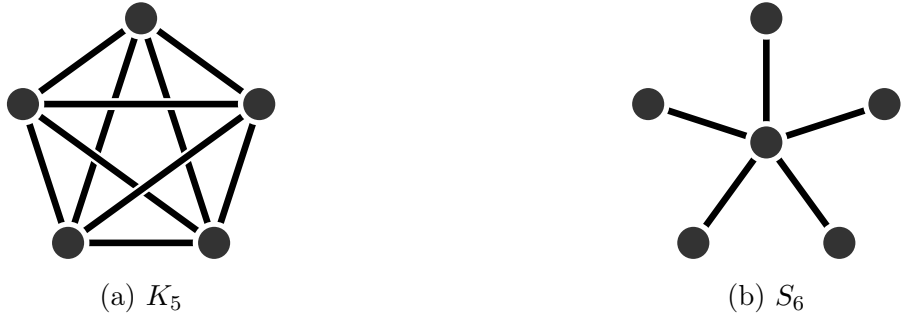

From \hyperref[chap3:sec:monogamyOfEntanglement]{Section~\ref*{chap3:sec:monogamyOfEntanglement}}, a bipartite quantum state $\rho$ is $k$-extendible if there exists a quantum state $\sigma$ on the star graph $S_{k+1}$ such that for all edges $e$, the reduced quantum states on $e$ is,
\begin{equation*}
    \sigma_e = \rho.
\end{equation*}

\section{Quantum cloning: star-graph extendibility of isotropic states} \label{chap5:sec:quantumCloningStarGraphExtendibilityOfIsotropicStates}

The quantum cloning problem can be seen as an extendibility problem, when considering the Choi matrix of the quantum cloning channel. Indeed, a perfect $1 \to N$ quantum cloning channel $\Phi: \mathcal{M}_d \to {\1( \mathcal{M}_d \1)}^{\otimes N}$ would have a Choi matrix $C_\Phi$ living on a star graph $S_{N+1}$ such that for all edges $e$, the reduced Choi matrix on $e$ is,
\begin{equation*}
    (C_{\Phi})_e = d \cdot \omega.
\end{equation*}
Because of the no-cloning \hyperref[chap4:thm:noCloning]{Theorem~\ref*{chap4:thm:noCloning}}, or alternatively the monogamy of the entanglement from \hyperref[chap3:sec:monogamyOfEntanglement]{Section~\ref*{chap3:sec:monogamyOfEntanglement}}, the normalized reduced Choi matrix on $e$ is required to be instead an isotropic state,
\begin{equation*}
    \frac{1}{d} (C_{\Phi})_e = \lambda_e \cdot \omega + (1 - \lambda_e) \mathrm{I},
\end{equation*}
where $\mathrm{I}$ is the maximally mixed state on $\mathbb{C}^d \otimes \mathbb{C}^d$

The symmetric $1 \to N$ quantum cloning problem can then be solved using a \emph{semi-definite programming} (\textsc{sdp}) optimization problem:
\begin{align*}
    \max_{C_\Phi,p} \quad & p \\
    \text{s.t.} \quad & (C_{\Phi})_e = d \1( p \cdot \omega + (1 - p) \mathrm{I} \1), & \textrm{for all edges}~e \\
    & \Tr_{[N+1] \setminus \{0\}} \1[ C_\Phi \1] = I_d, \quad C_\Phi \geq 0, \quad p \in [0,1].
\end{align*}
Note that in the previous optimization problem, the conditions $\Tr_{[N+1] \setminus \{0\}} \1[ C_\Phi \1] = I_d$ and $p \in [0,1]$ are superfluous, and the optimization problem reduces to,
\begin{align*}
    \max_{C_\Phi,p} \quad & p \\
    \text{s.t.} \quad & (C_{\Phi})_e = d \1( p \cdot \omega + (1 - p) \mathrm{I} \1), & \textrm{for all edges}~e \\
    & C_\Phi \geq 0.
\end{align*}
From \hyperref[chap4:thm:optimalSymmetricQuantumCloningChannel]{Therem~\ref*{chap4:thm:optimalSymmetricQuantumCloningChannel}}, this optimization problem has optimal solution:
\begin{equation*}
    p_{\text{opt}} \coloneqq \frac{d + N}{N (d + 1)}.
\end{equation*}

\section{\texorpdfstring{$K_N$}{Kn}-Extendibility of isotropic states}

The optimization problem of \hyperref[chap5:sec:quantumCloningStarGraphExtendibilityOfIsotropicStates]{Section~\ref*{chap5:sec:quantumCloningStarGraphExtendibilityOfIsotropicStates}}, can be extended to the following quantum state marginal problem. Given a complete graph $K_N$, what is the largest $p \in [0, 1]$, such that there exists a quantum state $\rho$ on $K_N$, with all reduced quantum state $\rho_e = p \cdot \omega + (1 - p) \mathrm{I}$, for all edges $e$? This problem depends on the number of quantum systems $N$ and their dimension $d$, and can be stated as the primal semi-definite problem:
\begin{equation} \label{chap5:eq:SDP}
    \begin{aligned}
        p(N,d) \coloneqq \max_{\rho,p} \quad & p \\
    \text{s.t.} \quad & \rho_e = p \cdot \omega + (1 - p) \mathrm{I}, & \textrm{for all edges}~e \\
    & \rho \geq 0.
    \end{aligned}
\end{equation}

\subsection{Lower bound} \label{chap5:sec:lowerBound}

A \emph{perfect matching} on a graph is a set of edges, such that every vertex is contained in exactly one of those edges.
\begin{proposition} \label{chap5:prop:perfectMatchingCount}
    There are $(2N - 1)!!$ perfect matchings on $K_{2N}$, and if $e$ is an edge on $K_{2N}$, then there are $(2N-3)!!$ perfect matchings on $K_{2N}$ containing $e$.
\end{proposition}
\begin{proof}
    Let $a_N$ be the number of perfect matching on $K_{2N}$; clearly $a_1 = 1$. Now assume $N > 1$ and let $v$ be a vertex in $K_{2N}$. This vertex can be matched with $2N - 1$ other vertices, let $u$ be such other vertex matched with $v$. Remove $u$ and $v$ from $K_{2N}$, the resulting graph $K_{2N} \setminus \{u, v\}$ is the complete graph $K_{2(N-1)}$. Thus, by induction on $N$, the number of perfect matchings on $K_{2N}$ satisfies the recursive relation:
    \begin{equation*}
        a_N = (2N - 1) a_{N-1} \quad \implies \quad a_N = (2N - 1)!!.
    \end{equation*}

    Assume $e = (u, v)$, thus the number of perfect matchings containing $e$ is the number of perfect matchings on $K_{2N} \setminus \{u, v\}$, that is $(2N - 3)!!$.
\end{proof}
\begin{remark*}
    There is no perfect matching on $K_N$ for odd $N$ (see \hyperref[chap5:fig:perfectMatching]{Figure~\ref*{chap5:fig:perfectMatching}}).
\end{remark*}

\begin{figure}[h]
    \centering
    \begin{subfigure}[b]{0.45\textwidth}
        \centering

        \caption{$K_6$}
     \end{subfigure}
     \hfill
     \begin{subfigure}[b]{0.45\textwidth}
         \centering
        %
         \caption{$K_5$}
    \end{subfigure}
    \caption
        [Perfect matching in red for complete graphs $K_N$, with even $N$. For $K_N$ with odd $N$, some vertices are not matched in gray.]
        {Perfect matching \tikz{\draw[fill = red!70!white] (0,0) rectangle (1.5ex,1.5ex);} for complete graphs $K_N$, with even $N$. For $K_N$ with odd $N$, some vertices are not matched \tikz{\draw[fill = gray] (0,0) rectangle (1.5ex,1.5ex);}.} \label{chap5:fig:perfectMatching}
\end{figure}

A lower bound on the optimisation problem \hyperref[chap5:eq:SDP]{Eq.~(\ref*{chap5:eq:SDP})} would be the following. For even $N$, let $E_1, \ldots, E_{(N - 1)!!}$ be all the perfect matchings on $K_N$, and for each perfect matching $E_k$, define the quantum state $\rho^{(k)}$ on $K_N$ by,
\begin{equation*}
    \rho^{(k)} \coloneqq \bigotimes_{e \in E_k} \omega_e.
\end{equation*}
For odd $N$, let $v$ be vertex of $K_N$, and let $E_{v, 1}, \ldots, E_{v, {(N - 2)!!}}$ be all the perfect matchings on $K_N \setminus \{ v \}$, define the quantum state $\rho^{(v, k)}$ on $K_N$ by,
\begin{equation*}
    \rho^{(v, k)} \coloneqq \bigotimes_{e \in E_{v, k}} \omega_e \otimes \mathrm{I}_v.
\end{equation*} 
That is a quantum state maximally entangled on the perfect matching edges, and maximally mixed on the remaining vertex in the odd case. Let $\rho$ be the quantum state defined on $K_N$, as a uniform combination of the previously constructed states:
\begin{align*}
    & (N~\textrm{even}) & \rho &\coloneqq \frac{1}{N - 1} \sum_{1 \leq k \leq (N - 1)!!} \rho^{(k)} \\[1em]
    & (N~\textrm{odd}) & \rho &\coloneqq \frac{1}{N} \sum_{\substack{v \in K_N \\ 1 \leq k \leq (N - 2)!!}} \rho^{(v,k)},
\end{align*}
where the corresponding normalization factors can be found using \hyperref[chap5:prop:perfectMatchingCount]{Proposition~\ref*{chap5:prop:perfectMatchingCount}}.
Then for all edges $e$ in $K_N$, the reduced quantum state $\rho_e$ is
\begin{equation*}
    \rho_e =
    \begin{cases}
        \frac{1}{N-1} \omega + \frac{N-2}{N-1} \mathrm{I} &N~\text{even} \\
        \frac{1}{N} \omega + \frac{N-1}{N} \mathrm{I} &N~\text{odd}.
    \end{cases}
\end{equation*}
The lower bound becomes
\begin{equation*}
    p(N, d) \geq
    \begin{cases}
        \frac{1}{N-1} &N~\text{even} \\
        \frac{1}{N} &N~\text{odd}.
    \end{cases}
\end{equation*}
In particular, the lower bound is independent of the dimension $d$.

\subsection{Symmetries}

Since it is a complete graph, any permutation of the vertices of $K_N$ is a graph isomorphism. Let $\rho$ be a quantum state on $K_N$ such that for all edges $e$, the reduced quantum state $\rho_e$ is,
\begin{equation*}
    \rho_e = p \cdot \omega + (1 - p) \mathrm{I}.
\end{equation*}
Then $\rho$ is invariant by vertex permutation, that is for all permutations $\sigma \in \mathfrak{S}_N$, the equality
\begin{equation*}
    \psi(\sigma) \cdot \rho \cdot \psi \1( \sigma^{\shortminus 1} \1) = \rho,
\end{equation*}
holds.

Recall that a maximally entangled state $\omega$ is invariant by $\bar{U} \otimes U$, for all unitary matrix $U$, i.e. $(\bar{U} \otimes U) \cdot \omega = \omega$. On the star graph $S_N$ with central vertex $v$, a quantum state $\rho$ with reduced quantum states $\rho_e = p \cdot \omega + (1 - p) \mathrm{I}$, for all edges $e$, would be invariant under conjugation by $\bar{U}$ on $v$ and $U^{\otimes (N-1)}$ on $S_N \setminus \{v\}$, i.e.
\begin{equation*}
    \2( \bar{U}_v \otimes {\1(U^{\otimes (N-1)} \1)}_{S_N \setminus \{v\}} \2) \cdot \rho \cdot \2( U^\T_v \otimes {\1({(U^*)}^{\otimes (N-1)} \1)}_{S_N \setminus \{v\}} \2) = \rho. 
\end{equation*}
But on the complete graph $K_N$, no distinct vertex can be chosen, instead, using the invariant on a maximally entangled state by $O \otimes O$, for all orthogonal matrix $O \in \mathrm{O}_d$, i.e. $(O \otimes O) \cdot \omega = \omega$, a quantum state $\rho$ with reduced quantum states $\rho_e = p \cdot \omega + (1 - p) \mathrm{I}$, for all edges $e$, is invariant under conjugation by $O^{\otimes N}$, that is,
\begin{equation*}
    O^{\otimes N} \cdot \rho \cdot {\1( O^\T \1)}^{\otimes N} = \rho.
\end{equation*}

Hence, with respect to the optimisation problem \hyperref[chap5:eq:SDP]{Eq.~(\ref*{chap5:eq:SDP})}, the two commutation relations,
\begin{align}
    \1[ \rho, \sigma \1] &= 0, & & \forall \sigma \in \mathfrak{S}_N \label{chap5:eq:symmetricCommutation} \\
    \1[ \rho, O^{\otimes N} \1] &= 0, & & \forall O \in \mathrm{O}_d. \label{chap5:eq:orthogonalCommutation}
\end{align}
hold for all $N$ and $d$.

\subsection{Dual \texorpdfstring{\textsc{sdp}}{SDP}}

The Lagrangian \cite{boyd2004convex} associated with the semi-definite problem \hyperref[chap5:eq:SDP]{Eq.~(\ref*{chap5:eq:SDP})} is defined as,
\begin{equation*}
    L \1( p, \rho, h_e, Z \1) \coloneqq p + \sum_{e~\text{edge}} \scalar{h_e}{\rho_e - p \cdot \omega + (1 - p) \mathrm{I}} + \scalar{Z}{\rho}.
\end{equation*}
The \emph{min-max principle} states that,
\begin{equation*}
    \max_{\rho,p} \min_{h_e, Z} L \1( p, \rho, h_e, Z \1) \leq \min_{h_e, Z} \max_{\rho,p} L \1( p, \rho, h_e, Z \1).
\end{equation*}
Then, the \emph{dual} semi-definite problem of \hyperref[chap5:eq:SDP]{Eq.~(\ref*{chap5:eq:SDP})} is,
\begin{align*}
    p^*(N,d) \coloneqq \max_{h_e,Z} \quad & - \frac{\Tr \2[ \sum\limits_{e~\text{edge}} h_e \otimes I^{\otimes (N-2)}_d \2]}{d^N} \\
    \text{s.t.} \quad & \sum_{e~\text{edge}} \Tr \1[ h_e \cdot (\omega - \mathrm{I}) \1] = 1 \\
    & Z + \sum_{e~\text{edge}} h_e \otimes I^{\otimes (N-2)}_d = 0, \quad  Z \geq 0,
\end{align*}
where for each edge $e$ of $K_N$, the matrix $h_e \in \mathcal{M}_{d^2}$ is Hermitian, as well as the matrix $Z \in \mathcal{M}_{d^N}$. Recall that for any Hermitian matrix $M$ the smallest $\lambda \in \mathbb{R}$ such that $\lambda \cdot I \geq M$ is equal to the largest eigenvalue of $M$. Then the dual optimization problem can be simplified to,
\begin{align*}
    p^*(N,d) = \min_{h_e} \quad & \lambda_{\text{max}} \3( \sum_{e~\text{edge}} h_e \otimes I^{\otimes (N-2)}_d \3) \\
    \text{s.t.} \quad & \sum_{e~\text{edge}} \Tr \1[ h_e \cdot \omega \1] = 1 \\
    & \Tr \1[ h_e \1] = 0, & \textrm{for all edges}~e.
\end{align*}

From the commutation relation \hyperref[chap5:eq:orthogonalCommutation]{Eq.~(\ref*{chap5:eq:orthogonalCommutation})} applied to the dual optimization problem, define for all edge $e$, the twirled Hermitian matrix $\widetilde{h}_e$:
\begin{equation*}
    \widetilde{h}_e \coloneqq \int_O \1( O \otimes O \1) \1( h_e \1) \1( O^\T \otimes O^\T \1) ~\mathrm{d}O,
\end{equation*}
where the integral is taken with respect to the normalized Haar measure on the orthogonal group $\mathrm{O}_d$. The constraints of the dual optimization problem are also satisfied with $\widetilde{h}_e$, and since by convexity of $\lambda_{\text{max}}$,
\begin{equation*}
    \lambda_{\text{max}} \3( \sum_{e~\text{edge}} h_e \otimes I^{\otimes (N-2)}_d \3) \geq \lambda_{\text{max}} \3( \sum_{e~\text{edge}} \widetilde{h}_e \otimes I^{\otimes (N-2)}_d \3),
\end{equation*}
hence the dual optimization problem can be restricted to twirled $\widetilde{h}_e$.

The twirled $\widetilde{h}_e$ commutes with all the orthogonal matrices $O \otimes O$, i.e. they are in the commutant of the algebra
\begin{equation*}
    \Span_{\mathbb{C}} \set{O \otimes O}{O \in \mathrm{O}_d}.
\end{equation*}
From \hyperref[chap2:thm:commutantSchurWeylDualityBn]{Theorem~\ref*{chap2:thm:commutantSchurWeylDualityBn}}, the twirled $\widetilde{h}_e$ is in the complex algebra spanned by the tensor representation of the Brauer monoid:
\begin{equation*}
    \Span_{\mathbb{C}} \set{\psi(p)}{p \in \mathbb{B}_2}.
\end{equation*}
Hence, the twirled $\widetilde{h}_e$ are a linear combination of the unormalized maximally mixed, maximally entangled, and flip state:
\begin{equation*}
    \widetilde{h}_e = \alpha_e (d^2 \cdot \mathrm{I}_e) + \beta_e (d \cdot \omega_e) + \gamma_e (d \cdot \mathrm{F}_e),
\end{equation*}
subject to the constraint $\alpha_e = - \frac{\beta_e + \gamma_e}{d}$. The dual optimization problem simplifies to,
\begin{align*}
    p^*(N,d) = \min_{\beta_e,\gamma_e} \quad & \lambda_{\text{max}} \3( \sum_{e~\text{edge}} \2( - \frac{\beta_e + \gamma_e}{d} (d^2 \cdot \mathrm{I}_e) + \beta_e (d \cdot \omega_e) + \gamma_e (d \cdot \mathrm{F}_e) \2) \otimes I^{\otimes (N-2)}_d \3) \\
    \text{s.t.} \quad & (d^2 - 1) \3( \sum_{e~\text{edge}} \beta_e \3) + (d-1) \3( \sum_{e~\text{edge}} \gamma_e \3) = d.
\end{align*}

From the commutation relation \hyperref[chap5:eq:orthogonalCommutation]{Eq.~(\ref*{chap5:eq:symmetricCommutation})}, applied to the dual optimization problem, all the $\beta_e$ and $\gamma_e$ must be equal. Writing $\beta \coloneqq \beta_e$ and $\gamma \coloneqq \gamma_e$ for all edges $e$, the dual optimization problem simplifies further to
\begin{align*}
    p^*(N,d) = \min_{\beta,\gamma} \quad & \lambda_{\text{max}} \3( \sum_{e~\text{edge}} \2( - \frac{\beta + \gamma}{d} (d^2 \cdot \mathrm{I}_e) + \beta (d \cdot \omega_e) + \gamma (d \cdot \mathrm{F}_e) \2) \otimes I^{\otimes (N-2)}_d \3) \\
    \text{s.t.} \quad & \frac{N (N - 1)}{2} \1( \beta (d^2 - 1) + \gamma (d-1) \1) = d.
\end{align*}

Using the constraint $\frac{N (N - 1)}{2} \1( \beta (d^2 - 1) + \gamma (d-1) \1) = d$ to eliminate $\gamma$, let the function $f$ be defined by $f(x) \coloneqq \lambda_{\text{max}} \1( H(x) \1)$ with
\begin{equation*}
    H(x) \coloneqq \sum_{e~\text{edge}} \3( \2( \frac{2}{N (N - 1) (1 - d)} - x \2) \1( (d^2 \cdot \mathrm{I}_e) - d (d \cdot \mathrm{F}_e) \1) + x \1( (d \cdot \mathrm{F}_e) - (d \cdot \omega) \1) \3) \otimes I^{\otimes (N-2)}_d.
\end{equation*}
Then the dual optimization problem becomes finally the minimization,
\begin{equation*}
    p^*(N,d) = \min_x f(x).
\end{equation*}

Recall from \hyperref[chap2:thm:SchurWeylDuality]{Theorem~\ref*{chap2:thm:SchurWeylDuality}}, that under the action of the tensor representation of the symmetric group algebra $\mathbb{S}_N$, the space of $N$-fold tensors over $\mathbb{C}^d$ decomposes as
\begin{equation*}
    {\1( \mathbb{C}^d \1)}^{\otimes N} \simeq \bigoplus_{\lambda \in \mathrm{Irr}(\mathfrak{S}_N)} V^{\oplus m_\lambda}_\lambda,
\end{equation*}
where $\mathrm{Irr}(\mathfrak{S}_N) \coloneqq \set{\lambda \vdash N}{\lambda^\prime_1 \leq d}$. Similarly, as a consequence of \hyperref[chap2:thm:commutantSchurWeylDualityBn]{Theorem~\ref*{chap2:thm:commutantSchurWeylDualityBn}}, under the action of the tensor representation of the Brauer algebra $\mathbb{B}_N(d)$, the space of $N$-fold tensors over $\mathbb{C}^d$ decomposes as
\begin{equation*}
    {\1( \mathbb{C}^d \1)}^{\otimes N} \simeq \bigoplus_{\lambda \in \mathrm{Irr}(\mathbb{B}_N(d))} W^{\oplus n_\lambda}_\lambda,
\end{equation*}
where $\mathrm{Irr}\1(\mathbb{B}_N(d)\1) \coloneqq \set{\lambda \vdash N - 2k}{k \in \1\{ 0, \ldots, \lfloor \frac{N}{2} \rfloor \1\} \text{ and } \lambda^\prime_1 + \lambda^\prime_2 \leq d}$ \cite{wenzl1988structure}. Both decompositions are indexed by some Young diagram $\lambda$.

Given a irreducible representation $W_\lambda$ of the Brauer algebra $\mathbb{B}_N(d)$, the restriction of $W_\lambda$ to the symmetric group algebra $\mathbb{S}_N$, as a subalgebra decompose into
\begin{equation*}
    \mathrm{Res}^{\mathbb{B}_N(d)}_{\mathbb{S}_N} \1( W_\lambda \1) \simeq \bigoplus_{\mu \in \mathrm{Irr}(\mathfrak{S}_N)} V^{\oplus m_{\lambda, \mu}}_\lambda.
\end{equation*}
The multiplicities have no known concise formula. Even the explicit characterization of the set $\emph{\Omega}$,
\begin{equation*}
    \Omega \coloneqq \set{(\lambda, \mu) \in \mathrm{Irr}\1(\mathbb{B}_N(d)\1) \times \mathrm{Irr}(\mathfrak{S}_N)}{m_{\lambda, \mu} \neq 0},
\end{equation*}
is still unknown. However Okada \cite{okada2016pieri} characterizes it using an algorithm, and found the subset $\emph{\Gamma} \subseteq \Omega$,
\begin{equation*}
        \Gamma \coloneqq \set{(\lambda, \mu) \in \mathrm{Irr}\1(\mathbb{B}_N(d)\1) \times \mathrm{Irr}(\mathfrak{S}_N)}{\lambda = (1^m), r(\mu) = m \text{ for some } m \in \{1, \dots, d\}},
\end{equation*}
where $r(\mu)$ is the number of rows with odd size in the Young diagram $\mu$. Moreover if $\1( (1^m), \mu \1) \in \Omega$ then $\1( (1^m), \mu \1) \in \Gamma$.

The \emph{content} of a Young diagram $\lambda$, denoted $\emph{c(\lambda)}$, is the the sum of the entries of the boxes of the Young tableau $T$, where the entry of the box in row $i$ and column $j$ is given by $j - i$. For example if $\lambda \coloneqq (3, 3, 1)$, then
\begin{align*}
    \lambda &= \ydiagram{3, 3, 1} & T &=
    \begin{ytableau}
        0 & 1 & 2 \\
        \shortminus 1 & 0 & 1 \\
        \shortminus 2
    \end{ytableau}
\end{align*}
and the content of $\lambda$ is $c(\lambda) = 1$.

The next lemmas describe how particular central elements of the symmetric group $\mathfrak{S}_N$ and the Brauer monoid $\mathbb{B}_N$, act on their irreducible representations.
\begin{lemma}[\cite{doty2019canonical}] \label{chap5:lem:eigenvalueCentralElementsSn}
    Let $V_\lambda$ be an irreducible representation of $\mathfrak{S}_N$, with $\lambda \in \mathrm{Irr}(\mathfrak{S}_N)$, and define $J$ an element of the symmetric group algebra $\mathbb{S}_N$, by
    \begin{equation*}
        J \coloneqq \sum_{1 \leq i, j \leq N} (i \: j).
    \end{equation*}
    Then the restriction of $J$ on $V_\lambda$ is the multiple of the identity,
    \begin{equation*}
        J|_{V_\lambda} = c(\lambda) \cdot I.
    \end{equation*}
\end{lemma}
\begin{lemma}[\cite{doty2019canonical}] \label{chap5:lem:eigenvalueCentralElementsBn}
    Let $W_\lambda$ be an irreducible representation of $\mathbb{B}_N(d)$, with $\lambda \in \mathrm{Irr}\1(\mathbb{B}_N(d)\1)$, and define $J$ an element of the Brauer algebra $\mathbb{B}_N(d)$, by
    \begin{equation*}
        J \coloneqq \sum_{1 \leq i, j \leq N} (i \: j) - (i \: j)^\Tpartial,
    \end{equation*}
    where $(i \: j)^\Tpartial$ denotes the partial transposition of the diagram $(i \: j)$. Then the restriction of $J$ on $W_\lambda$ is the multiple of the identity,
    \begin{equation*}
        J|_{W_\lambda} = \1( c(\lambda) - k(d-1) \1) \cdot I.
    \end{equation*}
\end{lemma}

Using \hyperref[chap5:lem:eigenvalueCentralElementsSn]{Lemma~\ref*{chap5:lem:eigenvalueCentralElementsSn}}, \hyperref[chap5:lem:eigenvalueCentralElementsBn]{Lemma~\ref*{chap5:lem:eigenvalueCentralElementsBn}}, the decomposition of the space of $N$-fold tensors over $\mathbb{C}^d$ under the action of the tensor representation of the Brauer algebra $\mathbb{B}_N(d)$, and decomposing of the restriction of the irreducible representations of the Brauer algebra $\mathbb{B}_N(d)$ to the symmetric group algebra $\mathbb{S}_N$, the function $f$ becomes,
\begin{equation*}
        f(x) = \max_{(\lambda,\mu) \in \Omega} \quad \underbrace{\frac{1}{d-1} \3( \frac{2 d \cdot c(\mu)}{N (N - 1)} - 1 \3) + x \3( c(\lambda) + d \cdot c(\mu) - k(d-1) - \frac{N (N - 1)}{2} \3)}_{\normalsize f_{\lambda, \mu}(x)},
\end{equation*}
where $f_{\lambda, \mu}(x)$ is an affine function. The dual optimization problem is then the minimum value of the max over a set of affine functions, i.e.,
\begin{equation} \label{chap5:eq:dualSDP}
    p^*(N,d) = \min_x \max_{(\lambda,\mu) \in \Omega} f_{\lambda, \mu}(x).
\end{equation}

\subsubsection{Case $d > N$ or either $d$ or $N$ is even}

An upper bound for the dual optimization problem \hyperref[chap5:eq:dualSDP]{Eq.~(\ref*{chap5:eq:dualSDP})}, can be found by setting $x = \frac{2}{N (N - 1)(1-d)}$. Then \hyperref[chap5:eq:dualSDP]{Eq.~(\ref*{chap5:eq:dualSDP})} becomes, since $x < 0$,
\begin{equation*}
    p^*(N,d) \leq \min_{\lambda \in \mathrm{Irr}\1(\mathbb{B}_N(d)\1)} \frac{2}{N (N - 1)(1-d)} \1( c(\lambda) - k(d - 1) \1).
\end{equation*}

\begin{lemma}\label{chap5:lem:upperBound}
    If $d > N$, or either $d$ or $N$ is even, then,
    \begin{equation*}
        p^*(N,d) \leq
        \begin{cases}
            \frac{1}{N} & \text{ if $N$ is odd} \\
            \frac{1}{N-1} & \text{ if $N$ is even}.
        \end{cases}
    \end{equation*}
\end{lemma}
\begin{proof}
    It is enough to prove that
    \begin{equation*}
        \min_{\lambda \in \mathrm{Irr}\1(\mathbb{B}_N(d)\1)} c(\lambda) - k(d - 1) \leq
        \begin{cases}
            \frac{(1 - d) (N - 1)}{2} & \text{ if $N$ is odd} \\[0.5em]
            \frac{(1 - d) N}{2} & \text{ if $N$ is even},
        \end{cases}
    \end{equation*}
    
    If $d > N$, the minimization can be restricted to only single-column partitions $\lambda \coloneqq \1( 1^{(N - 2k)} \1)$, for all $k \in \{ 0, \dots, \lfloor \frac{N}{2} \rfloor \}$, which is always possible when $d > N$. Let $|\lambda| \coloneqq N - 2k$, then
    \begin{align*}
        \min_{\lambda \in \mathrm{Irr}\1(\mathbb{B}_N(d)\1)} c(\lambda) - k(d - 1) &\leq \min_{k \in \{0, \dots, \lfloor \frac{N}{2} \rfloor \}} c \1( 1^{(N - 2k)} \1) - k(d - 1) \\
        &= \min_{k \in \{0, \dots, \lfloor \frac{N}{2} \rfloor \}} - \frac{|\lambda| \big{(}|\lambda| - 1 \big{)}}{2} - (d-1) \frac{N - |\lambda|}{2} \\
        &=
        \begin{cases}
            \frac{(1 - d) (N - 1)}{2} &\text{if $N$ is odd} \\[0.5em]
            \frac{(1 - d) N}{2} &\text{if $N$ is even}.
        \end{cases}
    \end{align*}

    Otherwise, if $d$ is even, let $k^* = \lceil \frac{N}{2} \rceil - \frac{d}{2}$. Then the single-column partition $\lambda \coloneqq \1( 1^{(N - 2(k + k^*))} \1)$ satisfies $\lambda^\prime_1 + \lambda^\prime_2 \leq d$ for all $k \in \{ 0, \dots, \lfloor \frac{N}{2} \rfloor - k^* \}$, and,
    \begin{align*}
        \min_{\lambda \in \mathrm{Irr}\1(\mathbb{B}_N(d)\1)} c(\lambda) - k(d - 1) &\leq \min_{k \in \{ 0, \dots, \lfloor \frac{N}{2} \rfloor - k^* \}} c \1( 1^{(N - 2k)} \1) - (k+k^*)(d - 1) \\
        &=
        \begin{cases}
            \frac{(1 - d) (N - 1)}{2} &\text{if $N$ is odd} \\[0.5em]
            \frac{(1 - d) N}{2} &\text{if $N$ is even}.
        \end{cases}
    \end{align*}
    
    The same result holds if $N$ is even.
\end{proof}

\begin{theorem} \label{chap5:thm:largeDimension}
    If $d > N$, or either $d$ or $N$ is even, then,
    \begin{equation*}
        p^*(N,d) =
        \begin{cases}
            \frac{1}{N} & \text{ if $N$ is odd} \\
            \frac{1}{N-1} & \text{ if $N$ is even}.
        \end{cases}
    \end{equation*}
\end{theorem}
\begin{proof}
    Using \hyperref[chap5:sec:lowerBound]{Section.~\ref*{chap5:sec:lowerBound}} and \hyperref[chap5:lem:upperBound]{Lemma.~\ref*{chap5:lem:upperBound}}, the dual optimization problem is lower and upper bounded by,
    \begin{equation*}
        \begin{cases}
            \frac{1}{N} & \text{ if $N$ is odd} \\
            \frac{1}{N-1} & \text{ if $N$ is even}.
        \end{cases}
    \end{equation*}
\end{proof}

\subsubsection{Case $N \geq d$ and both $d$ and $N$ are odd}

Let $\tilde{x} \coloneqq \frac{2}{N (N - 1)(1-d)}$, then the affine functions $f_{\lambda,\mu}$ of \hyperref[chap5:eq:dualSDP]{Eq.~(\ref*{chap5:eq:dualSDP})} evaluated at the negative coordinate $\tilde{x}$ become,
\begin{align*}
     f_{\lambda,\mu}(\tilde{x}) &= \frac{1}{d-1} \3( \frac{2 d \cdot c(\mu)}{N (N - 1)} - 1 \3) + \tilde{x} \3( c(\lambda) + d \cdot c(\mu) - k(d-1) - \frac{N (N - 1)}{2} \3) \\
     &= \frac{1}{n-1} + \tilde{x} \cdot h(\lambda),
\end{align*}
where $h(\lambda)$ is defined by $h(\lambda) \coloneqq \frac{1}{2} \sum^{\lambda_1}_{i=0} \lambda^\prime_i (d - \lambda^\prime_i + 2 (i - 1))$. At this coordinate the affine functions $f_{\lambda,\mu}$ do not depend on $\mu$ anymore. Let $g$ be the function defined by,
\begin{equation*}
    g(\lambda) \coloneqq f_{\lambda,\mu} \1( \tilde{x} \1).
\end{equation*}
The offsets of the affine functions $f_{\lambda,\mu}$ do not depend on $\lambda$ either. Let $a$ be the function defined by,
\begin{equation*}
    a(\mu) \coloneqq \frac{1}{d-1} \3( \frac{2 d \cdot c(\mu)}{N (N - 1)} - 1 \3).
\end{equation*}

If $N \geq d$, $k \coloneqq \lfloor \frac{N-d}{2} / d \rfloor$ and $m \coloneqq \frac{N-d}{2} \bmod d$, let the two partitions $\lambda_1, \lambda_2$ in $\mathrm{Irr}\1(\mathbb{B}_N(d)\1)$ and the three partitions $\mu_1, \mu_2, \mu_3$ in $\mathrm{Irr}(\mathfrak{S}_N)$ be defined by\footnote{In the definition above, $\mu_3$ is given using the column notation $\mu^\prime_3$. Using the row notation it becomes $\mu_3 = \1( (2k+3)^m, (2k+1)^{d-m} \1)$: $m$ rows of size $(2k+3)$ and $d-m$ rows of size $(2k+1)$.},
\begin{equation} \label{chap5:eq:5partitions}
    \begin{aligned}
        \lambda_1 &\coloneqq \1( 1^{d} \1) & \mu_1 &\coloneqq \1( N \1) \\
        \lambda_2 &\coloneqq \1( 1 \1) & \mu_2 &\coloneqq \1( N-d+1, 1^{d-1} \1) \\
        & & \mu^\prime_3 &\coloneqq \1( d^{2k+1}, m^2 \1).
    \end{aligned}
\end{equation}

\begin{lemma} \label{chap5:lem:twoRegimes}
    If $N \geq d$ and both $d$ and $N$ are odd, then $\lambda_1, \mu_2$ from \hyperref[chap5:eq:5partitions]{Eq.~(\ref*{chap5:eq:5partitions})} satisfy
    \begin{equation*}
        g(\lambda_1) - a(\mu_2) = \frac{2 d + 2 - N}{N - 1}.
    \end{equation*}
    In particular $g(\lambda_1) < a(\mu_2)$ if $N \geq 2d + 3$, and $g(\lambda_1) > a(\mu_2)$ if $N \leq 2d + 1$.
\end{lemma}
\begin{proof}
    The content of $\mu_2$ is $\frac{(N-d+1)(N-d)}{2} - \frac{d(d-1)}{2}$. Also $h(\lambda_1) = 0$, so $g(\lambda_1) = \frac{1}{N-1}$. Then
    \begin{align*}
        g(\lambda_1) - a(\mu_2) &= \frac{1}{N-1} - \frac{1}{d-1} \3( \frac{2d \cdot c(\mu_2)}{n (n-1)} - 1 \3) \\
        &=-d \frac{(N-d+1)(N-d)-d(d-1)}{N(N-1)(d-1)} + \frac{1}{d-1} + \frac{1}{N-1} \\
        &=\frac{2 d + 2 - N}{N - 1}.
    \end{align*}
\end{proof}

\begin{lemma} \label{chap5:lem:lambdaMuOrders}
    If $N \geq d$ and both $d$ and $N$ are odd, then the partitions from \hyperref[chap5:eq:5partitions]{Eq.~(\ref*{chap5:eq:5partitions})} satisfy $(\lambda_i, \mu_j) \in \Gamma$, and the relations,
    \begin{equation*}
        g(\lambda) < g(\lambda_2) < g(\lambda_1) \quad \text{ and } \quad a(\mu) < a(\mu_1),
    \end{equation*}
    for all $\lambda \neq \lambda_1,\lambda_2$ in $\mathrm{Irr}\1(\mathbb{B}_N(d)\1)$ and all $\mu \neq \mu_1$ in $\mathrm{Irr}(\mathfrak{S}_N)$. Moreover for all $\1( \lambda_1, \mu \1) \in \Omega$,
    \begin{equation*}
        a(\mu) \leq a(\mu_2).
    \end{equation*}
\end{lemma}
\begin{proof}
    By definition of $\Gamma$, all $(\lambda_i, \mu_j)$ are in $\Gamma$.

    Let $\lambda$ in $\mathrm{Irr}\1(\mathbb{B}_N(d)\1)$ then
    \begin{equation*}
        g(\lambda) = \frac{1}{N-1} + \tilde{x} \cdot h(\lambda).
    \end{equation*}
    But since $\lambda^\prime_1 \leq d$ holds for all $\lambda$ in $\mathrm{Irr}\1(\mathbb{B}_N(d)\1)$, then $h(\lambda) \geq 0$. In particular, $h(\lambda_2) = \frac{d-1}{2}$, and $h(\lambda) = 0$ iff $\lambda = \lambda_1$. Assume there exists $\lambda$ in $\mathrm{Irr}\1(\mathbb{B}_N(d)\1)$ such that
    \begin{equation*}
        g(\lambda_2) \leq g(\lambda) \leq g(\lambda_1),
    \end{equation*}
    then necessarily the first term of $h(\lambda)$ is either $h(\lambda_1)$ or $h(\lambda_2)$, since it minimized for the columns $(1)$ and $(1^d)$. But since all the terms in $h(\lambda)$ are positive, then either $g(\lambda) = g(\lambda_1)$ or $g(\lambda) = g(\lambda_2)$.

    Because $\mu_1$ is the $N$-box Young diagram that maximizes the content function, then
    \begin{equation*}
        a(\mu) < a(\mu_1),
    \end{equation*}
    for all $\mu \neq \mu_1$ in $\mathrm{Irr}(\mathfrak{S}_N)$.
    
    Assume there exists $\mu$ such that $\1( \lambda_1, \mu \1) \in \Omega$ and,
    \begin{equation*}
        a(\mu_2) \leq a(\mu) \leq a(\mu_1).
    \end{equation*}
    Since $\1( (1^d), \mu \1) \in \Omega$ implies $\1( (1^d), \mu \1) \in \Gamma$, then $\1( \lambda_1, \mu \1) \in \Gamma$, and by definition $r(\mu) = d$. Thus necessarily $c(\mu_2) \leq c(\mu)$, which implies that the first row of $\mu_2$ is of size at most $n-d+1$. But the content of a Young diagram is a non-increasing function of the first row's size, i.e. given two Young diagrams $\lambda$ and $\mu$ with the same number of boxes, if $c(\lambda) \leq c(\mu)$ then $\lambda_1 \leq \mu_1$. Then $\mu_2$ and $\mu$ share the same first row, and $\mu_2 = \mu$.
\end{proof}

\begin{theorem} \label{chap5:thm:smallDimension}
    If $N \geq d$ and both $d$ and $N$ are odd, then,
    \begin{equation*}
        p^*(N,d) = \min \3\{ \frac{2 d + 1}{2 d N + 1}, \frac{1}{N - 1} \3\}.
    \end{equation*}
\end{theorem}
\begin{proof}
    Let $\lambda_1, \lambda_2$ and $\mu_1, \mu_2, \mu_3$ be the partitions from \hyperref[chap5:eq:5partitions]{Eq.~(\ref*{chap5:eq:5partitions})}. Two cases are considered in the proof, $p^*(N,d) = g(\lambda_1)$ when $N \geq 2d + 3$, and $p^*(N,d)$ lies at intersection of the affine functions $f_{\lambda_2, \mu_1}$ and $f_{\lambda_1, \mu_2}$, when $N \leq 2d + 1$ (see \hyperref[chap5:fig:twoRegimes]{Figure~\ref*{chap5:fig:twoRegimes}}).
    
    Since 
    \begin{align*}
        c(\mu_3) &= \sum_{i=1}^{2k+1} \3( - \frac{d(d-1)}{2}+(i-1)d \3) + \sum_{i=2k+2}^{2k+3} \3( - \frac{m(m-1)}{2}+(i-1)m \3) \\
        &= \frac{d(2k+1)(2k+1-d)}{2}+m(4k+4-m),
    \end{align*} 
    and $N-d=2kd+2m$ with $m \in \{ 0, \dots, d-1 \}$, then
    \begin{align*}
        g(\lambda_1) - a(\mu_3) &= \frac{1}{N-1} - \frac{1}{d-1} \3( \frac{2d \cdot c(\mu_3)}{N (N - 1)} -1 \3) \\
        &= \frac{N(d+N-2) - 2 d \cdot c(\mu_3)}{N(N-1)(d-1)} \\
        &= \frac{(d + 2)(2m^2 - 2dm - N + dN)}{N(N-1)(d-1)} \\
        &\geq \frac{(d + 2)(-(d-1)(d+1)/2 + d(d-1))}{N(N-1)(d-1)} \\
        &= \frac{(d + 2)(d - 1)}{2N(N-1)} > 0,
    \end{align*}
    where in the first inequality $N \geq d$ was used, and that the minimum of $2m^2 - 2dm$ on the domain $m \in \{0,\dots,d-1\}$ is achieved for $m=\frac{d-1}{2}$. Geometrically this means that the point $(0,a(\mu_3))$ is always lower than $(\tilde{x},g(\lambda_1))$.
    
    Suppose that $N \geq 2d + 3$, then the relation
    \begin{align*}
        g(\lambda_1) < a(\mu_2),
    \end{align*}
    holds by \hyperref[chap5:lem:twoRegimes]{Lemma~\ref*{chap5:lem:twoRegimes}}. Therefore $p^*(N,d) \geq g(\lambda_1)$ since the optimal point should be above the intersection of the affine functions $f_{\lambda_1, \mu_2}$ and $f_{\lambda_1, \mu_3}$ that is, above $g(\lambda_1)$. But since $g(\lambda) \leq g(\lambda_1)$ for all $\lambda$ in $\mathrm{Irr}\1(\mathbb{B}_N(d)\1)$ by \hyperref[chap5:lem:lambdaMuOrders]{Lemma~\ref*{chap5:lem:lambdaMuOrders}}, it must be $p^*(N,d) = g(\lambda_1)$.
    
    Suppose that $N \leq 2d + 1$, then the relation
    \begin{align*}
        a(\mu) \leq a(\mu_2),
    \end{align*}
    holds for all $\1( \lambda_1, \mu \1) \in \Omega$, by \hyperref[chap5:lem:lambdaMuOrders]{Lemma~\ref*{chap5:lem:lambdaMuOrders}}. Then $p^*(N,d)$ lies above the affine function $f_{\lambda_1, \mu_2}$. But $g(\lambda) \leq g(\lambda_1)$ for all $\lambda$ in $\mathrm{Irr}\1(\mathbb{B}_N(d)\1)$, by \hyperref[chap5:lem:lambdaMuOrders]{Lemma~\ref*{chap5:lem:lambdaMuOrders}}, then $p^*(N,d)$ must lies on the affine function $f_{\lambda_1, \mu_2}$, at the intersection with another affine function $f_{\lambda, \mu}$ with $g(\lambda) \leq a(\mu)$. Among all such affine functions there are no functions with $\lambda = \lambda_1$ due to \hyperref[chap5:lem:lambdaMuOrders]{Lemma~\ref*{chap5:lem:lambdaMuOrders}}. Because $g(\lambda) \leq g(\lambda_2)$ for all $\lambda \neq \lambda_1$ in $\mathrm{Irr}\1(\mathbb{B}_N(d)\1)$ and $a(\mu) \leq a(\mu_1)$ for all $\mu \in \mathrm{Irr}(\mathfrak{S}_N)$, by \hyperref[chap5:lem:lambdaMuOrders]{Lemma~\ref*{chap5:lem:lambdaMuOrders}}, it must be that this function is $f_{\lambda_2, \mu_1}$. Therefore $p^*(N,d)$ lies at intersection of the affine functions $f_{\lambda_2, \mu_1}$ and $f_{\lambda_1, \mu_2}$. In order to find the intersection of $f_{\lambda_2, \mu_1}$ and $f_{\lambda_1, \mu_2}$, the equation $f_{\lambda_2, \mu_1}(x^*) = f_{\lambda_1, \mu_2}(x^*)$ must be solved, which gives $x^* = \frac{4d}{(d-1)(N-1)(2dN+1)}$ and $p^*(N,d) = \frac{2d+1}{2dN+1}$.

    In conclusion, when $N \geq 2d + 3$ then $p^*(N,d) = \frac{1}{N-1}$, and when $N \leq 2 d + 1$ then $p^*(N,d) = \frac{2d+1}{2dn+1}$, which is equivalent to the statement of the theorem.
\end{proof}

\subsection{Optimal value}

Using \hyperref[chap5:thm:largeDimension]{Theorem~\ref*{chap5:thm:largeDimension}} and \hyperref[chap5:thm:smallDimension]{Theorem~\ref*{chap5:thm:smallDimension}}, the dual optimization problem \hyperref[chap5:eq:dualSDP]{Eq.~(\ref*{chap5:eq:dualSDP})} has solution,
\begin{equation*}
    p^*(N,d) =
    \begin{cases}
        \frac{1}{N + N \bmod 2 - 1} &\text{ if $d > N$ or either $d$ or $N$ is even} \\
        \min \1\{ \frac{2 d + 1}{2 d N + 1}, \frac{1}{N - 1} \1\} &\text{ if $N \geq d$ and both $d$ and $N$ are odd}.
    \end{cases}
\end{equation*}
Hence $p^*(N,d)$ is decreasing with respect to $N$, but is not monotonic with respect to $d$.

Slater's condition for strong duality of \textsc{sdp} holds true in this case, from \hyperref[chap5:sec:lowerBound]{Section.~\ref*{chap5:sec:lowerBound}}, hence both optimal values of the optimization problem \hyperref[chap5:eq:SDP]{Eq.~(\ref*{chap5:eq:SDP})} and the dual optimization problem \hyperref[chap5:eq:dualSDP]{Eq.~(\ref*{chap5:eq:dualSDP})} are equal, i.e. $p(N,d) = p^*(N,d)$. The first values of $p(N,d)$ are summarized in the following table (in \tikz{\draw[fill = gray!30!white] (0,0) rectangle (1.5ex,1.5ex);} the values of $p(N,d)$ for which the isotropic states $\rho_e$ are separable):
\begin{center}
    \begin{NiceTabular}{c|cccccccc}[columns-width = 2em, cell-space-limits = 0.25em]
        \CodeBefore
            \rectanglecolor{gray!30!white}{2-3}{2-9}
            \rectanglecolor{gray!30!white}{3-5}{3-9}
            \rectanglecolor{gray!30!white}{4-5}{4-9}
            \rectanglecolor{gray!30!white}{5-7}{5-9}
            \rectanglecolor{gray!30!white}{6-7}{6-9}
            \rectanglecolor{gray!30!white}{7-9}{7-9}
            \rectanglecolor{gray!30!white}{7-9}{8-9}
        \Body
        \diagbox{$d$}{$N$} & 2 & 3 & 4 & 5 & 6 & 7 & 8 & 9 \\ \hline
        2 &	1 &	\nicefrac{1}{3}  & \nicefrac{1}{3} & \nicefrac{1}{5}   & \nicefrac{1}{5} & \nicefrac{1}{7}   & \nicefrac{1}{7} & \nicefrac{1}{9} \\
        3 &	1 &	\nicefrac{7}{19} & \nicefrac{1}{3} & \nicefrac{7}{31}  & \nicefrac{1}{5} & \nicefrac{7}{43}  & \nicefrac{1}{7}& \nicefrac{1}{8} \\
        4 &	1 &	\nicefrac{1}{3}  & \nicefrac{1}{3} & \nicefrac{1}{5}   & \nicefrac{1}{5} & \nicefrac{1}{7}   & \nicefrac{1}{7}& \nicefrac{1}{9} \\
        5 &	1 &	\nicefrac{1}{3}  & \nicefrac{1}{3} & \nicefrac{11}{51} & \nicefrac{1}{5} & \nicefrac{11}{71} & \nicefrac{1}{7}& \nicefrac{11}{91} \\
        6 &	1 &	\nicefrac{1}{3}  & \nicefrac{1}{3} & \nicefrac{1}{5}   & \nicefrac{1}{5} & \nicefrac{1}{7}   & \nicefrac{1}{7}& \nicefrac{1}{9} \\
        7 & 1 & \nicefrac{1}{3}  & \nicefrac{1}{3} & \nicefrac{1}{5}   & \nicefrac{1}{5} & \nicefrac{5}{33}  & \nicefrac{1}{7}& \nicefrac{15}{127} \\
        8 & 1 & \nicefrac{1}{3}  & \nicefrac{1}{3} & \nicefrac{1}{5}   & \nicefrac{1}{5} & \nicefrac{1}{7}   & \nicefrac{1}{7}& \nicefrac{1}{9} \\
        9 & 1 & \nicefrac{1}{3}  & \nicefrac{1}{3} & \nicefrac{1}{5}   & \nicefrac{1}{5} & \nicefrac{1}{7}   & \nicefrac{1}{7}& \nicefrac{19}{163} \\
    \end{NiceTabular}
\end{center}

In general, it is not known which quantum state $\rho$ gives the optimal value $p(N, d)$, but using the commutation relation \hyperref[chap5:eq:orthogonalCommutation]{Eq.~(\ref*{chap5:eq:orthogonalCommutation})}, it must be in the algebra,
\begin{equation*}
    \Span_{\mathbb{C}} \set{\psi(\sigma)}{\sigma \in \mathbb{B}_{N}}.
\end{equation*}

For example when $N=3$ and $d=3$ the optimal quantum state $\rho$ for the optimisation problem \hyperref[chap5:eq:SDP]{Eq.~(\ref*{chap5:eq:SDP})} is:
\begin{align*}
    \rho = \frac{1}{57} \3[
    &

        \caption{$N \geq 2d + 3$}
        \label{chap5:fig:Ngeq2dPlus3}
     \end{subfigure}
     \hfill
     \begin{subfigure}[b]{0.45\textwidth}
        \centering
        %
         \caption{$N \leq 2 d + 1$}
         \label{chap5:fig:Nleq2dPlus1}
    \end{subfigure}
    \captionsetup{margin=0.5cm}
    \caption
        [The optimal value of the dual optimization problem lies at an intersection of the set of affine functions $f_{\lambda, \mu}(x)$. When $N \geq 2 d + 3$ (Figure~\ref*{chap5:fig:Ngeq2dPlus3}, with $N = 9$ and $d = 3$), then $p^*(N,d) = g(\lambda_1)$. When $N \leq 2 d + 1$ (Figure~\ref*{chap5:fig:Nleq2dPlus1}, with $N = 5$ and $d = 3$), then $p^*(N,d)$ lies at the intersection of $f_{\lambda_2, \mu_1}$ and $f_{\lambda_1, \mu_2}$. In in red the $\max$ over all affine functions, in in gray the affine functions from the set $\Gamma$, and in in blue the affine functions from the set $\Omega \setminus \Gamma$.]
        {The optimal value of the dual optimization problem lies at the intersection of the set of affine functions $f_{\lambda, \mu}(x)$. When $N \geq 2 d + 3$ (Figure~\ref*{chap5:fig:Ngeq2dPlus3}, with $N = 9$ and $d = 3$), then $p^*(N,d) = g(\lambda_1)$. When $N \leq 2 d + 1$ (Figure~\ref*{chap5:fig:Nleq2dPlus1}, with $N = 5$ and $d = 3$), then $p^*(N,d)$ lies at the intersection of $f_{\lambda_2, \mu_1}$ and $f_{\lambda_1, \mu_2}$. In \tikz{\draw[fill = red!60!white] (0,0) rectangle (1.5ex,1.5ex);} the $\max$ over all affine functions, in \tikz{\draw[fill = gray!80!white] (0,0) rectangle (1.5ex,1.5ex);} the affine functions from the set $\Gamma$, and in \tikz{\draw[fill = blue!80!white] (0,0) rectangle (1.5ex,1.5ex);} the affine functions from the set $\Omega \setminus \Gamma$.} \label{chap5:fig:twoRegimes}
\end{figure}

\titleformat{\chapter}
    [display]
    {\centering\Huge\bfseries}
    {\vspace{0pt plus 0.25fill} \appendixname\ \thechapter}
    {15pt}
    {}
    [
    \vfill
    \flushleft\normalsize\bfseries
    {
        \startcontents[chapters]
        \printcontents[chapters]{}{1}[2]{\rule[0.25em]{47pt}{2pt} \: {\Large Appendix contents} \: \leaders\hbox{\rule[0.25em]{2pt}{2pt}\kern-1pt}\hfill\null}
    }
    \clearpage
    ]
 
 \chapter{Elements of representation theory} \label{appA}

This Appendix provides a comprehensive overview of the fundamental principles of representation theory that are used throughout the present thesis. Its aim is to focus primarily on the investigation of representations of finite groups, with a specific emphasis on the symmetric group, and on those of infinite matrix groups, with a specific emphasis on the unitary group.

The sections devoted to finite groups provide a complete set of proofs, while more detailed treatments of the subject can be found in the books \cite{fulton2013representation,serre1977linear,sagan2013symmetric,james2006representation}. Regarding matrix groups, only the principal results are presented, with books \cite{goodman2009symmetry,hall2013lie,brocker2013representations,baker2003matrix,bump2004lie,adams1982lectures,simon1996representations} serving as references for the corresponding proofs. A thorough investigation of the representation theory with a particular emphasis on its application in the domain of physics, may be found in the books \cite{landsberg2012tensors,boerner1970representations,zhelobenko1973compact,sternberg1995group,tung1985group}.

An alternative approach to address the representations of the symmetric group, which has been proposed by A. Okounkov and A. Vershik, can be found in \cite{vershik2005new,ceccherini2010representation}.

\section{Finite groups}

\subsection{Representations of finite groups}

Let $G$ be a finite group, a \emph{representation} of $G$ is a pair $\emph{(\rho, V)}$, where $V$ is a complex vector space with dimension $d$ and $\rho: G \to \mathrm{GL}(V)$ is a group homomorphism, i.e. $\rho(g \cdot h) = \rho(g) \rho(h)$ holds for all $g$ and $h$ in $G$. Specifically, $\rho$ satisfies the condition $\rho(1_G) = I$, where $I$ is the identity matrix on $V$. As a corollary of this property, it follows that $\rho \1( g^{\shortminus 1} \1) = \rho(g)^{\shortminus 1}$ for every $g \in G$.

Consider two representations $(\rho, V)$ and $(\sigma, W)$ of $G$, an \emph{intertwining map} between these representations is a linear map $\phi: V \to W$ that satisfies $\phi \circ \rho(g) = \sigma(g) \circ \phi$ for all $g \in G$. If such a map is an isomorphism, then the representations are said to be \emph{equivalent}, and $\rho(g) = \phi^{\shortminus 1} \circ \sigma(g) \circ \phi$ for all $g \in G$. The set of all intertwining maps between $(\rho, V)$ and $(\sigma, W)$ is denoted $\emph{\mathrm{Hom}_G(V, W)}$, and forms a complex vector space.

If $(\rho_V, V)$ is a representation of $G$, then the \emph{dual representation} $(\rho_{V^*}, V^*)$ is defined for all $g \in G$ by
\begin{equation*}
    \rho_{V^*} (g) \coloneqq \rho_V {\1( g^{\shortminus 1} \1)}^\T,
\end{equation*}
where $\cdot^\T$ denotes the transpose.

The natural bilinear pairing $\scalar{\cdot}{\cdot}$ between complex vector spaces $V$ and $V^*$ is defined as $\scalar{f}{x} \coloneqq f(x)$ for all $x \in V$ and $f \in V^*$. Then, the dual representation $(\rho_{V^*}, V^*)$ preserves this natural bilinear pairing, specifically, $\scalar{\rho_{V^*} (g) f}{\rho_V (g) x} = \scalar{f}{x}$ holds for any $g \in G$ and for all $x \in V$ and $f \in V^*$.

\begin{example*}
    Consider a finite group $G$ of order $n$, and let $(\rho, V)$ be a representation of $G$. Let $g$ be an arbitrary element of $G$, and let $A \in \mathrm{GL}(V)$ be the image of $g$ under $\rho$, i.e. $\rho(g) = A$. It follows that
    \begin{equation*}
        A^n = I.
    \end{equation*}
    The characteristic polynomial $X^n - 1$ of $A$ factors into $n$ distinct eigenvalues, which are the $n$-th roots of unity. As a result, it can be deduced that $A$ is diagonalizable. Furthermore, if $G$ is an abelian group, all of the matrices $\rho(g)$ are simultaneously diagonalizable. In \hyperref[appA:sec:reducibility]{Section~\ref*{appA:sec:reducibility}}, the concept of block diagonalization will be explored in the context of non-abelian groups.
\end{example*}

Given a finite group $G$ and a representation $(\rho, V)$ of $G$, the complex vector space $V$ is said to carry the representation $\rho$ of G. For clarity, it is appropriate to succinctly refer to the complex vector space $V$ as a representation of $G$ when no ambiguity surrounding the mapping $\rho$ exists. For any $g \in G$ and $v \in V$, the action of $g$ on $v$ can be denoted by $\emph{g \cdot v}$, as a compact alternative to the expression $\rho(g)(v)$.

\begin{remark*}
    The representation of a finite group $G$ is a mapping of the group onto a set of matrices that operate on a complex vector space $V$. This mapping preserves the underlying structure of the group, such that the group law of $G$ is equivalent to matrix multiplication on $V$. As a result, the study of representations of $G$ can be approached using the mathematical framework of linear algebra, rather than through direct examination of the group itself.
\end{remark*}

\subsection{Reducibility} \label{appA:sec:reducibility}

Let $V$ be a representation of a finite group $G$, and let $W$ be a subspace of $V$. If $W$ is invariant under $G$, that is to say, for all $g \in G$ and $w \in W$, the element $g \cdot w$ belongs to $W$, then $W$ is called a \emph{subrepresentation}. If $V$ possesses exactly two distinct subrepresentations, namely the trivial zero subspace and $V$ itself, then the representation is said to be \emph{irreducible}, otherwise the representation is considered to be \emph{reducible}.

\begin{theorem}[Maschke] \label{appA:thm:maschke}
    Let $V$ be a representation of a finite group $G$, and suppose that $W$ is a subrepresentation of $V$. Then, there exists a complementary subrepresentation $W^\perp$ of $V$, such that $V$ is the direct sum of $W$ and $W^\perp$, i.e. $V = W \oplus W^\perp$.
\end{theorem}
\begin{proof}
    Let $\scalar{\cdot}{\cdot}$ be any Hermitian inner product on the complex vector space $V$. The $G$-invariant Hermitian inner product $\scalar{\cdot}{\cdot}_G$ on $V$ is defined by
    \begin{equation*}
        \scalar{v}{w}_G \coloneqq \frac{1}{|G|} \sum_{g \in G} \scalar{g \cdot v}{g \cdot w}.
    \end{equation*}
    Then $\scalar{g \cdot v}{g \cdot w}_G = \scalar{v}{w}_G$ holds for any $g \in G$.
    
    Let $W^\perp = \set{v \in V}{\scalar{v}{w}_G = 0 \text{ for all } w \in W}$ be the orthogonal complement of $W$ with respect to the inner product $\scalar{\cdot}{\cdot}_G$. It follows that $V = W \oplus W^\perp$. Furthermore, for any $g \in G$ and for all $v \in W^\perp$ and $w \in W$, the $G$-invariance of $\scalar{\cdot}{\cdot}_G$ implies
    \begin{equation*}
        \scalar{g \cdot v}{w}_G = \scalar{v}{g^{\shortminus 1} \cdot w}_G = 0,
    \end{equation*}
    as $g^{\shortminus 1} \cdot w \in W$. Therefore $g \cdot v \in W^\perp$ for all $v \in W^\perp$ and $g \in G$, and thus $W^\perp$ is a subrepresentation.
\end{proof}
In accordance with Maschke's \hyperref[appA:thm:maschke]{Theorem~\ref*{appA:thm:maschke}}, it can be established that if a finite group $G$ possesses a reducible representation, it must necessarily be decomposed into the direct sum of at least two subrepresentations. These subrepresentations may in turn be further decomposed into further direct sums, if they are themselves reducible. This process of decomposition may be repeated until the representation is fully decomposed into the direct sum of irreducible representations of $G$. This result holds true for any finite group and is referred to as the concept of \emph{complete reducibility} of finite groups.

\begin{remark*}
    Maschke's \hyperref[appA:thm:maschke]{Theorem~\ref*{appA:thm:maschke}} asserts that for any finite group $G$ and its representation $(\rho, V)$, the matrices $\rho(g)$ for all $g \in G$ are simultaneously block-diagonalizable. Given the decomposition of the underlying complex vector space $V$ into $V = W \oplus W^\perp$, where $W$ is a subrepresentation of $V$, it follows that for each $g \in G$, the matrix $\rho(g)$ block-diagonalizes as 
    \begin{equation*}
        \rho(g) =
        \begin{pmatrix}
            \rho_W(g) & 0 \\
            0 & \rho_{W^\perp}(g)
        \end{pmatrix}.
    \end{equation*}
    Here, $\rho_W(g)$ and $\rho_{W^\perp}(g)$ are the restrictions of $\rho(g)$ onto $W$ and $W^\perp$, respectively. This result highlights the crucial property of simultaneously block-diagonalizability of matrices in representation theory.
\end{remark*}

Maschke's \hyperref[appA:thm:maschke]{Theorem~\ref*{appA:thm:maschke}} plays a crucial role in the field of representation theory of finite groups. The theorem states that every representation of a finite group can be decomposed into a direct sum of irreducible representations. This implies that to comprehend any representation of a finite group, it is sufficient to have a complete understanding of its irreducible representations. This result leads to the fundamental question of determining the number of irreducible representations for a given finite group, which is addressed in  \hyperref[appA:sec:characterTheory]{Section~\ref*{appA:sec:characterTheory}}.

\begin{lemma}[Schur] \label{appA:lem:Schur}
    If $V$ and $W$ are two irreducible representations of a finite group $G$, and $\phi \in \mathrm{Hom}_G(V, W)$ is a nonzero intertwining map, then $\phi$ is an isomorphism.
\end{lemma}
\begin{proof}
    Let $v \in \ker \phi$, then for all $g \in G$, since $\phi$ is an intertwining map it follows that
    \begin{equation*}
        g \cdot \phi(v) = \phi (g \cdot v) = 0.
    \end{equation*}
    Thus $\ker \phi$ is a subrepresentation of $V$, and since $V$ is irreducible and $\phi$ is nonzero, necessarily $\ker \phi = \{0\}$. Therefore $\phi$ is injective.
    
    Conversely, let $v \in V$, then for all $g \in G$, since $\phi$ is an intertwining map it follows that
    \begin{equation*}
        g \cdot \phi(v) = \phi (g \cdot v) \in \Ima \phi.
    \end{equation*}
    Thus $\Ima \phi$ is a subrepresentation of $W$, and since $W$ is irreducible and $\phi$ is nonzero, necessarily $\Ima \phi = W$. Hence $\phi$ is surjective.
\end{proof}

Schur's \hyperref[appA:lem:Schur]{Lemma~\ref*{appA:lem:Schur}} is a foundational principle in the field of representation theory. It has two particularly noteworthy implications for a finite group $G$.

Firstly, if $V$ is an irreducible representation of a $G$, and $\phi$ belongs to $\mathrm{Hom}_G(V)$, then $\phi$ is a homothety, meaning that $\phi$ is proportional to the identity matrix, with a scalar factor $\lambda \in \mathbb{C}$, i.e. \begin{equation*}
    \phi = \lambda \cdot I.
\end{equation*}

Secondly, if $V$ is a representation of $G$, there exists a unique decomposition of $V$ into a direct sum of non-equivalent irreducible representations, expressed as
\begin{equation*}
    V \simeq V^{\oplus n_1}_1 \oplus \cdots \oplus V^{\oplus n_k}_k,    
\end{equation*}
where $n_i$ is referred to as the \emph{multiplicity} of the irreducible representation $V_i$.

\begin{remark*}
    In the decomposition of complex vector space $V$ as $V \simeq V^{\oplus n_1}_1 \oplus \cdots \oplus V^{\oplus n_k}_k$, both the subspaces $V_i$ and the multiplicities $n_i$ are unique. However, it must be noted that the direct sum decomposition of each $V^{\oplus n_i}_i$ into $n_i$ copies of $V_i$ is not guaranteed to be unique in general.
\end{remark*}

For any finite group $G$, there is a unique, up to isomorphism, representation referred to as the \emph{trivial representation}, which is invariant under the action of all elements of $G$. This representation is characterized by the property that all elements of $G$ act as the identity. Furthermore, this representation is one-dimensional, and thus irreducible.

Let $V$ be a representation of $G$. The invariant subspace of $V$ with respect to $G$, denoted as $\emph{V_G}$, can be defined as follows:
\begin{equation*}
    V_G \coloneqq \set{v \in V}{g \cdot v = v, \forall g \in G}.
\end{equation*}
Then $V_G$ can be decomposed into a direct sum of irreducible representations. In particular, each summand of this direct sum is isomorphic to the trivial representation.
\begin{proposition} \label{appA:prop:projectorTrivialIsotypicComponent}
    Let $V$ be a representation of a finite group $G$. Consider the map $\phi: V \to V$ defined as follows: for any $v \in V$, 
    \begin{equation*}
        \phi(v) \coloneqq \frac{1}{|G|} \sum_{g \in G} g \cdot v.
    \end{equation*}
    Then $\phi$ is a projector onto $V_G$.
\end{proposition}
\begin{proof}
    Let $v \in V_G$, which implies that $\phi(v) = v$ according to the definition of $V_G$. As a result, $V_G \subset \Ima \phi$. For any $v \in V$, the action of any element $g \in G$ on $\phi(v)$ results in $g \cdot \phi(v) = \phi(v)$, thus $\Ima \phi \subset V_G$. Furthermore, for all $v \in V$,
    \begin{align*}
        \phi \circ \phi (v) &= \frac{1}{|G|^2} \sum_{g \in G} \sum_{g \in G} (g \cdot h) \cdot v \\
        &= \frac{|G|}{|G|^2} \sum_{g \in G} g \cdot v \\
        &= \phi(v),
    \end{align*}
    which implies that $\phi \circ \phi = \phi$.
\end{proof}

A representation $V$ of a finite group $G$ is considered to be \emph{unitary} if there exists a Hermitian inner product $\scalar{\cdot}{\cdot}$ defined on $V$ such that the property of unitarity is satisfied, that is, for all $g \in G$, it holds that $\scalar{g \cdot u}{g \cdot v} = \scalar{u}{v}$. Given any Hermitian inner product $\scalar{\cdot}{\cdot}$ in $V$, a new Hermitian inner product $\emph{\scalar{\cdot}{\cdot}_G}$ can be defined on $V$ as follows:
\begin{equation*}
    \scalar{u}{v}_G \coloneqq \frac{1}{|G|} \sum_{g \in G} \scalar{g \cdot u}{g \cdot v}.
\end{equation*}
Then for all $g \in G$, it holds that $\scalar{g \cdot u}{g \cdot v}_G = \scalar{u}{v}_G$. This means that every representation of a finite group can be considered to be unitary with respect to $\scalar{\cdot}{\cdot}_G$.

\begin{example*}
    Let $\mathfrak{S}_3$ be the symmetric group over $3$ elements, and $V$ be the complex vector space $\mathbb{C}^3$ with basis $e_1$, $e_2$ and $e_3$. The \emph{natural representation} of $\mathfrak{S}_3$ associates each element $\sigma$ in $\mathfrak{S}_3$ with its permutation matrix $P_\sigma$. This matrix represents the permutation of the $3$ coordinates of $V$ according to $\sigma$. Then the invariant subspace $V_{\mathfrak{S}_3}$ is the one-dimensional subspace:
    \begin{equation*}
        \Span_{\mathbb{C}} (1, 1, 1). 
    \end{equation*}
\end{example*}

Let $V$ be an irreducible representation of a finite group $G$, and consider any nonzero vector $v$ in $V$. It follows that $v$ generates $V$ under the action of $G$, that is, $G \cdot v = V$. Otherwise, if the orbit $G \cdot v$ were a proper subspace of $V$, then by the definition of irreducibility, $G \cdot v$ would be an invariant subspace of $V$, contradicting the assumption that $V$ is irreducible.

\subsection{Character theory} \label{appA:sec:characterTheory}

Let $(\rho, V)$ be a representation of a finite group $G$. The \emph{character} of $(\rho, V)$ is the map $\emph{\chi_V} : G \to \mathbb{C}$ defined on $g \in G$ as follows:
\begin{equation*}
    \chi_V (g) \coloneqq \Tr \1[ \rho(g) \1].
\end{equation*}
A \emph{class function} on $G$ is complex function $f$ that remains constant on the conjugacy classes of $G$, i.e. $f \1(h \cdot g \cdot h^{\shortminus 1} \1) = f(g)$, for all $g, h \in G$. The collection of all class functions on $G$ forms a complex vector space with dimension equal to the number of conjugacy classes of $G$. Due to the cyclic property of the trace, the characters are class functions. The concept of characters plays a central role in the representation theory of finite groups. It provides a means of calculating important quantities, such as the dimension of a representation. For instance, let $V$ be a representation of $G$, and let $\phi$ be the projector onto $V_G$, defined in \hyperref[appA:prop:projectorTrivialIsotypicComponent]{Proposition~\ref*{appA:prop:projectorTrivialIsotypicComponent}}, then $\chi_V (1_G) = \dim V$ and $\chi_V (\phi) = \dim V_G$, where $\chi_V (\phi)$ is defined linearly on each summand of $\phi$.

\begin{proposition} \label{appA:prop:characterProperties}
    Consider two representations, $(\rho_V, V)$ and $(\rho_W, W)$, of a finite group $G$. Then
    \begin{equation*}
        \chi_{V \oplus W} = \chi_V + \chi_W, \qquad \chi_{V \otimes W} = \chi_V \cdot \chi_W, \quad \text{and} \quad \chi_{V^*} = \bar{\chi}_V,
    \end{equation*}
    where $\bar{\cdot}$ denotes the complex conjugate.
\end{proposition}
\begin{proof}
    Let $g$ be in $G$. By the properties of the trace,
    \begin{align*}
        \chi_{V \oplus W} (g) &= \Tr \1[ \rho_V (g) \1] + \Tr \1[ \rho_W (g) \1] \\
        \chi_{V \otimes W} (g) &= \Tr \1[ \rho_V (g) \1] \cdot \Tr \1[ \rho_W (g) \1].
    \end{align*}
    Then $\chi_V (g) = \Tr \1[ \rho_V (g) \1]$ is the sum of the eigenvalues of the diagonalizable $\rho_V (g)$. Furthermore, since the eigenvalues, $\lambda_i$, of $\rho_V (g)$ are roots of unity, it follows that the inverse of each eigenvalue is equal to its complex conjugate, i.e. $\lambda^{\shortminus 1}_i = \bar{\lambda}_i$. Finally
    \begin{equation*}
        \chi_{V^*} (g) = \Tr \2[ \rho_V \1( g^{\shortminus 1} \1) \2] = \bar{\chi}_V (g).
    \end{equation*}
\end{proof}
\begin{corollary*}
    Let $(\rho_V, V)$ be a representations of a finite group $G$, then $\bar{\chi}_V(g) = \chi_V \1( g^{\shortminus 1} \1)$, for all $g \in G$.
\end{corollary*}

From the character of a tensor product and a dual space, the isomorphism $\mathrm{Hom}(V, W) \simeq V^* \otimes W$ between two finite dimensional complex vector spaces $V$ and $W$, implies
\begin{equation*}
    \chi_{\mathrm{Hom}(V, W)} = \bar{\chi}_V \cdot \chi_W.
\end{equation*}

Consider a finite group $G$ and the set of class functions defined on it. Let the Hermitian inner product $\emph{\scalar{\cdot}{\cdot}_G}$ on this set of class functions be defined as follows: given two class functions $f$ and $g$ on $G$, then, 
\begin{equation*}
    \scalar{f}{g}_G \coloneqq \frac{1}{|G|} \sum_{g \in G} \bar{f}(g) \cdot h(g).
\end{equation*}
\begin{theorem} \label{appA:thm:numberIrreducibleRepresentations}
    Given a finite group $G$, the characters of its irreducible representations serve as an orthonormal basis for the space of class functions of $G$, with respect to the Hermitian inner product $\scalar{\cdot}{\cdot}_G$.
\end{theorem}
\begin{proof}
    Let $V$ and $W$ be two irreducible representations of $G$. By Schur's \hyperref[appA:lem:Schur]{Lemma~\ref*{appA:lem:Schur}}, the dimension of $\mathrm{Hom}_G (V, W)$ is either $1$ if $V \simeq W$ and $0$ otherwise. But the dimension of $\mathrm{Hom}_G (V, W)$ is also equal to the character of the projector onto ${\1( V^* \otimes W \1)}_G$. That is
    \begin{equation*}
        \frac{1}{|G|} \sum_{g \in G} \bar{\chi}_V (g) \cdot \chi_W (g) =
        \begin{cases}
            1 &\text{if } V \simeq W \\
            0 &\text{otherwise}.
        \end{cases}
    \end{equation*}
    Thus the characters of the irreducible representations of $G$ are orthonormal with respect to the Hermtitan inner product $\scalar{\cdot}{\cdot}_G$.
    
    Hence, the number of irreducible representations of $G$ is finite, and in fact smaller than the number of conjugacy classes of $G$. Let $\1( \rho_{V_i}, {V_i} \1)$ denotes the irreducible representations of $G$, and $\chi_{V_i}$ the corresponding characters. Define
    \begin{equation*}
        V \coloneqq \Span_{\mathbb{C}} \1\{ \chi_{V_i} \1\},
    \end{equation*}
    the span of these characters over the field of complex numbers.

    Let $W$ be the complex vector space of all functions from $G$ to $\mathbb{C}$. The basis of this complex vector space are the elements $\delta_g$, where $g \in G$, defined on all $h \in G$ by
    \begin{equation*}
        \delta_g(h) =
        \begin{cases}
            1 &\text{ if } g = h \\
            0 &\text{ otherwise}.
        \end{cases}
    \end{equation*}
    Let $g \in G$ and define the representation of $G$ on $W$ by $(g \cdot f)(h) = f(g^{\shortminus 1} \cdot h)$, for all $f \in W$ and $h \in G$.
    
    Let $f \in V^\perp$ be a class function in the orthogonal complement of $V$, for all irreducible representations $V_i$ of $G$, define the linear map $\phi_i: V_i \to V_i$ by
    \begin{equation*}
        \phi_i \coloneqq \frac{1}{|G|} \sum_{g \in G} \bar{f} (g) \cdot \rho_{V_i} (g).
    \end{equation*}
    Let $h \in G$, then since $f$ is a class function it holds that 
    \begin{align*}
        \rho_{V_i} \1( h^{\shortminus 1} \1) \cdot \phi_i \cdot \rho_{V_i} (h) &= \frac{1}{|G|} \sum_{g \in G} \bar{f} (g) \cdot \rho_{V_i} \1( h^{\shortminus 1} \cdot g \cdot h \1) \\
        &= \frac{1}{|G|} \sum_{g \in G} \bar{f} \1( h \cdot g \cdot h^{\shortminus 1} \1) \cdot \rho_{V_i}(g) \\
        &= \phi_i,
    \end{align*}
    and in particular $\phi_i \in \mathrm{Hom}_G(V_i)$. From Schur's \hyperref[appA:lem:Schur]{Lemma~\ref*{appA:lem:Schur}} there exists $\lambda \in \mathbb{C}$ such that $\phi_i = \lambda \cdot I$, with the equality $\Tr \1[ \phi_i \1] = \lambda \cdot \dim V_i$, that is
    \begin{align*}
        \lambda &= \frac{1}{|G| \cdot \dim V_i} \sum_{g \in G} \bar{f} (g) \cdot \Tr \1[ \rho_{V_i} (g) \1] \\
        &= \frac{1}{|G| \cdot \dim V_i} \sum_{g \in G} \bar{f} (g) \cdot \chi_{V_i} (g) \\
        &= \frac{1}{\dim V_i} \scalar{f}{\chi_{V_i}}_G.
    \end{align*}
    Thus $\lambda = 0$, since $f \in V^\perp$. In particular, from the decomposition into a direct sum of irreducible representations $W \simeq V^{\oplus n_1}_1 \oplus \cdots \oplus V^{\oplus n_k}_k$, it holds that
    \begin{align*}
        \frac{1}{|G|} \sum_{g \in G} \bar{f} (g) \cdot (g \cdot \delta_{1_G}) &= \frac{1}{|G|} \sum_{g \in G} \bar{f} \1( g^{\shortminus 1} \1) \cdot \delta_g \\
        &= 0.
    \end{align*}
    But since $\set{\delta_g}{g \in G}$ forms a basis of $W$, necessarily $f(g) = 0$ for all $g \in G$. Thus $V^\perp = \{ 0 \}$.
\end{proof}
\begin{corollary*}
    Let $V$ be a representation of a finite group $G$. The representation $V$ is irreducible if and only if $\scalar{\chi_V}{\chi_V}_G = 1$. Otherwise, $V$ decomposes into a direct sum of irreducible representations $V \simeq V^{\oplus n_1}_1 \oplus \cdots \oplus V^{\oplus n_k}_k$, with
    \begin{equation*}
        \scalar{\chi_V}{\chi_V}_G = \sum^k_{i = 1} n^2_i \quad \text{and} \quad \scalar{\chi_{V_i}}{\chi_V}_G = n_i.
    \end{equation*}
\end{corollary*}

In light of \hyperref[appA:thm:numberIrreducibleRepresentations]{Theorem~\ref*{appA:thm:numberIrreducibleRepresentations}}, it has been established that the cardinality of the set of irreducible representations of a finite group is equal to the cardinality of the set of conjugacy classes of this group. The orthogonality of characters of irreducible representations of the finite group $G$ serves as the foundation for constructing the \emph{character table} of $G$, which assigns to each irreducible representation of $G$ a unique collection of numbers, namely the characters of each conjugacy class of $G$. Additionally, given a representation $V$ of a finite group, with its decomposition into a direct sum of irreducible representations $V \simeq V^{\oplus n_1}_1 \oplus \cdots \oplus V^{\oplus n_k}_k$, the character $\chi_i$ can be used to determine the multiplicity $n_i$, i.e., the number of times a particular irreducible representation $V_i$ appears in the representation $V$.

Despite \hyperref[appA:thm:numberIrreducibleRepresentations]{Theorem~\ref*{appA:thm:numberIrreducibleRepresentations}}, in general for a finite group $G$, there is no known correspondence between the conjugacy classes of $G$ and the irreducible representations of $G$. However, in \hyperref[appA:sec:RepresentationsOfTheSymmetricGroup]{Section~\ref*{appA:sec:RepresentationsOfTheSymmetricGroup}}, this correspondence will be explicitly demonstrated for the specific case of the symmetric group $\mathfrak{S}_n$.

Consider a finite group $G$. The \emph{group algebra} of $G$, denoted by $\emph{\mathbb{C}[G]}$, is the complex vector space with the elements of $G$ as its basis. The multiplication in $\mathbb{C}[G]$ is defined as the group law of $G$ on the basis. The \emph{regular representation} of $G$ is obtained by considering $\mathbb{C}[G]$ as a representation, where for all $g \in G$ and $x = \sum_{h \in G} c_h \cdot h$ an element of $\mathbb{C}[G]$, the action of $g$ on $x$ becomes,
\begin{equation*}
    g \cdot x = \sum_{h \in G} c_h \cdot (g \cdot h) = \sum_{h \in G} c_{(g^{\shortminus 1} \cdot h)} \cdot h.
\end{equation*}
\begin{remark*}
    In this formulation, the group $G$ serves as both the complex vector space of the representation as the group algebra, and the group homomorphism through its group law. 
\end{remark*}

Consider a finite group $G$. For any element $g \in G$, the action of $g$ on $G$ through the group law results in a bijective mapping. It follows that the regular representation of $g$ has no fixed points except in the case where $g$ is the identity element $1_G$ of $G$. Additionally, all elements of $G$ in the regular representation are represented as permutation matrices.

\begin{proposition} \label{appA:prop:groupAlgebraDecompositionMultiplicity}
    Let $\mathbb{C}[G]$ be the group algebra of a finite group $G$, and consider $\mathbb{C}[G] \simeq V^{\oplus n_1}_1 \oplus \cdots \oplus V^{\oplus n_k}_k$, its decomposition into a direct sum of irreducible representations $V_i$ with multiplicities $n_i$. Then for each irreducible representation $V_i$, its multiplicity in $\mathbb{C}[G]$ is equal the dimension $\dim V_i$.
\end{proposition}
\begin{proof}
    The character of $g \in G$ for the regular representation is equal to the cardinality of the set of fixed points of the action of $g$ on $G$, thus
    \begin{equation*}
        \chi_{\mathbb{C}[G]} (g) =
        \begin{cases}
            |G| &\text{if } g = 1_G\\
            0 &\text{otherwise}.
        \end{cases}
    \end{equation*}
    Consider an irreducible representation $V_i$, which appears in the decomposition of the group algebra $\mathbb{C}[G]$ as follows: $\mathbb{C}[G] \cong V^{\oplus n_1}_1 \oplus \cdots \oplus V^{\oplus n_k}_k$, then
    \begin{align*}
        n_i &= \scalar{\chi_{V_i}}{\chi_{\mathbb{C}[G]}}_G \\
        &= \frac{1}{|G|} \sum_{g \in G} \bar{\chi}_{V_i} (g) \cdot \chi_{\mathbb{C}[G]} (g) \\
        &= \frac{1}{|G|} \bar{\chi}_{V_i} (1_G) \cdot \chi_{\mathbb{C}[G]} (1_G) \\
        &= \dim V_i.
    \end{align*}
\end{proof}
\begin{corollary*}
    Consider the group algebra $\mathbb{C}[G]$ of a finite group $G$, and its decomposition into a direct sum of irreducible representations $\mathbb{C}[G] \simeq V^{\oplus n_1}_1 \oplus \cdots \oplus V^{\oplus n_k}_k$. Then, it follows that:
    \begin{equation*}
        \dim \mathbb{C}[G] = \sum_i \dim (V_i)^2.
    \end{equation*}
\end{corollary*}

The decomposition of the group algebra $\mathbb{C}[G]$ of a finite group $G$, into a direct sum of irreducible representations is a powerful tool for understanding the structure of finite groups, and the equation
\begin{equation*}
    \mathbb{C}[G] \simeq V^{\oplus \dim V_1}_1 \oplus \cdots \oplus V^{\oplus \dim V_k}_k,
\end{equation*}
provides important information about the size and multiplicity of irreducible representations in the regular representation.

\begin{theorem*}
    Let $\mathbb{C}[G]$ denote the group algebra of a finite group $G$, and let $\mathbb{C}[G] \simeq V^{\oplus n_1}_1 \oplus \cdots \oplus V^{\oplus n_k}_k$ be the decomposition of $\mathbb{C}[G]$ into a direct sum of irreducible representations. For each irreducible representation $W$ of $G$, there exists an index $i \in \{ 1, \ldots, k \}$ such that $W$ is equivalent to $V_i$.
\end{theorem*}
\begin{proof}
    Let $U, V$ and $W$ be three complex vector spaces of finite dimension, and define the inclusion and projection maps
    \begin{equation*}
        V \mathrel{\mathop{\leftrightarrows}^{\pi_v}_{\iota_V}} V \oplus W \mathrel{\mathop{\rightleftarrows}^{\pi_W}_{\iota_W}} W,
    \end{equation*}
    on all $v \in V$ and $w \in W$ by
    \begin{align*}
        \pi_V &: (v, w) \longmapsto v & \iota_V &: v \longmapsto v + 0 \\
        \pi_W &: (v, w) \longmapsto w & \iota_W &: w \longmapsto 0 + w.
    \end{align*}
    Now let $G$ be a finite group, and assume that $U, V$ and $W$ are representations of $G$. Using the previous inclusion and projection maps, the following isomorphisms holds:
    \begin{align*}
        \mathrm{Hom}(U, V \oplus W) &\simeq \mathrm{Hom}(U, V) \oplus \mathrm{Hom}(U, W) \\
        \mathrm{Hom}_G(U, V \oplus W) &\simeq \mathrm{Hom}_G(U, V) \oplus \mathrm{Hom}_G(U, W).
    \end{align*}

    Then, in the decomposition of $\mathbb{C}[G]$ into a direct sum of irreducible representations,
    \begin{equation*}
        \mathrm{Hom}_G \1( \mathbb{C}[G], W \1) \simeq \bigoplus^k_{i=1} {\mathrm{Hom}_G(V_i, W)}^{\oplus n_i}.
    \end{equation*}
    In particular, the equality of dimensions
    \begin{equation*}
        \dim \mathrm{Hom}_G \1( \mathbb{C}[G], W \1) = \bigoplus^k_{i=1} n_i \cdot \dim \mathrm{Hom}_G(V_i, W),
    \end{equation*}
    holds, and from the Schur's \hyperref[appA:lem:Schur]{Lemma~\ref*{appA:lem:Schur}}, $\dim \mathrm{Hom}_G(V_i, W) = 1$ if and only if $V_i \simeq W$, otherwise $\dim \mathrm{Hom}_G(V_i, W) = 0$. It follows that the dimension of $\mathrm{Hom}_G \1( \mathbb{C}[G], W \1)$ is positive if and only if there exists an index $i \in \{ 1, \ldots, k \}$ such that the irreducible representation $W$ is equivalent to $V_i$.

    Let $f: \mathrm{Hom}_G \1( \mathbb{C}[G], W \1) \to W$ be the linear map defined on $h$ in $\mathrm{Hom}_G \1( \mathbb{C}[G], W \1)$ by $f(h) = h(1_G)$, and let such $h$ in $\mathrm{Hom}_G \1( \mathbb{C}[G], W \1)$ satisfying $f(h) = 0$, then for all $g \in G$,
    \begin{equation*}
        h(g) = h(g \cdot e) = g \cdot h(e) = 0,
    \end{equation*}
    thus $f$ is injective. Let $x \in W$ and consider the function $h: G \to W$ defined on $g \in G$ by $h(g) = g \cdot x$, thus $h$ can be extended linearly on $\mathbb{C}[G]$, and in particular $h \in \mathrm{Hom}_G \1( \mathbb{C}[G], W \1)$, but since $f(h) = x$, then $f$ is surjective.
    Therefore $f$ is an isomorphism and the equality of dimensions
    \begin{equation*}
        \dim \mathrm{Hom}_G \1( \mathbb{C}[G], W \1) = \dim W,
    \end{equation*}
    holds. In particular $\dim \mathrm{Hom}_G \1( \mathbb{C}[G], W \1)$ is positive.
\end{proof}

The group algebra $\mathbb{C}[G]$, of a finite group $G$, provides a unified way to study the representations of $G$ by encapsulating all the irreducible representations within its structure.

\begin{example*}
    As established in \hyperref[appA:thm:numberIrreducibleRepresentations]{Theorem~\ref*{appA:thm:numberIrreducibleRepresentations}}, the number of irreducible representations of $\mathfrak{S}_3$ is equal to the number of its conjugacy classes, in the present case $3$. These three irreducible representations consist of the trivial representation, the \emph{sign representation}, and the \emph{standard representation}. The sign representation is the one-dimensional representation on which any $\sigma \in \mathfrak{S}_3$ is sent to $\emph{\mathrm{sign}(\sigma)}$, where $\mathrm{sign}(\sigma)$ is the \emph{signature} of the permutation $\sigma$. The standard representation is the two-dimensional orthogonal complement of $\Span_{\mathbb{C}} (1, 1, 1)$ in the natural representation. Then
    \begin{equation*}
        \mathbb{C} \1[ \mathfrak{S}_3 \1] \simeq V_{\text{trivial}} \oplus V_{\text{sign}} \oplus V^2_{\text{standard}}.
    \end{equation*}
    The character table of $\mathfrak{S}_3$ is defined as the table with the irreducible representations of $\mathfrak{S}_3$ as its rows, and the conjugacy classes of $\mathfrak{S}_3$ as its columns. The entries in the table are the character values for the corresponding irreducible representations and conjugacy classes of $\mathfrak{S}_3$. The conjugacy classes of $\mathfrak{S}_3$ are $\1\{ 1_{\mathfrak{S}_3} \1\}, \1\{ (1\:2), (1\:3), (2\:3) \1\}$ and $\1\{ (1\:2\:3), (3\:2\:1) \1\}$. The character table of $\mathfrak{S}_3$ becomes:
    \begin{equation*}
        \begin{array}{c|ccc}
             \mathfrak{S}_3 & \1\{ 1_{\mathfrak{S}_3} \1\} & \1\{ (1\:2), (1\:3), (2\:3) \1\} & \1\{ (1\:2\:3), (3\:2\:1) \1\} \\
             \hline
             \chi_{\text{trivial}} & 1 & 1 & 1\\
             \chi_{\text{sign}} & 1 & -1 & 1\\
             \chi_{\text{standard}} & 2 & 0 & -1
        \end{array}
    \end{equation*}
    Note that as expected from \hyperref[appA:thm:numberIrreducibleRepresentations]{Theorem~\ref*{appA:thm:numberIrreducibleRepresentations}}, the rows are orthogonal, with respect to Hermitian inner product $\scalar{\cdot}{\cdot}_G$.
\end{example*}

Similarly to \hyperref[appA:prop:projectorTrivialIsotypicComponent]{Propositition~\ref*{appA:prop:projectorTrivialIsotypicComponent}}, given a finite group $G$, it is possible to define a projector onto each direct sum of equivalent irreducible representations $V^{\oplus n_i}_i$ that appears in the decomposition of its group algebra, $\mathbb{C}[G]$, into a direct sum of irreducible representations, $\mathbb{C}[G] \simeq V^{\oplus n_1}_1 \oplus \cdots \oplus V^{\oplus n_k}_k$. The direct sum of equivalent irreducible representations $V^{\oplus n_i}_i$, in the decomposition of its group algebra $\mathbb{C}[G] \simeq V^{\oplus n_1}_1 \oplus \cdots \oplus V^{\oplus n_k}_k$, are called the \emph{isotypic components}.

\begin{theorem} \label{appA:thm:projectorIsotypicComponent}
    Given a finite group $G$ and its group algebra $\mathbb{C}[G]$ decomposed into a direct sum of irreducible representations $\mathbb{C}[G] \simeq V^{\oplus n_1}_1 \oplus \cdots \oplus V^{\oplus n_k}_k$, for each irreducible representation $V_i$, let the map $\phi_i$, defined on $\mathbb{C}[G]$ by
    \begin{equation*}
        \phi_i \coloneqq \frac{\dim V_i}{|G|} \sum_{g \in G} \bar{\chi}_{V_i} (g) \cdot g.
    \end{equation*}
    Then $\phi_i$ is a projector onto $V^{\oplus n_i}_i$.
\end{theorem}
\begin{proof}
    Let $h \in G$, then since $\chi_{V_i}$ is a class function it holds that 
    \begin{align*}
        h^{\shortminus 1} \cdot \phi_i \cdot h &= \frac{\dim V_i}{|G|} \sum_{g \in G} \bar{\chi}_{V_i} (g) \cdot \1( h^{\shortminus 1} \cdot g \cdot h \1) \\
        &= \frac{\dim V_i}{|G|} \sum_{g \in G} \bar{\chi}_{V_i} \1( h \cdot g \cdot h^{\shortminus 1} \1) \cdot g \\
        &= \phi_i,
    \end{align*}
    and in particular $\phi_i \in \mathrm{Hom}_G \1( \mathbb{C}[G] \1)$. Consider the decomposition of the group algebra $\mathbb{C}[G]$ into a direct sum of irreducible representations $\mathbb{C}[G] \simeq V^{\oplus n_1}_1 \oplus \cdots \oplus V^{\oplus n_k}_k$, for all $j \in \{ 1, \ldots, k \}$, from Schur's \hyperref[appA:lem:Schur]{Lemma~\ref*{appA:lem:Schur}}, the restriction of $\phi_i$ to $V_j$ is an homothety $\lambda \cdot I$, with $\lambda \in \mathbb{C}$, such that from \hyperref[appA:thm:numberIrreducibleRepresentations]{Theorem~\ref*{appA:thm:numberIrreducibleRepresentations}}
    \begin{align*}
        \lambda &= \frac{\dim V_i}{|G| \cdot \dim V_j} \sum_{g \in G} \bar{\chi}_{V_i} (g) \cdot \chi_{V_j} (g) \\
        &= \frac{\dim V_i}{\dim V_j} \scalar{\chi_{V_i}}{\chi_{V_j}}_G \\
        &=
        \begin{cases}
            1 &\text{ if } V_i \simeq V_j \\
            0 &\text{ otherwise}.
        \end{cases}
    \end{align*}
    Thus, the restriction of $\phi_i$ to $V_i$ is the identity if $V_j$ is equivalent to $V_i$, and the zero map otherwise. But, since $V_j$ is a representation, for all $g \in G$, the action of $g$ on $V_j$ is a subspace of $V_j$. That is, the map $\phi_i$ does not cause any intertwining between the representations $V_i$.
    
    As a consequence, $\phi_i$ is the identity on $V^{\oplus n_i}_i$ and the zero map on $V^{\oplus n_j}_j$ such that $i \neq j$, and thus $\phi_i$ is a projector.
\end{proof}

\subsection{Representations of \texorpdfstring{$\mathfrak{S}_n$}{the symmetric group}} \label{appA:sec:RepresentationsOfTheSymmetricGroup}

The \emph{symmetric group}, denoted as $\emph{\mathfrak{S}_n}$, is the group of order $n!$ consisting of permutations of the set $\{1, \ldots, n \}$, as illustrated in previous examples. The cyclic notation convention is employed throughout this thesis. For instance, $(1\:2\:3)(4\:5)$ indicates the permutation
\begin{equation*}
    \begin{pmatrix}
    1 & 2 & 3 & 4 & 5 \\[0.5em]
    2 & 3 & 1 & 5 & 4
    \end{pmatrix},
\end{equation*}
of $\mathfrak{S}_5$. According to \hyperref[appA:thm:numberIrreducibleRepresentations]{Theorem~\ref*{appA:thm:numberIrreducibleRepresentations}}, the number of irreducible representations of a finite group $G$ is equal to the number of conjugacy classes in $G$. For the symmetric group $\mathfrak{S}_n$, the conjugacy class of a permutation is uniquely determined by the length of its decomposition into a product of disjoint cycles, which is referred to as the \emph{cycle type} of the permutation. The cycle type of a permutation can be described by an integer \emph{partition} of $n$, where each part represents the length of a cycle in the permutation. For instance, the permutation $(1\:2\:3)(4\:5)$ in $\mathfrak{S}_5$ has cycle type $(3, 2)$. It is worth noting that this justifies the use of cyclic notation in this thesis.

The notation $\emph{\lambda \vdash n}$ is used to represent an ordered partition of the positive integer $n$ into $l$ parts, denoted by $\lambda = (\lambda_1, \ldots, \lambda_l)$. This partition is defined as a collection of non-increasing positive integers that sum up to $n$, i.e.
\begin{equation*}
    \lambda_1 \geq \cdots \geq \lambda_l \geq 1 \qquad \text{and} \qquad \sum^l_{i = 1} \lambda_i = n.
\end{equation*}
A partition $\lambda \vdash n$ may be represented as a \emph{Young diagram}, which is a collection of $n$ empty boxes arranged in left-justified rows such that the $i$-th row contains $\lambda_i$ boxes. This representation is illustrated in \hyperref[appA:ex:youngDiagramandTableaux]{Example~\ref*{appA:ex:youngDiagramandTableaux}}. The \emph{conjugate} of the partition $\lambda$, denoted $\emph{\lambda^\prime}$, is the partition corresponding to transposing the Young diagram representing $\lambda$. A \emph{Young tableau} is a Young diagram where each box is assigned an integer between $1$ and $n$. A Young tableau is classified as \emph{standard} if the entries in each row and each column are strictly increasing, \emph{semistandard} if the entries in each row are weakly increasing and the entries in each column are strictly increasing, and \emph{without repetitions} if each entry appears exactly once. Notably, standard Young tableaux are without repetitions. The \emph{canonical Young tableau} is a Young diagram that is filled with consecutive integers from left to right and top to bottom.

\begin{example} \label{appA:ex:youngDiagramandTableaux}
    Consider the permutation $\sigma \in \mathfrak{S}_9$, expressed as a product of disjoint cycles, arranged in non-increasing order, by $(1\:2\:3\:4)(5\:6\:7)(8)(9)$. The cycle type of $\sigma$ is the partition $\lambda \vdash 9$ given by $\lambda = (4,3,1,1)$, where the $i$-th entry of $\lambda$ denotes the length of the $i$-th cycle in the disjoint cycle decomposition of $\sigma$. Moreover, the partition $\lambda$ can be represented using a Young diagram, which consists of $4$ rows with $4$, $3$, $1$, and $1$ boxes, respectively
    \begin{equation*}
        \ydiagram{4, 3, 1, 1}
    \end{equation*}
    The conjugate of the partition $\lambda$, namely $\lambda^\prime$ is represented using a Young diagram, which consists of $4$ rows with $4$, $2$, $2$ and $1$ boxes, respectively:
    \begin{equation*}
        \ydiagram{4, 2, 2, 1}
    \end{equation*}
    Two instances of Young tableaux that correspond to the partition $\lambda$ are provided as examples
    \begin{equation*}
        \begin{ytableau}
            1 & 2 & 3 & 4 \\
            2 & 3 & 4 \\
            3 \\
            4
        \end{ytableau}
        \qquad \text{and} \qquad
        \begin{ytableau}
            1 & 2 & 3 & 4 \\
            3 & 4 & *(lightgray) 4 \\
            4 \\
            5
        \end{ytableau}
    \end{equation*}
    The previous Young tableau presented on the left satisfies the standard definition, whereas the one on the right solely complies with the semistandard definition, primarily due to the weakly increasing gray box. Neither of them adheres to the requirement of being without repetitions. For the partition $\lambda$, the canonical Young tableau can be expressed as follows
    \begin{equation*}
        \begin{ytableau}
            1 & 2 & 3 & 4 \\
            5 & 6 & 7 \\
            8 \\
            9
        \end{ytableau}
    \end{equation*}
    It is both standard and thus without repetition.
\end{example}

Consider a partition $\lambda \vdash n$ and a Young tableau $T$ without repetitions, on this partition. An action of the symmetric group $\mathfrak{S}_n$ on the Young tableau $T$ is defined as follows: given a permutation $\sigma \in \mathfrak{S}_n$, the action of $\sigma$ on $T$ is obtained by permuting the boxes in $T$ according to the permutation induced by the entries of the boxes. For $\sigma \in \mathfrak{S}_n$, the action of $\sigma$ on $T$ is denoted by $\emph{\sigma(T)}$.

A permutation $\sigma$ in $\mathfrak{S}_n$ is said to preserve each row of $T$ if each box of a given row is permuted by $\sigma$ on that same row. Similarly, a permutation $\sigma$ in $\mathfrak{S}_n$ is said to preserve each column of $T$ if each box of a given column is permuted by $\sigma$ on that same column. Two subgroups of $\mathfrak{S}_n$ can be defined as the sets of permutations that preserve the rows and columns of $T$, namely,
\begin{align*}
    \emph{R_T} &\coloneqq \set{\sigma \in \mathfrak{S}_n}{\sigma \text{ preserves each row of } T} \\
    \emph{C_T} &\coloneqq \set{\sigma \in \mathfrak{S}_n}{\sigma \text{ preserves each column of } T}.
\end{align*}
Consider two elements of the group algebra $\mathbb{C} \1[ \mathfrak{S}_n \1]$, the row and column symmetrizers of a given Young tableau $T$, defined as follows:
\begin{align*}
    \emph{r_T} &\coloneqq \sum_{\sigma \in R_T} \sigma \\
    \emph{c_T} &\coloneqq \sum_{\sigma \in C_T} \mathrm{sign}(\sigma) \cdot \sigma.
\end{align*}
The \emph{Young symmetrizer} associated with a given Young tableau $T$ is an element $\emph{s_T}$ of $\mathbb{C} \1[ \mathfrak{S}_n \1]$, defined as the product of the row and column symmetrizers of $T$, i.e. $s_T \coloneqq r_T \cdot c_T$. It is worth noting that the row and column symmetrizers satisfy certain properties under the action of the symmetric group $\mathfrak{S}_n$. Specifically, for any $\sigma \in \mathfrak{S}_n$, 
\begin{equation*}
    \sigma \cdot r_T \cdot \sigma^{\shortminus 1} = r_{\sigma(T)} \quad \text{ and } \quad \sigma \cdot c_T \cdot \sigma^{\shortminus 1} = c_{\sigma(T)}.
\end{equation*}
As a consequence, it follows that $\sigma \cdot s_T \cdot \sigma^{\shortminus 1} = s_{\sigma(T)}$.

\begin{example*}
    Consider a partition $\lambda \vdash 3$. When $\lambda = (3)$. In the case where $\lambda = (3)$, the associated canonical Young tableau $T$ is a single row with entries $1$, $2$, and $3$
    \begin{equation*}
        T =
        \begin{ytableau}
            1 & 2 & 3
        \end{ytableau}
    \end{equation*}
    As the identity permutation is the unique element of $\mathfrak{S}_3$ that leaves each column of $T$ unchanged, the Young symmetrizer can be expressed as
    \begin{align*}
        s_T &= r_T \cdot c_T \\
        &= r_T \\
        &= \sum_{\sigma \in \mathfrak{S}_3} \sigma.
    \end{align*}
    In the case where $\lambda = (1, 1, 1)$, the associated canonical Young tableau $T$ is a single column with entries $1$, $2$, and $3$
    \begin{equation*}
        T =
        \begin{ytableau}
            1 \\
            2 \\
            3
        \end{ytableau}
    \end{equation*}
    As the identity permutation is the unique element of $\mathfrak{S}_3$ that leaves each row of $T$ unchanged, the Young symmetrizer can be expressed as
    \begin{align*}
        s_T &= r_T \cdot c_T \\
        &= c_T \\
        &= \sum_{\sigma \in \mathfrak{S}_3} \mathrm{sign}(\sigma) \cdot \sigma.
    \end{align*}
    The two preceding Young symmetrizers correspond, to within a constant factor of $\tfrac{1}{6}$, to the projectors onto the irreducible representations $V_{\text{trivial}}$ and $V_{\text{sign}}$ in the decomposition of $\mathbb{C} \1[ \mathfrak{S}_3 \1]$ into a direct sum of irreducible representations:
        \begin{equation*}
        \mathbb{C} \1[ \mathfrak{S}_3 \1] \simeq V_{\text{trivial}} \oplus V_{\text{sign}} \oplus V^2_{\text{standard}},
    \end{equation*}
    as established by \hyperref[appA:thm:projectorIsotypicComponent]{Theorem~\ref*{appA:thm:projectorIsotypicComponent}}.
\end{example*}

In order to obtain a comprehensive collection of inequivalent irreducible representations of the symmetric group $\mathfrak{S}_n$, it is necessary to first establish several requisite technical lemmas.

\begin{lemma} \label{appA:lem:YoungSymmetrizerProductSameShapes}
    Let $\lambda \vdash n$ be a partition of $n$, let $T$ be a Young tableau on $\lambda$ without repetitions, and let $x \in \mathbb{C} \1[ \mathfrak{S}_n \1]$. Then, there exists $\mu \in \mathbb{C}$ such that the following equation holds: $s_T \cdot x \cdot s_T = \mu \cdot s_T$.
\end{lemma}
\begin{proof}
    To begin, it will be proven that if an element $y \in \mathbb{C} \1[ \mathfrak{S}_n \1]$ satisfies the condition $\sigma \cdot y \cdot \tau = \mathrm{sign}(\tau) \cdot y$, for all permutations $\sigma \in R_T$ and $\tau \in C_T$, then it follows that $y = \mu \cdot s_T$, for some scalar $\mu \in \mathbb{C}$. Consider such an element $y \coloneqq \sum_{\pi \in \mathfrak{S}_n} c_\pi \cdot \pi$ of the group algebra $\mathbb{C} \1[ \mathfrak{S}_n \1]$, then for all permutations $\sigma \in R_T$ and $\tau \in C_T$,
    \begin{equation*}
        \sigma \cdot \3( \sum_{\pi \in \mathfrak{S}_n} c_\pi \cdot \pi \3) \cdot \tau = \sum_{\pi \in \mathfrak{S}_n} c_\pi \cdot (\sigma \cdot \pi \cdot \tau) = \mathrm{sign}(\tau) \cdot \sum_{\pi \in \mathfrak{S}_n} c_\pi \cdot \pi.
    \end{equation*}
    This implies that, $c_{(\sigma \cdot \pi \cdot \tau)} = \mathrm{sign}(\tau) \cdot c_\pi$ for all permutation $\pi \in \mathfrak{S}_n$. Then in particular case if the identity permutation, $c_{(\sigma \cdot \tau)} = \mathrm{sign}(\tau) \cdot c_{1_{\mathfrak{S}_n}}$, and
    \begin{align*}
        \sum_{\substack{\sigma \in R_T \\ \tau \in C_T}} c_{(\sigma \cdot \tau)} \cdot (\sigma \cdot \tau) &= c_{1_{\mathfrak{S}_n}} \cdot \sum_{\substack{\sigma \in R_T \\ \tau \in C_T}} \mathrm{sign}(\tau) \cdot (\sigma \cdot \tau) \\
        &= c_{1_{\mathfrak{S}_n}} \cdot s_T.
    \end{align*}
    Assuming that $c_\pi = 0$ when $\pi \not\in \set{(\sigma \cdot \tau)}{\sigma \in R_T \text{ and } \tau \in C_T}$, and given that $R_T \cap C_T = \1\{ 1_{\mathfrak{S}_n} \1\}$, it follows that $y = c_{1_{\mathfrak{S}_n}} \cdot s_T$. Consider a permutation $\pi \not\in \set{(\sigma \cdot \tau)}{\sigma \in R_T \text{ and } \tau \in C_T}$, and suppose that there exist no distinct $i, j \in \{ 1, \ldots, n \}$ such that $i$ and $j$ belong to the same row of $T$ and the same column of $\pi(T)$. Then, all entries in the first row of $T$ must appear in distinct columns of $\pi(T)$. Hence, there exist two permutations $\sigma \in R_T$ and $\tau \in C{\pi(T)}$ such that the first row of $\sigma(T)$ and $(\tau \cdot \pi)(T)$ are identical. By iterating this process on the remaining rows of $T$, it follows that there exist two permutations $\sigma \in R_T$ and $\tau \in C_{\pi(T)}$ such that $\sigma(T) = (\tau \cdot \pi)(T)$. Consequently $\sigma = \tau \cdot \pi$ and $\pi \in \set{(\sigma \cdot \tau)}{\sigma \in R_T \text{ and } \tau \in C_T}$ since $C_{\pi(T)} = \pi \cdot C_T \cdot \pi^{\shortminus 1}$. Therefore, there exists a transposition $(i\:j) \in R_T \cap C_{\pi(T)}$, and in particular $(i\:j) \in C_{\pi(T)}$. Hence, $(i\:j) = \pi \cdot \tau \cdot \pi^{\shortminus 1}$ for some permutation $\tau \in C_T$, then equality
    \begin{equation*}
        c_{1_{\mathfrak{S}_n}} = (i\:j) \cdot (i\:j) = (i\:j) \cdot \pi \cdot \tau \cdot \pi^{\shortminus 1},
    \end{equation*}
    holds, and finally $\pi = (i\:j) \cdot \pi \cdot \tau$. Now since $(i\:j) \in R_T$ and $c_{(\sigma \cdot \pi \cdot \tau)} = \mathrm{sign}(\tau) \cdot c_\pi$ for all permutations $\sigma \in R_T$, $\tau \in C_T$ and $\pi \in \mathfrak{S}_n$, it holds that
    \begin{align*}
        c_\pi &= c_{((i\:j) \cdot \pi \cdot \tau)} \\
        &= \mathrm{sign}\1( (i\:j) \1) \cdot c_\pi \\
        &= - c_\pi,
    \end{align*}
    which implies that $c_\pi = 0$ when $\pi \not\in \set{(\sigma \cdot \tau)}{\sigma \in R_T \text{ and } \tau \in C_T}$.
    
    Let $\sigma \in R_T$ and $\tau \in C_T$, then, the following equations hold:
    \begin{equation*}
        \sigma \cdot r_T = r_T \cdot \sigma = r_T \quad \text{ and } \quad \tau \cdot c_T = c_T \cdot \tau = \mathrm{sign}(\tau) \cdot c_T.
    \end{equation*}
    Additionally $\sigma \cdot s_T \cdot \tau = \mathrm{sign}(\tau) \cdot s_T$. Furthermore, for $y$ an element of the group algebra $\mathbb{C} \1[ \mathfrak{S}_n \1]$, it follows that
    \begin{equation*}
        \sigma \cdot r_T \cdot y \cdot c_T \cdot \tau = \mathrm{sign}(\tau) \cdot r_T \cdot y \cdot c_T.
    \end{equation*}
    It is inferred that there exists a $\mu \in \mathbb{C}$ for which the equality $r_T \cdot y \cdot c_T = \mu \cdot s_T$ holds. Furthermore, when $y = c_T \cdot x \cdot r_T$,
    \begin{equation*}
        s_T \cdot x \cdot s_T = r_T \cdot c_T \cdot x \cdot r_T \cdot c_T = \mu \cdot s_T.
    \end{equation*}
\end{proof}
\begin{corollary*}
    Given a partition $\lambda \vdash n$, and an arbitrary Young tableau $T$ on $\lambda$ without repetitions, the product of the corresponding Young symmetrizer $s_T$ is nonzero, i.e. $s_T \cdot s_T \neq 0$.
\end{corollary*}
\begin{lemma} \label{appA:lem:YoungSymmetrizerProductDistinctShapes}
    Let $\lambda \vdash n$ and $\mu \vdash n$ be two distinct partitions of $n$, let $T_\lambda$ and $T_\mu$ be two Young tableaux without repetitions on partitions $\lambda$ and $\mu$, respectively, and let $x \in \mathbb{C} \1[ \mathfrak{S}_n \1]$. Then $s_{T_\mu} \cdot x \cdot s_{T_\lambda} = 0$.
\end{lemma}
\begin{proof}
    It is possible to assert, without loss of generality, that by interchanging $\lambda$ and $\mu$, there is a $k \in \{ 1, \ldots, n \}$ for which $\lambda_k > \mu_k$ holds, and for all $i \in \{ 1, \ldots, k - 1 \}$, the equalities $\lambda_k = \mu_k$ holds.

    Suppose that there exists no pair of distinct indices $i, j \in \{ 1, \ldots, n \}$ such that $i$ and $j$ occupy the same row in $T_\lambda$ and the same column in $T_\mu$. Under this condition, it follows that every element within the first row of $T_\lambda$ occupies a unique column within $T_\mu$. When $k > 1$, there exist two permutations $\sigma \in R_{T_\lambda}$ and $\tau \in C_{T_\mu}$ such that $\sigma(T_\lambda)$ and $\tau(T_\mu)$ have identical first rows. By iteratively applying this procedure to the first $k-1$ rows of $T_\lambda$ and $T_\mu$, it follows that there exist two permutations $\sigma \in R_{T_\lambda}$ and $\tau \in C_{T_\mu}$ such that $\sigma(T_\lambda)$ and $\tau(T_\mu)$ have the same first $k-1$ rows.
    
    The equality $\lambda_k > 0$ holds, otherwise $\lambda = \mu$. Consequently, the $k$-th row of $\sigma(T_\lambda)$ contains $\lambda_k$ entries, which appear in $\lambda_k$ distinct columns of $\tau(T_\mu)$, and are located between the $k$-th and $n$-th rows of $\tau(T_\mu)$. However, this arrangement cannot occur since $\lambda_k > \mu_k$. Therefore, there exists a transposition $(i\:j) \in R_{T_\lambda} \cap C_{T_\mu}$ such that
    \begin{equation*}
        (i\:j) \cdot r_{T_\lambda} = r_{T_\lambda} \quad \text{ and } \quad c_{T_\mu} \cdot (i\:j) = -c_{T_\mu}.
    \end{equation*}
    Furthermore, since $(i\:j) \cdot (i\:j)$ is the identity element in the symmetric group $\mathfrak{S}_n$, then
    \begin{equation*}
        c_{T_\mu} \cdot r_{T_\lambda} = c_{T_\mu} \cdot (i\:j) \cdot (i\:j) \cdot r_{T_\lambda} = -c_{T_\mu} \cdot r_{T_\lambda},
    \end{equation*}
    which implies that $c_{T_\mu} \cdot r_{T_\lambda} = 0$.
    
    Let $\sigma$ be an element of the symmetric group $\mathfrak{S}_n$. Then $c_{T_\mu} \cdot \sigma \cdot r_{T_\lambda} \cdot \sigma^{\shortminus 1} = c_{T_\mu} \cdot r_{\sigma(T_\lambda)} = 0$, as the previous equality is independent of the entries of $T_\lambda$. Thus, it follows that $c_{T_\mu} \cdot \sigma \cdot r_{T_\lambda} = 0$ for any permutation $\sigma \in \mathfrak{S}_n$. Consequently for all $x \in \mathbb{C} \1[ \mathfrak{S}_n \1]$ it follows that $c_{T_\mu} \cdot x \cdot r_{T_\lambda} = 0$ and also $s_{T_\mu} \cdot x \cdot s_{T_\lambda} = 0$.
\end{proof}
\begin{corollary*}
    Given two distinct partitions $\lambda \vdash n$ and $\mu \vdash n$, and two arbitrary Young tableaux without repetitions on partitions $\lambda$ and $\mu$, respectively, the product of the corresponding Young symmetrizers $s_{T_\mu}$ and $s_{T_\lambda}$ is zero, i.e. $s_{T_\mu} \cdot s_{T_\lambda} = 0$.
\end{corollary*}

Let $\lambda \vdash n$ be a partition of $n$, and let $T$ be a Young tableau on $\lambda$ without repetitions. The operation of right multiplication by the Young symmetrizer $s_T$ on the group algebra $\mathbb{C} \1[ \mathfrak{S}_n \1]$, i.e., $\mathbb{C} \1[ \mathfrak{S}_n \1] \cdot s_T$, defines a complex vector space that constitutes a representation of the symmetric group $\mathfrak{S}_n$ through its left action on this space.

\begin{lemma} \label{appA:lem:spetchModuleIrreducibility}
    For any partition $\lambda \vdash n$, and any Young tableau $T$ on $\lambda$ without repetitions, the representation $\mathbb{C} \1[ \mathfrak{S}_n \1] \cdot s_T$ is an irreducible representation of the symmetric group $\mathfrak{S}_n$.
\end{lemma}
\begin{proof}
    Consider two permutations $\sigma \in R_T$ and $\tau \in C_T$. Then, their product $\sigma \cdot \tau$ is equal to the identity element $1_{\mathfrak{S}_n}$ if and only if both $\sigma$ and $\tau$ are equal to $1_{\mathfrak{S}_n}$. This equivalence is due to the fact that $R_T \cap C_T = \1\{ 1_{\mathfrak{S}_n} \1\}$. Consequently $s_T \neq 0$ and $\mathbb{C} \1[ \mathfrak{S}_n \1] \cdot s_T$ is nonzero.

    Consider $V$ as a subrepresentation of $\mathbb{C} \1[ \mathfrak{S}_n \1] \cdot s_T$, then using \hyperref[appA:lem:YoungSymmetrizerProductSameShapes]{Lemma~\ref*{appA:lem:YoungSymmetrizerProductSameShapes}}, for all $x \in V$ there exists a scalar $\mu \in \mathbb{C}$ such that $s_T \cdot x = \mu \cdot s_T$. Thus
    \begin{equation*}
        s_T \cdot V \subset \mathbb{C} \cdot s_T.
    \end{equation*}
    Since $\mathbb{C} \cdot s_T$ has dimension one, the subspace $s_T \cdot V$ is either equal to $\mathbb{C} \cdot s_T$ or to the zero space $\{ 0 \}$. In the former case, the inclusion $\mathbb{C}\1[ \mathfrak{S}_n \1] \cdot s_T \cdot V \subseteq V$ follows, since $V$ is a representation of the symmetric group $\mathfrak{S}_n$, and the equality $V = \mathbb{C}\1[ \mathfrak{S}_n \1] \cdot s_T$ holds. Similarly, in the latter case, since $V$ is a representation of the symmetric group $\mathfrak{S}_n$, the inclusion $V \cdot V \subseteq V$ holds. Moreover, as $V$ is a subrepresentation of $\mathbb{C}\1[ \mathfrak{S}_n \1] \cdot s_T$, i.e., $V \subseteq \mathbb{C} \1[ \mathfrak{S}_n \1] \cdot s_T$, it follows that the product of $\mathbb{C} \1[ \mathfrak{S}_n \1] \cdot s_T$ and $V$ yields the trivial space $\{ 0 \}$, and the product of $V$ with itself also yields $\{ 0 \}$, namely $\mathbb{C} \1[ \mathfrak{S}_n \1] \cdot s_T \cdot V = \{ 0 \}$ and $V \cdot V = \{ 0 \}$. Consider an element $x$ in the subrepresentation $V$, defined as $x \coloneqq \sum_{\pi \in \mathfrak{S}_n} c_\pi \cdot \pi$. The adjoint of $x$, denoted as $x^*$, is defined by
    \begin{equation*}
        x^* \coloneqq \sum_{\sigma \in \mathfrak{S}_n} \bar{c}_\sigma \cdot \sigma^{\shortminus 1}.
    \end{equation*}
    Then both $x^*$ and $x \cdot x^*$ belong to $V$. As $V \cdot V = \{ 0 \}$ it follows that $x \cdot x^* = 0$. Consequently, for all permutations $\sigma \in \mathfrak{S}_n$ it holds that $c_\sigma \cdot x^* = 0$. In particular
    \begin{equation*}
        c_e \cdot x^* = \sum_{\sigma \in \mathfrak{S}_n} \bar{c}_\sigma \cdot \sigma = 0.
    \end{equation*}
    Therefore $x = 0$, and thus $V = \{ 0 \}$.
\end{proof}

\begin{lemma} \label{appA:lem:spetchModuleEquivalent}
    For every partition $\lambda \vdash n$, every Young tableau $T$ on $\lambda$ without repetitions, and every permutation $\sigma$ in the symmetric group $\mathfrak{S}_n$, the representations $\mathbb{C} \1[ \mathfrak{S}_n \1] \cdot s_T$ and $\mathbb{C} \1[ \mathfrak{S}_n \1] \cdot s_{\sigma(T)}$ are equivalent.
\end{lemma}
\begin{proof}
    Consider the linear map $\phi$ from $\mathbb{C} \1[ \mathfrak{S}_n \1] \cdot s_T$ to $\mathbb{C} \1[ \mathfrak{S}_n \1] \cdot s_{\sigma(T)}$. Given $x \in \mathbb{C} \1[ \mathfrak{S}_n \1] \cdot s_T$, with $x = y \cdot s_T$ for some $y \in \mathbb{C} \1[ \mathfrak{S}_n \1]$, then $\phi$ is defined on $x$ by $\phi(x) = x \cdot \sigma^{\shortminus 1}$. Observe that
    \begin{align*}
        \phi(x) &= x \cdot \sigma^{\shortminus 1} \\
        &= y \cdot s_T \cdot \sigma^{\shortminus 1} \\
        &= y \cdot \sigma^{\shortminus 1} \cdot \sigma \cdot s_T \cdot \sigma^{\shortminus 1} \\
        &= y \cdot \sigma^{\shortminus 1} \cdot s_{\sigma(T)}.
    \end{align*}
    It follows that $\phi(x)$ belongs to $\mathbb{C} \1[ \mathfrak{S}_n \1] \cdot s_{\sigma(T)}$, and that $\phi$ is an isomorphism.

    Consider a permutation $\tau \in \mathfrak{S}_n$ and $x$ is an element of $\mathbb{C} \1[ \mathfrak{S}_n \1] \cdot s_T$. Then, it follows that
    \begin{align*}
        \phi \circ (\tau \cdot x) &= \tau \cdot x \cdot \sigma^{\shortminus 1} \\
        &= \tau \circ \phi(x),
    \end{align*}
    and hence $\phi$ satisfies the property of being intertwining.
\end{proof}

\begin{lemma} \label{appA:lem:spetchModuleInequivalent}
    Let $\lambda \vdash n$ and $\mu \vdash n$ be two distinct partitions of $n$. Let $T_\lambda$ and $T_\mu$ be two Young tableaux without repetitions on partitions $\lambda$ and $\mu$, respectively. Then, the two representations $\mathbb{C} \1[ \mathfrak{S}_n \1] \cdot s_{T_\lambda}$ and $\mathbb{C} \1[ \mathfrak{S}_n \1] \cdot s_{T_\mu}$ are inequivalent.
\end{lemma}
\begin{proof}
    According to \hyperref[appA:lem:YoungSymmetrizerProductSameShapes]{Lemma~\ref*{appA:lem:YoungSymmetrizerProductSameShapes}}, it follows that the left action of the Young symmetrizer $s_{T_\mu}$ on $\mathbb{C} \1[ \mathfrak{S}_n \1] \cdot s_{T_\mu}$ results in
    \begin{equation*}
        s_{T_\mu} \cdot \mathbb{C} \1[ \mathfrak{S}_n \1] \cdot s_{T_\mu} = \mathbb{C} \cdot s_{T_\mu}.
    \end{equation*}
    However, from \hyperref[appA:lem:YoungSymmetrizerProductDistinctShapes]{Lemma~\ref*{appA:lem:YoungSymmetrizerProductDistinctShapes}}, the left action of the Young symmetrizer $s_{T_\mu}$ on $\mathbb{C} \1[ \mathfrak{S}_n \1] \cdot s_{T_\lambda}$ is zero, i.e.
    \begin{equation*}
        s_{T_\mu} \cdot \mathbb{C} \1[ \mathfrak{S}_n \1] \cdot s_{T_\lambda} = 0.
    \end{equation*}
    As a consequence, two irreducible representations $\mathbb{C} \1[ \mathfrak{S}_n \1] \cdot s_{T_\lambda}$ and $\mathbb{C} \1[ \mathfrak{S}_n \1] \cdot s_{T_\mu}$ are not isomorphic, hence not equivalent.
\end{proof}

From \hyperref[appA:lem:spetchModuleEquivalent]{Lemma~\ref*{appA:lem:spetchModuleEquivalent}}, for a given partition $\lambda \vdash n$, the representations $\mathbb{C} \1[ \mathfrak{S}_n \1] \cdot s_T$ of the symmetric group $\mathfrak{S}_n$, are mutually equivalent across all Young tableaux $T$ on $\lambda$ without repetition, thereby solely dependent on $\lambda$, and consequently, it is possible to designate any representation of this type as $\emph{V_\lambda}$.

\begin{theorem} \label{appA:thm:spetchModule}
    All irreducible representations of $\mathfrak{S}_n$ can be expressed as $V_\lambda$ for some partition $\lambda \vdash n$.
\end{theorem}
\begin{proof}
    The irreducibility of the representation $V_\lambda$ is guaranteed for any partition $\lambda \vdash n$, as stated in \hyperref[appA:lem:spetchModuleIrreducibility]{Lemma~\ref*{appA:lem:spetchModuleIrreducibility}}. It follows from \hyperref[appA:lem:spetchModuleEquivalent]{Lemma~\ref*{appA:lem:spetchModuleEquivalent}} and \hyperref[appA:lem:spetchModuleInequivalent]{Lemma~\ref*{appA:lem:spetchModuleInequivalent}} that the number of inequivalent irreducible representations $V_\lambda$ for some partitions $\lambda \vdash n$ is equivalent to the number of conjugacy classes of the symmetric group $\mathfrak{S}_n$. This number, as stated in \hyperref[appA:thm:numberIrreducibleRepresentations]{Theorem~\ref*{appA:thm:numberIrreducibleRepresentations}}, corresponds to the number of irreducible representations of $\mathfrak{S}_n$.
\end{proof}

The \hyperref[appA:prop:groupAlgebraDecompositionMultiplicity]{Proposition~\ref*{appA:prop:groupAlgebraDecompositionMultiplicity}} establishes that upon decomposing the group algebra $\mathbb{C} \1[ \mathfrak{S}_n \1]$ into a direct sum of irreducible representations, namely $\mathbb{C} \1[ \mathfrak{S}_n \1] \simeq V^{\oplus n_1}_1 \oplus \cdots \oplus V^{\oplus n_k}_k$, the dimension of each irreducible representation $V_i$ coincides with its multiplicity $n_i$.

\begin{theorem} \label{appA:thm:spetchModuleDimension}
    Consider a partition $\lambda \vdash n$. Let $V_\lambda$ be the corresponding irreducible representation of a symmetric group $\mathfrak{S}_n$. The dimension of $V_\lambda$ is equivalent to the cardinality of the set of standard Young tableaux associated with $\lambda$.
\end{theorem}
\begin{proof}
    Consider a partition $\lambda \vdash n$ and let $f(\lambda)$ denote the number of standard Young tableaux associated with $\lambda$. The standard Young tableaux on this partition can be ordered lexicographically based on their entries. Specifically, let $T_1, T_2, \ldots, T_{f(\lambda)}$ be the standard Young tableaux on this partition, ordered such that $T_i < T_j$ if and only if the entries of $T_i$ are smaller in lexicographic order than those of $T_j$, from left to right and top to bottom. Note that the first tableau in this order, denoted by $T_1$, is the canonical Young tableau on this partition.
    
    Let $i$ and $j$ be arbitrary elements of $\1\{ 1, \ldots, f(\lambda) \1\}$, with $i < j$. Further, let $k$ and $l$ denote the first row and column, respectively, at which the two standard Young tableaux $T_i$ and $T_j$ differ, proceeding from left to right and top to bottom. Notably, $k$ and $l$ cannot be the first row or column, respectively, since the Young tableaux are standard. Let $a$ denote the entry of $T_i$ located in the $k$-th row and $l$-th column, and $b$ denote the entry of $T_j$ located in the same position.Given that $i<j$ and $T_i$ and $T_j$ are standard Young tableaux, it follows that $a<b$. Let $m$ and $n$ denote the row and column, respectively, of the entry $a$ in $T_j$. By virtue of $T_j$ being a standard Young tableau, it cannot hold that $m>k$ and $n\geq l$, as the entries of $T_j$ in rows greater than $k$ and columns greater than $l$ are strictly larger than $b$. Furthermore, it cannot be the case that $m<k$, or that $m=k$ and $n<l$, since $T_i$ and $T_j$ coincide on these rows and columns. Lastly, it cannot hold that $m=k$ and $n>l$, as $T_j$ is a standard Young tableau and $a<b$. Thus $m > k$ and $n < l$ hold necessarily. Notably, the entries located on the $k$-th row and $n$-th column of both $T_i$ and $T_j$ are equal and denoted by $c$. Consequently, the transposition $(a\:c)$ belongs to $R_{T_i} \cap C_{T_j}$, since $a$ and $c$ share the same $k$-th row of $T_i$ and the same $n$-th column of $T_j$. Then, the equations
    \begin{equation*}
        (a\:c) \cdot r_{T_i} = r_{T_i} \quad \text{ and } \quad c_{T_j} \cdot (a\:c) = -c_{T_j},
    \end{equation*}
    hold. Moreover, as $(a\:c) \cdot (a\:c)$ yields the identity element in the symmetric group $\mathfrak{S}_n$,
    \begin{align*}
        s_{T_j} \cdot s_{T_i} &= r_{T_j} \cdot c_{T_j} \cdot r_{T_i} \cdot c_{T_i} \\
        &= r_{T_j} \cdot c_{T_j} \cdot (a\:c) \cdot (a\:c) \cdot r_{T_i} \cdot c_{T_i} \\
        &= - s_{T_j} \cdot s_{T_i},
    \end{align*}
    implying that $s_{T_j} \cdot s_{T_i} = 0$.
    
    From \hyperref[appA:lem:spetchModuleEquivalent]{Lemma~\ref*{appA:lem:spetchModuleEquivalent}}, the irreducible representations $\mathbb{C} \1[ \mathfrak{S}_n \1] \cdot s_{T_i}$ and $\mathbb{C} \1[ \mathfrak{S}_n \1] \cdot s_{T_j}$ are equivalent for all $i,j \in \1\{ 1, \ldots, f(\lambda) \1\}$. 
    Let $x_1, x_2, \ldots, x_{f(\lambda)}$ denote certain elements of the group algebra $\mathbb{C} \1[ \mathfrak{S}_n \1]$, such that
    \begin{equation*}
        \sum^{f(\lambda)}_{i = 1} x_i \cdot s_{T_i} = 0.
    \end{equation*}
    By applying the result of \hyperref[appA:lem:YoungSymmetrizerProductSameShapes]{Lemma~\ref*{appA:lem:YoungSymmetrizerProductSameShapes}} and observing that $T_1 \leq T_i$ for all $i \in \1\{ 1, \ldots, f(\lambda) \1\}$, the existence of a nonzero $\mu \in \mathbb{C}$ is established, such that $\sum^{f(\lambda)}_{i = 1} x_i \cdot s_{T_i} \cdot s_{T_1} = \mu \cdot x_1 \cdot s_{T_1}$, which further implies that $x_1 = 0$. Right multiplication on both sides of the equation with $s_{T_i}$ yields $x_i = 0$ for all $i \in \1\{ 1, \ldots, f(\lambda) \1\}$. Consequently, $\dim V_\lambda \geq f(\lambda)$.

    Consider a Young diagram $\lambda$. Let $\mu$ be another Young diagram obtained by either removing a single box from $\lambda$, denoted as $\mu = \lambda - 1$, or adding a single box to $\lambda$, denoted as $\mu = \lambda + 1$. Then $f(\lambda)$ can be expressed as the sum of $f(\mu)$ for all possible Young diagrams $\mu$ obtained by removing a single box from $\lambda$, i.e.
    \begin{equation*}
        f(\lambda) = \sum_{\mu = \lambda - 1} f(\mu).
    \end{equation*}
    
    The proof of the identity $(n + 1) f(\lambda) = \sum_{\mu = \lambda + 1} f(\mu)$ shall be established through induction on $n$. The base case $n = 0$ holds, as $f \1( (0) \1) = f \1( (1) \1)$. For the inductive step, assume $n > 0$, it follows that
    \begin{align*}
        (n + 1) f(\lambda) &= n \cdot f(\lambda) + f(\lambda) \\
        &= n \sum_{\mu = \lambda - 1} f(\mu) + f(\lambda) \\
        &= \sum_{\mu = \lambda - 1} \sum_{\nu = \mu + 1} f(\nu) + f(\lambda),
    \end{align*}
    wherein the second summand enumerates partitions $\nu \vdash n$. There exist two distinct cases, namely $\nu = \lambda$ and $\nu \neq \lambda$. Consider the sets $\lambda_+$ and $\lambda_-$, defined as follows:
    \begin{align*}
        \lambda_+ &= \set{\mu}{\mu = \lambda + 1} \\
        \lambda_- &= \set{\mu}{\mu = \lambda - 1}.
    \end{align*}
    It is worth noting that for each box within a given Young diagram that may be extracted such that it does not alter its validity as a Young diagram, a box can be added to the next row of the Young diagram. Furthermore, it is always possible to add a box to the first row of the Young diagram. As a result, it follows that $\lambda_+ = \lambda_- + 1$. The occurrence of the first case, where $\nu = \lambda$, is equal to the cardinality of the set $\lambda_-$, namely $|\lambda_-|$. Thus
    \begin{align*}
        \sum_{\mu = \lambda - 1} \sum_{\nu = \mu + 1} f(\nu) + f(\lambda) &= \sum_{\mu = \lambda - 1} \sum_{\substack{\nu = \mu + 1 \\ \nu \neq \lambda}} f(\nu) + \1( |\lambda_-| + 1 \1) f(\lambda) \\
        &= \sum_{\mu = \lambda + 1} \sum_{\substack{\nu = \mu - 1 \\ \nu \neq \lambda}} f(\nu) + |\lambda_+| \cdot f(\lambda) \\
        &= \sum_{\mu = \lambda + 1} \sum_{\nu = \mu - 1} f(\nu) \\
        &= \sum_{\mu = \lambda + 1} f(\mu),
    \end{align*}
    concluding the induction proof.

    The formula $\sum_{\lambda \vdash n} f(\lambda)^2 = n!$ shall also be established through induction on $n$. As the base case, it is observed that $f \1( (0) \1) = 1$, and hence the formula is verified for $n = 0$. For the inductive step, suppose $n > 0$, then 
    \begin{align*}
        \sum_{\lambda \vdash n} f(\lambda)^2 &= \sum_{\lambda \vdash n} \sum_{\mu = \lambda - 1} f(\lambda) f(\mu) \\
        &= \sum_{\lambda \vdash (n-1)} \sum_{\mu = \lambda + 1} f(\lambda) f(\mu) \\
        &= n \sum_{\lambda \vdash (n-1)} f(\lambda)^2 \\
        &= n \cdot (n-1)!,
    \end{align*}
    concluding the induction proof.

    According to \hyperref[appA:prop:groupAlgebraDecompositionMultiplicity]{Proposition~\ref*{appA:prop:groupAlgebraDecompositionMultiplicity}} the following equation is satisfied: $\sum_{\lambda \vdash n} \dim(V_\lambda)^2 = n!$. However, for all partitions $\lambda \vdash n$, it is true that $\dim V_\lambda \geq f(\lambda)$. Hence, $\dim V_\lambda = f(\lambda)$.
\end{proof}

\begin{example*}
    There exist three partitions of $3$, which are denoted as $(3)$, $(2, 1)$ and $(1, 1, 1)$. Each partition has a corresponding Young diagram
    \begin{equation*}
        \begin{array}{ccccc}
            \ydiagram{3} & \qquad & \ydiagram{2, 1} & \qquad & \ydiagram{1, 1, 1} \\[3em]
            (3) && (2, 1) && (1, 1, 1)
        \end{array}
    \end{equation*}
    For the partitions $(3)$ and $(1, 1, 1)$, there exists only one standard Young tableau, which corresponds to the canonical Young tableau
    \begin{equation*}
        \begin{ytableau}
            1 & 2 & 3
        \end{ytableau}
        \qquad \text{ and } \qquad
        \begin{ytableau}
            1 \\
            2 \\
            3
        \end{ytableau}
    \end{equation*}
    However the partition $(2, 1)$ admits a pair of distinct standard Young tableaux, denoted as follows
    \begin{equation*}
        \begin{ytableau}
            1 & 2 \\
            3
        \end{ytableau}
        \qquad \text{ and } \qquad
        \begin{ytableau}
            1 & 3 \\
            2
        \end{ytableau}
    \end{equation*}
    
    From \hyperref[appA:thm:spetchModuleDimension]{Theorem~\ref*{appA:thm:spetchModuleDimension}}, it follows that the irreducible representations $V{(3)}$ and $V_{(1, 1, 1)}$, of $\mathfrak{S}_3$, have dimension 1, while $V_{(2, 1)}$ has dimension 2, as expected. The determination of the irreducible representations of $\mathfrak{S}_3$ corresponding to the canonical Young tableaux $T_{(3)}, T_{(2, 1)}$, and $T_{(1, 1, 1)}$ necessitates an examination of their respective Young symmetrizers $s_{T_{(3)}}, s_{T_{(1, 1, 1)}}$ and $s_{T_{(2, 1)}}$. This examination is conducted in conjunction with the construction of $V_\lambda$'s From \hyperref[appA:thm:spetchModule]{Theorem~\ref*{appA:thm:spetchModule}},
    \begin{align*}
        \mathbb{C} \1[ \mathfrak{S}_3 \1] \cdot s_{T_{(3)}} &= \mathbb{C} \1[ \mathfrak{S}_3 \1] \cdot \1( 1_{\mathfrak{S}_3} + (1\:2) + (1\:3) + (2\:3) + (1\:2\:3) + (3\:2\:1) \1) \\
        \mathbb{C} \1[ \mathfrak{S}_3 \1] \cdot s_{T_{(1, 1, 1)}} &= \mathbb{C} \1[ \mathfrak{S}_3 \1] \cdot \1( 1_{\mathfrak{S}_3} - (1\:2) - (1\:3) - (2\:3) + (1\:2\:3) + (3\:2\:1) \1) \\
        \mathbb{C} \1[ \mathfrak{S}_3 \1] \cdot s_{T_{(2, 1)}} &= \mathbb{C} \1[ \mathfrak{S}_3 \1] \cdot \1( 1_{\mathfrak{S}_3} + (1\:2) - (1\:3) + (3\:2\:1) \1).
    \end{align*}
    The action of the symmetric group $\mathfrak{S}_3$ on the irreducible representation $\mathbb{C} \1[ \mathfrak{S}_n \1] \cdot s_{T_{(3)}}$ yields a trivial action. It follows that $V_{(3)}$ is equivalent to $V_{\text{trivial}}$, while the irreducible representations $V_{(1, 1, 1)}$ and $V_{(2, 1)}$ are necessarily equivalent to $V_{\text{sign}}$ and $V_{\text{standard}}$, respectively.
\end{example*}

Let $\mathbb{C} \1[ \mathfrak{S}_n \1] \simeq V^{\oplus n_1}_1 \oplus \cdots \oplus V^{\oplus n_k}_k$ be the decomposition the group algebra $\mathbb{C} \1[ \mathfrak{S}_n \1]$ of the symmetric group $\mathfrak{S}_n$ into a direct sum of irreducible representations. By \hyperref[appA:thm:projectorIsotypicComponent]{Theorem~\ref*{appA:thm:projectorIsotypicComponent}} the projectors onto the isotypic components $V^{\oplus n_i}_i$ are given by
\begin{equation*}
    \phi_i \coloneqq \frac{\dim V_i}{n!} \sum_{\sigma \in \mathfrak{S}_n} \bar{\chi}_{V_i} (\sigma) \cdot \sigma.
\end{equation*}
\begin{theorem*}
    Let the group algebra $\mathbb{C} \1[ \mathfrak{S}_n \1]$ of the symmetric group $\mathfrak{S}_n$ and its decomposition into a direct sum of irreducible representations $\mathbb{C} \1[ \mathfrak{S}_n \1] \simeq \bigoplus_{\lambda \vdash n} V^{\oplus n_\lambda}_\lambda$. For each partition $\lambda \vdash n$, the map $\phi_\lambda$ defined on $\mathbb{C} \1[ \mathfrak{S}_n \1]$ by
    \begin{equation*}
        \phi_\lambda \coloneqq \frac{\dim V_\lambda}{n!} \sum_T s_T,
    \end{equation*}
    where the sum is taken over all Young tableaux without repetitions on $\lambda$, is the projector onto $V^{\oplus n_\lambda}_\lambda$. 
\end{theorem*}
\begin{proof}
    Let $\sigma \in \mathfrak{S}_n$, then
    \begin{align*}
        \sigma^{\shortminus 1} \cdot \sum_T s_T \cdot \sigma &= \sum_T s_{\sigma(T)} \cdot  \\
        &= \sum_T s_T,
    \end{align*}
    and in hence $\sum_T s_T \in \mathrm{Hom}_{\mathfrak{S}_n} \2( \mathbb{C} \1[ \mathfrak{S}_n \1] \2)$. Consider the decomposition of the group algebra $\mathbb{C} \1[ \mathfrak{S}_n \1]$ of the symmetric group $\mathfrak{S}_n$ into a direct sum of irreducible representations $\mathbb{C} \1[ \mathfrak{S}_n \1] \simeq \bigoplus_{\lambda \vdash n} V^{\oplus n_\lambda}_\lambda$, for all $\mu \vdash n$ partition of $n$, from Schur's \hyperref[appA:lem:Schur]{Lemma~\ref*{appA:lem:Schur}}, the restriction of $\sum_T s_T$ to $V_\mu$ is an homothety $\lambda \cdot I$, with $\lambda \in \mathbb{C}$. From \hyperref[appA:lem:YoungSymmetrizerProductDistinctShapes]{Lemma~\ref*{appA:lem:YoungSymmetrizerProductDistinctShapes}}, if $\lambda$ and $\mu$ are distinct, then for all $x \in V_\mu$,
    \begin{equation*}
        \sum_T s_T \cdot x = 0, 
    \end{equation*}
    and thus the restriction of $\sum_T s_T$ to $V_\mu$ is the zero map. Let $T_\lambda$ be any Young tableau without repetitions on $\lambda$, then the coefficient of $1_{\mathfrak{S}_n}$ in $s_{T_{\lambda}}$ is $1$, since $R_T \cap C_T = \1\{ 1_{\mathfrak{S}_n} \1\}$, and thus $\Tr \2[\sum_T s_T \2] = n!$. So the restriction of $\sum_T s_T$ to $V_\lambda$ is the homothety
    \begin{equation*}
        \frac{n!}{\dim V_\lambda} \cdot I.
    \end{equation*}

    As a consequence, since $\sum_T s_T$ does not cause any intertwining between the representations $V_\lambda$, the map $\phi_\lambda$ is the identity on $V^{\oplus n_\lambda}_\lambda$ and the zero map elsewhere, and thus is a projector.
\end{proof}

In the previous section, the complex vector spaces associated with the irreducible representations of the symmetric group $\mathfrak{S}_n$ were discussed. However, the corresponding matrices representing the permutation elements were not described. The construction of such matrices is far from trivial, and in fact, there exist a variety of constructions that may be employed based on the desired properties of the resulting matrices, i.e. integer matrix components, rational matrix components, or orthogonal matrices. 

\subsection{Restricted and induced representations}

The present Section concerns the correlation existing between a finite group $G$ and a subgroup $H$. Can a representation of $H$ be derived from a representation of $G$ or vice versa? Furthermore, if the original representation is irreducible, what conclusions can be drawn concerning the derived representation?

Consider a representation $(\rho, V)$ of a finite group $G$, and let $H$ be a subgroup of $G$. Given an element $h \in H$, it is worth noting that $H$ is a subgroup of $G$ and hence $h \in G$. Furthermore, as the representation $\rho$ induces an action of $G$ on $V$, it follows that the action of $\rho$ also restricts to only elements $h \in H$. The \emph{restricted representation} of $(\rho, V)$ on the subgroup $H$ is denoted $\emph{\1( \mathrm{Res}^G_H(\rho), V \1)}$ and defined on the element $h \in H$ as follows: 
\begin{equation*}
    \mathrm{Res}^G_H(\rho)(h) \coloneqq \rho(h).
\end{equation*}
In the case where $(\rho, V)$ is an irreducible representation of the group $G$, it is not always the case that the complex vector space $V$ manifests irreducibility as a representation of a subgroup $H$. This can occur when the size of $V$ becomes too large to maintain irreducibility with respect to the smaller subgroup $H$.

If $\chi$ is the character of the representation $(\rho, V)$ of $G$, the \emph{restricted character} of the representation $\mathrm{Res}^G_H(\rho)$, denoted $\emph{\mathrm{Res}^G_H(\chi)}$, becomes $\mathrm{Res}^G_H(\chi)(h) = \chi(h)$, for all $h \in H$.

\begin{example*}
    Let $\mathfrak{S_4}$ be the group of all permutations of $\{ 1, 2, 3, 4 \}$. From \hyperref[appA:sec:RepresentationsOfTheSymmetricGroup]{Section~\ref*{appA:sec:RepresentationsOfTheSymmetricGroup}}, the symmetric group $\mathfrak{S_4}$ has $5$ conjugacy classes, and consequently $5$ irreducible representations, namely
    \begin{equation*}
        V_{(4)}, \quad  V_{(3,  1)}, \quad V_{(2,  2)}, \quad V_{(2,  1,  1)} \quad \text{ and } \quad V_{(1,  1,  1,  1)}.
    \end{equation*}
    Let $\emph{D_4}$ be the \emph{dihedral group} of order $8$, whose elements are the symmetries of the square, generated by the $\tfrac{\pi}{4}$ counterclockwise rotation $r$ and the vertical reflection $s$:
    \begin{equation*}
        D_4 = \langle r, s \rangle = \1\{ 1_{D_4}, r, r^2, r^3, s, s r, s r^2, s r^3 \1\}.
    \end{equation*}
    It has $5$ conjugacy classes, which are given by 
    \begin{equation*}
        \1\{ 1_{D_4} \1\}, \quad \1\{ r, r^3 \1\}, \quad \1\{ r^2 \1\}, \quad \1\{ s r, s r^3 \1\} \quad \text{ and } \quad \1\{ s, s r^2 \1\},
    \end{equation*}
    and as many irreducible representations: $W_1, W_2, W_3, W_4$ and $W_5$. Then character table of $D_4$ can be expressed as follows
    \begin{equation*}
\begin{array}{c|ccccc}
             D_4 & \1\{ 1_{D_4} \1\} & \1\{ r, r^3 \1\} & \1\{ r^2 \1\} & \1\{ s r, s r^3 \1\} & \1\{ s, s r^2 \1\} \\
             \hline
             \chi_{W_1} & 1 & 1 & 1 & 1 & 1 \\
             \chi_{W_2} & 1 & -1 & 1 & -1 & 1 \\
             \chi_{W_3} & 1 & -1 & 1 & 1 & -1 \\
             \chi_{W_4} & 1 & 1 & 1 & -1 & -1 \\
             \chi_{W_5} & 2 & 0 & -2 & 0 & 0
        \end{array}
    \end{equation*}
    As can be inferred from \hyperref[appA:thm:numberIrreducibleRepresentations]{Theorem~\ref*{appA:thm:numberIrreducibleRepresentations}}, the rows exhibit orthogonality with respect to the Hermitian inner product $\scalar{\cdot}{\cdot}_{D_4}$. By indexing the vertices of the square counterclockwise with $1$, $2$, $3$ and $4$, starting from upper left corner, i.e.
    \begin{equation*}
        %
    \end{equation*}
    the group $D_4$ can be made a subgroup of $\mathfrak{S_4}$, where the $8$ symmetries of the square become,
    \begin{equation*}
        \begin{array}{r|l}
            (1)(2)(3)(4) &\text{ identity} \\
            (1, 2, 3, 4) &\text{ $\tfrac{\pi}{4}$ counterclockwise rotation} \\
            (1, 3)(2, 4) &\text{ $\tfrac{2 \pi}{4}$ counterclockwise rotation} \\
            (4, 3, 2, 1) &\text{ $\tfrac{3 \pi}{4}$ counterclockwise rotation} \\
            (1, 2)(3, 4) &\text{ vertical reflection} \\
            (1, 4)(2, 3) &\text{ horizontal reflection} \\
            (1, 3)(2)(4) &\text{ main diagonal reflection } \\
            (2, 4)(1)(3) &\text{ secondary diagonal reflection}
        \end{array}
    \end{equation*}
    In the general case, the restricted representations $\mathrm{Res}^{\mathfrak{S_4}}_{D_4}(V_\lambda)$ of $D_4$, for every partition $\lambda \vdash 4$, do not constitute irreducible representations. According to Maschke's \hyperref[appA:thm:maschke]{Theorem~\ref*{appA:thm:maschke}}, these representations can be decomposed into a direct sum of irreducible representations. Specifically, 
    \begin{align*}
        \mathrm{Res}^{\mathfrak{S_4}}_{D_4} \1( V_{(4)} \1) &\simeq W_1 \\
    	\mathrm{Res}^{\mathfrak{S_4}}_{D_4} \1( V_{(3, 1)} \1) &\simeq W_3 \oplus W_5 \\
    	\mathrm{Res}^{\mathfrak{S_4}}_{D_4} \1( V_{(2, 2)} \1) &\simeq W_1 \oplus W_2 \\
    	\mathrm{Res}^{\mathfrak{S_4}}_{D_4} \1( V_{(2, 1, 1)} \1) &\simeq W_4 \oplus W_5 \\
    	\mathrm{Res}^{\mathfrak{S_4}}_{D_4} \1( V_{(1, 1, 1, 1)} \1) &\simeq W_2.
    \end{align*}
\end{example*}

The passage from a representation $(\rho, V)$ of a subgroup $H$ of a finite group $G$ to the group $G$ necessitates a more intricate approach. This arises due to the inability of $\rho$ to induce an action of $G \setminus H$ on $V$. Nevertheless, it is feasible to look at $\emph{G/H}$, the \emph{quotient group}, comprising elements in the form of $g \cdot h$ with $g \in G$ and $h \cdot H$, where $\rho(h)$ exhibits well-defined behavior.

Let $V, W$ be two complex vector space of finite dimension. The complex tensor product between $V$ and $W$, namely $V \otimes W$, is a condensed notation for the explicit notation $V \otimes_{\mathbb{C}} W$. It is worth noting that for any $v \in V$ and $w \in W$, the relation
\begin{equation*}
    (c \cdot v) \otimes w = v \otimes (c \cdot w),
\end{equation*}
holds for all $c \in \mathbb{C}$.

Given a subgroup $H$ of a finite group $G$, suppose that $(\rho, V)$ is a representation of $H$. In this case, it is possible to define the tensor product complex vector space $\mathbb{C}[G] \otimes_{\mathbb{C}[H]} V$ as follow: for all $g \in \mathbb{C}[G]$ and $v \in V$, the following relation,
\begin{equation*}
    (g \cdot h) \otimes v = g \otimes (h \cdot v),
\end{equation*}
holds for all $h \in H$. The \emph{induced representation} of $(\rho, V)$ on the group $G$ is denoted $\emph{\1( \mathrm{Ind}^G_H(\rho), \mathbb{C}[G] \otimes_{\mathbb{C}[H]} V \1)}$ and defined on element $g \in G$ by
\begin{equation*}
    \mathrm{Ind}^G_H(\rho)(g)(x) \coloneqq g \cdot x,
\end{equation*}
for all $x \in \mathbb{C}[G] \otimes_{\mathbb{C}[H]} V$.

In order to gain a more comprehensive understanding of the operation of group $G$ on $\mathbb{C}[G] \otimes_{\mathbb{C}[H]} V$, consider a \emph{coset representation} of $G$, i.e. $G = g_1 \cdot H \uplus \cdots \uplus g_k \cdot H$, where $\uplus$ denotes the disjoint union, with $g_i \in G$. For any element $g \in G$, there exists $i \in \{ 1, \ldots, k \}$ and $h \in H$ such that $g = g_i \cdot h$. Consequently, all elements $g \otimes v \in \mathbb{C}[G] \otimes_{\mathbb{C}[H]} V$ satisfy
\begin{equation*}
    g \otimes v = (g_i \cdot h) \otimes v = g_i \otimes (h \cdot v).
\end{equation*}
Thus the action of the group $G$ on the tensor product $\mathbb{C}[G] \otimes_{\mathbb{C}[H]} V$ is exclusively defined for the elements $g_i \otimes v \in \mathbb{C}[G] \otimes_{\mathbb{C}[H]} V$. Specifically, for all $g \in G$, the representation $\mathrm{Ind}^G_H(\rho)(g)$ takes the form of a $k \times k$ block matrix that corresponds to the coset representation $G = g_1 \cdot H \uplus \cdots \uplus g_k \cdot H$.
The action of an element $g \in G$ upon $g_i \otimes v \in \mathbb{C}[G] \otimes_{\mathbb{C}[H]} V$ is given by
\begin{align*}
    g \cdot (g_i \otimes v) &= (g \cdot g_i) \otimes v \\
    &= (g_j \cdot h) \otimes v
    = g_j \otimes (h \cdot v),
\end{align*}
where $h$ is the element of $H$ such that $g \cdot g_i = g_j \cdot h$ in the coset representation $G = g_1 \cdot H \uplus \cdots \uplus g_k \cdot H$. In other words, $h$ is defined as $h \coloneqq g^{\shortminus 1}_j \cdot g \cdot g_i$. Additionally, it follows that $h$ acts on a vector $v \in V$ via the representation $\rho \1( g^{\shortminus 1}_j \cdot g \cdot g_i \1) (v)$. Finally,
\begin{equation*}
    \mathrm{Ind}^G_H(\rho)(g) =
    \begin{pmatrix}
        \rho \1( g^{\shortminus 1}_1 \cdot g \cdot g_1 \1) & \rho \1( g^{\shortminus 1}_1 \cdot g \cdot g_2) & \ldots & \rho \1( g^{\shortminus 1}_1 \cdot g \cdot g_k \1) \\[0.5em]
        \rho \1( g^{\shortminus 2}_1 \cdot g \cdot g_1 \1) & \rho \1( g^{\shortminus 2}_1 \cdot g \cdot g_2) & \ldots & \rho \1( g^{\shortminus 2}_1 \cdot g \cdot g_k \1) \\[0.5em]
        \vdots & \vdots & \ddots & \vdots \\[0.5em]
        \rho \1( g^{\shortminus k}_1 \cdot g \cdot g_1 \1) & \rho \1( g^{\shortminus k}_1 \cdot g \cdot g_2) & \ldots & \rho \1( g^{\shortminus k}_1 \cdot g \cdot g_k \1)
    \end{pmatrix}.
\end{equation*}

If $\chi$ is the character of the representation $(\rho, V)$ of $H$, the \emph{induced character}  of the representation $\mathrm{Ind}^G_H(\rho)$, denoted $\emph{\mathrm{Ind}^G_H(\chi)}$ becomes, for all $g \in G$,
\begin{equation*}
    \mathrm{Ind}^G_H(\chi) = \Tr \1[ \mathrm{Ind}^G_H(\rho)(g) \1] = \sum^k_{i = 1} \chi \1( g^{\shortminus 1}_i \cdot g \cdot g_i \1),
\end{equation*}
in the coset representation $G = g_1 \cdot H \uplus \cdots \uplus g_k \cdot H$.

\begin{example*}
    Let $\mathbb{Z} / 4 \mathbb{Z}$ be the cyclic group of order $4$ generated by the element $z$. Formally, $\mathbb{Z} / 4 \mathbb{Z}$ is defined by $\mathbb{Z} / 4 \mathbb{Z} \coloneqq \1\{ 1_{\mathbb{Z} / 4 \mathbb{Z}}, z, z^2, z^3 \1\}$. Let  $\mathbb{Z} / 2 \mathbb{Z}$ the subgroup of $\mathbb{Z} / 4 \mathbb{Z}$ of order $2$ generated by $z^2$, i.e. $\mathbb{Z} / 2 \mathbb{Z} \coloneqq \1\{ 1_{\mathbb{Z} / 2 \mathbb{Z}}, z^2 \1\}$. Consider the one-dimensional irreducible representation $(\rho, \mathbb{C})$ of $\mathbb{Z} / 2 \mathbb{Z}$, defined as follows:
    \begin{equation*}
        \rho \1( 1_{\mathbb{Z} / 4 \mathbb{Z}} \1) = 1 \qquad \text{ and } \qquad \rho \1( z^2 \1) = -1.
    \end{equation*}
    The induced representation $\mathrm{Ind}^{\mathbb{Z} / 4 \mathbb{Z}}_{\mathbb{Z} / 2 \mathbb{Z}}(\rho)$ of $\mathbb{Z} / 4 \mathbb{Z}$ is defined on the tensor product $\mathbb{C}[\mathbb{Z} / 4 \mathbb{Z}] \otimes_{\mathbb{C}[\mathbb{Z} / 2 \mathbb{Z}]} \mathbb{C}$ generated by the vectors
    \begin{equation*}
        1_{\mathbb{Z} / 4 \mathbb{Z}} \otimes 1, \quad z \otimes 1, \quad z^2 \otimes 1 \quad \text{ and } \quad z^2 \otimes 1.
    \end{equation*}
    But by definition of this tensor product over the group algebra $\mathbb{C}[\mathbb{Z} / 2 \mathbb{Z}]$, the two following relations hold:
    \begin{equation*}
        z^2 \otimes 1 = 1_{\mathbb{Z} / 4 \mathbb{Z}} \otimes \1( z^2 \cdot 1 \1) = -1 \cdot \1( 1_{\mathbb{Z} / 4 \mathbb{Z}} \otimes 1 \1),
    \end{equation*}
    and
    \begin{equation*}
        z^3 \otimes 1 = z \otimes \1( z^2 \cdot 1 \1) = -1 \cdot \1( z \otimes 1 \1).
    \end{equation*}
    Thus $\mathbb{C}[\mathbb{Z} / 4 \mathbb{Z}] \otimes_{\mathbb{C}[\mathbb{Z} / 2 \mathbb{Z}]} \mathbb{C}$ is a $2$-dimensional complex vector space with basis $1_{\mathbb{Z} / 4 \mathbb{Z}} \otimes 1$ and $z \otimes 1$. The action of the elements of $\mathbb{Z} / 4 \mathbb{Z}$ on the tensor product $\mathbb{C}[\mathbb{Z} / 4 \mathbb{Z}] \otimes_{\mathbb{C}[\mathbb{Z} / 2 \mathbb{Z}]} \mathbb{C}$, defined with respect to the generator $z$, is given by
    \begin{equation*}
        \mathrm{Ind}^{\mathbb{Z} / 4 \mathbb{Z}}_{\mathbb{Z} / 2 \mathbb{Z}}(\rho)(z) =
        \begin{pmatrix}
            0 & -1 \\
            1 & 0
        \end{pmatrix}.
    \end{equation*}
\end{example*}

Let $C(G)$ denotes the complex vector space of class functions of a finite group $G$, and $C(H)$ denotes the complex vector space of class functions of a subgroup $H$ of $G$. There is an linear map $\phi: C(G) \to C(H)$ given by restriction of class functions. As linear map between Hermitian inner product complex vector spaces of finite dimension, there exists a unique adjoint map $\phi^*: C(H) \to C(G)$ satisfying
\begin{equation*}
    \scalar{\phi^*(g)}{f}_G = \scalar{g}{\phi(f)}_H,
\end{equation*}
for all $f \in C(G)$ and $g \in C(H)$.

\begin{theorem}[Frobenius reciprocity] \label{appA:thm:FrobeniusReciprocity}
    Consider a finite group $G$ and a subgroup $H$ of $G$. Let $(\rho, V)$ and $(\sigma, W)$ be representations of $G$ and $H$, respectively. Then
    \begin{equation*}
        \scalar{\mathrm{Ind}^G_H \1( \chi_W \1)}{\chi_V}_G = \scalar{\chi_W}{\mathrm{Res}^G_H \1( \chi_V \1)}_H.
    \end{equation*}
\end{theorem}
\begin{proof}
    By definition of the Hermitian inner product of class functions on finite groups, the properties of the characters of \hyperref[appA:prop:characterProperties]{Proposition~\ref*{appA:prop:characterProperties}}, and the restricted and induced characters in the coset representation $G = g_1 \cdot H \uplus \cdots \uplus g_k \cdot H$,
    \begin{align*}
        \scalar{\mathrm{Ind}^G_H \1( \chi_W \1)}{\chi_V}_G &= \frac{1}{|G|} \sum_{g \in G} \mathrm{Ind}^G_H \1( \bar{\chi}_W \1) (g) \cdot \chi_V(g) \\
        &= \frac{1}{|G|} \sum_{g \in G} \sum^k_{i = 1} \bar{\chi}_W \1( g^{\shortminus 1}_i \cdot g \cdot g_i \1) \cdot \chi_V(g) \\
        &= \frac{1}{|G|} \frac{1}{|H|} \sum_{g \in G} \sum_{h \in G} \bar{\chi}_W \1( h^{\shortminus 1} \cdot g \cdot h \1) \cdot \chi_V(g) \\
        &= \frac{1}{|G|} \frac{1}{|H|} \sum_{g \in G} \sum_{h \in G} \chi_W \2( {\1( h^{\shortminus 1} \cdot g \cdot h \1)}^{\shortminus 1} \2) \cdot \chi_V(g) \\
        &= \frac{1}{|G|} \frac{1}{|H|} \sum_{g \in G} \sum_{h \in G} \chi_W \1( h^{\shortminus 1} \cdot g^{\shortminus 1} \cdot h \1) \cdot \chi_V(g).
    \end{align*}
    Through a variable substitution, and the explicit definition $\chi_W$ as the zero class function on $G \setminus H$,
    \begin{align*}
        \frac{1}{|G|} \frac{1}{|H|} \sum_{g \in G} \sum_{h \in G} \chi_W \1( g^{\shortminus 1} \1) \cdot \chi_V \1( h \cdot g \cdot h^{\shortminus 1} \1) &= \frac{1}{|G|} \frac{1}{|H|} \sum_{g \in G} \sum_{h \in G} \chi_W \1( g^{\shortminus 1} \1) \cdot \chi_V(g) \\
        &= \frac{1}{|H|} \sum_{g \in G} \chi_W \1( g^{\shortminus 1} \1) \cdot \chi_V(g) \\
        &= \frac{1}{|H|} \sum_{g \in H} \chi_W \1( g^{\shortminus 1} \1) \cdot \chi_V(g) \\
        &= \frac{1}{|H|} \sum_{g \in H} \bar{\chi}_W(g) \cdot \mathrm{Res}^G_H \1( \chi_V \1) (g) \\
        &= \scalar{\chi_W}{\mathrm{Res}^G_H \1( \chi_V \1)}_H.
    \end{align*}
\end{proof}

It is worth noting that the induced representation is not the inverse operation of the restricted representation. However, they are regarded as adjoint in the sens made precise by the Frobenius reciprocity \hyperref[appA:thm:FrobeniusReciprocity]{Theorem~\ref*{appA:thm:FrobeniusReciprocity}}.

In the general case, the relationship between the irreducible representations of a finite group $G$ and those of its subgroups $H$ is not established. Nevertheless, a distinct circumstance emerges when $H$ is a \emph{normal} subgroup of $G$.

\begin{proposition*}
   Consider a finite group $G$ and a normal subgroup $H$ of $G$ Let $(\rho, V)$ be a representation of the quotient group $G/N$. Define a representation $(\sigma, V)$ of $G$ as follows: for all $g \in G$, let $\bar{g}$ denote the representative of $g$ in $G/N$, and define $\sigma(g) \coloneqq \rho(\bar{g})$. Then $\sigma$ is irreducible if and only if $\rho$ is irreducible.
\end{proposition*}
\begin{proof}
    Let $\chi$ be the character of $\sigma$, then by definition of $\sigma$,
    \begin{align*}
        \scalar{\chi}{\chi}_G &= \frac{1}{|G|} \sum_{g \in G} \bar{\chi} (g) \cdot \chi(g) \\
        &= \frac{1}{|G|} \sum_{\bar{g} \in G/N} \sum_{h \in H} \bar{\chi} \1( \bar{g} \cdot h \1) \cdot \chi \1( \bar{g} \cdot h \1) \\
        &= \frac{|H|}{|G|} \sum_{\bar{g} \in G/N} \bar{\chi} \1( \bar{g} \1) \cdot \chi \1( \bar{g} \1) \\
        &= \frac{1}{|G/N|} \sum_{\bar{g} \in G/N} \bar{\chi} \1( \bar{g} \1) \cdot \chi \1( \bar{g} \1),
    \end{align*}
    where the last equality is the Hermitian inner product of the character of $\rho$.
\end{proof}

\section{Matrix groups}

\subsection{Representations of connected compact matrix groups}

The preceding Sections were specifically devoted to the study of the representation theory of finite groups. Specifically, the concepts of complete reducibility of group representations and the interplay between irreducible representations and characters were thoroughly examined. Moving forward, this section aims to extend the understanding of representation theory by examining the representations of infinite matrix groups, which are closed subgroups of the group consisting of invertible matrices.

A representation of a matrix group $G$ is a pair $(\rho, V)$, where $V$ is a complex vector space with dimension $d$, and $\rho: G \to \mathrm{GL}(V)$ is a group continuous homomorphism. The continuity of $\rho$ is equivalent to each of the component maps $g \mapsto {\1( \rho(g) \1)}_{i,j}$ being continuous, with $i, j \in \{ 1, \ldots, d \}$. If the component maps are rational functions of the matrix components, then the representation $(\rho, V)$ is referred to as being \emph{rational}, while if the component maps are polynomial functions of the matrix components, then the representation $(\rho, V)$ is referred to as being \emph{polynomial}.

\begin{theorem*}
    The rational representations of a matrix group have component maps polynomial in the matrix components and the inverse of the determinant.
\end{theorem*}

When considering finite groups, the normalized summation $\frac{1}{|G|} \sum_{g \in G}$ is frequently used as a means of averaging over the group. This technique is especially pertinent in the context of reducibility, in \hyperref[appA:sec:reducibility]{Section~\ref*{appA:sec:reducibility}} and character theory, in \hyperref[appA:sec:characterTheory]{Section~\ref*{appA:sec:characterTheory}}.

Generalising this practice to infinite matrix groups is not possible. However, for certain matrix groups, namely the compact matrix group, it is possible to perform an integration over a normalized \emph{Haar measure} as $\int_G ~\mathrm{d}g$.

Consider a compact matrix group $G$, and let $\emph{\mathrm{L}^2(G)}$ denote the Hilbert space of all complex function on $G$, with finite $2$-norm, for the Hermitian inner product defined on any $\phi, \psi \in \mathrm{L}^2(G)$ by
\begin{equation*}
    \scalar{\phi}{\psi} \coloneqq \int_G \phi(g) \bar{\psi}(g) ~\mathrm{d}g.
\end{equation*}
In the case where $G$ is a finite group equipped with the \emph{counting measure} as its Haar measure, an isomorphism between the Hilbert space $\mathrm{L}^2(G)$ and the group algebra $\mathbb{C}[G]$ is established via the map
\begin{equation*}
    f \mapsto \sum_{g \in G} f(g) \cdot g.
\end{equation*}

The collection of results presented herein, alongside their corollaries, are commonly referred to as the \emph{Peter-Weyl Theorem}.
\begin{theorem}[Peter-Weyl] \label{appA:thm:PeterWeyl}
    Let $G$ be a compact matrix  group, then the following hold.
    \begin{enumerate}
        \item For any non-identity matrix $g$ in $G$, there exists an irreducible representation of $G$ which maps $g$ to a non-identity matrix.
        \item The matrix components of inequivalent irreducible representations of $G$, scaled by a factor of the square root of the dimension of the corresponding irreducible representation, constitute an orthonormal basis of the Hilbert space $\mathrm{L}^2(G)$.
        \item The Hilbert space $\mathrm{L}^2(G)$ is isomorphic to a Hilbert space direct sum of irreducible representations of $G$, where the multiplicities correspond to the dimensions of the corresponding irreducible representations.
    \end{enumerate}
\end{theorem}
\begin{corollary*}
    The compact matrix groups exhibit the properties of complete reducibility and character orthogonality.
\end{corollary*}

The third assertion of the Peter–Weyl's \hyperref[appA:thm:PeterWeyl]{Theorem~\ref*{appA:thm:PeterWeyl}} constitutes the compact matrix group analog of \hyperref[appA:prop:groupAlgebraDecompositionMultiplicity]{Proposition~\ref*{appA:prop:groupAlgebraDecompositionMultiplicity}}.

\begin{example*}
    Let $\emph{\mathrm{GL}_2}$ be the \emph{complex general linear group}, consisting of invertible matrices of degree $2$, and which is locally compact but not compact, resulting in a non-finite Haar measure. Let the representation $\1( \rho, \mathbb{C}^2 \1)$ of $\mathrm{GL}_2$ defined on $g \in \mathrm{GL}_2$ as follows:
    \begin{equation*}
        \rho(g) \coloneqq
        \begin{pmatrix}
            1 & \log | \det g | \\
            0 & 1
        \end{pmatrix}.
    \end{equation*}
    As this representation admits only one-dimensional invariant subspace, namely $\Span_{\mathbb{C}} (1, 0)$, the representation is not completely reducible. Indeed, if for all $g \in \mathrm{GL}_2$, there exists $\lambda_g \in \mathbb{C}$ and $(c_1, c_2) \in \mathbb{C}^2$ such that $g \cdot (c_1, c_2) = \lambda_g \cdot (c_1, c_2)$, then it follows from
    \begin{align*}
        c_1 + \log \1( | \det g | \1) \cdot c_2 &= \lambda_g \cdot c_1 \\
        c_2 &= \lambda_g \cdot c_2,
    \end{align*}
    that $c_2 = 0$.
\end{example*}

A \emph{torus} $T$ of a compact matrix group $G$ is defined as a compact and connected abelian subgroup of $G$. A torus of $G$ is said to be \emph{maximal} if there is no other torus of $G$ that contains it as a proper subgroup.

Let $\emph{\mathbb{S}^1}$ denotes the \emph{unit $1$-sphere} defined by $\mathbb{S}^1 \coloneqq \set{e^{i \theta}}{\theta \in [0, 2\pi)}$ on the complex plane, and isomorphic to the unit circle on the real plane.

\begin{theorem} \label{appA:thm:torusIsomorphism}
    For all tori $T$, there exists some $k \in \mathbb{N}$ such that $T$ is isomorphic to a direct product of $k$ copies of the unit $1$-sphere $\mathbb{S}^1$.
\end{theorem}

In what follows, a given element $t$ in a torus $T$ is implicitly expressed by means of the isomorphism in \hyperref[appA:thm:torusIsomorphism]{Theorem~\ref*{appA:thm:torusIsomorphism}}, whereby $t$ is written as $(z_1, \ldots, z_k)$ with $z_i \in \mathbb{S}^1$, given $k \in \mathbb{N}$, the number of direct products of copies of $\mathbb{C}^*$ via the isomorphism.

\begin{example*}
    Let $\emph{\mathrm{SO}_{2n}}$ be the \emph{special orthogonal group}, which consists of the special orthogonal matrices of even degree $2n$. An instance of $\mathrm{SO}_{2n}$'s maximal tori is the subgroup $T$ consisting of the matrices of the following structure
    \begin{equation*}
        \begin{pmatrix}
            \cos \theta_1 & \sin \theta_1 & \cdots & 0 & 0 \\
            - \sin \theta_1 & \cos \theta_1 & \cdots & 0 & 0 \\
            \vdots & \vdots & \ddots & \vdots & \vdots \\
            0 & 0 & \cdots & \cos \theta_n & \sin \theta_n \\
            0 & 0 & \cdots & - \sin \theta_n & \cos \theta_n \\
        \end{pmatrix},
    \end{equation*}
    where $\theta_i$ belongs to the interval $[0, 2 \pi)$. In other words, $T$ is the set of all block-diagonal matrices with $2 \times 2$ rotation matrix blocks. Each block is isomorphic to an element of $\mathbb{S}^1$.
\end{example*}

Let the automorphism of a compact matrix group $G$, given by the map $h \mapsto g^{\shortminus 1} \cdot h \cdot g$ for all $h\in G$. This automorphism induces a transformation on the set of maximal tori of $G$ through conjugation. More specifically, it maps a torus $T$ to $g^{\shortminus 1} \cdot T \cdot g$, which is itself a maximal torus of $G$.

\begin{theorem}[Eli Cartan] \label{appA:thm:EliCartan}
    Let $G$ be connected compact matrix group, then the following hold.
    \begin{enumerate}
        \item Each element of $G$ is an element of some maximal torus of $G$.
        \item All maximal tori of $G$ are conjugate to each other in $G$.
    \end{enumerate}
\end{theorem}
\begin{corollary*}
    Let $T$ be a maximal torus of a connected compact matrix group $G$, then
    \begin{equation*}
        G = \bigcup_{g \in G} g^{\shortminus 1} \cdot T \cdot g.
    \end{equation*}
\end{corollary*}

Consider a connected compact matrix group $G$ and let $T$ be a maximal torus of $G$. According to Eli Cartan's \hyperref[appA:thm:EliCartan]{Theorem~\ref*{appA:thm:EliCartan}}, for all $g \in G$ and $t \in T$ there exists $h \in G$ such that $g = u^{\shortminus 1} \cdot t \cdot u$. Let $(\rho, V)$ be a representation of $G$, then
\begin{equation*}
    \rho(g) = \rho \1( u^{\shortminus 1} \1) \cdot \rho(t) \cdot \rho(u),
\end{equation*}
which implies that $\rho(g)$ and $\rho(t)$ have the same spectrum.

Let $\emph{N_G(T)}$ denote the \emph{normalizer} of $T$. The \emph{Weyl group} of $G$ with respect to $T$ is defined as the quotient of $N_G(T)$ by $T$, denoted by $\emph{W(T)} \coloneqq N_G(T) / T$. The Weyl group acts on $T$ through conjugation, where for any $w \in W(T)$ and $t \in T$, the action is given by $w \cdot t \coloneqq w^{\shortminus 1} \cdot t \cdot w$.

\begin{proposition*}
    Consider a connected and compact matrix group $G$ and let $T$ be a maximal torus of $G$. The Weyl group $W(T)$ is a finite group. Furthermore, two elements $g, h \in T$ are conjugate in $G$ if and only if there exists an element $w \in W(T)$ satisfying $w \cdot g = h$.
\end{proposition*}
\begin{corollary*}
    Let $G$ be a connected compact matrix group and let $T$ be a maximal torus of $G$. The space of continuous class functions on $G$ is isomorphic to the space of continuous complex functions on $T$ which are invariant under the action $W(T)$. Specifically every continuous complex functions on $T$ which are invariant under the action $W(T)$, extends uniquely to a continuous class functions on $G$.
\end{corollary*}
\begin{corollary*}[Weyl integration formula]
    Let $T$ be a maximal torus of a connected compact matrix group $G$, and let $f$ be a continuous class function on $G$. Then there exists a continuous real function $h$ on $T$ such that
    \begin{equation*}
        \int_G f(g) ~\mathrm{d}g = \int_T f(t) h(t) ~\mathrm{d}t.
    \end{equation*}
\end{corollary*}

Let $G$ be a connected compact matrix group, with a maximal torus $T$. The Weyl group $W(T)$ can be identified to the group of automorphisms of $T$ that are induced by inner automorphisms of $G$. Given any element $g \in G$, it holds that $N_G \1( g^{\shortminus 1} \cdot T \cdot g \1) = g^{\shortminus 1} \cdot N_G(T) \cdot g$. Consequently, the Weyl groups of $G$ with respect to distinct maximal tori of $G$ are all isomorphic.

\begin{example*}
    Let $\mathrm{U}_n$ be the unitary group of degree $n$. The group of diagonal unitary matrices forms a maximal torus of $\mathrm{U}_n$, denoted by $T$. Any unitary matrix can be decomposed into a diagonal form via conjugation by another unitary matrix, so $\mathrm{U}_n$ consists of diagonalizable matrices. The quotient of $\mathrm{U}_n$ by the conjugacy relation yields the conjugacy classes of diagonal matrices with diagonal entries being roots of unity, which are the eigenvalues of the matrices. In other words, each conjugacy class is represented by a diagonal matrix of the form
    \begin{equation*}
        \mathrm{diag} \1( e^{i\theta_1}, \dots, e^{i\theta_n} \1),
    \end{equation*}
    where $\theta_i$ is a real number in the interval $[0, 2 \pi)$. The spectrum of a matrix is preserved under conjugation, but representatives of these conjugacy classes are unique only up to a permutation of the diagonal components. The permutation group acting here is precisely the Weyl group $W(T)$. In the case of $\mathrm{U}_n$, this is the entire symmetric group $S_n$. Hence, then quotient of the unitary group $\mathrm{U}_n$ by the conjugacy relation is equal to the quotient of the maximal torus $T$ by the Weyl group $W(T)$.
\end{example*}

Let $G$ denote a compact and connected matrix group, and let $T$ be a maximal torus of $G$. Then, $T$ is an abelian group. The irreducible representations of $T$ are one-dimensional. More specifically, they are continuous homomorphisms $\rho: T \to \mathbb{C}^*$, that map each element $z \in {\1( \mathbb{S}^1 \1)}^k$ of $T$, to a product of the form
\begin{equation*}
    \rho(z_1, \ldots, z_k) = \prod^k_{i=1} z^{\lambda_i}_i,
\end{equation*}
for some $\lambda$ in $\mathbb{Z}^k$, referred to as a \emph{weight}. It follows that for each weight $\lambda$ in $\mathbb{Z}^k$, there exists an irreducible representation of $T$ of this form.

Consider a representation $(\rho, V)$ of a group $G$, and let $\1( \mathrm{Res}^G_T(\rho), V \1)$ denote the restricted representation of $G$ on $V$ with respect to the subgroup $T$. By the Peter–Weyl's \hyperref[appA:thm:PeterWeyl]{Theorem~\ref*{appA:thm:PeterWeyl}}, the restricted representation $\1( \mathrm{Res}^G_T(\rho), V \1)$ can be decomposed into a direct sum of one-dimensional irreducible representations of $T$. This decomposition takes the form of
\begin{equation*}
    V \simeq \bigoplus_{\lambda \in \mathbb{Z}^k} V^{\oplus n_\lambda}_\lambda,
\end{equation*}
where each $V_\lambda$ is referred to as the \emph{weight space} corresponding to the weight $\lambda$.

\begin{example*}
    Let $\emph{\mathrm{SU}_2}$ be the \emph{special unitary group}, which consists of the special unitary matrices of degree $2$. Consider a maximal torus of $\mathrm{SU}_2$, denoted by $T$, which is a subgroup of $\mathrm{SU}_2$ consisting of diagonal matrices of the form:
    \begin{equation*}
        \begin{pmatrix}
            e^{i \theta} & 0 \\
            0 & e^{-i \theta}
        \end{pmatrix},
    \end{equation*}
    where $\theta \in [0, 2 \pi)$. Thus, $T$ is isomorphic to $\mathbb{S}^1$. Let $(\rho, V)$ be a representation of $\mathrm{SU}_2$ of dimension $d$, and consider the restricted representation $\1( \mathrm{Res}^{\mathrm{SU}_2}_T(\rho), V \1)$ and its decomposition into a direct sum of weight spaces:
    \begin{equation*}
        V \simeq \bigoplus^d_{i = 1} V_i.
    \end{equation*}
    In this basis, the action of the restricted representation $\mathrm{Res}^{\mathrm{SU}_2}_T(\rho)$ to an element of the torus $T$ is given by
    \begin{equation*}
        \rho
        \begin{pmatrix}
            e^{i \theta} & 0 \\
            0 & e^{-i \theta}
        \end{pmatrix}
        = \mathrm{diag} \1( e^{i \lambda_1 \theta}, \dots, e^{i \lambda_d \theta} \1),
    \end{equation*}
    for some weights $\lambda_i \in \mathbb{Z}$.
\end{example*}

Suppose $G$ is a connected compact group of $n \times n$ matrices having a maximal torus $T$, and consider the decomposition of $\mathbb{C}^n$ into a direct sum of weight spaces, for the representation of $T$ that maps each element $t \in T$ to itself:
\begin{equation*}
    \mathbb{C}^n \simeq \bigoplus_{\lambda \in \mathbb{Z}^k} V^{\oplus n_\lambda}_\lambda.
\end{equation*}
The Weyl group $W(T)$ acts on this decomposition by permuting the weights $\lambda$.

\begin{example*}
    Let $\mathrm{SU}_3$ be the special unitary group of degree $3$. An instance of $\mathrm{SU}_3$'s maximal tori is the subgroup $T$ consisting of the diagonal matrices,
    \begin{equation*}
        \begin{pmatrix}
            e^{i \theta_1} & 0 & 0 \\
            0 & e^{i \theta_2} & 0 \\
            0 & 0 & e^{-i (\theta_1 + \theta_2)}
        \end{pmatrix},
    \end{equation*}
    where $\theta_1$ and $\theta_2$ belong to the interval $[0, 2 \pi)$. Therefore, $T$ is isomorphic to $\mathbb{S}^2$, via the map
    \begin{equation*}
        \begin{pmatrix}
            e^{i \theta_1} & 0 & 0 \\
            0 & e^{i \theta_2} & 0 \\
            0 & 0 & e^{-i (\theta_1 + \theta_2)}
        \end{pmatrix}
        \longmapsto \1( e^{i \theta_1}, e^{i \theta_2} \1).
    \end{equation*}
    Using this isomorphism, the action of each $(z_1, z_2) \in \mathbb{S}^2$, to the elements $e_1, e_2$ and $e_3$ forming the standard basis of $\mathbb{C}^3$, is given by
    \begin{align*}
        (z_1, z_2) \cdot e_1 &= z^1_1 \cdot z^0_2 \\
        (z_1, z_2) \cdot e_2 &= z^0_1 \cdot z^1_2 \\
        (z_1, z_2) \cdot e_3 &= z^{\shortminus 1}_1 \cdot z^{\shortminus 1}_2,
    \end{align*}
    This action yields the $3$ weights $(1, 0), (0, 1)$ and $(-1, -1)$, and their corresponding weight spaces. The Weyl group $W(T)$, being isomorphic to the symmetric group $S_3$, acts on these weight spaces by permuting the weights via the above action.
\end{example*}

Concluding the section, a summary of the various matrix groups of degree $n$ that have been used, and their properties is presented, assuming that the degree $n$ is larger than $1$:
\begin{equation*}
    \begin{array}{l|r}
         \large\textbf{Matrix group} & \large\textbf{Topology} \\[0.5em]
         \hline\hline\rule{0pt}{1.5em}
         \text{Complex general linear group} & \text{not compact} \\[0.25em]
         \mathrm{GL}_n \coloneqq \set{ g \in M_n(\mathbb{C}) }{\det g \neq 0} & \text{connected / not simply connected} \\[0.5em]
         \hline\rule{0pt}{1.5em}
         \text{Complex special linear group} & \text{not compact} \\[0.25em]
         \mathrm{SL}_n \coloneqq \set{ g \in M_n(\mathbb{C}) }{\det g  = 1} & \text{simply connected} \\[0.5em]
         \hline\rule{0pt}{1.5em}
         \text{Unitary group} & \text{compact} \\[0.25em]
         \mathrm{U}_n \coloneqq \set{ g \in M_n(\mathbb{C}) }{g g^* = I} & \text{connected / not simply connected} \\[0.5em]
         \hline\rule{0pt}{1.5em}
         \text{Special unitary group} & \text{compact} \\[0.25em]
         \mathrm{SU}_n \coloneqq \set{ g \in M_n(\mathbb{C}) }{g g^* = I \text{ and }\det g = 1} & \text{simply connected} \\[0.5em]
         \hline\rule{0pt}{1.5em}
         \text{Orthogonal group} & \text{compact} \\[0.25em]
         \mathrm{O}_n \coloneqq \set{ g \in M_n(\mathbb{R}) }{g g^\T = I} & \text{not connected} \\[0.5em]
         \hline\rule{0pt}{1.5em}
         \text{Special orthogonal group} & \text{compact} \\[0.25em]
         \mathrm{SO}_n \coloneqq \set{ g \in M_n(\mathbb{R}) }{g g^\T = I \text{ and }\det g = 1} & \text{connected / not simply connected }\\[0.5em]
    \end{array}
\end{equation*}
With the inclusions $\mathrm{O}_n \subseteq \mathrm{U}_n \subseteq \mathrm{GL}_n$ and $\mathrm{SO}_n \subseteq \mathrm{SU}_n \subseteq \mathrm{SL}_n$.

\subsection{Representations of \texorpdfstring{$\mathrm{U}_d$}{the unitary group}}

The \emph{unitary group} of degree $d$, denoted $\emph{\mathrm{U}_d}$, is the group of $d \times d$ unitary matrices, i.e., matrices $U$ with complex entries such that $U U^* = U^* U = I$, where $U^*$ denotes the conjugate transpose of $U$.

An equivalent definition of the unitary group is as the matrix group that preserves the standard Hermitian inner product $\scalar{\cdot}{\cdot}$ defined on the complex vector space $\mathbb{C}^d$. That is, for all vectors $x, y \in \mathbb{C}^d$ and for all $U \in \mathrm{U}_d$, it holds that \begin{equation*}
    \scalar{Ux}{Uy} = \scalar{x}{y}.
\end{equation*}

Let $T$ be the group of diagonal unitary matrices, a maximal torus of the unitary group $\mathrm{U}_d$. Let $(\rho, V)$ be a representation of $\mathrm{U}_d$, and $\1( \mathrm{Res}^{\mathrm{U}_d}_T(\rho), V \1)$ the restricted representation of $\mathrm{U}_d$ on $V$ with respect to the subgroup $T$, such that
\begin{equation*}
    V \simeq \bigoplus_{\lambda \in \mathbb{Z}^d} V^{\oplus n_\lambda}_\lambda,
\end{equation*}
is the decomposition of $V$ into a direct sum of one-dimensional irreducible representations of $T$, with some weights $\lambda \in \mathbb{Z}^d$. The maximal element in the set of these weights, with respect to the lexicographic order, is denoted the \emph{highest weight}. 

\begin{example*}
    Let $\mathrm{U}_d$ be the unitary group of degree $d$, and let $T$ be the group of diagonal unitary matrices, an instance of $\mathrm{U}_d$'s maximal tori. Let $\1( \rho, \mathbb{C}^d \1)$ be the representation of $\mathrm{U}_d$ that maps each element $U \in \mathrm{U}_d$ to itself, i.e. $\rho(U) = U$ for all $U \in \mathrm{U}_d$. Introduce the element $\mathrm{diag} \1( e^{i\theta_1}, \ldots, e^{i\theta_d} \1)$ of $T$, and let $e_1, \ldots, e_d$ constitute the standard basis of the vector space $\mathbb{C}^d$. The relationship between these elements is given by the equation:
    \begin{equation*}
        \mathrm{diag} \1( e^{i\theta_1}, \ldots, e^{i\theta_d} \1) \cdot e_i = e^{i\theta_i} \cdot e_i.
    \end{equation*}
    The decomposition of $\mathbb{C}^d$ into a direct sum of one-dimensional irreducible representations of $T$ can be expressed as
    \begin{equation*}
        V \simeq \bigoplus^d_{i = 1} V_i,
    \end{equation*}
    where the index $i$ denotes the weight $\lambda \in \mathbb{Z}^d$, such that $\lambda_i = 1$ and $\lambda_j = 0$ for every $j \neq i$. Consequently, the highest weight is $(1, 0, \ldots, 0)$.
\end{example*}

The unitary group $\mathrm{U}_d$ has a trivial irreducible representation of dimension $1$, which assigns the scalar value $1$ to every unitary matrix $U \in \mathrm{U}_d$. Additionally, there exists a $d$-dimensional irreducible representation of $\mathrm{U}_d$, wherein each unitary matrix $U \in \mathrm{U}_d$ is mapped to itself. However, no irreducible representation of intermediate dimension, i.e., between $1$ and $n$, can be found. In order to discover new irreducible representations, it is necessary to increase the dimension. One possible approach to achieving this is to consider a complex tensor product space.

\begin{theorem} \label{appA:thm:highestWeightIrreducibleRepresentation}
    For each weight $\lambda \in \mathbb{Z}^d$ such that $\lambda_1 \geq \cdots \geq \lambda_d$, there exists a unique irreducible representation $V_\lambda$ of the unitary group $\mathrm{U}_d$, with highest weight $\lambda$. These are all inequivalent and they are all the irreducible representations of $\mathrm{U}_d$.
\end{theorem}

Consider the complex tensor product space ${\1( \mathbb{C}^d \1)}^{\otimes n}$. Define a representation of the unitary group $\mathrm{U}_d$ on this space by
\begin{equation*}
    U \mapsto U^{\otimes n}.
\end{equation*}
Additionally, define a representation of the symmetric group $\mathfrak{S}_n$ on the this space by permuting the tensor positions, and extend this action linearly to the group algebra $\mathbb{C} \1[ \mathfrak{S}_n \1]$.

\begin{example*}
    Consider the complex tensor product space ${\1( \mathbb{C}^d \1)}^{\otimes 3}$ Let $v_1, v_2, v_3 \in \mathbb{C}^d$ be arbitrary vectors, and $\sigma \in \mathfrak{S}_n$ be arbitrary permutation, then
    \begin{equation*}
        \sigma \cdot (v_1 \otimes v_2 \otimes v_3) = v_{\sigma^{\shortminus 1}(1)} \otimes v_{\sigma^{\shortminus 1}(2)} \otimes v_{\sigma^{\shortminus 1}(3)}.
    \end{equation*}
    Define $T$ as the canonical Young tableau corresponding to the partition $(3)$ of $3$. The action of the Young symmetrizer $s_T$ on the complex tensor product space ${\1( \mathbb{C}^d \1)}^{\otimes 3}$ gives rise to a new complex tensor product space, $s_T \cdot {\1( \mathbb{C}^d \1)}^{\otimes 3}$, which is referred to as the \emph{symmetric subspace} of ${\1( \mathbb{C}^d \1)}^{\otimes 3}$, and denoted by $\emph{\vee^3 \mathbb{C}^d}$. For instance,
    \begin{align*}
        s_T \cdot (v_1 \otimes v_2 \otimes v_3) &= \sum_{\sigma \in \mathfrak{S}_3} \sigma \cdot (v_1 \otimes v_2 \otimes v_3) \\
        &= (v_1 \otimes v_2 \otimes v_3) + (v_2 \otimes v_1 \otimes v_3) + (v_3 \otimes v_2 \otimes v_1) + \\
        &\phantom{{}={}} (v_1 \otimes v_3 \otimes v_2) + (v_3 \otimes v_1 \otimes v_2) + (v_2 \otimes v_3 \otimes v_1).
    \end{align*}
    The symmetric group $\mathfrak{S}_3$ acts trivially on $\vee^3 \mathbb{C}^d$, i.e. for all $\sigma \in \mathfrak{S}_3$ and $x \in \vee^3 \mathbb{C}^d$, then $\sigma \cdot x = x$. Let $T$ be the canonical Young tableau associated with the partition $(1, 1, 1)$ of $3$. The complex tensor product space $s_T \cdot {\1( \mathbb{C}^d \1)}^{\otimes 3}$, is referred to as the \emph{antisymmetric subspace} of ${\1( \mathbb{C}^d \1)}^{\otimes 3}$, and denoted by $\emph{\wedge^3 \mathbb{C}^d}$. For instance,
    \begin{align*}
        s_T \cdot (v_1 \otimes v_2 \otimes v_3) &= \sum_{\sigma \in \mathfrak{S}_3} \sigma \cdot (v_1 \otimes v_2 \otimes v_3) \\
        &= (v_1 \otimes v_2 \otimes v_3) - (v_2 \otimes v_1 \otimes v_3) - (v_3 \otimes v_2 \otimes v_1) - \\
        &\phantom{{}={}} (v_1 \otimes v_3 \otimes v_2) + (v_3 \otimes v_1 \otimes v_2) + (v_2 \otimes v_3 \otimes v_1).
    \end{align*}
    The symmetric group $\mathfrak{S}_3$ acts on $\wedge^3 \mathbb{C}^d$ through multiplication by the signature, i.e. for all $\sigma \in \mathfrak{S}_3$ and $x \in \wedge^3 \mathbb{C}^d$, then $\sigma \cdot x = \mathrm{sign}(\sigma) \cdot x$.
\end{example*}

Let $\mathrm{U}_d$ be the unitary group of degree $d$. Given a partition $\lambda \vdash n$ consisting of $l$ parts, it is possible to associate $\lambda$ with a weight in $\mathbb{Z}^d$, provided that $l \leq d$, by appending $d-l$ zeros to the right of the partition.

\begin{theorem*}
    Consider $\lambda \vdash n$ a partition of $n$ with $l$ parts. Let $T$ a Young tableau without repetitions on $\lambda$. Then, the complex tensor product space $s_T \cdot {\1( \mathbb{C}^d \1)}^{\otimes n}$ is the zero space, if $l > d$, and an irreducible representation of $\mathrm{U}_d$, for the representation $U \mapsto U^{\otimes n}$, with highest weight $\lambda$, if $l \leq d$.
\end{theorem*}
\begin{corollary*}
    All irreducible representations of the unitary group $\mathrm{U}_n$, highest weights $\lambda \in \mathbb{N}^d$, is equivalent to $s_T \cdot {\1( \mathbb{C}^d \1)}^{\otimes n}$, with $T$ a Young tableau without repetitions on $\lambda$.
\end{corollary*}

From \hyperref[appA:thm:highestWeightIrreducibleRepresentation]{Theorem~\ref*{appA:thm:highestWeightIrreducibleRepresentation}}, for a given partition $\lambda \vdash n$ with at most $d$ part, the irreducible representations $s_T \cdot {\1( \mathbb{C}^d \1)}^{\otimes n}$ of the of the unitary group $\mathrm{U}_d$, are mutually equivalent across all Young tableaux without repetition $T$ on $\lambda$, thereby solely dependent on $\lambda$, and consequently, it is possible to designate any representation of this type as $\emph{V^d_\lambda}$.

\begin{theorem*}
    Consider a partition $\lambda \vdash n$ with at most $d$ parts. Let $V^d_\lambda$ be the corresponding irreducible representation of the unitary group $\mathrm{U}_d$. The dimension of $V^d_\lambda$ is equivalent to the cardinality of the set of semistandard Young tableaux associated with $\lambda$, with entries in $\{ 1, \ldots, d \}$
\end{theorem*}

Since the entries in each columns of a semistandard Young tableau are strictly increasing, there is no semistandard Young tableau with entries in $\{ 1, \ldots, d \}$, and with more than $d$ rows.

\begin{example*}
    Let $\mathrm{U}_2$ be the unitary group of degree $2$. Consider $\lambda$ the partition $(1, 1, 1)$ of $3$, then the antisymmetric subspace $V^3_\lambda \simeq \wedge^3 \mathbb{C}^2$ is the zero space. Let $\lambda$ be the partition $(2, 1)$ of $3$, then the irreducible representation $V^2_\lambda$ of $\mathrm{U}_2$ has dimension $2$, since the semistandard Young tableaux associated with $\lambda$, with entries in $\{ 1, 2 \}$ are
    \begin{equation*}
        \begin{ytableau}
            1 & 1 \\
            2
        \end{ytableau}
        \qquad \text{and} \qquad
        \begin{ytableau}
            1 & 2 \\
            2
        \end{ytableau}
    \end{equation*}
\end{example*}

For all $k \in \mathbb{Z}$ there is a $1$-dimensional irreducible representation of $\mathrm{U}_d$ defined by $U \mapsto (\det U)^k$, and denoted $\emph{{\det}^k}$ This irreducible representation is rational for all $k \in \mathbb{Z}$ and polynomial for $k \in \mathbb{N}$.

\begin{theorem*}
    All irreducible representations of $\mathrm{U}_d$ with highest weight $\lambda \in \mathbb{Z}^d$ are equivalent to the representation
    \begin{equation*}
        {\det}^{\shortminus k} \cdot V^d_\mu,
    \end{equation*}
    where $\mu \vdash n$ is a partition of $n$ with at most $d$ parts, and $k \in \mathbb{N}$ is an integer such that $\lambda_i = \mu_i - k$.
\end{theorem*}

The set of unitary matrices $\mathrm{U}_d$ is dense, with respect to the Zariski topology, within the set of complex invertible matrices $\mathrm{GL}_d$.

\begin{theorem*}
    The irreducible rational representations of $\mathrm{GL}_d$ are the same as the irreducible representations of $\mathrm{U}_d$.
\end{theorem*}


\newpage

\titleclass{\chapter}{top}
\titleformat{\chapter}
    [display]
    {\centering\Huge\bfseries}
    {\appendixname\ \thechapter}
    {0pt}
    {\huge}

\renewcommand\emph[1]{\oldemph{#1}}

\printbibliography

\newpage

\end{document}